 \documentclass{cambridge7A}
 \usepackage[backend=bibtex]{biblatex}
 \usepackage[figuresright]{rotating}
 \usepackage{floatpag}
 \rotfloatpagestyle{empty}
 \usepackage{mathrsfs}
 \usepackage{xspace}
 \usepackage{amsthm,amsmath}
 \usepackage{graphicx} 
 \usepackage{txfonts}
 \usepackage{multind}
 \usepackage{hyperref}
 \usepackage{color}
 \usepackage[nomessages]{fp}
 \newcommand{\avoidexercises}[1]{#1}
 \avoidexercises{\usepackage[lastexercise,answerdelayed]{exercise}}
 \usepackage[all]{xy}
 \usepackage{alltt}
 \usepackage[shortlabels]{enumitem}
 \usepackage{mik-makro}
 \usepackage{mathtools}
 \usepackage{ifthen}
 \usepackage{stackengine} 
\usepackage{wasysym}
\usepackage{wrapfig}
 \makeindex{authors}
 \makeindex{subject} 
 \makeglossary
 \theoremstyle{plain}
 \newtheorem{theorem}{Theorem}[chapter]
 \newtheorem{lemma}[theorem]{Lemma}
 \newtheorem{claim}[theorem]{Claim}

 \newtheorem{corollary}[theorem]{Corollary}
 
 \newtheorem*{theorem*}{Theorem}
 \newtheorem*{lemma*}{Lemma} 
 \newtheorem*{proposition*}{Proposition}
 \newtheorem*{corollary*}{Corollary}
 \newtheorem*{conjecture*}{Conjecture}
 \theoremstyle{definition}
 \newtheorem{definition}[theorem]{Definition}
 \newtheorem{example}[theorem]{Example}

 \newtheorem*{definition*}{Definition}
 \newtheorem*{example*}{Example}
 \newtheorem*{prob*}{Problem}
 \newtheorem*{remark*}{Remark}
 \newtheorem*{notation*}{Notation}
 \newtheorem*{exer*}{Exercise}
 \def\makeRRlabeldot#1{\hss\llap{#1}}
 \renewcommand\theenumii{{\rm (\roman{enumii})}}
 \renewcommand\theenumi{{\rm (\arabic{enumi})}}
 \renewcommand\theenumiii{{\rm (\alph{enumiii})}}
 \renewcommand\theenumiv{{\rm (\Alph{enumiv})}}

\begin{document}
\newcommand{\setexpr}[3]{\set{#1 : \text{ for }#2 \in \atoms \text{ such that }#3}}
\newcommand{\setexprtrue}[2]{\set{#1 : \text{ for }#2 \in \atoms}}

\newcommand{\aequiv}[1]{\stackrel{#1}\sim}
\newcommand{\locations}{\mathsf{Loc}}

\newcommand{\bibnotes}[1]{
\bigskip
{\footnotesize
\noindent {\bf Bibliographic notes for Section~\arabic{section}. }
#1}
}
\avoidexercises{
\renewcommand{\ExerciseHeader}{\noindent \textbf{\ExerciseName~\ExerciseHeaderNB\ExerciseHeaderTitle
\ExerciseHeaderOrigin. }}

\renewcommand{\AnswerHeader}{\noindent {\textbf{Solution to \ExerciseName\ \ExerciseHeaderNB.\\}}}
}

\newcommand{\mikexercise}[3][]{
\avoidexercises{
\begin{Exercise} 
	\ifthenelse{\equal{#1}{}}{}{(#1)}
	#2 
\end{Exercise}

\medskip
\begin{Answer}
	#3
\end{Answer}
}
}

\newcommand{\pv}[1]{{\mathtt{#1}}}
\newcommand{\mypic}[1]{
\begin{center}
	\includegraphics[page=#1,scale=0.7]{algebra-pics}
\end{center}
}

\newcommand{\sets}{\mathsf{set}}
\newcommand{\setbuil}{\mathsf{setbuil}}
\newcommand{\defset}{\mathsf{hdef}}
\newcommand{\hofset}{\mathsf{hof}}

\newcommand{\deflift}{{\mathsf{def}}}
\newcommand{\tti}{{\mathtt I}}
\newcommand{\ttj}{{\mathtt J}}
\newcommand{\qatom}{(\mathbb Q, <)}
\newcommand{\sem}[1]{[\![#1]\!]}
\newcommand{\fraisse}{Fra\"iss\'e } 

\newcommand{\saut}[1]{\mathrm{G}_{#1}}
\newcommand{\aut}{\mathrm{G}}
\newcommand{\secfm}{\bfseries{\scshape{fm\ }}}
\newcommand{\tneq}{\(\neq\)}
\newcommand{\tin}{\(\in\)}
\newcommand{\tnotin}{\(\notin\)}
\newcommand{\tcup}{\(\cup\)}
\newcommand{\temptyset}{\(\emptyset\)}
\newcommand{\tcap}{\(\cap\)}
\newcommand{\ttimes}{\(\times\)}

\newcommand{\structclass}{\mathscr A}
\newcommand{\structa}{\mathfrak A}
\newcommand{\structb}{\mathfrak B}
\newcommand{\structc}{\mathfrak C}
\newcommand{\structh}{\mathfrak H}

\newcommand{\atoms}{\mathbb A}
\newcommand{\universe}[1]{\mathrm{universe(#1)}}
\newcommand{\Fff}{\mathscr F}
\newcommand{\limmodel}{\underline \structa}
\newcommand{\arity}{\mathrm{arity}}

\newcommand{\atoma}{{a}}
\newcommand{\atomb}{{b}}
\newcommand{\atomc}{{c}}
\newcommand{\atomd}{{d}}

\newcommand{\atomq}{{q}}
\newcommand{\atomzero}{{\underline 0}}
\newcommand{\atomone}{{\underline 1}}
\newcommand{\atomtwo}{{\underline 2}}
\newcommand{\atomthree}{{\underline 3}}

\newcommand{\atomseta}{{A}}
\newcommand{\atomsetb}{{B}}
\newcommand{\atomsetc}{{C}}
\newcommand{\atomsets}{{S}}
\newcommand{\atomsett}{{T}}

\newcommand{\Qfin}{Q_{\mathrm{fin}}}
\newcommand{\Afin}{A_{\mathrm{fin}}}

\newcommand{\bind}{\nabla}
\newcommand{\powerset}{{\mathsf P}}
\newcommand{\decproblem}[3]
{\begin{center}
\fbox{
\begin{tabular}{rl}
\textsc{Problem}: & #1 \\
\textsc{Input}: & #2\\
\textsc{Output}: & #3
\end{tabular}
}
\end{center}}

\newcounter{openquestioncounter}
\newenvironment{openquestion}{
\medskip

\refstepcounter{openquestioncounter}
\smallskip\noindent{\textbf{{Open question \arabic{openquestioncounter}. }}}}{\hfill $\Box$
\medskip 
}

\newenvironment{ourexample}{
\medskip

\refstepcounter{ourexamplecounter}
\smallskip\noindent{\textbf{{Example \arabic{ourexamplecounter}. }}}}{\hfill $\Box$
\medskip
}

\newcounter{exercisecounter}

\newcommand{\eqdef}{\stackrel {\mathrm{def}} = }

\newcommand{\myunderbrace}[2]{\underbrace{#1}_{\mathclap{\txt{\scriptsize #2}}}}

\newcommand{\myoverbrace}[2]{\overbrace{#1}^{\mathclap{\txt{\scriptsize #2}}}}

\definecolor{cyan}{cmyk}{1,0,0,0}
\newcommand{\red}[1]{{\color{red}#1}}
\newcommand{\blue}[1]{{\color{cyan}#1}}

\newcommand{\infixsim}{\stackrel {\text{infix}} \sim}
\newcommand{\prefixsim}{\stackrel {\text{prefix}} \sim}
\newcommand{\suffixsim}{\stackrel {\text{suffix}} \sim}

\newcommand{\equivalentbecause}[2]{& \quad #1 \quad & \text{\scriptsize{(#2)}}}
\newcommand{\equalbecause}[1]{\equivalentbecause{=}{#1}}

\newcommand{\until}{\mathsf U}
\newcommand{\finally}{\mathsf F}
\newcommand{\globally}{\mathsf G}
\newcommand{\nextx}{\mathsf X}

\newcommand{\mso}{{\sc mso}\xspace}
\newcommand{\fo}{{\sc fo}\xspace}
\newcommand{\fotwo}{{\sc fo$^2$}\xspace}
\newcommand{\dav}{{\sc da}\xspace}

\newcommand{\ltlf}{{\sc ltl}{\small [$\finally$]}\xspace}
\newcommand{\ltlx}{{\sc ltl}{\small [$\nextx$]}\xspace}
\newcommand{\ltlunary}{{\sc ltl}{\small [$\finally,\finally^{-1}$]}\xspace}
\newcommand{\tlef}{{\sc ef}\xspace}

\newcommand{\msoequiv}[1]{\equiv_{#1}}
\newcommand{\foequiv}[1]{{  \equiv_{{#1}}}}
\newcommand{\ef}{Ehrenfeucht-Fra\"iss\'e\xspace}

\newcommand{\incite}[2][]{
\begin{tabular}{l}
	\cite{#2}
	\begin{minipage}[l]{10cm}
		\citeauthor{#2}, \citetitle{#2}, \citeyear{#2}
	\ifthenelse{\equal{#1}{}}{}{, #1}
	\end{minipage}
\end{tabular}
}

\newcommand{\footcitelong}[2][]{\footnote{ 	\\ \vspace{-0.3cm}\incite[#1]{#2}
}
}

\newcommand{\exercisehead}{
	
\vspace{1cm}
\begin{center}
  \bf Exercises
\end{center}
}

\newcommand{\momega}{!}

\newcommand{\Rat}{{\mathbb Q}}

\newcommand{\flatt}{{\mathrm{flat}}}

\newcommand{\finsort}[1]{#1_{\!+}}
\newcommand{\omegasort}[1]{#1_{\!\omega}}
 \newcommand{\finsortbis}[1]{#1[+]}
 \newcommand{\omegasortbis}[1]{#1[\omega]}

\newcommand{\cc}{\circ}

\newcommand{\lale}{L\"auchli-Leonard\xspace}
\newcommand{\tbe}{Trakhtenbrot-B\"uchi-Elgot\xspace}
\newcommand{\schutz}{Sh\"utzenberger\xspace}
\newcommand{\buchi}{B\"uchi\xspace}
\newcommand{\konig}{K\"onig\xspace}
\newcommand{\omegaop}{{\omega*}}

\newcommand{\range}[3]{#1(#2..#3 )}
\newcommand{\rightrange}[2]{#1(#2.. )}
\newcommand{\Ii}{\mathcal I}

\newcommand{\monad}{\mathsf T}
\newcommand{\smonad}{\mathsf S}
\newcommand{\fmonad}{\mathsf F}
\newcommand{\unit}[1]{\mathsf{unit}_{#1}}
\newcommand{\multnoarg}{\mathsf{mult}}
\newcommand{\mult}{\multnoarg}
\newcommand{\ident}{\mathrm{id}}
\newcommand{\alg}{\mathbf A}
\newcommand{\balg}{\mathbf B}

\newcommand{\multiset}[1]{\red{\{}#1 \red{\}} }
\newcommand{\setcat}{\mathsf{Set}}

\newcommand{\sorted}[1]{{{#1}}}
\newcommand{\semterm}[2]{#1^{#2}}
\newcommand{\ranked}[1]{#1}
\newcommand{\rSigma}{\ranked \Sigma}
\newcommand{\rGamma}{\ranked \Gamma}

\newcommand{\shelahtype}[1]{\scriptsize \red{$#1$}}

\newcommand{\langclass}{\mathscr L}
\newcommand{\algclass}{\mathscr A}
\newcommand{\eqclass}{\mathscr E}
\newcommand{\feqclass}{\mathscr F}
\newcommand{\basis}{\mathscr B}
\newcommand{\langmap}{\mathsf{L}}
\newcommand{\algmap}{\mathsf{A}}
\newcommand{\eqmap}{\mathsf{E}}

\newcommand{\msoone}{{\sc mso$_1$}\xspace}
\newcommand{\msotwo}{{\sc mso$_2$}\xspace}

\newcommand{\hmonad}{{\ranked {\mathsf H}}}
\newcommand{\fomega}{\fmonad_{\infty}}
\newcommand{\fthin}{\fmonad_{\text{thin}}}

\newcommand{\msocomplexity}[1]{|#1|_{\text{\mso}}}
\newcommand{\bueltr}{Trakhtenbrot-B\"uchi-Elgot\xspace}

\newcommand{\ltl}{\textsc{ltl}\xspace}
\newcommand{\ltlfun}{\textsc{fltl}\xspace}
\newcommand{\higman}{\hookrightarrow}
\newcommand{\facelabel}[2]{\ar@{}[#1]|{\text{\red{(#2)}}}}
\newcommand{\faceref}[1]{\red{(#1)}} 
 \mypic{107}

 \pagebreak
 \title 
{Languages recognised by finite semigroups, and their generalisations to objects such as trees and graphs, with an emphasis on  definability in monadic second-order logic}
\author{Miko{\l}aj Boja\'nczyk\\[3\baselineskip]
\today \\ The latest version can be downloaded from:
{\scriptsize https://www.mimuw.edu.pl/~bojan/2019-2020/algebraic-language-theory-2020 }
 }
\bookabstract{This is the guide for authors who are preparing written,
 rather than edited, books.}
 \bookkeywords{\LaTeX; authored books; CUP style; cambridge7A.cls.}
\frontmatter 

\pagebreak

\tableofcontents

\chapter*{Preface}

These are lecture notes on the algebraic approach to regular languages. The present version is from \today, while the most recent version can be found here:
\begin{center}
  {\tt {\scriptsize https://www.mimuw.edu.pl/$~$bojan/2019-2020/algebraic-language-theory-2020 }
  }
\end{center}
The classical algebraic approach is for finite words; it  uses semigroups instead of automata. However, the algebraic approach can be extended to structures beyond words,  e.g.~infinite words, or trees or graphs.  The purpose of this book is to describe the algebraic approach in a way that covers these extensions.

\mainmatter

\part{Words}
 \chapter{Semigroups, monoids and their structure}
\label{chapter:semigroups}
In this chapter, we define semigroups and monoids, and use them to recognise languages of finite words.

\begin{definition}[Semigroup]\label{def:monoid}
A \emph{semigroup} consists of an underlying set $S$ together  with a binary multiplication operation
\begin{align*}
   (a,b) \in S^2 \quad \mapsto \quad ab \in S,
\end{align*}
which is \emph{associative} in the sense that 
\begin{align*}
a (bc) = (ab)c \qquad \text{for all }a,b,c \in S.
\end{align*}
\end{definition}
 
The definition says that the order of evaluation in a semigroup is not important, i.e. that different ways of bracketing a sequence of elements in the semigroup will yield the same result as far as  semigroup multiplication is concerned. For example,
\begin{align*}
((ab)c)(d(ef)) = ((((ab)c)d)e)f.
\end{align*}
Therefore, it makes sense to omit the brackets and write simply
\begin{align*}
abcdef.
\end{align*}
This means that  semigroup multiplication  can be seen as an operation of type $S^+ \to S$, i.e.~it is defined not just on pairs of semigroups elements, but also  on finite nonempty words consisting of semigroup elements. 

A \emph{semigroup homomorphism} is a function between  (underlying sets of) semigroups that preserves the structure of semigroups, i.e. a  function 
\begin{align*}
h: \myunderbrace{S}{semigroup} \to \myunderbrace{T}{semigroup}
\end{align*}  which  is consistent with the multiplication operation in the sense that 
\begin{align*}
    h(a \cdot b) = h(a)\cdot h(b),
\end{align*}
where the semigroup multiplication on the left is in $S$, and the semigroup multiplication on the right is in $T$.  An equivalent definition of a  semigroup homomorphism, which views semigroup multiplication as defined on entire words and not just pairs of letters, says that the following diagram must commute:
\begin{align*}
   \xymatrix{
      S^+ \ar[r]^{h^+} \ar[d]_{\text{multiplication in $S$}} &
      T^+ \ar[d]^{\text{multiplication in $T$}}\\
      S \ar[r]_{h} & T
   }
   \end{align*}
In the above, $h^+$ is the natural lifting of $h$ to words, which applies $h$ to every letter.

A \emph{monoid} is the special case of a  semigroup where there  is an identity element, denoted by $1 \in S$, which satisfies
\begin{align*}
1a = a= a1 \qquad \text{for all }a \in S.
\end{align*}
The identity element, if it exists, must be unique. This is because if there are two candidates for the identity, then multiplying them reveals the true identity. The multiplication operation in a monoid can be thought of as having type $S^* \to S$, with  the empty word $\varepsilon$ being mapped to $1$.  A \emph{monoid homomorphism} is a semigroup homomorphism that preserves the identity element. In terms of commuting diagrams, a monoid homomorphism is a function which makes the following diagram commute:
\begin{align*}
   \xymatrix{
      S^* \ar[r]^{h^*} \ar[d]_{\text{multiplication in $S$}} &
      T^* \ar[d]^{\text{multiplication in $T$}}\\
      S \ar[r]_{h} & T
   }
   \end{align*}
Clearly there is a pattern behind the diagrams. This pattern will be explored in the second part of this book, when talking about monads.

\begin{example}\label{example:semigroups}
   Here are some examples of monoids and semigroups.
   \begin{enumerate} 
      \item If $\Sigma$ is a set, then the set  $\Sigma^+$ of nonempty words over $\Sigma$, equipped with concatenation, is a semigroup, called the \emph{free\footnote{The reason for this name is the following universality property. The free semigroup  is generated by $\Sigma$, and it is the biggest  semigroup generated by $\Sigma$  in the following sense. For every semigroup $S$ that is generated by $\Sigma$, there exists a (unique) surjective semigroup homomorphism $h : \Sigma^+ \to S$ which is the identity on the $\Sigma$ generators.    } semigroup over generators $\Sigma$}.  The \emph{free monoid} is the set $\Sigma^*$ of possibly empty words.  
      \item Every group is a monoid.
      \item For every set $Q$, the set of all functions  $Q \to Q$, equipped with function composition, is a monoid. The monoid identity is the identity function.
      \item For every set $Q$, the set of all binary relations on $Q$ is a monoid, when equipped with relational composition
      \begin{align*}
      a \circ b =  \set{(p,q) : \text{there is some $r \in Q$ such that $(p,r) \in a$ and $(r,q) \in b$}}.
      \end{align*}
      The monoid identity is the identity function. The monoid from the previous item is a sub-monoid of this one, i.e.~the inclusion map is a monoid homomorphism.
      \item Here are all semigroups of size two, up to semigroup isomorphism:
      \begin{align*}
      \myunderbrace{(\set{0,1}, +)}{addition mod 2} \quad (\set{0,1}, \min) \quad \myunderbrace{(\set{0,1}, \pi_1)}{$(a,b) \mapsto a$} 
      \quad \myunderbrace{(\set{0,1}, \pi_2)}{$(a,b) \mapsto b$}
      \quad \myunderbrace{(\set{0,1}, (a,b) \mapsto 1)}{all multiplications are $1$}
      \end{align*}
      The first two are monoids. 
   \end{enumerate}
\end{example}

\paragraph*{Compositional functions.}
Semigroup homomorphisms are closely related with functions that are compositional in the sense defined below. Let $S$ be a semigroup, and let $X$ be a set (without a semigroup structure). A function 
\begin{align*}
    h : S \to X
\end{align*}
is called \emph{compositional} if for every $a, b \in S$, the value $h(a \cdot b)$ is uniquely determined
by the values $h(a)$ and $h(b)$.  If $X$ has a semigroup structure, then every semigroup homomorphism $S \to X$ is a compositional function. The following lemma shows that the converse is also true for surjective functions. 
\begin{lemma}\label{lem:compositional-monoid}
Let $S$ be a semigroup, let $X$ be a set, and let $h : S \to X$ be a surjective compositional function. Then there exists (a unique) semigroup structure on $X$ which makes $h$ into a semigroup homomorphism.	
\end{lemma}
\begin{proof}
   Saying that $h(a \cdot b)$ is uniquely determined by $h(a)$ and $h(b)$, as in the definition of compositionality, means that there is a binary operation $\circ$ on $X$, which is not yet known to be associative, that satisfies
   \begin{align}\label{eq:compositionality-assumption}
      h(a \cdot b) = h(a) \circ h(b) \qquad \text{for all }a,b \in S.
   \end{align}
   The semigroup structure on $X$ uses $\circ$ as the semigroup operation. It remains to prove associativity of $\circ$.  Consider three elements of $X$, which can be written as $h(a), h(b), h(c)$ thanks to the assumption on surjectivity of $h$. We have
   \begin{align*}
   (h(a) \circ h(b)) \circ h(c) 
   \stackrel{\text{\eqref{eq:compositionality-assumption}}} =
   (h(ab)) \circ h(c) 
   \stackrel{\text{\eqref{eq:compositionality-assumption}}} =
   h(abc).
   \end{align*}
   The same reasoning shows that $h(a) \circ (h(b) \circ h(c))$ is equal to $h(abc)$, thus establishing associativity. 
\end{proof}

\paragraph*{Commuting diagrams.} We finish this section with  an alternative description of semigroups which uses commuting diagrams. Similar descriptions will be  frequently used in this book, e.g.~for generalisations of semigroups for infinite words, so we want to start using them as early as possible.  

As mentioned before,  the binary multiplication operation in a semigroup $S$ can be extended to  an  operation of type  $S^+ \to S$. The following lemma explains, using commuting diagrams, which operations of  type  $S^+ \to S$ arise this way. 

\begin{lemma}\label{lem:commuting-diagram-semigroup}
   An operation  $\mu : S^+ \to S$ arises from some semigroup operation on $S$ if and only if the following two diagrams commute: 
   \begin{align*}
      \xymatrix{
         S \ar[dr]^{\text{identity}} \ar[d]_{\txt{\scriptsize view a letter as \\ \scriptsize  a one-letter word}}\\
         S^+ \ar[r]_\mu & S
      }
      \qquad \qquad
         \xymatrix@C=4cm
         {(S^+)^+ \ar[r]^{\text{multiplication in free semigroup $S^+$}} \ar[d]_{\mu^+} &
          S^+ \ar[d]^{\mu} \\ 
          S^+ \ar[r]_{\mu}
          &  S}
      \end{align*}
   In the above, $\mu^+$ stands for the coordinate-wise lifting of $\mu$ to words of words. 
\end{lemma}

For monoids, the same lemma holds, with $+$ replaced by $*$. There is no need to add an extra diagram for the monoid identity, since the monoid identity can be defined as the image under $\mu$ of the empty word $\varepsilon$. The axioms 
\begin{align*}
1 \cdot a = a = a \cdot 1 
\end{align*}
can then be derived as
\begin{align*}
1 \cdot a = \mu(\varepsilon) \cdot \mu(a) = \mu(\varepsilon a) = \mu(a) = a,
\end{align*}
with a symmetric reasoning used for $a \cdot 1$. 


\exercisehead

\mikexercise{Show a function between two monoids that is a semigroup homomoprhism, but not a monoid homomorphism.}
{
Let $M$ be any monoid, e.g.~the trivial monoid of size one, and let $M'$ be its extension obtained by adding an extra identity element $1$. The inclusion embedding 
\begin{align*}
M \hookrightarrow M'
\end{align*}
is a semigroup homomorphism, but not a monoid homomorphism.
}

\mikexercise{Show that there are exponentially many semigroups of size $n$. }  {
Take any finite set $X$ and any  function
\begin{align*}
f : X^2 \to \set{0, \text{true}}.
\end{align*} 
The number of choices for $f$ is exponential in $X^2$. 
We can extend $f$ to an associative multiplication operation on the set 
\begin{align*}
S = X + \set{0, \text{true}}
\end{align*}
by defining all multiplications outside $X^2$ to  give value $0$.  Semigroups obtained this way correspond to languages where all words have length at most 2. 
}
\mikexercise{\label{ex:has-some-idempotent-ramsey}
	Show that for every  semigroup homomorphism $h : \Sigma^+ \to S$, with $S$ finite, there exists some $N \in \set{1,2,\ldots}$ such that every word of length at least $N$ can be factorised as $w = w_1 w_2 w_3$ where $h(w_2)$ is  an idempotent\footnote{This exercise can be seen as the semigroup version of the pumping lemma.}.
}
{}
\mikexercise{\label{ex:powerset-semigroup}Show that if $S$ is a semigroup, then the same is true for the \emph{powerset semigroup}, whose elements are possibly empty subsets of $S$, and where  multiplication is defined coordinate-wise:
\begin{align*}
   A \cdot B = \set{a \cdot b : a \in A, b \in B} \qquad \text{for $A,B \subseteq S$.}
  \end{align*}  
}{}

\mikexercise{\label{ex:category-of-semigroups} Let us view semigroups as a category, where the objects are semigroups and the morphisms are semigroup homomorphisms. What are the product and co-products of this category?}{For the case of products, we take $S$ to be the Cartesian product of $S_1 \times S_2$. This is the semigroup where underlying set is the product of the underlying sets, and the semigroup operation is defined coordinate-wise. It is not hard to check that this is a semigroup, and the two projection functions 
\begin{align*}
   \xymatrix{
      & S 
      \ar[dl]_{\pi_1}
      \ar[dr]^{\pi_2}
       & \\
      S_1
      &  &  S_2
   }
 \end{align*}
 are semigroup homomorphisms, which satisfy the defining property of products in a category, namely:
\begin{align*}
   \red{\forall} \blue{\exists!} \qquad
   \xymatrix{
      & S 
      \ar[dl]_{\pi_1}
      \ar[dr]^{\pi_2}
      \ar@[blue][dd]_{f}
       & \\
      S_1
      \ar@[red][dr]_{\red f_1}
      &  &  S_2
      \ar@[red][dl]^{\red f_2}\\
      & \red T & 
   } \qquad \qquad
 \end{align*}

 The co-products are more tricky.
 Here we use $S$ to be the following semigroup. Its elements are words over the disjoint union of alphabets $S_1 + S_2$, equipped with concatenation,  modulo the least  equivalence relation that contains all pairs of the form
\begin{align*}
w ab v  \sim wcv   \text{ where $c \in S_i$ is the multiplication $ab$ in $S_i$},
\end{align*}
and $w,v$ are possibly empty words. This equivalence relation is semigroup congruence, and therefore it makes sense to consider the quotient semigroup. One can check that the semigroup $S$ defined this way, along with the two natural inclusions $\iota_1$ and $\iota_2$, satisfies the 
 defining diagram of co-products, namely
\begin{align*}
  \red{\forall} \blue{\exists!} \qquad
  \xymatrix{
     & S 
      & \\
     S_1
     \ar[ur]^{\iota_1}
     &  &  S_2
     \ar[ul]_s{\iota_2}\\
     & \red T 
     \ar@[blue][uu]_{f}
     \ar@[red][ul]^{\red f_1}
     \ar@[red][ur]_{\red f_2}
     & 
  } \qquad \qquad
\end{align*}}

\mikexercise{Let $\Sigma$ be an alphabet, and let 
\begin{align*}
X \subseteq \Sigma^+ \times \Sigma^+
\end{align*}
be a  set of words pairs. Define $\sim_X$ to be least congruence on $\Sigma^+$ which contains all pairs from $X$. This is the same as the symmetric transitive  closure of 
\begin{align*}
\set{(wxv,wyv) : w,v \in \Sigma^*, (x,y)  \in X }.
\end{align*}
Show that the following problem -- which is called the \emph{word problem for semigroups} -- is undecidable: given finite $\Sigma, X$ and $w,v \in \Sigma^+$, decide if $w \sim_X v$.
}{
A Turing machine is called reversible if every configuration has at most one successor and at most one predecessor configuration. The halting problem is undecidable for reversible Turing machines; this can be shown by taking any Turing machine and extending it so that it stores its computation history. For a reversible Turing machine, one can come up with a set $X$ such that  for every words $w,v$ that represent configurations of the machine, we have 
\begin{align*}
w \sim_X v \quad \text{iff} \quad \text{there is a computation from $w$ to $v$ or from $v$ to $w$}
\end{align*}}

\mikexercise{Define the \emph{theory of semigroups} to be the set of first-order sentences, which use one ternary relation $x=y\cdot z$, that are true in every semigroup. Show that the theory of semigroups is undecidable, i.e.~it is undecidable if a first-order sentence is true in all semigroups.}{}

\mikexercise{Show that the theory of finite semigroups is different from the theory of (all) semigroups, but still undecidable. }{}
\section{Recognising languages}

In this book, we are  interested in monoids and semigroups as an alternative to finite automata for the purpose of recognising  languages\footnote{
   The semigroup approach to languages can be credited to 
   \incite{SD_1955-1956__9__A10_0}
   On page 10 of this paper, which is primarily devoted to codes, \schutz remarks that semigroups can be used to recognise languages and defines the syntactic congruence (in fact, the syntactic pre-order). Apparently, the syntactic congruence dates back to 
   \incite{dubreil1941}
   but I have been unable to obtain a copy of this paper. 
   These are the  early days of automata theory, and \schutz's paper is  contemporary to 
   \incite[Theorem 4,]{moore1956gedanken}
   which is the first place that I know where  minimisation of automata appears.
}. Since  languages are usually defined for possibly empty words, we use monoids and not semigroups when recognising languages.

\begin{definition}
   Let $\Sigma$ be a finite alphabet. 
 A language $L \subseteq \Sigma^*$ is \emph{recognised} by a monoid homomorphism 
 \begin{align*}
   h : \Sigma^* \to M
 \end{align*}
 if the membership relation  $w \in L$ is determined uniquely by $h(w)$. In
  other words, there is an \emph{accepting subset}
   $F \subseteq M$ 
   such that 
   \begin{align*}
      w \in L \quad\text{iff}
      \quad h(w) \in F 
      \qquad \text{for every }w \in \Sigma^*.
   \end{align*}
\end{definition}

We say that a language is recognised by a monoid if it is recognised by some monoid homomorphism into that monoid. 
The following theorem shows that, for the purpose of  recognising languages, finite monoids and finite automata are equivalent.

\begin{theorem}
   \label{thm:regular-languages-monoids}
   The following  conditions are equivalent for every $L \subseteq \Sigma^*$:
   \begin{enumerate}
      \item $L$ is recognised by a finite nondeterministic automaton;
      \item  $L$ is recognised by a finite monoid.
   \end{enumerate}
\end{theorem}
\begin{proof}\ 

   \begin{description}
      \item[2 $\Rightarrow$ 1] From a monoid homomorphism  one creates a deterministic automaton, whose states are elements of the monoid, the initial state is the identity, and the transition function is 
      \begin{align*}
      (m, a) \mapsto m \cdot \text{(homomorphic image of $a$)}.
      \end{align*}
       After reading an input word $w$, the state of the automaton is equal to the  homomorphic image of $w$ under the recognising homomorphism, and therefore the accepting subset for the monoid homomorphisms can be used. This automaton computes the monoid multiplication according to the choice of parentheses illustrated  in this example:
       \begin{align*}
       (((((ab)c)d)e)f)g.  
       \end{align*}
      \item[1 $\Rightarrow$ 2] Let $Q$ be the states of a nondeterministic automaton recognising $L$. Define a function\footnote{This  transformation from a nondeterministic (or deterministic) finite automaton to a monoid incurs an exponential blow-up, which is unavoidable in the worst case. }  
      \begin{align*}
         \delta : \Sigma^* \to \text{monoid of binary relations on $Q$}
      \end{align*}
      which sends a word $w$ to the binary relation 
      \begin{align*}
      \set{(p,q) \in Q^2 : \text{some run over $w$ goes from $p$ to $q$}}.
      \end{align*}
      This is a monoid homomorphism. It recognises the language:  a word is in the language if and only if its image under the homomorphism  contains at least one pair of the form (initial state, accepting state).
   \end{description}
\end{proof}

\paragraph*{The syntactic monoid of a language.}
Deterministic finite automata have minimisation, i.e.~for every language there is a minimal deterministic automaton, which can be found inside every  other deterministic  automaton that recognises the language. The same is true for monoids, as proved in the following theorem.

\begin{theorem}\label{thm:syntactic-monoid}
 For every language\footnote{The language need not be regular, and the alphabet need not be finite.} $L \subseteq \Sigma^*$ there is a surjective monoid homomorphism 
 \begin{align*}
   h : \Sigma^* \to  M,
 \end{align*}
 called the syntactic homomorphism of $L$, which recognises it and  is minimal in the sense explained in the following quantified diagram\footnote{Here is how to read the diagram. \red{For every red extension} of the black diagram \blue{there exists a unique blue extension} which makes the diagram commute. Double headed arrows denote surjective homomorphisms, which means that $\red \forall$ quantifies over surjective homomorphisms, and the same is true for $\blue \exists!$.}
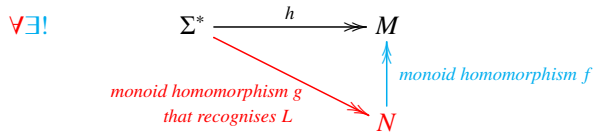
 \begin{align*}
   \red{\forall} \blue{\exists!} \qquad
   \xymatrix@C=2cm{ 
      \Sigma^* 
      \ar@{->>}[r]^h 
      \ar@{->>}@[red][dr]_-{ \txt{\scriptsize \red{monoid homomorphism $g$ \qquad}\\ \scriptsize{\red{that recognises $L$}}}} 
   & M \\
    & \red N
    \ar@[cyan]@{->>}[u]_{\blue {\text{ monoid homomorphism $f$}}} }
 \end{align*}

\end{theorem}

\begin{proof} The proof is the same as for the Myhill-Nerode theorem about minimal automata, except that the corresponding congruence is two-sided. 
   Define the \emph{syntactic congruence} of $L$ to be the equivalence relation $\sim$ on $\Sigma^*$ which identifies two words $w,w' \in \Sigma^*$ if 
   \begin{align*}
      uwv \in L \quad\text{iff} \quad uw'v \in L \qquad \text{for all }u,v \in \Sigma^*.
   \end{align*}
   Define $h$ to be the function that maps a word to its equivalence class under syntactic congruence. It is not hard to see that $h$ is compositional, and therefore by (the monoid version of) Lemma~\ref{lem:compositional-monoid}, one can equip the set of equivalence classes of syntactic congruences  with a monoid structure -- call $M$ the resulting monoid -- which turns $h$ into a monoid homomorphism. 
   
   It remains to show minimality of $h$, as expressed by the diagram in the lemma. Let then $\red g$ be as in the diagram. Because $\red g$ recognises the language $L$, we have 
   \begin{align*}
   \red g (w) = \red g (w') \quad \text{implies} \quad w \sim w',
   \end{align*}
   which, thanks to surjectivity of $\red g$, yields some function $\blue f$ from $\red  N$ to $\blue M$, which makes the diagram commute, i.e.~$h = \blue f \circ \red g$. Furthermore, $\blue f$ must be a monoid homomorphism, because 

\begin{eqnarray*}
\blue f ( a_1 \cdot a_2)  
\equalbecause{by surjectivity of $\red g$, each $a_i$ can be presented as $\red g(w_i)$ for some  $w_i$}
\\ 
\blue f (\red g(w_1) \cdot \red g(w_2)) 
\equalbecause{$\red g$ is a monoid homomorphism}\\
  \blue f (\red g(w_1 w_2))  
  \equalbecause{the diagram commutes} \\
  h(w_1 w_2) 
  \equalbecause{$h$ is a monoid homomorphism} \\
  h(w_1) \cdot h(w_2)
  \equalbecause{the diagram commutes}\\
  \blue f(\red g(w_1)) \cdot \blue f(\red g(w_2)) & = &\\
  \blue f (a_1) \cdot \blue f (a_2).
\end{eqnarray*}

\end{proof}

\mikexercise{\label{ex:deterministic-to-monoids}
	Show that the translation from deterministic finite automata to monoids is exponential in the worst case. }
{For $n \in \set{1,2,\ldots}$, consider the language of words in $\set{a,b}^*$ where  the $n$-th letter is $a$. This language is recognised by a deterministic finite automaton with $\Oo(n)$ states. To recognise the language with a monoid homomorphism, we need to remember the first $n$ letters of a word. 
}

\mikexercise{
   Show that the translation from (left-to-right) deterministic finite automata to monoids is exponential in the worst case, even if there is a  right-to-left deterministic automaton of same  size.
}{
An example is  words of length $2n+1$, over alphabet $\set{a,b}$, where the middle letter is $a$. For the same reasons as in Exercise~\ref{ex:deterministic-to-monoids}, this language needs an exponential size monoid to be recognised. On the other hand, the language is clearly recognised by a deterministic automaton (running in either direction), with $\Oo(n)$ states.}

\mikexercise{\label{ex:commutative-regular-languages}
Show that a language $L \subseteq \Sigma^*$ is recognised by a finite commutative monoid if and only if it can be defined by a finite Boolean combination of conditions of the form ``letter $a$ appears exactly $n$ times'' or ``the number of appearances of letter $a$ is congruent to $\ell$ modulo $n$''.
}
{No surprises here -- these are the commutative languages, i.e.~languages $L \subseteq \Sigma^*$ which are commutative in the sense that 
\begin{align*}
w ab v \in L \quad \text{iff} \quad wbav \in L \qquad  \text{for all $w,a,b,v \in \Sigma^*$}.
\end{align*}
There is, however, an extended description of commutative languages, which is given below. 

Consider a language  $L \subseteq \Sigma^*$ recognised by a monoid homomorphism
\begin{align*}
 h : \Sigma^* \to M,
\end{align*}
where $M$ is a commutative monoid. 
By commutativity of $M$, we have
\begin{align*}
h(w) = \prod_{a \in \Sigma} h(a)^{\#_a(w)}  \qquad \text{for every $w \in \Sigma^*$,}
\end{align*} 
where $\#_a(w) \in \set{0,1,\ldots}$ is the number of appearances of letter $a \in \Sigma$ in $w$. For every $a \in M$ the sequence
\begin{align*}
a^0, a^1, a^2, \ldots  \in M
\end{align*}
is easily seen to be ultimately periodic, which means that after cutting of a finite prefix of the sequence we get a sequence that is periodic. This in turn implies that for every $a,b \in M$, the set
\begin{align*}
\set{ i \in \set{0,1,\ldots} : a^i = b }
\end{align*}
is defined by a formula $\varphi(i)$ which is a finite  Boolean combination of formulas which have one of the following forms:
\begin{enumerate}
   \item $i = k$ for some $k \in \set{0,1,\ldots}$; or
   \item $i \equiv k \mod m$ for some $k \in \set{0,1,\ldots}$.
\end{enumerate}
Putting these observations together, we see that for every $b \in M$, the set
\begin{align*}
\set{w \in \Sigma^* : b = \prod_{a \in \Sigma} h(a)^{\#_a(w)}  },
\end{align*}
which is equal to the inverse image $h^{-1}(b)$, 
is defined by a finite Boolean combination of formulas of one of the following forms:
\begin{enumerate}
   \item[(i)] $\#_a(w) = k$ for some $a \in \Sigma$ and $k \in \set{0,1,\ldots}$; or
   \item[(ii)] $\#_a(w) \equiv k \mod m$ for some $a \in \Sigma$ and  $k \in \set{0,1,\ldots}$.
\end{enumerate}
It follows that the following conditions are equivalent for every language:
\begin{itemize}
   \item recognised by a finite commutative monoid;
   \item regular and commutative as a language;
   \item defined by a finite Boolean combination of conditions as in (i) and (ii).
\end{itemize}

}

\mikexercise{
Prove that surjectivity of $\red g$ is important in Theorem~\ref{thm:syntactic-monoid}. 	
}{
Consider the language $(aa)^*$ of even length words, which is recognised by the homomorphism 
\begin{align*}
h : a^* \to \Int_2  = (\set{0,1}, +_{\mathrm{mod} 2}).
\end{align*}
Define $\Int_2 + \bot$ to be the extension of $\Int_2$ with an absorbing element $\bot$. Define 
\begin{align*}
g : a^* \to \Int_2 + \bot
\end{align*}
to be the same function as $h$, except that the co-domain is bigger. In particular, $g$ is not surjective.  We claim that there is no monoid homomorphism $f$ which makes the following diagram commute
\begin{align*}
\xymatrix{
   a^* \ar[r]^h
   \ar[dr]_g & \Int_2 \\
    & \Int_2 + \bot 
    \ar[u]_f
}
\end{align*}
Since $\bot$ is absorbing in $\Int_2 + \bot$, then the image of $f(\bot)$ must be an absorbing in $\Int_2$, and there are no absorbing elements in $\Int_2$. It follows that there is no $f$ which makes the diagram commute. }

\mikexercise{
   Show that for every language, not necessarily regular, its syntactic homomorphism is the function
   \begin{align*}
   w \in \Sigma^* \qquad \mapsto \qquad \myunderbrace{(q \mapsto qw)}{state transformation \\ \scriptsize in the sytactic automaton},
   \end{align*}
   where the syntactic automaton is the deterministic finite automaton from the Myhill-Nerode theorem.
}{
   The syntactic congruence identifies two words $w_1$ and $w_2$ if 
\begin{align*}
 uw_1v \in L \iff uw_2v \in L \qquad \text{for all $u,v \in \Sigma^*$.}
\end{align*}
To prove the exercise, we will show that two words are equivalent in the above sense if and only if they induce the same state transformations in the syntactic automaton.  Clearly if the words have the same state transformations, then they are equivalent. We are left with proving the opposite, i.e.~different state transformations imply non-equivalence under the syntactic congruence.

   Suppose that  $w_1$ and $w_2$ have different state transformations. This  means that there is some state $q$ of the syntactic automaton  such that  
   \begin{align*}
   q w_1  \neq   q w_2.
   \end{align*}
   Like any state of the syntactic automaton, $q$ is reached from the initial state by reading some word $u$. Since the states $q w_1$ and $q w_2$ are different,  there must be some word $u$ such that exactly one of the states 
\begin{align*}
   (q w_1)u \qquad (q w_2)u
\end{align*}
is accepting.  Summing up, we have found two words $u,v$ such that exactly one of the words
\begin{align*}
u w_1 v \qquad u w_2 v
\end{align*}
   is in the language, thus proving that $w_1$ and $w_2$ are not equivalent under the syntactic congruence.
}

\mikexercise{
\label{ex:half-eilenberg-star}    
Let $\Ll$ be a class of regular languages with the following closure properties:
\begin{itemize}
    \item $\Ll$ is  closed under Boolean combinations;
    \item $\Ll$ is closed under inverse images of homomorphisms $h : \Sigma^* \to \Gamma^*$;
    \item  Let $L \subseteq \Sigma^*$ be a language in  $\Ll$. For every  $u,w \in \Sigma^*$, $\Ll$ contains the inverse image of $L$ under the following operation:
    \begin{align*}
    v \mapsto uvw. 
    \end{align*}
\end{itemize}
Show that if $L$ belongs to $\Ll$, then the same is true for every language recognised by its syntactic monoid.
}{}
\section{Green's relations and the structure of finite semigroups}
\label{sec:greens-relations}
In this section, we describe some of the structural theory of finite semigroups. 
This theory is based on Green's relations\footfullcite{green51}, which are pre-orders in a semigroup that correspond to prefixes, suffixes and infixes.

\newcommand{\prefclass}[1]{[#1)}
\newcommand{\sufclass}[1]{(#1]}
\newcommand{\jclass}[1]{(#1)}
\newcommand{\hclass}[1]{[#1]}

We begin with idempotents, which are  ubiquitous in the analysis of finite semigroups.  A semigroup element $e$ is called  \emph{idempotent} if it satisfies 
\begin{align*}
ee = e.
\end{align*}
\begin{example}
	In a group, there is a unique idempotent, namely the group identity. There can be several idempotents, for example  all elements are idempotent in  the semigroup 
	\begin{align*}
	(\set{1,\ldots,n}, \max).
	\end{align*}
\end{example}
One can think of idempotents as being a relaxed version of identity elements. 

\begin{lemma}[Idempotent Power Lemma]\label{lem:idempotent-lemma}
	Let $S$ be a finite semigroup. For every $a \in S$, there is exactly one idempotent in the set
	\begin{align*}
	\set{a^1,a^2,a^3, \ldots} \subseteq S.
	\end{align*}
\end{lemma}
\begin{proof}
	Because the semigroup is finite, the sequence $a^1,a^2,\ldots$ must contain a repetition, i.e.~there must exist $n,k \in \set{1,2,\ldots}$ such that 
	\begin{align*}
	a^{n} = a^{n+k} = a^{n + 2k} = \cdots. 
	\end{align*}
	After multiplying both sides of the above equation by $a^{nk-n}$ we get
	\begin{align*}
	a^{nk} = a^{nk+k} = a^{nk + 2k} = \cdots,
	\end{align*}
	and therefore $a^{nk} = a^{nk+nk}$ is an idempotent. To prove uniqueness of the idempotent, suppose $n_1,n_2 \in \set{1,2,\ldots}$ are powers such that  that $a^{n_1}$ and $a^{n_2}$ are idempotent. The we have 
	\begin{align*}
		\myunderbrace{ a^{n_1} = (a^{n_1})^{n_2}
		}{because $a^{n_1}$\\
	\scriptsize is idempotent}  
	= a^{n_1 n_2} =
	\myunderbrace{  (a^{n_2})^{n_1} = a^{n_2}
		}{because $a^{n_2}$\\
	\scriptsize is idempotent}
	\end{align*}
\end{proof}
Finiteness is important in the above lemma. For example the infinite semigroup
\begin{align*}
(\set{1,2,\ldots},+)
\end{align*}
contains no idempotents.
For $a \in S$, we use the name \emph{idempotent power} for the element $a^n$, and we use the name \emph{idempotent exponent} for the number $n$. The idempotent power is unique, but the idempotent exponent is not. It is easy to see that there is always an idempotent exponent which is at most the size of the semigroup, and idempotent exponents are closed under multiplication. Therefore, if a semigroup has $n$ elements, then the factorial $n!$ is an idempotent exponent for every element of  the semigroup. This motivates the following notation:  we write $a^{\momega
}$ for the idempotent power of  $a$. The  notation usually used in the semigroup literature is $a^\omega$, but we will use $\omega$ for infinite words.

The analysis presented in the rest of this chapter will hold in any semigroup which satisfies the  conclusion of the Idempotent Power Lemma.

\section*{Green's relations}
We now give the main definition of this chapter.
\begin{definition}[Green's relations] Let $a,b$ be elements of a semigroup $S$. We say that $a$ is a \emph{prefix} of $b$  if there exists a  solution $x$ of
	\begin{align*}
	ax =b.
	\end{align*}
	The solution $x$ can be an element of the semigroup, or  empty (i.e.~$a=b$). Likewise we define the suffix and infix relations, but with the equations
\begin{align*}
\myunderbrace{xa = b}{suffix}  \qquad
\myunderbrace{xay = b}{infix}.
\end{align*}
In the case of the infix relation,  one or both of $x$ and $y$ can be empty.
\end{definition}

Figure~\ref{fig:j-classes-three} illustrates  Green's relations on the example of the monoid of partial functions on a three element set.
The prefix, suffix and infix relations are pre-orders, i.e.~they are transitive and reflexive\footnote{Another description of the prefix pre-order is  that $a$ is a prefix of $b$ if
\begin{align}\label{eq:ideal-inclusion}
a S^{1} \supseteq b S^1.
\end{align}
In the above,  $S^1$ is the monoid which is obtained from $S$ by adding an identity element, unless it was already there. The sets $a S^{1}, bS^{1}$ are called \emph{right ideals}. Because of the description in terms of inclusion of right ideals, the semigroup literature  uses the notation
\begin{align*}
a \ge_{\Rr} b   \eqdef aS^{1}  \supseteq bS^1
\end{align*}
for the prefix relation. Likewise, $a \ge_{\Ll} b$ is used for the suffix relation, which is defined in terms of left ideals. Also, for some mysterious reason, $a \ge_{\Jj} b$ is used for the infix relation.   We avoid this notation, because it makes longer words smaller. 
}.  They need not be anti-symmetric, for example in a group every element is a prefix (and also a suffix and infix) of every other element. We say that two elements of a semigroup are in the same \emph{prefix class} if they are prefixes of each other. Likewise we define \emph{suffix classes} and \emph{infix classes}. 

Clearly every prefix  class is contained in some infix class, because prefixes are special cases of infixes. Therefore, every infix class is partitioned into prefix classes. For the same reasons, every infix class is partitioned into suffix classes. The following lemma describes the structure of these partitions.

\begin{figure}[!p]
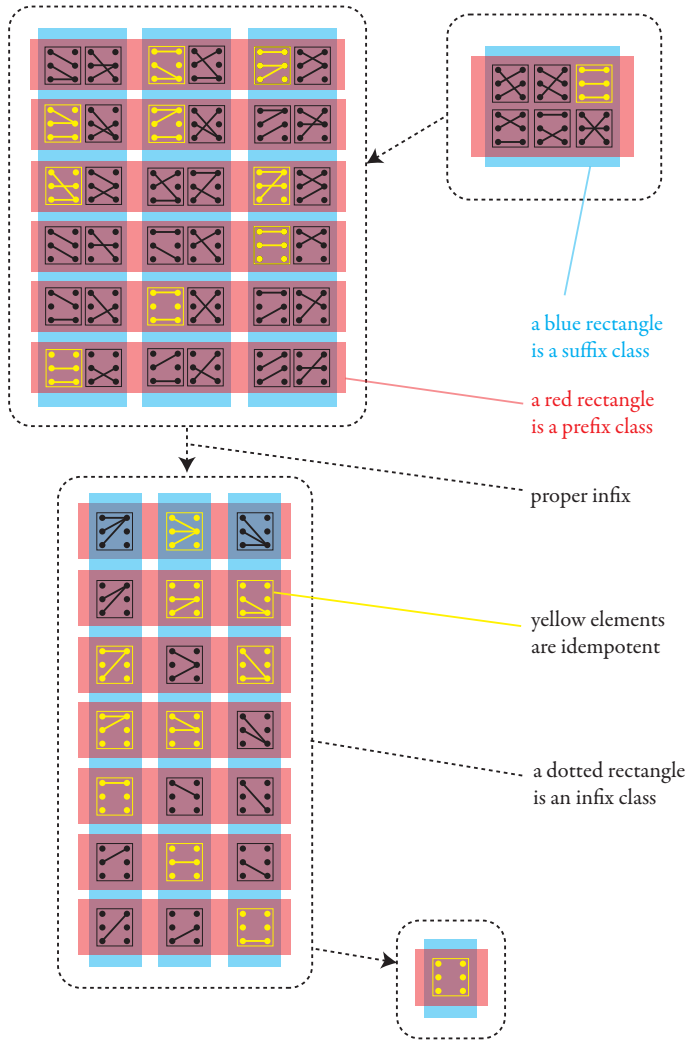

\mypic{1}
	\caption{The monoid of partial functions from  a three element set to itself, partitioned into prefix, suffix and infix classes. In this particular example, the infix classes are totally ordered, which need not be the case in general. }
	\label{fig:j-classes-three}   
\end{figure}

\clearpage

\begin{lemma}[Egg-box lemma] The following hold in every finite semigroup.
	 \begin{enumerate}
		\item \label{Egg-box:prefix-incomparable} all distinct prefix classes in a given infix class are incomparable: 
		\begin{align*}
		\text{$a,b$ are infix equivalent, and  $a$ is a prefix of $b$} \  \Rightarrow \ \text{$a,b$ are prefix equivalent}
		\end{align*}
		\item \label{Egg-box:prefix-suffix-must-intersect} if a prefix class and a suffix class are contained in the same infix class, then they have nonempty intersection;
		\item \label{Egg-box:same-size-prefix} all prefix classes in the same infix class have the same size.
	\end{enumerate}
\end{lemma}
Of course, by symmetry, the lemma remains true after swapping prefixes with suffixes.
\begin{proof}
	\ 
	\begin{enumerate}
		\item This item says that distinct prefix classes in the same infix class are incomparable with respect to the prefix relation. This item of the Egg-box Lemma is the one that will be used most often. 
		
		Suppose that $a,b$ are infix equivalent and $a$ is a prefix of $b$, as witnessed by  solutions $x,y,z$ to the equations
		\begin{align*}
		   b= ax \qquad   a = ybz.
		\end{align*}
		As usual, each of $x,y,z$ could be empty. This can be illustrated as 
		
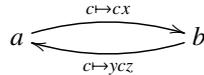
\begin{align*}
			\xymatrix@C=2cm{
				a
				\ar@/^/[r]^{c \mapsto cx}
				&
				b
				\ar@/^/[l]^{c \mapsto ycz}
			}
			\end{align*}
		Consider the idempotent exponent $\momega \in \set{1,2,\ldots}$  which arises from  Idempotent Power Lemma. We have:
		\begin{eqnarray*}
			b \equalbecause{follow $\momega + \momega$ times the loop around $a$, then go to $b$} \\
			y^{\momega + \momega}  a (xz)^{\momega + \momega} x \equalbecause{$y^\momega$ is an idempotent} \\
			y^{\momega} a (xz)^{\momega + \momega} x  \equalbecause{follow $\momega$ times the loop around $a$}\\			
			a (xz)^\momega x,
		\end{eqnarray*}
		which establishes that $b$ is a prefix of $a$, and therefore $a,b$ are in the same prefix class. 
		\item We now show that prefix and suffix classes in the same infix class must intersect. Suppose that $a,b$ are in the same infix class, as witnessed by 
		\begin{align*}
			a =xby.
		\end{align*}
		With respect to the infix relation, $by$ is between $b$ and $a=xby$, and therefore it must be in the same infix class as both of them.
		We have 
		\begin{align*}	
		\overbrace{x\myunderbrace{b \ \ y}{$b$ is a prefix of $by$}}^{\text{$by$ is a suffix of $xby = a$}},
		\end{align*}
		and therefore, thanks to the previous item, $by$ is prefix equivalent to $b$ and suffix equivalent to $a$. This  witnesses that the prefix class of $b$ and the suffix class of $a$ have nonempty intersection.
		\item We now show that all prefix classes  in the same infix class have the same size. Take some two prefix classes in the same infix class, given by representatives $a,b$. We can assume that $a,b$ are in the same suffix class, thanks to the previous item. Let 
		\begin{align*}
		a = xb \qquad b = ya
		\end{align*}
		be witnesses for the fact that $a,b$ are in the same suffix class. The following claim implies that the two prefix classes under consideration have the same size.
		\begin{claim}\label{claim:mutually-inverse-bijections}
			The following maps are mutually inverse bijections
				\begin{align*}
				\xymatrix@C=2cm{
					\text{prefix class of $a$} 
					\ar@/^/[r]^{c \mapsto yc}
					&
					\text{prefix class of $b$}
					\ar@/^/[l]^{c \mapsto xc}
				}
				\end{align*}
		\end{claim}
		\begin{proof}
			Suppose that $c$ is in the prefix class of $a$,  as witnessed by a decomposition $c=az$. If we apply sequentially both maps in the statement of the claim to $c$, then we get 
			\begin{align*}
			xyc = xyaz \stackrel{ya=b}= xbz \stackrel{xb=a}=  az \stackrel{az=c}= c.
			\end{align*}
			This, and a symmetric argument for the case when $c$ is in the prefix class of $b$, establishes that the maps in the statement of the claim are mutually inverse. It remains to justify that the images of the maps are as in the statement of the claim, i.e.~the image of the top map is the prefix class of $b$, and the image of the bottom map is the prefix class of $a$.  Because the two maps are mutually inverse, and they prepend elements to their inputs, it follows that each of the maps has its image contained in the infix class of $a,b$. To show that the image of the top map is in the prefix class of $b$ (a symmetric argument works for the bottom map), we observe that every element of this image is of the form $yaz$, and therefore it has $b=ya$ as a prefix, but it is still in the same infix class as $a,b$ as we have observed before, and therefore it must be prefix equivalent to $b$ thanks to the item~\ref{Egg-box:prefix-incomparable} of the lemma. 
		\end{proof}
	\end{enumerate}
\end{proof}

The Egg-box Lemma establishes that each infix class has the structure of a rectangular grid (which apparently reminded hungry author  of a box of eggs), with the rows being prefix classes and the columns being suffix classes. Let us now look at the eggs in the box:  define an $\Hh$-class to be an intersection of some prefix class and some suffix class, both taken from some common infix class.
By item~\ref{Egg-box:prefix-suffix-must-intersect} of the Egg-box Lemma, this intersection is nonempty. The following lemma shows that all $\Hh$-classes in the same infix class have the same size.

\begin{lemma}\label{lem:mutually-inverse-bijections-h}
	If  $a,b$ are in the same infix class, then  there exist possibly empty $x,y$ such that the following is a bijection
		\begin{align*}
		\xymatrix@C=2cm{
			\text{$\Hh$-class of $a$} 
			\ar@/^/[r]^{c \mapsto xcy}
			&
			\text{$\Hh$-class of $b$}
		}
		\end{align*}
\end{lemma}
\begin{proof}Consider first the special case of the lemma, when  $a$ and $b$ are in the same suffix class. Take the map from Claim~\ref{claim:mutually-inverse-bijections}, which maps bijectively the prefix class of $a$ to the prefix class of $b$. Since this map    preserves suffix  classes, it maps bijectively the $\Hh$-class of $a$ to the $\Hh$-class of $b$. By a symmetric argument, the lemma is also true when $a$ and $b$ are in the same prefix class. 
	
	For the general case, we use item~\ref{Egg-box:prefix-suffix-must-intersect} of the Egg-box Lemma, which says that there must be some intermediate element that is in the same prefix class as $a$ and in the same suffix class as~$b$, and we can apply the previously proved special cases to go from the $\Hh$-class of $a$ to the $\Hh$-class of the intermediate element, and then to the $\Hh$-class of $b$.
\end{proof}

The following lemma shows a dichotomy for an $\Hh$-class: either it is a group, or the  multiplying any two elements from the $\Hh$-class not only  falls outside the $\Hh$-class, but even outside  the corresponding infix class. 

\begin{lemma}[$\Hh$-class Lemma]\label{lem:h-classes} The following conditions are equivalent for every $\Hh$-class $G$ in a finite semigroup: 
	\begin{enumerate}
		\item \label{hclass:contains-idempotent} $G$ contains an idempotent;
		\item \label{hclass:some-product-j} $ab$ is in the same infix class as $a$ and $b$ for some $a,b \in G$;
		\item \label{hclass:some-product} $ab \in G$ for some $a,b \in G$;
		\item \label{hclass:all-product}$ab \in G$ for all $a,b \in G$;
		\item \label{hclass:group}$G$ is a group (with multiplication inherited from the semigroup).
	\end{enumerate}
	
\end{lemma}
\begin{proof}
	Implications \ref{hclass:group} $\Rightarrow$ \ref{hclass:contains-idempotent} $\Rightarrow$ \ref{hclass:some-product-j} in the lemma are obvious, so we focus on the remaining implications.
	\begin{description}
		\item[\ref{hclass:some-product-j}$\Rightarrow$\ref{hclass:some-product}] Suppose that $ab$ is in the same infix class as $a$ and $b$. Since $a$ is a prefix of $ab$, and the two elements are in the same infix class, item~\ref{Egg-box:prefix-incomparable} of the Egg-box Lemma implies that $ab$ is in the prefix class of $a$, which is the same as the prefix class of $b$. For similar reasons, $ab$ is in the same suffix class as $a$ and $b$, and therefore $ab \in G$.
		\item[\ref{hclass:some-product}$\Rightarrow$\ref{hclass:all-product}] Suppose that there exist  $a,b \in G$ with $ab \in G$.   We need to show that $G$  contains the multiplication of every two  elements $c,d \in G$. Since $c$ is prefix equivalent to $a$ there is a decomposition $a=xc$, and for similar reasons there is a decomposition $b=dy$. Therefore,  $cd$ is an infix of 
		\begin{align*}
		\overbrace{xc}^{a}\overbrace{dy}^b  \in G,
		\end{align*}
		 and therefore it is in the same infix class as $G$. By the reasoning in the previous item, $cd \in G$. 
		\item[\ref{hclass:all-product}$\Rightarrow$\ref{hclass:group}] Suppose that $G$ is closed under multiplication, i.e.~it is a subsemigroup. We will show that it is a group. By the Idempotent Power Lemma, $G$ contains some idempotent, call it $e$. We claim that $e$ is an identity element in $G$, in particular it is unique. Indeed, let $a \in G$. Because  $a$ and $e$ are in the same suffix class, it follows that $a$ can be written as $xe$, and therefore 
	\begin{align*}
		ae = xee= xe = a.
	\end{align*}
	For similar reasons, $ea=a$, and therefore $e$ is the unique identity element in $G$. The group inverse is defined as follows. Take $\momega \in \set{1,2,\ldots}$ to be the idempotent exponent which arises from the Idempotent Power Lemma. For every  $ a \in G$, the power $a^\momega$ is an idempotent.  Since there is only one idempotent in $G$, we have $a^\momega = e$. Therefore, $a^{\momega-1}$ is a group inverse of $a$.  
	\end{description}
\end{proof}

\mikexercise{\label{ex:monoid-identity-infix-class}Show that for every finite monoid, the infix class of the monoid identity is a group.}{
Let $J$ be the infix class of the monoid identity $1$. Since $1$ is prefix of every monoid element, it follows from the Egg-box Lemma that $J$ is equal to the prefix class of $1$. For the same reasons, $J$ is equal to the suffix class of $1$. Therefore $J$ is a single $\Hh$-class. Since $J$ contains an idempotent, namely $1$, it must be a group by the $\Hh$-class Lemma.
}

\mikexercise{\label{ex:regular-j-class} Consider a finite semigroup.
	Show that an infix class contains an idempotent if and only if it is \emph{regular}, which means that there exist $a,b$ in the infix class such that $ab$ is also in the infix class.
}
{If an infix class contains an idempotent, then it clearly contains elements $a,b$ such that $ab$ is in the infix class. For the converse implication, suppose that $a,b$ and $ab$ are in the same infix class $J$. It follows that each of $a,b$ can be decomposed as a multiplication of two elements from $J$. By iterating this procedure, we see that $a$ can be decomposed as a multiplication of $n$ elements from $J$, for every $n \in \set{1,2,\ldots}$. Thanks Exercise~\ref{ex:has-some-idempotent-ramsey}, a multiplication of $n$ elements from $J$ must contain an idempotent infix.
}

\mikexercise{\label{ex:isomorphic-groups}
	Show that if $G_1,G_2$ are two  $\Hh$-classes in the same infix class of a finite semigroup, and they are both groups, then they are isomorphic as groups\footnote{Let us combine Exercises~\ref{ex:regular-j-class} and \ref{ex:isomorphic-groups}. By  Exercises~\eqref{ex:regular-j-class} and the $\Hh$-class lemma, an infix class is regular if and only if it  contains an  $\Hh$-class which is a   group.  By Exercise~\eqref{ex:isomorphic-groups}, the corresponding group is unique up to isomorphism. This group is called the \emph{\schutz group} of the regular infix class. }.
}
{Suppose that $G_1$ and $G_2$ are groups in the same infix class. Let $e_1, e_2$ be the identities in the groups $G_1,G_2$.  
From Claim~\ref{claim:mutually-inverse-bijections} it follows that there exist $x_1,x_2,y_1,y_2$ such that 
\begin{align}\label{eq:mutually-inverse}
	\xymatrix@C=2cm{
		G_1
		\ar@/^/[r]^{g \mapsto x_1gy_1}
		&
		G_2
		\ar@/^/[l]^{g \mapsto x_2 g y_2}
	}
	\end{align}
are mutually inverse bijections. By replacing 
\begin{align*}
 \myunderbrace{x_1 e_1}{new $x_1$}  \quad \myunderbrace{e_1 x_2 }{new $x_2$}  \quad
 \myunderbrace{e_1 y_1}{new $y_1$}  \quad \myunderbrace{y_2 e_1}{new $y_2$},
\end{align*}
we still get mutually inverse bijections between $G_1$ and $G_2$. Summing up, we can assume without loss of generality that $x_1,y_2$ end with $e_1$, while $x_2,y_1$ begin with $e_1$.

The element $y_1x_1$ begins and ends with $e_1$, and it is also an infix of $e_2$ (and therefore also of $e_1$) thanks to 
\begin{align*}
	e_2 = e_2 e_2 = \underbrace{x_1x_2e_2y_2y_1}_{e_2} \underbrace{x_1x_2e_2y_2y_1}_{e_2}.
	\end{align*}
	It follows from the Egg-box Lemma that $y_1 x_1$ is both in the prefix class and suffix class of $e_1$, which means that $y_1 x_1 \in G_1$. Since $G_1$ is a group, there must be some $a,b \in G_1$ such that
	\begin{align*}
	a y_1 x_1 a = e_1
	\end{align*}
	Let $\overline{y_1 x_1}$  be the group inverse of $y_1 x_1$, in the group $G_1$. Define $\alpha :G_1 \to G_2$ to be the composition of the following two functions 
\begin{align*}
	\xymatrix@C=2cm{
		G_1 \ar[r]^{g \mapsto   g \overline{y_1 x_1}} &
		G_1 \ar[r]^{g \mapsto x_1 g y_1} & 
		G_2.
	}
\end{align*}
The first function is a permutation of $G_1$, while the second function is a bijection of $G_1$ and $G_2$. It follows that $\alpha$ is a bijection. We now claim that $\alpha$ is a homomorphism:
\begin{align*}
\alpha(gh) = x_1  gh \overline{y_1 x_1} y_1 = 
\rlap{$\overbrace{\phantom{x_1 g \overline{y_1 x_1} y_1}}^{\alpha(g)}$}
x_1 g 
\rlap{$\underbrace{\phantom{\overline{y_1 x_1} y_1 x_1}}_{e_1}$}
\overline{y_1 x_1} y_1
\overbrace{ x_1 h \overline{y_1 x_1} y_1 }^{\alpha(h)} = 
\alpha(g) \alpha(h).
\end{align*}
Summing up, $\alpha$ is a bijective semigroup homomorphism between the groups $G_1$ and $G_2$. It follows that these groups are isomorphic as groups, because a bijective semigroup homomorphism also preserves the group structure. 
}

\mikexercise{\label{ex:prefix-trivial-identity} We say that semigroup is \emph{prefix trivial} if its prefix classes are singletons. Show that a finite semigroup  $S$ is prefix trivial if and only if it satisfies the identity 
\begin{align*}
	(xy)^\momega = (xy)^{\momega} x \qquad \text{for all }x,y \in S.
\end{align*}
}{}

\mikexercise{\label{ex:definite}
	Define the \emph{syntactic semigroup} of a language to be the subset of the  syntactic monoid which is the image of the nonempty words under the syntactic homomorphism. The syntactic semigroup may be equal to the syntactic monoid. 	
We say that a language $L \subseteq \Sigma^*$ is definite if it is a finite Boolean combination of languages of the form $w \Sigma^*$, for $w \in \Sigma^*$. Show that a language is definite if and only if its syntactic semigroup $S$ satisfies the identity
\begin{align*}
	x^\momega = x^\momega y \qquad \text{for all }x,y \in S.
\end{align*}
}{}

\mikexercise{Show two regular languages such that one is definite and the other is not, but both have isomorphic syntactic monoids.}{}

\mikexercise{Consider semigroups  $S$ which satisfy the following property: (*) that there is an infix class $J \subseteq S$ such that every $a \in S$ is an infix of $J$, or an absorbing zero element. Show that every finite semigroup is sub-semigroup of a product of finite semigroups that satisfy (*).
}
{
Let  $S$ be a finite semigroup, and let  $ a \in S$. 
     Define  $S_a$ to be the set obtained from $S$ by removing all elements that are not infixes of $a$, and adding a fresh $0$ element.  We can view this a semigroup, with multiplication defined by 
\begin{align*}
w \in (S_a)^\cc \quad \mapsto \quad 
\begin{cases}
    0 & \text{if $w$ contains at least one letter $0$}\\
    \text{multiplication in $S$} & \text{otherwise}.
\end{cases}
\end{align*}
This semigroup satisfies condition (*). There is also a natural homomorphism $h_a : S \to S_a$, which is the identity on infixes of $a$, and maps the remaining elements to $0$. By taking combining these homomorphisms, we get  a 
homomorphism 
\begin{align*}
S \to \prod_{a \in S} S_a,
\end{align*}
which is injective, because each $h_a$ is injective on infixes of $a$. 
 It follows that $S$ is (isomorphic) to a sub-algebra of the product of the $\cc$-semigroups $S_a$. 
}

\mikexercise{Show that every finite semigroup satisfies 
\begin{align*}
\forall x_1\  \forall x_2 \ \exists y_1 \ \exists y_2 \ 
\myunderbrace{z_1 = z_1 z_1 = z_1 z_2 \land z_2 = z_2 z_2 = z_2 z_1}{where $z_i = x_i y_i$},
\end{align*}
where quantifiers range over elements of the finite semigroup.

}{}

\mikexercise{
Show that the following problem is decidable:
\begin{itemize}
	\item {\bf Input.} Two disjoint sets of variables
	\begin{align*}
	  X = \set{x_1,\ldots,x_n} \qquad Y = \set{y_1,\ldots,y_m}
	\end{align*}
	and two words $w,w' \in (X \cup Y)^+$.
	\item {\bf Question.} Is the following true in all finite semigroups:
	\begin{align*}
	\forall x_1 \ \cdots \ \forall x_n \ \exists y_1 \ \cdots \exists y_m \ \myunderbrace{w = w'}{same multiplication}
	\end{align*}
\end{itemize}

}{}
 \section{The Factorisation Forest Theorem}
\label{sec:fact-for}
In this section, we show how the  multiplication of a long sequence of elements in a semigroup can be organised as a tree, so that in each node of the tree the multiplication is very simple. The most natural way to do this is to have binary tree, as in the following example, which uses the two semigroup $\set{0,1}$ with addition modulo 2:
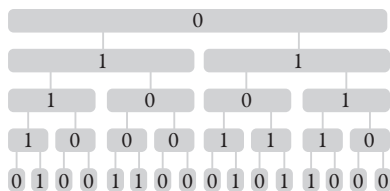
\mypic{25}
We use the name \emph{factorisation tree} for structures as in the above picture. More formally, a  \emph{factorisation tree over a semigroup $S$} is a tree, where nodes are labelled by semigroup elements, such that every node is either a leaf, or is labelled by  the semigroup multiplication of the labels of its children.  Since the semigroup in question need not be commutative, the children in a tree are ordered, i.e.~there is a first child, second child, etc. 

A \emph{binary factorisation} tree is one where every node has zero or two children. For every word in $S^+$, one can find a corresponding binary factorisation tree (i.e.~one where the word is obtained by reading the leaves left-to-right) whose height (i.e.~the maximal number of edges on a root-to-leaf path) is logarithmic in the length of the word. 
Binary factorisation trees are a natural data structure for several problems about regular languages. 

\begin{myexample}\label{ex:dynamic-change}
    Fix  a regular language $L \subseteq \Sigma^*$. Consider the following dynamic problem. We begin with some word  in $\Sigma^*$. We want to build a data structure that handles efficiently the following updates and queries: 
    \begin{description}
        \item[Query.] Is the current word in $L$?
        \item[Update.] Change the label of position $i$ to $a \in \Sigma$.
    \end{description}
    To solve this problem, as the data structure we can use a binary factorisation  tree with respect to some finite semigroup that recognises the language. If we assume that the  language is fixed and not part of the input, then the queries are processed in constant time, by checking if the root label of the factorisation tree belongs to the accepting set. The updates are processed in time proportional to the height of the factorisation tree, by updating all of the nodes on the path from the updated position to the root, as in the following picture:
    \mypic{26}
    If the factorisation tree is chosen to be balanced, then the updates are processed in logarithmic time. 
\end{myexample}

\begin{myexample}\label{ex:dynamic-infix}
    Fix  a regular language $L \subseteq \Sigma^*$. Consider the following dynamic problem. We begin with some word  in $\Sigma^*$. We want to build a data structure that handles efficiently the following  queries (there are no updates): 
    \begin{description}
        \item[Query.]  Given positions $i \le j$, does $L$ contain the infix from $i$ to $j$?
    \end{description}
    Of course, one obvious solution is to pre-compute in quadratic time a table of answers to all possible queries. If we want to solve the problem with linear time pre-computation, then we can use  a binary factorisation  tree, with respect to some semigroup recognising the language. Suppose that the tree has height $k$. Each node of the factorisation tree corresponds to an infix of the underlying word.  The infix from $i$ to $j$ can be partitioned into at most $2k$ intervals,  each of which corresponds to a node of the tree, as in the following picture:
    \mypic{27}
    Therefore, the queries can be processed in time proportional to the height of the tree, which can be assumed to be logarithmic in the length of the underlying word.  
\end{myexample}

In this section, we show a data structure which will allow constant time query processing in the problem from Example~\ref{ex:dynamic-infix}. We will also use a variant of factorisation trees, except that non-binary nodes will need to be used. The problem in Example~\ref{ex:dynamic-change} cannot be solved in constant time\footnote{Lower bounds for this problem can be seen in\incite[Fig.~1]{frandsenMiltersenSkym97}}.

\section*{Simon trees}
A Simon tree is a factorisation tree which allows nodes of degree higher than 2, but these nodes must have idempotent children.  The data structure is named after Imre Simon, who introduced it\footnote{Under the name Ramseyan factorisation forests, in \incite[69]{simonFactorizationForestsFinite1990}}.

\begin{definition}[Simon Tree]  Define a \emph{Simon tree} (for a given semigroup) to be a factorisation tree  where every non-leaf node has one (or both) of the following types:
    \begin{description}
        \item[binary:] there are  two children; or
        \item[idempotent:] all children have the same label, which is an idempotent.
    \end{description}
    \end{definition}

    Here is a picture of a Simon tree for the semigroup $\set{0,1}$ with addition modulo 2, with idempotent nodes drawn in red:
    \mypic{2}

The main result about Simon trees is that their height can be bounded by a constant that depends only on the semigroup, and not the underlying word. 
\begin{theorem}[Factorisation Forest Theorem]
    \label{thm:simon-factfor} Let $S$ be a  finite semigroup. Every word in $S^+$ admits a Simon tree of height\footnote{
        The first version of this theorem was proved  in~\cite[Theorem 6.1]{simonFactorizationForestsFinite1990}, with a bound of  $9|S|$. The optimal bound is $3|S|$, which was shown in 
        \incite[Theorem 1]{kufleitner2008}
         The  proof here  is based on Kufleitner, with some optimisations removed.
    }  $<5|S|$.
\end{theorem}

The rest of this chapter is devoted to proving the theorem. 
\paragraph*{Groups.}
We begin with  the special case of groups.
\begin{lemma}\label{lem:group-factfor}
    Let $G$ be a finite group. Every word in $G^+$ has a Simon tree of height  $<3|G|$.
\end{lemma}
\begin{proof}
Define the \emph{prefix set} of a word $w \in G^+$ to be the set of group elements that can be obtained by  multiplying  some nonempty prefix of $w$. 
By induction on the size of the prefix set, we show  that every $w \in G^+$ has a Simon tree of height strictly less than 3 times the size of the prefix set. Since the prefix set has maximal size $|G|$, this proves the lemma.

The induction base is when the prefix set is a singleton $\set g$. This means that the first letter is $g$, and every other letter $h$  satisfies $gh=g$. In a group, only the group identity $h=1$ can satisfy $gh=g$,  and therefore $h$ is the group identity. In other words, if the prefix set is $\set g$, then the word is of the form
\begin{align*}
g \myunderbrace{1  \cdots 1}{a certain number of times}.
\end{align*}
 Such a word admits a Simon tree as in the following picture:
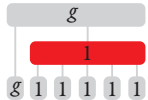
\mypic{28}
The height of this tree is 2, which is strictly less than three times the size of the prefix set. 

To prove the induction step, we show that every $w \in G^+$  admits a Simon tree, whose height is at most 3 plus the size from the induction assumption.
 Choose some $g$ in the prefix set of $w$. Decompose $w$ into factors as
    \begin{align*}
    w = \myunderbrace{w_1 w_2 \cdots w_{n-1}}{nonempty\\ \scriptsize factors} \myunderbrace{w_n}{could\\ \scriptsize be empty} 
    \end{align*}
    by cutting along all prefixes that multiply to $g$. For the same reasons as in the induction base, every factor $w_i$ with $1 < i < n$  yields the group identity under multiplication.
    
    \begin{claim}
        The induction assumption applies to all of $w_1,\ldots,w_n$.
    \end{claim}
    \begin{proof} 
        For the first factor $w_1$,  the induction assumption applies, because its prefix set omits $g$. 
        For the remaining blocks, we have a similar situation, namely
        \begin{align*}
            g \cdot \text{(prefix set of $w_i$)} \subseteq \text{(prefix set of $w$)} - \set g \qquad \text{for $i \in \set{2,3,\ldots,n}$},
            \end{align*}
        where the left side of the inclusion is the image of the prefix set under the operation  $x \mapsto gx$. Since this operation is a permutation of the group, it follows that the left size of the inclusion has smaller size than the prefix set of $w$, and therefore the induction assumption applies.             
    \end{proof}
    By the above claim, we can apply the induction assumption to compute Simon trees $t_1,\ldots,t_n$ for the factors $w_1,\ldots,w_n$. To get a Simon tree for the whole word, we join  these trees as follows:
    \mypic{8}
    The gray nodes are binary, and the red node is idempotent because every $w_i$ with $1 < i <n$  evaluates to the group identity.
\end{proof}

\paragraph*{Smooth words.} In the next step, we prove the theorem for words where all infixes have multiplication in the  same infix class. 
We say that a word $w \in S^+$ is \emph{smooth} if every  nonempty infix multiplies to  the same infix class. The following lemma constructs Simon trees for smooth words.
\begin{lemma}\label{lem:smooth} If a word is smooth, and the corresponding infix class is $J \subseteq S$, then it has a  Simon tree of height $<4|J|$.
\end{lemma}
\begin{proof}
    Define a \emph{cut} in a word to be the space between two consecutive letters; in other words this is a decomposition of the word into a nonempty prefix and a nonempty suffix. 
    For a cut, define its \emph{prefix} and \emph{suffix} classes as in the following picture:
    \mypic{7}
    For every cut,  both the prefix and suffix classes are contained in $J$, and therefore they have nonempty intersection thanks to item~\ref{Egg-box:prefix-suffix-must-intersect} of the Egg-box Lemma. This nonempty intersection is an $\Hh$-class, which is defined to be the \emph{colour} of the cut.
    The following claim gives the crucial property of cuts and their colours.  
    
    \begin{claim}\label{claim:consecutive-cuts}
        If  two cuts have the same colour $H$, then the infix between returns an element of  $H$ under multiplication.
    \end{claim}
    \begin{proof}
        Here is a picture of the situation:
        \mypic{6}
        The infix begins with a letter from the prefix class containing $H$.  Since the infix is still in the infix class $J$, by assumption on smoothness, it follows from item~\ref{Egg-box:prefix-incomparable} of the Egg-box Lemma that the result of multiplying   the infix is in the prefix class of $H$. For the same reason, the result is also in the suffix class of $H$. Therefore, it is in $H$.
    \end{proof}

    Define the \emph{colour set} of a word to be the set of colours of its cuts; this is a subset of the $\Hh$-classes in $J$. Thanks to Lemma~\ref{claim:mutually-inverse-bijections}, all $\Hh$-classes contained in $J$ have the same size, and therefore it makes sense to talk about  the $\Hh$-class size in $J$, without specifying which $\Hh$-class is concerned. 
    
    \begin{claim}
        Every $J$-smooth word has a Simon tree of height at most 
        \begin{align*}
        |\text{colour set of $w$}| \cdot (3 \cdot \text{$\Hh$-class size} + 1).
        \end{align*}
    \end{claim}

    Since the number of possible colours is the number of $\Hh$-classes, the maximal height that can arise from the claim is
    \begin{align*}
    3 \cdot |J| + \text{(maximal size of colour set)} < 4|J|,
    \end{align*}
which proves the lemma.
    It remains to prove the claim.
    
    \begin{proof}
    Induction on the size of the colour set. The induction base is when the colour set is empty. In this case the word has no cuts, and therefore it is a single letter, which is a Simon tree of height zero.

    Consider the induction step. Let $w$ be a smooth word. To prove the induction step, we will find a Simon tree whose height is at most the height from the induction assumption, plus
    \begin{align*}
        3 \cdot \text{($\Hh$-class size)} +1.
    \end{align*} Choose some colour in the colour set of $w$, which is an $\Hh$-class $H$. Cut  the word $w$ along all cuts with colour $H$, yielding a decomposition 
    \begin{align*}
    w = w_1  \cdots w_n.
    \end{align*}
    None of the words $w_1,\ldots,w_n$ contain a cut with colour $H$, so the induction assumption can be applied to yield corresponding Simon trees $t_1,\ldots,t_n$. 

    If $n \le 3$, then  the Simon trees from the induction assumption  can be combined  using  binary nodes, increasing the height by at most 2, and thus staying within the bounds of the claim. 
    
    Suppose now that $n \ge 4$.  
    By Claim~\ref{claim:consecutive-cuts}, multiplying an infix between any two cuts of colour $H$ returns a value in $H$. In particular, all  $w_2,\ldots,w_{n-1}$ yield results in $H$ under multiplication, and the same is true for $w_2w_3$. It follows that $H$ contains at least one multiplication of two elements from $H$, and therefore $H$ is a group thanks to item~\ref{hclass:some-product} of the $\Hh$-class Lemma. Therefore,  we can apply the group case from Lemma~\ref{lem:group-factfor} to join the trees $t_2,\ldots,t_{n-1}$. The final Simon tree looks like this:
    \mypic{10}
\end{proof}
\end{proof}

\paragraph*{General case.}
We now complete the proof of the Factorisation Forest Theorem. The proof is by induction on the \emph{infix height} of the semigroup, which is defined to be the longest chain that is strictly increasing in the infix ordering.  If the infix height is one, then the semigroup is a single infix class, and we can apply Lemma~\ref{lem:smooth} since all words in $S^+$ are smooth. For the induction step, suppose that $S$ has infix height at least two, and let $T \subseteq S$ be the elements which have a proper infix. It is not hard to see that $T$ is a subsemigroup, and its induction parameter is smaller. 

 Consider a  word $w \in S^+$. As in Lemma~\ref{lem:smooth}, define a cut to be a space between two letters. We say that a cut is \emph{smooth} if the letters preceding and following the cut give a two-letter word that is smooth.  

\begin{claim}\label{claim:smooth}
    A word in $S^+$ is smooth if and only if all of its cuts are smooth.
\end{claim}
\begin{proof}
    Clearly if a word is smooth, then all of its cuts must be smooth.  We prove the converse implication by induction on the length of the word. Words of length one or two are vacuously smooth. For the induction step, consider a word $w  \in S^+$  with all cuts being smooth.  Since all cuts are smooth, all letters are in the same infix class. We will show that $w$ is also in this infix class. Decompose the word as $w = vab$ where  $a,b \in S$ are the last two letters. By induction assumption, $va$ is smooth. Since the last cut is smooth,  $a$ and $ab$ are in the same infix class, and therefore they are in the same prefix class by the Egg-box Lemma. This means that there is some $x$ such that  $abx = a$. We have
    \begin{align*}
    va = vabx = wx
    \end{align*}
which establishes that $w$ is in the same infix class as $va$, and therefore in the same infix class as all the letters in $w$. 
\end{proof}

Take a word $w \in S^+$, and cut it along all cuts which are not smooth, yielding a factorisation
\begin{align*}
w = w_1 \cdots w_n.
\end{align*}
By Claim~\ref{claim:smooth}, all of the words $w_1,\ldots,w_n$ are smooth, and therefore Lemma~\ref{lem:smooth} can be applied to construct corresponding Simon trees of height strictly smaller than
\begin{align*}
4 \cdot \text{(maximal size of an infix class in $S-T$)}.
\end{align*}
Using binary nodes, group these trees into pairs, as in the following picture:
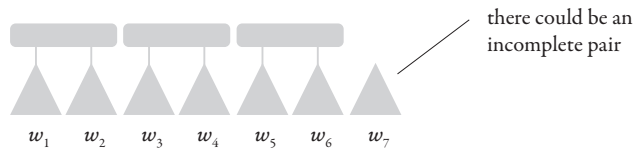
\mypic{29}
Each pair corresponds to a word with a non-smooth cut, and therefore multiplying each pair yields a result in $T$. Therefore, we can combine the paired trees into a single tree, using the induction assumption on a smaller semigroup. The resulting height is the height from the induction assumption on $T$, plus at most
\begin{align*}
    1 + 4 \cdot \text{(maximal size of an infix class in $S-T$)} < 5|S-T|,
    \end{align*}
    thus proving the induction step. 

\exercisehead

\mikexercise{
    Show that for  every semigroup homomorphism
    \begin{align*}
    h : \Sigma^+ \to S \qquad \text{with $S$ finite}
    \end{align*}
    there is some $k \in \set{1,2,\ldots}$ such that for every $n \in \set{3,4,\ldots}$, every word of length bigger than $n^k$ can be decomposed as 
    \begin{align*}
     w_0 w_1 \cdots w_n w_{n+1}
    \end{align*}
    such that all of the words $w_1,\ldots,w_n$ are mapped by $h$ to the same idempotent. 
}
{  The number $k$ is the bound on the height Simon trees for $S$, i.e.~at most $5|S|$. If a word $w$ has length more than $n^k$, and it has a Simon tree of height $k$, then some node in that tree must have degree at least $n$. Since $n \ge 3$, that node is an idempotent node, which gives the factorisation in the statement of the exercise.   }

\mikexercise{Show optimality for the previous exercise, in the following sense. Show that for every $k \in \set{1,2,\ldots}$ there is some semigroup homomorphism
\begin{align*}
    h : \Sigma^+ \to S \qquad \text{with $S$ finite}
    \end{align*}
    such that for every  $n \in \set{1,2,\ldots}$ there is a word of length at least $n^k$ which does not admit a factorisation $w_0 \cdots w_{n+1}$ where all of $w_1,\ldots,w_n$ are mapped by $h$ to the same idempotent. 
}
{
    Consider the semigroup homomorphism
    \begin{align*}
    h : \set{1,\ldots,k}^+ \to (\set{1,\ldots,k}, \max)
    \end{align*}
    which maps a word to its maximum.
    For $n \in \set{1,2,\ldots}$, define a word $w_{i,n}$ by induction on $i \in \set{0,1,\ldots,k}$ as follows:
    \begin{align*}
        w_{i,n} = \begin{cases}
            \varepsilon & \text{for $i=0$}\\
            w_{i,n} (i w_{i,n})^{n-1} & \text{for $i \in \set{1,\ldots,k}$}.
        \end{cases}
    \end{align*}
    When going from $i$ to $i+1$, the length of the word increases at least $n$ times, and therefore $w_{k,n}$ has length bigger than $n^k$. Nevertheless, there is no infix $v$ of $w_{k,n}$ which is a concatenation of $n$ idempotents. To see this, consider the maximal letter $i \in \set{1,\ldots,k}$ that appears in the infix $v$. By construction of the words, this letter appears $<n$ times. If $v$ is a concatenation of idempotents, then each of these idempotents must contain $i$, and therefore there are $<n$ such idempotents.
}

\mikexercise{\label{ex:typed-regular-expressions}Let  $h : \Sigma^* \to M$ be a monoid homomorphism. Consider a regular expression over $\Sigma$, which does not use Kleene star $L^*$ but only Kleene plus $L^+$.  Such a  regular expression is called  $h$-typed if every subexpression has singleton image under $h$, and furthermore subexpressions with Kleene plus have idempotent image.  Show that every language recognised by $h$ is defined by finite union of $h$-typed expressions.}{}
 \chapter{Logics on finite words, and the corresponding monoids}
\label{chap:logics}
In this chapter, we show how structural properties of a monoid correspond to  the logical power needed to define languages recognised by this monoid. We consider two kinds of logic: monadic second-order logic \mso and its fragments (notably first-order logic \fo), as well as linear temporal logic \ltl and its fragments. Here is a map of the results from this chapter, with horizontal  arrows being equivalences, and the vertical arrows being strict inclusions.

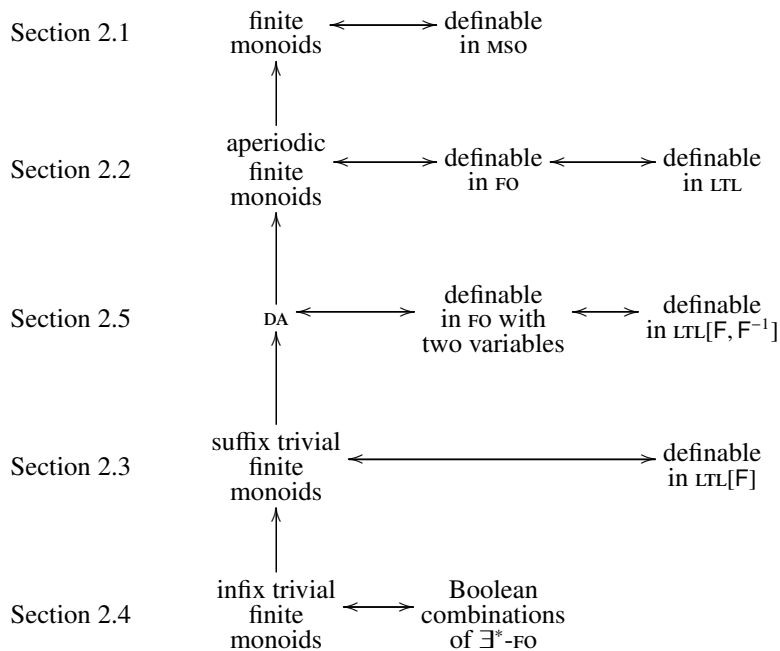
\begin{align*}
\xymatrix{
  \txt{Section~\ref{sec:mso-finite-words}} &
  \txt{finite\\ monoids} \ar[r] &
  \txt{definable\\ in \mso} \ar[l]
  \\
  \txt{Section~\ref{sec:schutz}} & \ar[u]
  \txt{aperiodic\\finite\\ monoids} \ar[r] &
  \txt{definable\\ in \fo} \ar[l] \ar[r] & 
  \txt{definable\\ in \ltl} \ar[l] 
  \\
  \txt{Section~\ref{sec:fotwo}} & \ar[u]
  \txt{\dav } \ar[r] &
  \txt{definable\\ in \fo with\\ two variables} \ar[l] \ar[r] & 
  \txt{definable\\ in \ltlunary} \ar[l] 
  \\
  \txt{Section~\ref{sec:finally-only}} &
  \ar[u]
  \txt{suffix trivial\\finite\\ monoids} \ar[rr] &
   & 
  \txt{definable\\ in \ltlf} \ar[ll]
  \\
  \txt{Section~\ref{sec:piecewise}} &
  \ar[u]
  \txt{infix trivial\\finite\\ monoids} \ar[r] &
  \txt{Boolean\\ combinations \\ of $\exists^*$-\fo} \ar[l] 
}
\end{align*}

\section{All monoids and monadic second-order logic}
\label{sec:mso-finite-words}
We begin with  monadic second-order logic (\mso), which is the 
 logic that captures exactly the class of  regular languages.

\paragraph*{Logic on words.} We assume that the reader is familiar with the basic notions of logic, such as formula, model, quantifier or free variable. The following description is meant to fix notation. We use the word  \emph{vocabulary} to denote a set of relation names, each one with associated arity in $\set{1,2,\ldots}$.   A \emph{model} over a vocabulary consists of an \emph{underlying set} (also called the \emph{universe} of the model), together with an interpretation, which maps each relation name from the vocabulary to a relation over the universe of same arity. We allow the universe to be empty.  For example, a directed graph is the same thing as a model where the universe is the vertices and the  vocabulary  has one binary relation $E(x,y)$ that represents the edge relation.

To express properties of models, we use  first-order logic \fo and \mso.  Formulas of first-order logic over a given vocabulary are constructed as follows:
\begin{align*}
\myunderbrace{\forall x \quad \exists x}{quantification over\\ \scriptsize elements of the universe}  \qquad 
\myunderbrace{\varphi \land \psi \quad \varphi \lor \psi \quad \neg \varphi}{Boolean operations}
\qquad 
\myunderbrace{R(x_1,\ldots,x_n)}{an $n$-ary relation name from\\ \scriptsize the vocabulary applied \\ \scriptsize  to a tuple of variables} 
\qquad 
\myunderbrace{x=y}{equality}.
\end{align*}
We  use the  notation
\begin{align*}
\mathbb A, a_1,\ldots,a_n  \models \varphi(x_1,\ldots,x_n)
\end{align*}
to say that formula $\varphi$ is true in the model $\mathbb A$, assuming that  free variable $x_i$ is mapped to $a_i \in \mathbb A$.
A \emph{sentence} is a formula without free variables.

Apart from first-order logic, we also use monadic second-order logic \mso; in fact \mso is the central logic for this book.
The logic \mso extends first-order logic by allowing quantification over subsets of the universe (in other words, monadic relations over the universe, hence the name). The syntax of the logic has   two kinds of variables: lower case variables $x,y,z,\ldots$ describe elements of the universe as in first-order logic, while upper case variables $X,Y,Z,\ldots$ describe subsets of the universe. Apart from the syntactic constructions of first-order logic, \mso also allows:
\begin{align*}
  \myunderbrace{\forall X \quad \exists Y}{quantification over\\ \scriptsize subsets of the universe}  \qquad \qquad 
  \myunderbrace{x \in X}{membership}.
\end{align*}
We do not use more powerful logics (e.g.~full second-order logic, which can also quantify over binary relations, ternary relations, etc.). This is because more powerful logic will not be subject to compositionality methods that are discussed in this book.

The following definition associates to each word a corresponding model.  With this  correspondence, we can use   logic to define properties of words.  
\begin{definition}[Languages definable in first-order logic and \mso]
\label{def:mso-finite-words}
  For a word $w \in \Sigma^*$, define its \emph{ordered model} as follows. The universe is the set of  positions in the word, and it is equipped with the following relations: 
\begin{align*}
\myunderbrace{x \le y}{position $x$ \\ \scriptsize is before \\ \scriptsize position $y$} \qquad \set{\myunderbrace{a(x)}{position $x$ \\ \scriptsize has label $a$ }}_{a \in \Sigma}.
\end{align*}
  For a sentence $\varphi$ of \mso   over the vocabulary used in the ordered model (this vocabulary depends only on the alphabet $\Sigma$), we  define  its \emph{language} to be
   \begin{align*}
    \set{w \in \Sigma^* : \text{the ordered model of $w$ satisfies $\varphi$}}.
   \end{align*}
   A language is called \emph{\mso definable} if it is of this form. If $\varphi$ is in first-order logic, i.e.~it does not use set quantification, then the language is called \emph{first-order definable}.
\end{definition}

\begin{myexample}\label{es:mso-definable-languages}
  The language $a^*b c^* \subseteq \set{a,b,c}^*$ is first-order definable, as witnessed by the sentence:
  \begin{align*}
  \myunderbrace{\exists x}
  {there is a \\
  \scriptsize position} \quad 
  \myunderbrace{b(x)}
  {which has  \\
  \scriptsize  label $b$}
  \land 
  \myunderbrace{\forall y \ \overbrace{y < x}^{y \le x \land x \neq y} \Rightarrow a(y)}
  {and every earlier position 
  \\ \scriptsize has label $a$}
   \land 
   \myunderbrace{\forall y \ y > x \Rightarrow c(y).}
  {and every later position 
  \\ \scriptsize has label $b$}
  \end{align*}
\end{myexample}
\begin{myexample}\label{ex:parity-mso} The language $(aa)^*a \subseteq a^*$ of words of odd length is \mso definable, as witnessed by the sentence:
  \begin{align*}
    \myunderbrace{\exists X }
    {there is a \\
    \scriptsize set of\\
    \scriptsize positions}
    \quad
    \myunderbrace{\forall x \ 
    \overbrace{\mathrm{first}(x)}^{\forall y \ y \ge x} \lor 
    \overbrace{\mathrm{last}(x)}^{\forall y \ y \le x}
     \Rightarrow x \in X}
    {which contains the first and last positions,}
     \land 
     \myunderbrace{\forall x \forall y \  
     \overbrace{x = y+1}^{
       \substack{
       x < y \land \\ \forall z \ z \le x \lor y \le z}}
     \Rightarrow (x \in X \iff y \not \in X)}
    {and contains every second position.}
    \end{align*}
    As we will see  in Section~\ref{sec:schutz}, this language is  not first-order definable.
\end{myexample}

One could imagine other ways of describing a word via a model, e.g.~a \emph{successor model} where $x \le y$ is replaced by a successor relation $x+1=y$. The successor relation can be defined in first-order logic in terms of order, but the converse is not true. Indeed,  there are languages that are first-order definable in the ordered model but are not first-order definable  in the successor model, see Exercise~\ref{ex:successor-idempotent-swap}. 
For the logic \mso, there is no difference between successor and order, since the order can be defined in terms of successor using the logic \mso as follows
\begin{align*}
x \le y \quad \text{iff} \quad \forall X\ 
\myunderbrace{
(x \in X \land (\forall y \ \forall z\ y \in X \land y+1=z \Rightarrow z \in X) )}
{$X$ contains $x$ and is closed under successors} \Rightarrow y \in X.
\end{align*}

We now present the seminal \bueltr Theorem, which says that  \mso describes exactly the regular languages. 
\begin{theorem}[\bueltr]
  A language $L \subseteq \Sigma^*$ is \mso definable  if and only if it is regular\footnote{
    This result was proved, independently, in the following papers:
     \incite[Theorems 1 and 2]{trakthenbrot1958} 
     \incite[Theorems 1 and 2]{Buchi60}
     \incite[Theorem~5.3]{Elgot61}
     }.
\end{theorem}
This result is seminal for   two reasons.

The first reason is that it motivates the  search  for other correspondences 
\begin{align*}
\text{machine model} \quad \sim \quad \text{logic,}
\end{align*}
which can concern either  restrictions or generalisations of the regular languages. In the case of restrictions, an important example is first-order logic; this restriction and others  will be described later in this chapter. In this book, we do not study the generalisations; we are only interested in regular languages. Nevertheless, it is worth mentioning Fagin's Theorem, which says that  {\sc np}  describes exactly the languages definable  in existential second-order logic\footcitelong[Theorem 6]{Fagin74}.   

The  \bueltr theorem is also seminal because it generalises well to structures beyond finite words. For example, there are obvious notions of \mso definable languages for: infinite words, finite trees, infinite trees, graphs, etc. It therefore makes sense to search for notions of regularity -- e.g.~based on generalisations of semigroups -- which have the same expressive power as \mso. This line of research will also be followed in this book.

The rest of Section~\ref{sec:mso-finite-words} proves the \bueltr Theorem.

The easy part is  that every regular language is \mso definable. Using the same idea as for the parity language in Example~\ref{ex:parity-mso}, the existence of a run of  nondeterministic finite automaton can be formalised in \mso. If the automaton has $n$ states, then the formula looks like this:
\begin{align*}
\myunderbrace{
  \exists X_1\  \exists X_2 \ \cdots \exists X_n}
  {existential set quantification} \qquad 
\myunderbrace{
\text{``the sets $X_1,\ldots,X_n$
 describe an accepting run''.}}
 {first-order formula}
\end{align*}
A corollary is that if we take any \mso definable language, turn it into an automaton using the hard implication, and come back to \mso using the easy implication, then we get an  \mso sentence of the form described above.

We now turn to the hard part, which says that every \mso definable language is regular. This implication is proved in the rest of Section~\ref{sec:mso-finite-words}. The proof is so generic that it will be reused multiple times in future chapters, for structures such as infinite words, trees or graphs.

For the definition of regularity, we use finite monoids. In other words, we will show that every \mso definable language is  recognised by a finite monoid. 
The idea is to construct the finite monoid by induction on formula size. In the induction, we also construct monoids for formulas with free variables, so we begin by dealing with those. 

\begin{definition}[Language of formulas with free variables]
  For an \mso  formula 
  \begin{align*}
  \varphi(\myunderbrace{X_1,\ldots,X_n}{all free variables are set variables})
  \end{align*}
  which uses the vocabulary of the ordered model for words over alphabet $\Sigma$, define its \emph{language}  to be the set of words $w$ over alphabet $\Sigma \times \set{0,1}^n$ such that 
  \begin{align*}
  \pi_\Sigma(w) \models \varphi(X_1,\ldots,X_n),
  \end{align*}
  where $\pi_\Sigma$ is the projection of $w$ onto the $\Sigma$ coordinate, and $X_i$ is the set of positions whose label has value 1 on the $i$-th bit of the bit vector from $\set{0,1}^n$.
\end{definition}
If the formula $\varphi$ in the above definition has no free variables, then the above notion of language coincides with Definition~\ref{def:mso-finite-words}. Therefore, the hard part of the \bueltr Theorem will follow immediately from Lemma~\ref{lem:bet-induction} below.
\begin{lemma}\label{lem:bet-induction}
  If $\varphi(X_1,\ldots,X_n)$ is a formula of \mso (where all free variables are set variables), then its language is recognised by a finite monoid.
\end{lemma}
\begin{proof}
  Before proving the lemma, we observe that first-order variables can be eliminated from \mso, and therefore we can assume that in  $\varphi$ and all of its sub-formulas, all free variables are set variables.  Suppose that we extend  \mso with the following predicates that express properties of sets
  \begin{align}\label{eq:set-model}
  \myunderbrace{X \subseteq Y}{set inclusion}
  \qquad \qquad
\myunderbrace{X \le Y}{
$x \le y$ holds for \\
\scriptsize every $x \in X$ and \\
\scriptsize every $y \in Y$}
\qquad \qquad
\myunderbrace{X \subseteq a}{
$a(x) $ holds for \\
\scriptsize every $x \in X$}.
  \end{align}
  The above predicates are second-order predicates in the sense that they express properties of sets; in contrast to the first-order predicates $x \le y$ and $a(x)$ which express properties of elements.
  Using the second-order predicates, we can eliminate the first-order variables: instead of quantifying over a position $x$, we can quantify over a set of positions $X$, and then say that this set is a singleton:
  \begin{align*}
  \myunderbrace{X \neq \emptyset}{$X$ is nonempty}
  \quad \myunderbrace{\forall Y \ Y \subseteq X \Rightarrow Y = \emptyset \lor Y =X}{and every proper subset of $X$ is empty}.
  \end{align*}
  Once elements are represented as singleton sets, the first-order predicates $x \le y$ and $a(x)$ can be simulated using  the second-order predicates from~\ref{eq:set-model}.

  Using the transformation described above, from now on we assume that \mso has only set variables, and it uses the second-order predicates from~\ref{eq:set-model}. For such formulas, we prove the lemma by induction on formula size. 

  \begin{itemize}
    \item \emph{Induction base.} In the induction base, we need to show that for every atomic formula as in~\eqref{eq:set-model}, its language is recognised by a finite monoid. Consider for example the formula $X \subseteq a$. The language of this formula consists of words over alphabet $\Sigma \times 2$ where  every position has a label that satisfies the following implication:
\begin{align*}
\text{second coordinate is $1$}\quad \Rightarrow \quad
\text{first coordinate is $a$}.
\end{align*}
This language is recognised by the homomorphism into the monoid 
\begin{align*}
(\set{0,1}, \min)
\end{align*}
which maps letters that satisfy the implication to $1$ and other letters to $0$.  Similar constructions can be done for the remaining predicates, in the case of $X \le Y$ the monoid is not going to be commutative.
\item {Boolean combinations.} For negation, the language of 
\begin{align*}
\neg \varphi(X_1,\ldots,X_n)
\end{align*}
is recognised by the same homomorphism as the language of $\varphi(X_1,\ldots,X_n)$, only the accepting set needs to be complemented. For conjunction
\begin{align*}
  \varphi_1(X_1,\ldots,X_n) \land \varphi_2(X_1,\ldots,X_n),
  \end{align*}
  one uses a product homomorphism 
  \begin{align*}
  (h_1,h_2) : (\Sigma \times 2^n)^* \to M_1 \times M_2 \quad \text{where $h_i : (\Sigma \times 2^n)^* \to M_i$ recognises $\varphi_i$},
  \end{align*}
  with the accepting set consisting of pairs that are accepting on both coordinates. Disjunction $\lor$ reduces to conjunction and negation using De Morgan's Laws.
  \item \emph{Set quantification.} By De Morgan's Laws, it is enough to consider existential set quantification
  \begin{align*}
 \exists X_n \ \varphi(X_1,\ldots,X_n) 
  \end{align*}
  The language of the quantified formula uses alphabet $\Sigma \times 2^{n-1}$. Let 
  \begin{align*}
  h : (\Sigma \times 2^{n})^* \to M
  \end{align*}
  be a homomorphism that recognises the language of the formula $\varphi(X_1,\ldots,X_n)$, which is obtained by induction assumption. To recognise the quantified formula, we  use a powerset construction. Define
  \begin{align*}
    \pi : (\Sigma \times 2^n)^* \to (\Sigma \times 2^{n-1})^*
    \end{align*}
    to be the letter-to-letter homomorphism which removes the last bit from every input position, and define 
  \begin{align*}
  H : (\Sigma \times 2^{n-1})^* \to \powerset M \qquad 
  H(w) = \set{ h(v) : \pi(v)=w}.
  \end{align*}
  It is not hard to see that the function $H$ is a homomorphism, with the monoid structure on the powerset 
$\powerset M$ defined by 
\begin{align*}
A \cdot B = \set{a \cdot b : a \in A, b \in B} \qquad \text{for }A,B \subseteq M.
\end{align*}
The powerset construction clearly preserves finiteness, although at the cost of an exponential blow up. The accepting set consists of those subsets of $M$ which have at least one accepting element. 
  \end{itemize}
\end{proof}
The construction in the above lemma is effective, which means that given a sentence of \mso, we can compute in finite time a recognising monoid homomorphism with an accepting set. Therefore, it is decidable if a sentence of \mso is true in at least one finite word: check if the image of the monoid homomorphism contains at least one accepting element.  

As mentioned before, 
the proof of the ``hard'' implication in the  \bueltr Theorem is very generic and will work without substantial changes in other settings, such as infinite words, trees or graphs. The ``easy part'' will become hard part in some generalisations -- e.g.~for some kinds of infinite words or for graphs -- because these generalisations lack a suitable automaton model.

\exercisehead

\mikexercise{\label{ex:u2-powerset-generates} Define $\Uu_2$ to be the monoid with elements $\set{a,b,1}$ and multiplication
\begin{align*}
      xy = \begin{cases}
    y & \text{if $x=1$}\\
    x & \text{otherwise}.
 \end{cases}
\end{align*}
Show that every finite monoid can be obtained from  $\Uu_2$ by applying Cartesian products, quotients (under semigroup congruences), sub-semigroups, and the powerset construction from Exercise~\ref{ex:powerset-semigroup}.
}{}

\mikexercise{\label{ex:prefix-structure}For an alphabet $\Sigma$, consider the model where the universe is the set $\Sigma^*$ of all finite words, and which is equipped with the following relations: 
\begin{align*}
  \myunderbrace{x \text{ is a prefix of }y}{binary relation} \qquad 
  \myunderbrace{\text{the last letter of $x$ is $a \in \Sigma$}}{one unary relation for each $a \in \Sigma$}
  \end{align*}
Show that a language $L \subseteq \Sigma^*$ is regular if and only if there is a first-order formula $\varphi(x)$ over the above vocabulary such that $L$ is exactly the words that satisfy $\varphi(x)$ in the above structure. 
}{}

\mikexercise{\label{ex:infix-structure} What happens if the prefix relation in Exercise~\ref{ex:prefix-structure} is replaced by the infix relation?}{Then we can describe accepting computations of Turing machines. 
}

\mikexercise{\label{ex:second-order-logic-not-congruent}Consider the fragment of second-order logic where one can quantify over: elements, unary relations, and binary relations. (This fragment is expressively complete.) Define $\equiv_k$ to be the equivalence on $\Sigma^*$ which identifies two words if they a satisfy the same sentences from the above fragment of second-order logic, up to quantifier rank $k$. Show that this equivalence relation has finite index, but it is not a semigroup congruence.}{}

\mikexercise{In the proof of the \bueltr Theorem, there was an exponential blowup incorred by every set set quantifier. Show that this is optimal, i.e.~for every $n$ there is an \mso formula with $\Oo(n)$ set quantifiers such that the smallest model of this formula is a word that has length which is a tower of $n$ exponentials.   }{}{}

\section{Aperiodic semigroups and  first-order logic}
\label{sec:schutz}
Having shown that \mso corresponds to all finite monoids, we  now begin the study of fragments of  \mso and the corresponding restrictions on finite monoids.  The first -- and arguably most important -- fragment is first-order logic. This fragment will be described in  the \schutz-McNaughton-Papert-Kamp Theorem. One part of the theorem says that  a language   is first-order definable if and only if it is recognised by a finite monoid  $M$ which  satisfies
\begin{align*}
\myunderbrace{a^\momega = a^{\momega}a \quad \text{for all $a \in M$}}{a monoid or semigroup which satisfies this is called \emph{aperiodic}},
\end{align*}
where $\momega \in \set{1,2,\ldots}$ is the idempotent exponent from  the Idempotent Power Lemma. 
In other words, in an aperiodic monoid the sequence $a, a^2, a^3, \ldots$ is eventually constant, as opposed to having some non-trivial periodic behaviour.
\begin{myexample}
    Consider the parity language $(aa)^* \subseteq a^*$.   We claim that this language is not recognised by any aperiodic monoid, and therefore it is not first-order definable. Of course the same is true for the complement of the language, namely the words of odd length which were discussed in Example~\ref{ex:parity-mso}.

    Suppose that the parity language is recognised by a homomorphism $h$ into some  finite monoid $M$.  By  Theorem~\ref{thm:syntactic-monoid} on syntactic monoids, there is a surjective homomorphism from the image of $h$, which is a sub-monoid of $M$, into the syntactic monoid. In other words, the syntactic monoid is a quotient (i.e.~image under a surjective homomorphism) of a sub-monoid of $M$. Since the syntactic monoid is the two-element group, which is not aperiodic, and since aperiodic monoids are closed under taking quotients and  sub-monoids, it follows that $M$ cannot be aperiodic. 

    The above argument shows  that  a  regular language is first-order definable if and only if its syntactic monoid  is  aperiodic. Since the syntactic monoid can be computed, and aperiodicity is clearly decidable, it follows that there is an algorithm which decides if a regular language is first-order definable. 
\end{myexample}

As can be guessed from the name,  the \schutz-McNaughton-Papert-Kamp Theorem is an amalgam of several results, which consider several formalisms. Apart from first-order logic and aperiodic monoids, these formalisms include linear temporal logic and star-free regular expressions, so we begin by defining those. 

\paragraph*{Linear temporal logic} Linear temporal logic\footnote{
  This logic, and the theorem about its expressive completeness for first-order logic, is due to 
\incite[Theorem II.1.]{Kamp68}
This theorem considers all words where the set of positions is a (possibly infinite) complete linear ordering, which covers the special case of finite words that is considered in this chapter.
} (\ltl) is an alternative to first-order logic which does not use quantifiers.  The logic \ltl only makes sense for structures equipped with a linear order; hence the name. 

\begin{definition}[Linear temporal logic]
  Let $\Sigma$ be a finite alphabet. Formulas of \emph{linear temporal logic} (\ltl) over $\Sigma$ are defined by the following grammar:
  \begin{align*}
  \myunderbrace{a \in \Sigma}{the current\\
  \scriptsize position has\\
  \scriptsize label $a$} \qquad \varphi \land \psi \qquad \varphi \lor \psi \qquad \neg \varphi \qquad \myunderbrace{\varphi \until \psi}{$\varphi$ until $\psi$}.
  \end{align*}
  The semantics for \ltl formulas is a ternary relation, denoted by 
\begin{align*}
\myunderbrace{w,}{word\\ \scriptsize in $\Sigma^+$} 
\myunderbrace{x}{position\\ \scriptsize in $w$} 
\models
\myunderbrace{\varphi}{\ltl formula},
\end{align*}
   which is defined as follows. A formula $a \in \Sigma$ is true in positions with label $a$. The semantics of Boolean combinations are defined as usual. For formulas of the form $\varphi \until \psi$, the semantics\footnote{We use a variant of the until operator which is sometimes called \emph{strict until}. Strict until is the variant that was originally used by Kamp, see~\cite[p. viii]{Kamp68}.
   } are 
    \begin{align*}
    w,x \models \varphi \until \psi \quad \eqdef \quad 
    \myunderbrace{\exists y\ x<y}
    {there is some \\
    \scriptsize position strictly\\
    \scriptsize after $x$}
    \land \ 
    \myunderbrace{w,y \models \psi}
    {which\\
    \scriptsize satisfies $\psi$} \land 
    \myunderbrace{
    \forall z \ x < z < y \ \Rightarrow w,z \models \varphi.}
    {and such that all intermediate \\
    \scriptsize positions satisfy $\varphi$}
    \end{align*}
  We say that an \ltl formula is true in a word, without specifying a position,   if the formula  is true in the first position of that word; this only makes sense for nonempty words. 
  A language $L \subseteq \Sigma^*$ is called \emph{\ltl definable} if there is an \ltl formula $\varphi$  that defines the language on nonempty words: 
  \begin{align*}
  w \in L \quad \text{iff} \quad 
  w \models \varphi \qquad \text{for every $w \in \Sigma^+$.}
  \end{align*}
  \end{definition}
For example, the formula
$a \until b$ defines the language $\Sigma a^* b \Sigma^*$. If we add the empty word to this language, then it is still defined by the formula $a \until b$, because the notion of \ltl definable language does not take into account the empty word.
\begin{myexample}
  To get a better feeling for \ltl, we discuss some extra operators that can be defined using until, and which will be used later in this chapter.
     We write $\bot$ for any vacuously false formula, such as~$a \land \neg a$. Likewise $\top$ denotes any vacuously true formula. Here are  some commonly used  extra operators:
    \begin{align*}
    \myunderbrace{\nextx \varphi \quad \eqdef \quad \bot \until \varphi,}{the next position satisfies $\varphi$} \hspace{1cm}
    \myunderbrace{\finally \varphi \quad \eqdef \quad \top \until \varphi,}{some strictly later position\\ \scriptsize satisfies $\varphi$}
    \hspace{1cm}
    \myunderbrace{\varphi \until^* \psi \quad \eqdef \quad \psi \lor (\varphi \until \psi).}{non-strict until}
    \end{align*}
    Similarly, we  define a non-strict version of the operator $\finally$, with $\finally^* \varphi = \varphi \lor \finally \varphi$.
    For example, the formula 
    \begin{align*}
    \finally^* (a \land \myunderbrace{\neg \finally \top}{last position})
    \end{align*}
    says that the last position in the word has label $a$. 
\end{myexample}

Almost by definition, every \ltl definable language is also first-order definable. Indeed, by unfolding the definition, one sees that for every \ltl formula there is  a first-order formula $\varphi(x)$ that is true in the same positions.

\paragraph*{Star-free languages.} We now present the final formalism that will appear in the \schutz-McNaughton-Papert-Kamp Theorem, namely star-free expressions\footnote{These were introduced in\incite[p. 190.]{Schutzenberger65}}. As the name implies, star-free expressions cannot use Kleene star. However,  in exchange they are allowed to use complementation (without star and complementation one could only define finite languages). For an alphabet $\Sigma$, the  star-free expressions are those that can be defined using the following operations on languages:
\begin{align*}
\myunderbrace{
  a \in \Sigma}
  {the language that\\
  \scriptsize contains only\\
  \scriptsize the word $a$} 
  \qquad  
  \myunderbrace{\emptyset}{empty \\ \scriptsize language} \qquad  \myunderbrace{L K}{concatenation} 
  \qquad \myunderbrace{L + K}{union} \qquad 
  \myunderbrace{\overline L.}{complementation \\ \scriptsize with respect \\ \scriptsize to $\Sigma^*$}
\end{align*}
Note that the alphabet needs to be specified to give meaning to the complementation operation. A language is called \emph{star-free} if it can be defined by a star-free expression. 
\begin{myexample}
    Assume that the alphabet is $\set{a,b}$. The expression $\bar \emptyset$ describes the full language $\set{a,b}^*$. Therefore 
    \begin{align*}
    \bar \emptyset \cdot a \cdot \bar \emptyset 
    \end{align*}
    describes all words with at least one $a$. Taking the complement of the above expression, we get a star-free expression for the language $b^*$. 
\end{myexample}

Like for \ltl formulas, almost by definition every star-free expression describes a first-order definable language. This is because to every star-free expression one can associate a first-order formula $\varphi(x,y)$ which selects a pair of positions $x \le y$ if and only if the corresponding infix (including $x$ and $y$) belongs to the language described by the  expression.

\paragraph*{Equivalence of the models.} The \schutz-McNaughton-Papert-Kamp  Theorem says that all of the formalisms discussed so far in this section are equivalent. 

\begin{theorem}[\schutz-McNaughton-Papert-Kamp]
  The following are equivalent\footnote{This theorem combines three equivalences.

    The equivalence of aperiodic monoids and star-free expressions was shown in \incite[p.~190.]{Schutzenberger65} The equivalence of star-free expressions and first-order logic was  shown in \incite[Theorem 10.5.]{McNaughtonPapert71}   The equivalence of first-order logic and \ltl, not just for finite words, was shown  in \incite[Theorem II.1.]{Kamp68}
  } for every  $L \subseteq \Sigma^*$:
  \begin{enumerate}
    \item \label{schutz:aperio} recognised by a finite  aperiodic monoid;
    \item \label{schutz:starfree} star-free;
    \item \label{schutz:fodef} first-order definable;
    \item \label{schutz:ltl} \ltl definable.
  \end{enumerate}
\end{theorem}

The rest of Section~\ref{sec:schutz} is devoted to proving the  theorem, according to the following plan:
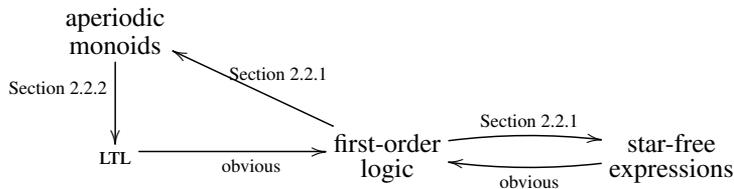
\begin{align*}
\xymatrix@C=2cm{
   \txt{aperiodic\\monoids} 
  \ar[d]_{\text{Section~\ref{sec:aperio-to-ltl}}} \\
  \txt{\ltl} 
  \ar[r]_{\text{obvious}}
   & 
  \txt{first-order\\logic}
  \ar@/^/[r]^{\text{Section~\ref{sec:fol-to-starfree}}}
  \ar[ul]_{\text{Section~\ref{sec:fol-to-aperio}}} &
  \txt{star-free\\ expressions}
  \ar@/^/[l]^{\text{obvious}}
}
\end{align*}

\subsection{From first-order logic to aperiodic monoids and star-free expressions}
\label{sec:fol-to-aperio}
\label{sec:fol-to-starfree}
In this section, we prove two inclusions: first-order logic is contained in both aperiodic monoids and star-free expressions. 

\paragraph*{\ef games.} In the proof, we  use \ef games, which are described as follows. An \ef  game is played by two players, called Spoiler and Duplicator.  A configuration of the game is a pair of words (one red and one blue), each one  with a $n$-tuple of distinguished word positions 
\begin{align*}
    \myunderbrace{\red w,}{in $\Sigma^*$}\myunderbrace{\red{x_1,\ldots,x_n}}{positions in $w$}
\qquad    
\myunderbrace{\blue w,}{in $\Sigma^*$}\myunderbrace{\blue{x_1,\ldots,x_n}}{positions in $w$}
\end{align*}
For such a configuration and $k \in \set{0,1,\ldots}$, the $k$-round game is played as follows.  If there is a quantifier-free formula that distinguishes the two sides (\red{red} and \blue{blue}), then Spoiler wins immediately and the game is stopped. Otherwise, the game continues as follows. If $k=0$, then Duplicator wins. If $k > 0$ then  Spoiler chooses one of the colours \red{red} or \blue{blue}, and a distinguished position $x_{n+1}$ in the word of the chosen colour. Duplicator responds with a matching distinguished position in the word of the other colour, and the game continues with $k-1$ rounds from the configuration with $n+1$  distinguished positions, which is obtained by adding the new distinguished positions. This completes the definition of \ef games.

The point of \ef games is that they characterise the expressive power of first-order logic, as stated in Theorem~\ref{thm:ehrenfeucht-fraisse} below. The   number of rounds in the games corresponds  to quantifier rank of a formula, which is the nesting depth of quantifiers, as illustrated in the following example
\begin{align*}
    \underbrace{\forall x\ \big(a(x)   \ \Rightarrow  \  \overbrace{(\exists y\ y < x \land b(y))}^{\text{quantifier rank 1}} 
    \land 
    \overbrace{(\exists y\ y > x \land b(y))}^{\text{quantifier rank 1}}
    \big)}_{
      \text{quantifier rank 2}
    }
  \end{align*}
The correspondence of logic and games is given in the following theorem:
\begin{theorem}\label{thm:ehrenfeucht-fraisse}
    For every configuration of the game and $k \in \set{0,1,\ldots}$, Duplicator has a winning strategy in the game if and only if the two sides of the configuration satisfy the same formulas of first-order logic with quantifier rank at most $k$.
\end{theorem}
\begin{proof}
    Straightforward induction on $k$.
\end{proof}

 For $k \in \set{0,1,2,\ldots}$ and $\red w, \blue w \in \Sigma^*$, we write 
\begin{align*}
\red w \foequiv k \blue w
\end{align*}
if the two words satisfy the same sentences of first-order logic with quantifier rank at most $k$, or equivalently, Duplicator has a winning strategy in the $k$-round game over the two words (with no distinguished positions).
The following lemma characterises  equivalence classes of $\foequiv {k+1}$ in terms  of equivalence classes of $\foequiv k$ by using only Boolean combinations and concatenation.  
\begin{lemma}\label{lem:star-free-ef} For every $k \in \set{0,1,\ldots}$ and finite alphabet $\Sigma$, the equivalence relation $\foequiv k$ on $\Sigma^*$ has finitely many equivalence classes. Furthermore, for every words $\red w, \blue w \in \Sigma^*$ we have  $\red w \foequiv {k+1} \blue w$ if and only if 
 \begin{align*}
   \red w \in LaK \iff \blue w \in LaK
 \end{align*}
 holds for every $a \in \Sigma$ and  every $L,K \subseteq \Sigma^*$ which are  equivalence classes of $\foequiv {k}$.
\end{lemma}
\begin{proof}
  Induction on $k$. 
The ``furthermore'' part immediately implies that there are finitely many equivalence classes, since there are finitely many choices for the letter $a$, and also finitely many choices for the equivalence classes $L,K$ thanks to the induction assumption. Note that the number of equivalence classes for $\equiv_{k+1}$ is exponential in the number of equivalence classes for $\equiv_k$. 

It remains to prove the ``furthermore'' part. 

    For the left-to-right implication, we observe that $LaK$ can be defined by a first-order sentence of quantifier rank $k+1$, which existentially quantifies over some  position $x$ with label $a$ and then checks (using quantifier rank $k$) that the part before $x$ belongs to $L$ and the part after $x$ belongs to $K$. Therefore, if $\red w$ and $\blue w$ satisfy the same sentences of quantifier rank $k+1$, they must belong to the same languages of the form $LaK$. 

    Consider now the right-to-left implication. Here it will be useful to consider variant of the \ef game, call it the \emph{local game}. Consider a configuration of the \ef game of the form 
    \begin{align*}
\red {w,x_1,\ldots,x_n}
    \qquad    \blue {w,x_1,\ldots,x_n}.
    \end{align*}
    where the red  distinguished positions are listed in strictly increasing order  $\red{x_1 < \cdots < x_n}$, and  the same is true for the blue positions. Let us partition  the positions $\red x$  of $\red w$ into the following $2n+1$ sets, some of which may be empty: 
    \begin{align}\label{eq:ef-partition}
    \myunderbrace{\red{X_0}}{$\red x < \red{x_1}$} \quad  
    \set{\red{x_1}}  \quad
    \myunderbrace{\red{X_1}}{$\red{x_1} < x < \red{x_2}$} 
    \quad 
    \set{\red{x_2}}  \quad
    \cdots
    \quad 
    \myunderbrace{\red{X_{n-1}}}{$\red{x_{n-1}} <\red  x < \red{x_n}$}  \quad 
    \set{\red{x_n}}  \quad
    \myunderbrace{\red{X_{n}}}{$ \red{x_n} < \red x$} \quad .
    \end{align}
    Similarly, we partition the positions in the blue word $\blue w$. We say that a strategy of player Spoiler is local if there is some $i \in \set{0,\ldots,n}$ such that all positions chosen by Spoiler in the strategy belong to $\red{X_i}$ (for positions in the red word $\red w$) or $\blue{X_i}$ (for positions in the blue word $\blue w$). 
    
    \begin{claim}
        If Spoiler has a winning strategy, then he also has a local one.
    \end{claim}
    \begin{proof}
        There is no benefit for Spoiler in using two different blocks of the partition described in~\eqref{eq:ef-partition}.
    \end{proof}
    A corollary of this claim is that if $n=1$, then Spoiler has a winning strategy in the $(k+1)$-round game for  the configuration 
    \begin{align*}
    \red{w, x_1}\qquad \blue{w,x_1}
    \end{align*}
    if and only if: (1) the  distinguished positions have different labels; or (2) Spoiler has a winning strategy in the $k$-round game for the parts strictly before the distinguished position; or (3) Spoiler has a winning strategy for the parts strictly after the distinguished position. This gives the right-to-left implication in the lemma.  
\end{proof}

We  use the lemma above to prove the inclusion of first-order logic in both star-free expressions and aperiodic monoids. 

\begin{description}
  \item[From first-order logic to star-free.] It is enough to show that every equivalence class of $\equiv k$ is star-free. This is proved by  induction on $k$.  For the induction base of $k=0$, there is only one equivalence class, namely all words, which is clearly a  star-free language. Consider now the induction step. 
  Consider an equivalence class $M$ of $\foequiv {k+1}$. Let 
   $X$ be  the set of triples 
  \begin{align*}
   \myunderbrace{L}{equivalence\\ \scriptsize class of $\foequiv k$}
    \qquad 
    \myunderbrace{a}{letter in $\Sigma$}
     \qquad 
    \myunderbrace{K}{equivalence\\ \scriptsize class of $\foequiv k$}.
  \end{align*}
  By 
   Lemma~\ref{lem:star-free-ef}, the  equivalence class $M$ is equal to  the following finite Boolean combination of concatenations
  \begin{align*}
  \bigcap_{
    \substack{(L,a,K) \in X\\
    M \subseteq LaK}
  } LaK 
  \qquad \cap \qquad
  \bigcap_{
    \substack{(L,a,K) \in X\\
    LaK \cap  M = \emptyset}
  } \overline{ LaK }.
  \end{align*}
  This is a star-free expression, if we assume that $L$ and $K$ are described by star-free expressions from the induction assumption. 
   Since every first-order definable language is a finite union of equivalence classes of $\foequiv k$ for some $k$, the result follows. 
   \item[From first-order logic to aperiodic monoids.] A corollary of Lemma~\ref{lem:star-free-ef} is the following compositionality property for first-order logic on words. 
   \begin{corollary}
       For every alphabet $\Sigma$ and $k \in \set{0,1,\ldots}$, the equivalence relation $\foequiv k$ on $\Sigma^*$ is a monoid congruence with finitely many equivalence classes. 
   \end{corollary}
   \begin{proof}
       Induction on $k$. To see that there are finitely many equivalence classes, we use Lemma~\ref{lem:star-free-ef}, which says that an equivalence class of $\equiv_{k+1}$ can be viewed as a set of triples (equivalence class of $\equiv_k$, letter from $\Sigma$, equivalence class of $\equiv_k$), and there are finitely many possible sets of such triples. We now show that $\equiv_{k+1}$ is a monoid congruence, i.e.
       \begin{align*}
       \red{w} \equiv_{k+1} \blue{w} \text{ and } \red{v} \equiv_{k+1} \blue v \qquad \text{implies} \qquad  \red{wv} \equiv_{k+1} \blue{wv}.
       \end{align*}
       By Lemma~\ref{lem:star-free-ef}, to prove the conclusion of the above implication, it is enough to show that $\red{wv}$ and $\blue{wv}$ belong to the same languages of the form $LaK$ as in the lemma. This follows immediately from the assumption of the implication, and the induction assumption of the lemma which  that $\equiv_{k}$ is a monoid congruence. (In the proof we also need the observation that $\equiv_{k+1}$ refines $\equiv_k$, which follows from the definition of $\equiv_k$.) 
   \end{proof} 
   By the above corollary, the function  $h_k$ which maps a word to its equivalence class under $\foequiv{ k}$ is a monoid homomorphism into a finite monoid. This homomorphism recognises every language that is defined by a first-order sentence of quantifier rank at most $k$, by definition of $\foequiv k$. Therefore, every first-order definable language is recognised by $h_k$ for some $k$.  It remains to show that the monoid used by such a homomorphism  is aperiodic. To prove this, we use Lemma~\ref{lem:star-free-ef} and a simple induction on $k$ to show that
   \begin{align*}
   w^{2^k-1} \foequiv k w^{2^k} \qquad  \text{for every  $w \in \Sigma^*$ and $k \in \set{1,2,\ldots}$.}
    \end{align*}
\end{description}

\subsection{From aperiodic monoids to \ltl}
\label{sec:aperio-to-ltl}

The last, and most important, step in the proof is constructing an \ltl formula based on an aperiodic monoid\footnote{The proof in this section is based on
\incite[Section 2]{wilke1999}
}. In this part of the proof, semigroups will be more convenient than monoids. We will use \ltl to define colourings, which are like languages but with possibly more than two values:  a function from $\Sigma^+$ to a finite set of colours is called \ltl definable if for every colour, the words sent that colour are an \ltl definable language. For example, a semigroup homomorphism into a finite semigroup is a colouring. 

\begin{lemma}
  \label{lem:wilke-homomorphism} Let $S$ be a finite aperiodic semigroup, and let $\Sigma \subseteq S$. The colouring 
  \begin{align*}
  w \in \Sigma^+ \qquad \mapsto \qquad \text{multiplication of $w$}
  \end{align*}
  is  \ltl definable.
  \end{lemma}
  By applying the lemma to the special case of $S$ being a monoid, and substituting each monoid element for the letters that get mapped to it in the recognising homomorphism, we immediately get the implication from finite aperiodic monoids to \ltl. 

It remains to prove the lemma. The  proof is by induction on two parameters: the size of the semigroup $S$, and the size of the subset $\Sigma$.  These parameters are ordered lexicographically, with the size of $S$ being more important. Without loss of generality, we assume that $\Sigma$ generates $S$, i.e.~every element of $S$ is the multiplication of some word in $\Sigma^+$. 

The induction base is treated in the following claim. 
\begin{claim}\label{claim:wilke-base}
  If either $S$ or $\Sigma$ has size one, then Lemma~\ref{lem:wilke-homomorphism} holds.
\end{claim}
\begin{proof}
  If the  semigroup has one element, there is nothing to do, since colourings with one possible colour are clearly \ltl definable. Consider the case when the $\Sigma$ contains only one element $a \in S$. By aperiodicity, the sequence
  \begin{align*}
  a, a^2, a^3, \ldots
  \end{align*}
  is eventually constant, because all powers bigger than the  threshold $\momega$ give the same result. The multiplication is therefore easily seen to be an  \ltl definable colouring, because for every $n \in \set{1,2,\ldots}$ the singleton language 
  \begin{align*}
    \set{a^n} \subseteq \set{a}^+
  \end{align*}
  is definable in \ltl. For example, when $n=3$, then the defining formula is
  \begin{align*}
  \myunderbrace{\finally \finally \top }{there are at least\\
  \scriptsize 3 positions} 
  \quad \land \quad 
  \myunderbrace{\neg (\finally \finally \finally \top) }{there are strictly \\
  \scriptsize  less than 4 positions}   .
\end{align*}
\end{proof}

We are left with the induction step. For $c \in S$, consider the function 
\begin{align*}
a \in S \mapsto ca \in S.
\end{align*}

\begin{claim}
  If $a \mapsto ca$ is a permutation of $S$, then it is the identity.
\end{claim}
\begin{proof}
  Suppose that $a \mapsto ca$ is a permutation of $S$, call it $\pi$. By aperiodicity, 
  \begin{align*}
    \pi^\momega \circ \pi = \pi^\momega.
  \end{align*}
  Since permutations form a group, we can multiply both sides by the inverse of $\pi^\momega$ and conclude that   $\pi$ is  the identity permutation.
\end{proof}
If the function $a \mapsto ca$ is the identity for every  $c \in \Sigma$, then the multiplication  of a word is the same as its last letter; and such a colouring is clearly \ltl definable.  We are left with the case when there is some $c \in \Sigma$ such that $a \mapsto ca$ is not the identity. Fix this $c$ for the rest of the proof.  
 Define $T$ to be the image of the function $a \mapsto ca$,  this is a proper subset of $S$ by assumption on $c$. 
\begin{claim}
  $T$ is a sub-semigroup of $S$. 
\end{claim}
\begin{proof}
  Multiplying two elements with prefix $c$ gives an element with prefix~$c$.
\end{proof}

In the rest of the proof, we use the following terminology for a word $w \in \Sigma^+$:
\mypic{30} 
We first describe the proof strategy.    
 For each black block, its multiplication can be computed in \ltl using the induction assumption on a smaller set of generators. The same is true for red blocks.  Define a \emph{red-black block} 
to be any union of a  red block  plus  the following (non-empty) black block; as illustrated below:
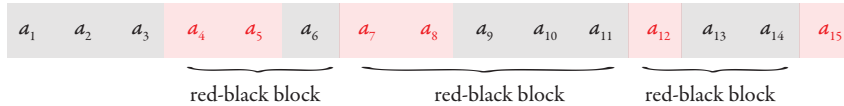
\mypic{31}
For red-black blocks,  the multiplication operation  can be computed in \ltl, by using multiplication for the  red and black blocks inside it. Also, for every red-black block, its multiplication is in $T$ because it begins with $c$ and has at least two letters.  Therefore, we can use the induction assumption on a smaller semigroup, to compute the multiplication of the union of all red-black blocks. Finally, the multiplication of the entire word is obtained by taking into account the blocks that are not part of any red-black block. 

The rest of this section is devoted to formalising the above proof sketch. In the formalisation, it will be convenient to reason with word-to-word functions. 
We say that  a function of type $\Sigma^* \to \Gamma^*$ is  an \label{page:ltl-transduction}  \emph{\ltl transduction} if it has the form 
\begin{align*}
a_1 \cdots a_n \in \Sigma^* \qquad \mapsto \qquad f(a_1 \cdots a_n)  f(a_2 \cdots a_n) \cdots f(a_{n})
\end{align*}
for some \ltl definable colouring $f : \Sigma^+ \to \Gamma + \varepsilon$. Under this definition, the output length is at most the input length for \ltl transductions.  By substituting formulas, one easily\footnote{The result would also hold for generalisation of \ltl transductions where the colouring $f$ has type $\Sigma^+ \to \Gamma^*$, but the proof is easier when the type is $\Sigma^+ \to \Gamma+\varepsilon$, and only the latter case is needed here.} shows the following composition properties: 
\begin{eqnarray*}
  \text{(\ltl colourings)} \circ \text{(\ltl transductions)}  &\subseteq& \text{\ltl colourings}\\
  \text{(\ltl transductions)} \circ \text{(\ltl transductions)}  &\subseteq& \text{\ltl transductions.}
\end{eqnarray*}

We use \ltl transductions to decorate an input word $w \in \Sigma^+$ with extra information that will serve towards computing its multiplication. 

\begin{enumerate}
  \item For each position that precedes a block (i.e.~the next position begins a new block), write  in that position the value of the next block.   For the remaining positions, do not write anything. Use two disjoint copies of $S$ to distinguish the values of the red and black blocks. Here is a picture:
    \mypic{32}
    In the above picture, $a_{i,j}$ denotes the multiplication of the infix $\set{i,\ldots,j}$. The function described in this step is an \ltl transduction, thanks to the induction assumption on smaller alphabets\footnote{To make this formal, we need a simple closure property of  \ltl that is described in Exercise~\ref{ex:run-ltl-on-prefix}. }. 
\item  Take the output of the function in the previous step, and for each red letter (the multiplication of a red block), multiply it with the next letter (which is the multiplication of a black block). As a result, we get the values of all red-black blocks which do not begin in the first position.   Here is a  picture:
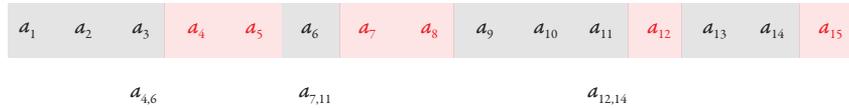
\mypic{33}
The function in this step is clearly an \ltl transduction. 
\end{enumerate}
 By induction assumption on a smaller semigroup, the multiplication operation $T^+ \to T$ is an \ltl colouring. By composing the functions described above with the semigroup multiplication in $T$, we see that 
\begin{align*}
w \in \Sigma^+ \quad \mapsto \quad \text{value of the union of red-black blocks}
\end{align*}
is an \ltl colouring. The values of the (at most two) blocks that do not participate in above union can also be computed using  \ltl colourings, and therefore the multiplication of the entire word can be computed.

\exercisehead

\mikexercise{\label{ex:three-variable} Show that for every sentence of first-order logic, there is a sentence that is equivalent on finite words, and which uses at most three variables (but these variables can be repeatedly quantified).}{
  Translate a first-order sentence to \ltl, and then back again to first-order logic. 
}

\mikexercise{\label{ex:h-trivial}
  Show that the following  are equivalent for a finite semigroup:
  \begin{enumerate}
    \item \label{aperiodic:power} aperiodic;
    \item \label{aperiodic:h-trivial} $\Hh$-trivial, which means that all $\Hh$-classes are singletons;
    \item \label{aperiodic:no-groups} no sub-semigroup is a non-trivial group.
  \end{enumerate}
}
{
\begin{description}
      \item[\ref{aperiodic:power} $\Rightarrow$ \ref{aperiodic:h-trivial}] Suppose that $a,b$ are in the same $\Hh$-class. This means that there exist $x,y$ such that $a=xb$ and $b=ay$. We have therefore 
      \begin{align*}
      a = x^n a y^n \qquad \text{for every $n \in \set{1,2,\ldots}$.}
      \end{align*}
      By~\ref{aperiodic:power}, there must be some $n$ such that $y^{n}=y^{n+1}$ and therefore
      \begin{align*}
      a = x^n a y^n = x^n a y^{n+1} = ay = b.
      \end{align*}
      \item[\ref{aperiodic:h-trivial} $\Rightarrow$ \ref{aperiodic:no-groups}] Every group is contained in some $\Hh$-class, and therefore all groups must be trivial.
      \item[\ref{aperiodic:no-groups} $\Rightarrow$ \ref{aperiodic:power}] Let $s \in S$, and choose $n$ so that $s^n$ is idempotent. It follows that 
      \begin{align*}
      s^{n}, s^{n+1}, \ldots, s^{n+n} = s^n
      \end{align*}
      are all in the same $\Hh$-class, and this $\Hh$-class contains an idempotent. Therefore, this $\Hh$-class is a group by the $\Hh$-class Lemma, and hence it must be trivial. This proves that $s^n= s^{n+1}$.
    \end{description}
}

\mikexercise{Consider the successor model of a word $w \in \Sigma^*$, which is defined like the ordered model, except that instead of $x < y$ we have $x+1=y$. Give an example of a regular language that is first-order definable using the ordered model, but not using the successor model.}
  {Consider the language 
  \begin{align*}
    a^* b a^* c a^* \subseteq \set{a,b,c}^*.
  \end{align*}
  By induction on $k$, one can show that  Duplicator wins the $k$-round \ef game over the successor models of 
  \begin{align*}
  \myunderbrace{a^{2^k}ba^{2^k}ca^{2^k}}{in the language}\qquad 
  \myunderbrace{a^{2^k}ca^{2^k}ba^{2^k}.}{not in the language}
  \end{align*}
}

\mikexercise{Show two languages which have the same syntactic monoid, and such that only one of them is first-order definable in the successor model. In particular, one of the closure properties from Exercise~\ref{ex:half-eilenberg-star} must fail for this logic.} 
{The languages are 
\begin{align*}
\myunderbrace{abaca}
{first-order definable\\
\scriptsize in the successor model} 
\hspace{2cm}
\myunderbrace{a^*ba^*ca^*}
{not first-order definable\\
\scriptsize in the successor model}
\end{align*}
}

\mikexercise{Let $\Sigma$ be a finite alphbet and let $\vdash, \dashv$ be fresh symbols. For $k,\ell \in \set{0,1,\ldots}$, we say that $w, w' \in \Sigma^*$ are  $(k,\ell)$-locally equivalent if 
\vspace{0.3cm}
\begin{align*}
\txt{
  $\vdash w \dashv$ has at least $i$\\
   occurrences of infix $v$
}
\qquad \text{iff} \qquad 
\txt{
   $\vdash w' \dashv$ has at least $i$\\
   occurrences of infix $v$
}
\end{align*}
\vspace{0.1cm}

\noindent holds  for every $i \in \set{0,\ldots,k}$ and every $v \in \Sigma^*$ of length at most $\ell$. Show that $L \subseteq \Sigma^*$ is first-order definable in the successor model if and only if it is a union of equivalence classes of $(k,\ell)$-local equivalence, for some $k,\ell$.
}{}

\mikexercise{\label{ex:run-ltl-on-prefix}  Let $\Gamma \subseteq \Sigma$ and let   $L \subseteq \Gamma^*$. If $L$ is  definable in \ltl, then the same is true for 
\begin{align*}
\set{w \in \Sigma^* : \text{$L$ contains the maximal prefix of $w$ which uses only letters from $\Sigma$}}.
\end{align*}
}{}


\mikexercise{\label{ex:definite-ltlx}
Consider  \ltlx, i.e.~the fragment of \ltl where the only operator is $\nextx$. Show that this fragment is equal to the definite languages from Exercise~\ref{ex:definite}.
}{}

\mikexercise{\label{ex:successor-idempotent-swap}
Show that if   a language  is first-order definable in the successor model, then the  syntactic semigroup  satisfies the following equality
\begin{align*}
eafbecf = ecfbeaf \qquad \text{for all }
\myunderbrace{e,f,}{idempotents}
a,b,c.
\end{align*}
}{}

\mikexercise{
  \label{ex:successor-more-than-swap} Show that the identity in Exercise~\ref{ex:successor-idempotent-swap}, together with aperiodicity, is equivalent to first-order definability in the successor model.
}{Any commutative language that is not first-order definable, e.g.~words of even length. However, if the semigroup satisfies the equality and is aperiodic, then we get a necessary and equivalent condition.    }

\mikexercise{Consider the following extension of \ltl with group operators.  Suppose that $G$ is a  finite group, and let
\begin{align*}
  \set{\varphi_g}_{g \in G},
\end{align*}
be a family of already defined formulas  such that every position in an input word is selected by  exactly one formula $\varphi_g$. Then we can create a new formula, which  is true in a word of length $n$ if 
\begin{align*}
1 = g_1 \cdots g_n,
\end{align*}
where  $g_i \in G$ is the unique group element whose corresponding formula selects  position $i$. Show that this logic defines all regular languages.
}{}

\section{Suffix trivial semigroups and temporal logic with \texorpdfstring{$\finally$}{F}  only}
\label{sec:finally-only}
In the previous section, we showed that first-order logic corresponds to the monoids without groups, which is the same thing as finite monoids with trivial $\Hh$-classes (Exercise~\ref{ex:h-trivial}). What about monoids with trivial suffix classes, prefix classes, or infix classes? Trivial infix classes will be described in Section~\ref{sec:piecewise}. In this section, we give a logical characterisation of trivial suffix classes. A symmetric statement holds for trivial prefix classes.

In the characterisation,  we use the fragment of \ltl where until is replaced by the following operators
\begin{align*}
\underbrace{\top \until \varphi}_{\finally \varphi} \qquad
\underbrace{\neg \finally \neg  \varphi}_{\globally \varphi} \qquad 
\underbrace{\varphi \lor \finally \varphi}_{\finally^* \varphi}\qquad 
\underbrace{\neg \finally^* \neg  \varphi}_{\globally^* \varphi}.
\end{align*}
Since all of the above operators can be defined in terms of $\finally$, we write \ltlf for the resulting logic.

\begin{theorem}\label{thm:ltl-f}\footnote{This theorem is based on 
  \incite[Theorem 6.1]{etessamiWilke2000}
  \incite[Theorem 4.2]{CohenPerrinPin93}
  The result itself is taken from~\cite{etessamiWilke2000}. However, the use of Green's relations in the proof is more in the spirit of~\cite{CohenPerrinPin93}, which considers an stronger logic that is obtained from \ltlf by adding a ``next'' operator.
  }
  The following conditions are equivalent for  $L \subseteq \Sigma^*$:
  \begin{enumerate}
    \item \label{suffix:suffix} is recognised by a finite suffix trivial monoid;
    \item \label{suffix:expression} is defined by a finite union of regular expressions of the form
    \begin{align*}
    \myunderbrace{\Sigma_0^* a_1 \Sigma_1^* a_2 \cdots a_n \Sigma_n^* \qquad \text{where $a_i \in \Sigma - \Sigma_i$ for $i \in \set{1,\ldots,n}$;}}{We call such an expression  \emph{suffix unambiguous}.\\
    \scriptsize A set $\Sigma_i \subseteq \Sigma$ is allowed to be empty, in which case $\Sigma_i^* = \set \varepsilon$.}
    \end{align*}
    
    \item \label{suffix:logic} is defined by a Boolean combination of \emph{\ltlf} formulas of the form  $\finally^* \varphi$.
  \end{enumerate}
\end{theorem}

To see why the formulas in item~\ref{suffix:logic} need to be guarded by $\finally^*$, consider the \ltlf formula $a$ which defines the language ``words beginning with $a$''. This language is not recognised by any finite suffix trivial monoid. 



\begin{proof} \ 
  \begin{description}
  \item[\ref{suffix:suffix}$\Rightarrow$ \ref{suffix:expression}]  We will show that for every finite suffix trivial monoid $M$, and every $F \subseteq  M$, the language
  \begin{align*}
  \set{w \in M^* : \text{the multiplication of $w$ is in $F$}}
  \end{align*}
  is defined by a finite union of suffix unambiguous expressions. It will follow that for every monoid homomorphism into $M$, the recognised language  is defined by a similar expression, with monoid elements substituted by the letters that map to them (such a substitution preserves suffix unambiguity). 
  
  Since our target  expressions are closed under finite unions,
  it is of course enough to consider the case when $F$ contains only one element, call it $a \in M$. The proof is  by induction on the position of $a$ in the  suffix ordering. 

  The induction base is when $a$ is a suffix of every monoid element, which means that $a$ is a suffix of the monoid identity. By Exercise~\ref{ex:monoid-identity-infix-class}, the infix class of the identity is a group, and a group must be trivial in a suffix trivial monoid. It follows that a word multiplies to $a$ if and only if it belongs to $a^*$, which is a suffix unambiguous expression.

   We now prove  the induction step. Consider a word  that multiplies to $a$. This word must be nonempty, since otherwise it  would multiply to the identity. Let $i$ be the maximal position in the word such that the suffix starting in $i$ also multiplies to $a$. By suffix triviality, every position $<i$ is labelled by a letter in 
   \begin{align*}
   \Sigma_0  = \set{b \in M : ba = a}.
   \end{align*}
   Let $b$ be the multiplication of the suffix that starts after $i$, not including $i$, and let $c$ be the label of position $i$. By choice of $i$, $b$ is a proper suffix of $a$ and $a=cb$.  Summing up, words that multiply to $a$ are defined by the expression
\begin{align*}
\bigcup_{
  \substack{b,c \in M\\
  \text{$b$ is a proper suffix of $a$}\\
  a=cb
}} \Sigma_0^* c \cdot  \text{(words that multiply to $b$)},
\end{align*}
Apply the induction assumption to $b$, yielding a finite union of suffix unambiguous expressions, and distribute the finite union across concatenation. It remains to justify that the resulting expressions are also suffix unambiguous. This is because none of the expressions that define words that multiply to $b$ can begin with $\Sigma_1^*$ with $c \in \Sigma_1$, since otherwise we would  contradict the assumption that $cb = a \neq b$. 

\item[\ref{suffix:expression} $\Rightarrow$ \ref{suffix:logic}] Since the formulas from item~\ref{suffix:logic}  are closed under union, it is enough to show that every suffix unambiguous expression 
\begin{align*}
  \Sigma_0^* a_1 \Sigma_1^* a_2 \cdots a_n \Sigma_n^*
\end{align*}
can be defined by a formula as in~\ref{suffix:logic}.  For $i \in \set{0,\ldots,n}$, define $L_i$ to be the suffix of the above expression that begins with $\Sigma_i^*$. By induction on $i$, starting with $n$ and progressing down to $0$, we show that $L_i$ can be defined by a formula $\varphi_i$ as in item~\ref{suffix:logic}.
In the induction base, we use the formula
\begin{align*}
\varphi_n =  \myunderbrace{  \globally^* \bigvee_{a  \in  \Sigma_n} a}{all positions have label in $\Sigma_n$}.
\end{align*}
For the induction step, we first define the language $a_i L_i$, using a formula of \ltlf (which is not in the shape from item~\ref{suffix:logic}):
\begin{align*}
\psi_i =  a_i \land (\finally \varphi_i) \land  \globally \bigvee_{j > i} \varphi_j.
\end{align*}
 Because the expression is suffix unambiguous, the formula $\psi_i$ selects at most one position in a given input word; this property will be used below.
The language $L_{i-1}$ is then defined by 
\begin{align*}
\varphi_{i-1}= \qquad \finally^* \psi_i \ \land \ 
\myunderbrace{
\globally^* ((\finally \psi_i) \Rightarrow \bigwedge_{a \in \Sigma_0} a)}
{if a position is to the left\\
\scriptsize of the unique position \\
\scriptsize satisfying $\psi_i$, then \\
\scriptsize it has label in $\Sigma_0$.}
\end{align*}
\item[\ref{suffix:logic} $\Rightarrow$ \ref{suffix:suffix}] Define the rank of a formula in \ltlf to be the nesting depth of the operator $\finally$. (We assume here that $\finally$ is the only temporal operator used in the formula, and the remaining operators such as $\globally$ or $\finally^*$ are replaced by their definitions using $\finally$.) For $k \in \set{0,1,\ldots}$, define $\approx_k$ to be the equivalence relation on $\Sigma^+$ which identifies two words if they satisfy the same formulas of  rank at most $k$.
  The key observation is the following pumping lemma.

  \begin{claim}
    For every  $k \in \set{0,1,2,\ldots}$  we have 
    \begin{align*}
      w (xy)^{i} u \quad  \approx_k \quad   w y (xy)^{j} u  
      \qquad \text{for every $w \in \Sigma^+,x,y,u \in \Sigma^*$ and $i,j \ge k$.}
      \end{align*}    
  \end{claim}
  \begin{proof} 
    Induction on $k$. For $k=0$, we observe that  the equivalence class under $\approx_0$ depends only on the first letter, and the two words on both sides in the claim have the same letter because $w$ is nonempty. 
    
    Consider now the induction step, when going from $k$ to $k+1$. By unravelling the definition of $\approx_{k+1}$, we need to show that if $i,j \ge k+1$, then words on both sides of the equivalence 
    \begin{align*}
      w (xy)^{i} u \quad  \approx_{k+1}\quad   w y (xy)^{j} u 
    \end{align*}
    have the same first letter, and for every nonempty proper suffix of a word on one side of the equivalence, there 
    there is a nonempty proper suffix on the other side of the equivalence, such that the two suffixes are equivalent under $\approx_k$.  Clearly the first letters are the same, because $w$ is nonempty. Consider now the suffixes. Suppose first that  $v$ is a nonempty proper suffix of the left side. If $v$ is a suffix of $(xy)^{k+1}u$, then the same $v$ is a suffix of the right side. Otherwise, we can use the induction assumption.  Consider now a nonempty proper suffix $v$ of the right side.  Here we argue in the same way as previously, except that there is one extra case, when 
    \begin{align*}
    v = z(xy)^{i}u \qquad \text{for some $z$ that is a suffix of $y$.}
    \end{align*}
    In this case, the $\approx_k$-equivalent suffix on the left side is 
    $z(xy)^k u$.
  \end{proof}
  
  By unravelling the definition of the syntactic monoid, in terms of two-sided congruences, we infer from the above claim  that for every rank $k$ formula $\varphi$ of \ltlf, the syntactic monoid $M$ of $\finally^* \varphi$  satisfies 
  \begin{align}\label{eq:suffix-trivial-identity}
    (xy)^\momega = y(xy)^\momega \qquad \text{for all $x,y \in M$}.
  \end{align}
  The same is also true for syntactic monoids of Boolean combinations of such formulas.  
  To finish the proof, we observe that  property~\eqref{eq:suffix-trivial-identity} is true in a finite monoid if and only if it is suffix trivial. Indeed, if a monoid is suffix trivial, then   $(xy)^\momega$ and $y(xy)^\momega$ must be in the same suffix class, and hence equal. Conversely, if $a,b$ are in the same suffix class, then there must be some $x,y$ such that $b=xa$ and $a=yb$; it follows that 
  \begin{align*}
   a = y(xy)^\momega b \stackrel{\text{\eqref{eq:suffix-trivial-identity}}} =  (xy)^\momega b = b.
  \end{align*}
\end{description} 
\end{proof}

\exercisehead
\mikexercise{\label{ex:ltlf-fresh-letter}
  Let $\Sigma$ be an alphabet and let $c \not \in \Sigma$ be a fresh letter. Show that $L \subseteq \Sigma^+$ satisfies the conditions of Theorem~\ref{thm:ltl-f} if and only $cL$ is definable in \ltlf.
}{}
\section{Infix trivial semigroups and piecewise testable languages}
\label{sec:piecewise}
Having discussed monoids that are $\Hh$-trivial,  prefix-trivial and suffix-trivial  in the previous sections, we turn to finite monoids that are infix-trivial.  For languages recognised by finite infix trivial monoids, a prominent role will be played embeddings of words (also known as the Higman ordering).

\begin{definition}[Embedding] We say that a word $w \in \Sigma^*$ \emph{embeds} in a word $v \in \Sigma^*$, denoted by $w \higman v$,  if there is an injective function from positions in $w$ to positions in $v$, which preserves the order on positions and the labels. 
\end{definition}
In other words, $w$ embeds in $v$ if and only if  $w$ can be obtained from $v$ by removing zero or more positions.  For example ``ape'' embeds into ``example''. It is easy to see that embedding is an ordering on words: it is reflexive, transitive and anti-symmetric (although it will cease to be anti-symmetric for infinite words).
We say that a language $L \subseteq \Sigma^*$ is \emph{upward closed} if 
\begin{align*}
      v \higman w \land  v \in L \Rightarrow w \in L.
      \end{align*}
Symmetrically, we define downward closed languages. The main result  
about embedding is that it is a well-quasi order, as explained in the following lemma. 

\begin{lemma}[Higman's Lemma]
  For every upward closed  $L \subseteq \Sigma^*$ there is a finite subset $U \subseteq L$ such that
  \begin{align*}
  L = \myunderbrace{\set{w \in \Sigma^*: v \higman w \text{ for some $v \in U$} }}{we call this the \emph{upward closure} of $U$}.
  \end{align*}
\end{lemma}
\begin{proof}
  Consider the set of minimal elements in $L$, i.e.~the set
  \begin{align*}
  U = \set{ w \in L : \text{there is no $v \in L$ such that $v \higman w$ and $v \neq w$}}.
  \end{align*}
  Because the embedding ordering is well-founded (there are no infinite decreasing chains) it follows that $L$ is equal to the upward closure of its minimal elements $U$. By definition, $U$ is an \emph{antichain},  which means that every two elements of $U$ are incomparable with respect to embedding. Therefore, to prove the lemma it remains to show that  antichains are finite.

  \begin{claim}
    There is no infinite antichain with respect to embedding.
  \end{claim}
  \begin{proof}
    Define a \emph{growth} in a finite or infinite sequence $w_1,w_2,\ldots$ to be a pair of indices $i < j$ such that $w_i \higman w_j$. We will show that every infinite sequence contains at least one growth. This implies that there cannot be any infinite antichains. 
    
    Suppose, toward a contradiction, that there is an infinite sequence of words without growths. 
Define the radix ordering on finite words as follows: shorter words come before longer ones, and same length words are ordered lexicographically.    Define a sequence $w_1,w_2,\ldots$ by induction as follows. The word $w_1$ is the least word, in the radix ordering (which is well-founded, so it makes sense to talk about least words), which can be extended to an infinite sequence without growths. For $n > 1$, define $w_n$ to be the least word in the radix ordering such that $w_1,\ldots,w_n$ can be extended to an infinite sequence without growths, in particular $w_1,\ldots,w_n$ has no growths. Sequences without growths are closed under limits, and therefore $w_1,w_2,\ldots$ has no growths.

Consider the sequence $w_1,w_2,\ldots$ defined in the previous paragraph. Because the alphabet is finite, there must be some letter, call it $a$, such that infinitely many words in the sequence, say with indexes $n_1 < n_2 < \cdots$, begin with the letter $a$. Define a new sequence as follows:
\begin{align*}
w_1,\ldots,w_{n_1 -1 }, \myunderbrace{\red{w_{n_1}},\red{w_{n_2}},\ldots}{the word $\red{w_n}$ is
\\ \scriptsize is obtained from \\\scriptsize the word $w_n$ by \\ \scriptsize removing the \\ \scriptsize first letter} .
\end{align*}
Since $\red{w_{n_1}}$ is shorter than $w_{n_1}$, it follows from the construction in the previous paragraph that  the above sequence must have some growth. However, it is easy to see that any growth in the above sequence would also translate to some growth in the sequence from the previous paragraph, hence a contradiction. 
  \end{proof}
\end{proof}

Here is a logical corollary of  Higman's lemma. 

\begin{theorem}
  A language is upward closed if and only if it can be defined in the ordered model by an $\exists^*$-sentence, i.e.~a sentence of the form
  \begin{align*}
    \myunderbrace{\exists x_1 \ \exists x_2 \ \cdots \exists x_n}{only existential quantifiers}  \quad
    \myunderbrace{\varphi(x_1,\ldots,x_n)}{quantifier-free}.
    \end{align*} 
\end{theorem}
\begin{proof}
  Clearly every $\exists^*$-sentence defines an upward closed language.  Higman's Lemma gives the converse implication, because the upward closure of every finite set is definable by an $\exists^*$-sentence.   
\end{proof}

Embeddings will also play an important role in the characterisation of languages recognised by monoids that are infix trivial. Before stating the characterisation, we introduce one more definition, namely zigzags.  For languages $L,K \subseteq \Sigma^*$, define  a \emph{zigzag between $L$ and $K$} to be a sequence
\begin{align*}
 \underbrace{w_1}_{\in L} \higman 
 \underbrace{w_2}_{\in K} \higman
 \underbrace{w_3}_{\in L} \higman 
 \underbrace{w_4}_{\in K} \higman
 \underbrace{w_5}_{\in L} \higman 
 \underbrace{w_6}_{\in K} \higman  \cdots.
\end{align*}
In other words, this is a sequence that is  growing with respect to embeddings, and such that odd-numbered elements are in $L$ and even-numbered elements are in $K$. The zigzag does not need to be strictly growing, but it will be if  $L$ and $K$ are disjoint.

We are now ready for the characterisation of infix trivial monoids.

\begin{theorem}\label{thm:piecewise}
  The following conditions are equivalent\footnote{
    Equivalence of items~\ref{jsimon-infix} and~\ref{jsimon-piecewise} was first proved in 
    \incite[p.~220.]{Simon75}
    Equivalence of items~\ref{jsimon-piecewise} and~\ref{jsimon-zigzag} was first proved in 
    \incite[Theorem 3.]{czewinskietal2013}
  } for every $L \subseteq \Sigma^*$:
  \begin{enumerate}
    \item \label{jsimon-infix} recognised by a finite monoid that is infix trivial;
    \item \label{jsimon-piecewise} is a finite Boolean combination of upward closed languages;
    \item \label{jsimon-zigzag} there is no infinite zigzag between $L$ and its complement.
  \end{enumerate}
\end{theorem}

We use the name \emph{piecewise testable} for languages as in item~\ref{jsimon-piecewise} of the above theorem.
Equivalence\footnote{Both conditions~\ref{jsimon-zigzag} and~\ref{jsimon-infix}  can be checked by  algorithms. For~\ref{jsimon-infix} this is immediate, while condition~\ref{jsimon-zigzag} is discussed in Exercise~\ref{ex:separate-piecewise}. Therefore, condition~\ref{jsimon-infix} would not be useful for a hypothetical person that  only cares about deciding if a regular language is piecewise testable.  } of items~\ref{jsimon-piecewise} and~\ref{jsimon-zigzag} is a corollary of the following lemma, when applied to $K = \Sigma^* - L$.

\begin{lemma}[Zigzag Lemma]
  \label{lem:zigzag} Let $L,K \subseteq \Sigma^*$. The following are equivalent:
  \begin{enumerate}
    \item \label{zigzag-finite-zigzags} there are zigzags between $L$ and $K$  of every finite length;
    \item \label{zigzag-infinite-zigzag} there is an infinite zigzag between $L$ and $K$;
    \item \label{zigzag-finite-separation} there is no piecewise testable  language $M \subseteq \Sigma^*$ such that 
    \begin{align*}
    \myunderbrace{L \subseteq M \quad \text{and} \quad  M \cap K = \emptyset.}{we say that $M$  separates $L$ and $K$}
    \end{align*}
  \end{enumerate}
\end{lemma}
\begin{proof}
  \ 
  \begin{description}
    \item[\ref{zigzag-finite-zigzags}$\Rightarrow$\ref{zigzag-infinite-zigzag}] Assume that zigzags between $L$ and $K$ can have  arbitrarily long finite lengths. Define a directed acyclic graph $G$ as follows. Vertices are words in $L$, and there is an edge $w \to v$ if 
    \begin{align*}
    w \higman u \higman v \qquad \text{for some $u \in K$.}
    \end{align*}
    For a vertex $v \in L$ of this graph, define its \emph{potential} 
    \begin{align*}
      \alpha(v) \in \set{0,1,\ldots,\omega}
    \end{align*}
    to be the least upper bound on the lengths of paths in the graph that start in $v$. This can be either a finite number, or $\omega$ if the paths have unbounded length.  
    
    We first show that some vertex must have potential $\omega$. By assumption on arbitrarily long zigzags, potentials have arbitrarily high values. By definition of the graph, $\alpha$ is monotone with respect to  (the opposite of the) embedding, in the following sense:
    \begin{align*}
    v \higman v' \quad \text{implies} \quad \alpha(v) \ge \alpha(v') \qquad \text{for every $v,v' \in L$}.
    \end{align*}
    By Higman's Lemma, the language  $L$, like any set of words, has finitely many minimal elements with respect to embedding. By monotonicity, one of these minimal words must therefore have potential $\omega$.  
    
    For the same reason as above, if a word has potential $\omega$, then one of its successors (words reachable in one step in the graph) must also have potential $\omega$; this is because there are finitely many successors that are minimal with respect to embedding. This way, we can construct an infinite path in the graph which only sees potential $\omega$, using the same reasoning as  in the proof of  \konig's Lemma. 

    \item[\ref{zigzag-infinite-zigzag}$\Rightarrow$\ref{zigzag-finite-separation}]Suppose that there is a zigzag between $L$ and $K$ of infinite length. Every  upward closed  set selects either no elements of the zigzag, or all but finitely many elements of the zigzag. It follows that every finite Boolean combination of upward closed sets must contain, or be disjoint with, two consecutive elements of the zigzag. Therefore, such a Boolean combination cannot  separate $L$ from $K$. 
    \item[\ref{zigzag-finite-separation}$\Rightarrow$\ref{zigzag-finite-zigzags}] We prove the contra-positive: if zigzags between $L$ and $K$ have bounded length, then $L$ and $K$ can be separated by a piecewise testable language. For  $w \in L$ define its \emph{potential} to be the maximal length of a zigzag between $L$ and $K$  that starts in $w$; likewise we define the potential for $w \in K$, but using  zigzags between $K$ and $L$.  Define $L_i \subseteq L$ to be the words in $L$ with potential exactly $i \in \set{1,2,\ldots}$, likewise define $K_i \subseteq K$.  Our assumption is that the potential is bounded, and therefore $L$ is a finite union of the languages $L_i$, likewise for $K$. 
    By induction on $i \in \set{0,1,\ldots}$, we will show that the languages
    \begin{align*}
    \underbrace{L_1 \cup \cdots \cup L_i}_{L_{\le i}} \qquad \text{and} \qquad 
    \underbrace{K_1 \cup \cdots \cup K_i}_{K_{\le i}}.
    \end{align*}
    can be separated by a piecewise testable language, call it $M_i$. 
In the induction base, both languages are empty, and can therefore be separated by the empty language, which is clearly piecewise testable. 
Consider  the induction step, where we go from $i-1$ to $i$. We write $L_{<i}$ instead of $L_{\le i-1}$. We will use the following sets
\begin{itemize}
  \item the upward closure of $L_{\le i}$;
  \item the downward closure of $L_i$;
  \item a piecewise testable set $M$ that contains $K_{<i}$ and is disjoint with $L_{<i}$.
\end{itemize}
The first two sets are piecewise testable because they are upward or downward closed, and the third set is obtained from the induction assumption.  These sets  are  depicted in the following picture, with $i=3$:
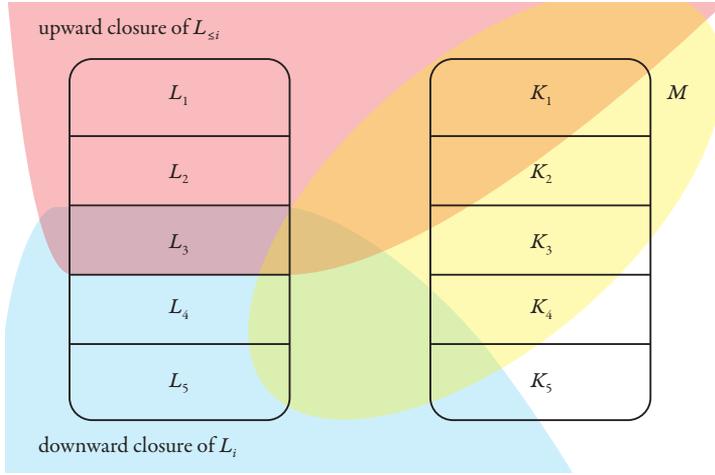
\mypic{36}
 The separator $M$ from the induction assumption contains $K_{<i}$ and is disjoint with $L_{<i}$. Therefore, the piecewise testable language 
\begin{align*}
M' = \myunderbrace{\text{(upward closure of $L_{\le i}$)}}{contains $L_{\le i}$}  -  \myunderbrace{(M -  \text{(downward closure of $L_i$)}}{disjoint with $L_{\le i}$})
\end{align*}
contains $L_{\le i}$. We will now show that $M'$ is disjoint with $K_{\le i}$, thus finding a separator as required in the induction. First observe that the upward closure of  $L_{\le i}$ is disjoint   with  $K_i$, because otherwise there would be some words
 \begin{align*}
 \myunderbrace{w}{$L_{\le i}$} \quad \higman \quad 
 \myunderbrace{v}{$K_i$},
 \end{align*}
 and therefore the word $w$ would have potential $i+1$, and would not belong to $L_{\le i}$. Therefore, $K_i$ is disjoint with $M' \subseteq L_{\le i}$. The downward closure of   $L_i$ is disjoint with $K_{<i}$, since otherwise there would be some words 
 \begin{align*}
  \myunderbrace{w}{$K_{<i}$} \quad \higman \quad 
  \myunderbrace{v}{$L_i$},
  \end{align*}
  contradicting the definition of $K_{<i}$. Therefore  $K_{<i}$ is contained in
  \begin{align*}
     \myunderbrace{M}{contains $K_{<i}$} - \myunderbrace{\text{(downward closure of $L_i$)}}{disjoint with $K_{<i}$},
  \end{align*}
  and thus $K_{<i}$ is disjoint with $M'$. 
  \end{description}
\end{proof}

The Zigzag Lemma proves the equivalence of the conditions about infinite zigzags and piecewise testability in  Theorem~\ref{thm:piecewise}. To finish the proof of the Theorem, we show that the syntactic monoid of $L$ is finite and  infix trivial (which is the same as saying that some recognising monoid is  finite and infix trivial) if and only if there is no infinite zigzag between $L$ and its complement.

Suppose first  that the syntactic monoid of $L$  is either infinite or finite but not infix trivial. If the syntactic monoid is infinite, then the language cannot be piecewise testable, since piecewise testable languages are necessarily regular. Assume therefore that the syntactic monoid is  finite but not infix trivial.  This means that the syntactic monoid is either not prefix trivial, or not suffix trivial. By symmetry, we only consider the case where the syntactic monoid is not suffix trivial. This means that there exist $a,b$ in the syntactic monoid such that
\begin{align*}
(ab)^\momega  \neq b(ab)^\momega.
\end{align*}
By unravelling the definition of the syntactic monoid, the above disequality  can be easily used to create an infinite zigzag between $L$ and its complement. 

It remains to show that if the syntactic monoid of $L$ is finite and infix trivial,  then there is no infinite zigzag between $L$ and its complement. Let $M$ be the syntactic monoid. For  $a,b \in M$, define a zigzag between $a$ and $b$ to be a zigzag between the languages
\begin{align*}
\set{w \in M^* : \text{$w$ multiplies to $a$}} \qquad \set{w \in M^* : \text{$w$ multiplies to $b$}}.
\end{align*}
If $M$ recognises $L$, then a zigzag between $L$ and its complement can be used, by extraction, to obtain a zigzag between some two distinct monoid elements $a,b \in M$. The following lemma shows that this cannot happen, thus completing the proof of Theorem~\ref{thm:piecewise}.

\begin{lemma}
  Let $M$ be finite and infix trivial, and let $a,b \in M$. If there is an infinite zigzag between $a$ and $b$, then $a=b$.
\end{lemma}
\begin{proof}
  The proof is by induction on  the infix ordering lifted to pairs:
  \begin{align*}
  (x,y)\preceq (a,b) \qquad \eqdef \qquad \text{$x$ is an infix of $a$ and $y$ is an infix of $b$}.
  \end{align*}
  The induction base is proved the same way as the induction step.
  Suppose that we have proved the lemma for all pairs $(x,y)\prec (a,b)$. 
 
  \begin{claim}\label{claim:zigzag-extraction}
    If there is an infinite zigzag between $a$ and $b$, then there exists $n \in \set{0,1,\ldots}$ and monoid elements $\set{a_i,b_i,c_i}_{i}$ such that
    \begin{align*}
      \begin{array}{cccccccccccc}
        a &=& &  c_1 & & c_2 & &  \cdots  & & c_n \\
        b &=&b_0 & c_1 & b_1 & c_2 & b_2 & \cdots & b_{n-1} & c_n & b_n\\
        a &=&a_0 & c_1 & a_1 & c_2 & a_2 & \cdots & a_{n-1} & c_n & a_n
      \end{array}
    \end{align*}
    and for every $i \in \set{0,\ldots,n}$ there is an infinite zigzag between $a_i$ and $b_i$.
  \end{claim}
  \begin{proof}
    Consider an infinite zigzag between $a$ and $b$ of the form
    \begin{align*}
    w_1 \higman w_2 \higman \cdots
    \end{align*}

    Let the letters in $w_1$ be $c_1,\ldots,c_n \in M$.   For  $j \ge 2$, define an \emph{important position} in $w_j$ to be any position that arises by starting in some position of  $w_1$, and then following the embeddings 
    \begin{align*}
    w_1 \higman w_2 \higman \cdots \higman w_j.
    \end{align*}
    By distinguishing the important positions in $w_j$, we get a factorisation
    \begin{align*}
    w_j =  \myunderbrace{w_{j,0}}{$M^*$}   c_1   \myunderbrace{w_{j,1}}{$M^*$}   c_2   \cdots c_{n-1} \myunderbrace{w_{j,n-1}}{$M^*$}   c_n  \myunderbrace{w_{j,n}}{$M^*$}.
    \end{align*}
    By definition of important positions, for every $i \in \set{0,\ldots,n}$  the following sequence is  growing with respect to embedding
    \begin{align*}
      w_{2,i} \higman w_{3,i} \higman \cdots.
      \end{align*}
    By extracting a subsequence, we can assume that for every $i \in \set{0,1\ldots,n}$, the above chain is a zigzag between $b_i$ and $a_i$, for some $b_i,a_i \in M$. This proves the conclusion of the claim.
  \end{proof}

  \begin{claim}\label{claim:zigzag-decomposition}
    If there is an infinite zigzag between $a$ and $b$, then either $a=b$, or there exist $c,c'  \in M$ such that $a = cc'$ and $cb=b=bc'$.
    \end{claim}
    \begin{proof}
      Apply Claim~\ref{claim:zigzag-extraction}, yielding monoid elements which satisfy the following equalities: 
      \begin{align*}
        \begin{array}{cccccccccccc}
          a &=& &  c_1 & & c_2 & &  \cdots  & & c_n \\
          b &=&b_0 & c_1 & b_1 & c_2 & b_2 & \cdots & b_{n-1} & c_n & b_n\\
          a &=&a_0 & c_1 & a_1 & c_2 & a_2 & \cdots & a_{n-1} & c_n & a_n
        \end{array}
      \end{align*}
      For every $i \in \set{0,\ldots,n}$, we can see that $(b_i,a_i) \preceq (b,a)$. If the inclusion is strict, then the induction assumption of the lemma yields $b_i=a_i$. Otherwise, the inclusion is not strict, and therefore 
      \begin{align*}
      (a_i,b_i) = (a,b).
      \end{align*}
      If the inclusion is strict for all $i$, then the third and second rows in the conclusion of Claim~\ref{claim:zigzag-extraction} are equal, thus proving $a=b$, and we are done. Otherwise, there is some $i \in \set{0,\ldots,n}$ such that $(b_i,a_i)=(b,a)$. By infix triviality, every interval in the second row that contains $i$ will have multiplication $b$. It follows that 
      \begin{align*}
       \myunderbrace{c_jb = b}{for all $j \le i$}  \qquad 
       \myunderbrace{bc_j = b}{for all $j > i$}  \qquad 
      \end{align*}
      It is now easy to see that the conclusion of the claim holds if we  define $c$ and $c'$ as follows:
      \begin{align*}
      a = \myunderbrace{c_1 \cdots c_i}{$c$} 
      \myunderbrace{c_{i+1} \cdots c_n}{$c'$}.
      \end{align*}
    \end{proof}
  
Apply the above claim, and a symmetric one with the roles of $a$ and $b$ swapped, yielding elements $c,c',d,d'$ such that
\begin{align}\label{eq:jtrivial-endgame}
a=cc' \quad cb=b=bc' \quad b=dd' \quad da=a=ad'.
\end{align}
We can now prove the conclusion of the lemma:
\begin{align*}
a \myoverbrace{=}{
  \eqref{eq:jtrivial-endgame}
} (dc)^\momega (c'd')^\momega 
\myunderbrace{=}{infix triviality}
(cd)^\momega (d'c')^\momega 
\myoverbrace{=}{
  \eqref{eq:jtrivial-endgame}
}  b.
\end{align*}
\end{proof}

\exercisehead

\mikexercise{Prove Higman's Lemma.}{}

\mikexercise{\label{ex:separate-piecewise} Give a polynomial time algorithm, which inputs two nondeterministic automata, and decides if their languages can be separated by a piecwise testable language.}{}

\mikexercise{Consider  $\omega$-words, i.e.~infinite words of the form \begin{align*}
  a_1 a_2 \cdots \qquad\text{where } a_1,a_2,\ldots \in \Sigma.
  \end{align*}
Embedding  naturally extends to $\omega$-words (in fact, any labelled orders).  Show that the embedding on $\omega$-words is also a well-quasi order, i.e.~every upward closed set is the upward closure of finitely many elements.
}{}

\section{Two-variable first-order logic}
\label{sec:fotwo}
We finish this chapter with one more monoid characterisation of a fragment of first-order logic. A corollary of the equivalence of first-order logic and \ltl (or of the equivalence of first-order logic and star-free expressions) is that, over finite words, first-order logic is equivalent to its three variable fragment. What about one or two variables?

First-order logic with one variable defines exactly the languages which are  Boolean combinations for sentences of the form $\exists x \ a(x)$. These languages are exactly the languages that are recognised by monoids that are aperiodic and commutative:
\begin{align*}
a^\momega = a^{\momega+1} \quad ab=ba \qquad \text{for all $a,b$.}
\end{align*}
The more interesting case is first-order logic with  two variables, which we denote by \fotwo. This logic  is characterised in the following theorem. 

\begin{theorem}\footnote{
  The class of monoids from  item 2 appears, under the name Df, in 
\incite[p.~47,]{Schtzenberger1976SurLP}
where it is used to characterise certain unambiguous regular expressions, see Exercise~\ref{ex:unambiguous-concatenation}. 
Subsequent articles use the name \dav, which we use here as well. The connection with two variable first-order logic, which is the content of the theorem, is from 
\incite[Theorem 4.]{TherienWilke98}
}
  For a language $L \subseteq \Sigma^*$, the following are equivalent:
  \begin{enumerate}
    \item \label{fo2:fo2} Definable in two variable first-order logic;
    \item \label{fo2:da} Recognised by a finite monoid $M$ with the following property:  $M$ is  aperiodic, and if an infix class $J \subseteq M$ contains an idempotent, then $J$ is a sub-semigroup of $M$.  
  \end{enumerate}
\end{theorem}

We use the name \dav for the monoids (more generally, finite semigroups) that satisfy the property in item~\ref{fo2:da}.
In the exercises, we add several other equivalent conditions for the above theorem, including the temporal logic \ltlunary and the following fragment of first-order logic:
\begin{align*}
\text{(definable by a $\exists^* \forall^*$-sentence)}
\quad \cap  \quad
\text{(definable by a $\forall^* \exists^*$-sentence)}.
\end{align*}

The rest of Section~\ref{sec:fotwo} is devoted to proving the theorem. We begin with an equational description of \dav, which uses the embedding ordering on words that featured prominently in the previous section. (A stronger equational description  is given in Exercise~\ref{eq:da-identity-simplified}.)

\begin{lemma}\label{lem:da-identity}
  A finite monoid $M$ is in \dav if and only if  
  it satisfies:
\begin{align*}
\myunderbrace{w^\momega = w^\momega v w^\momega}{same multiplication} \qquad \text{for all $w,v \in M^*$ with $v \higman w$.}
\end{align*}
  \end{lemma}

\begin{proof}
  We first prove that the identity implies that $M$ is  \dav. The identity clearly implies aperiodicity, by taking $w=v$. 
  Let $e$ be an idempotent. We need to show that if $a,b$ are infix equivalent to $e$, then the same is true for $ab$. Because $a,b$ are infixes of $e$, and $e$ is an idempotent, one can find a word $w$ in $S^+$ which multiplies to $e$ and contains both  $a$ and $b$. In particular, $ab \higman w$. By the identity in the lemma, we know that $e = eabe$, and therefore $ab$ is an infix of  $e$. 

  We now show that if $M$ is in \dav, then the identity is satisfied.  Let $v \higman w$ be as in the identity. Let $e$ be the multiplication of $w^\momega$, and let  $J$ be the infix class of $e$. This infix class is a monoid, by definition of \dav.  For every letter $a$ that appears in the word $w$, there is a suffix of $w^\momega w^\momega$ which begins with $a$ and has  multiplication in $J$. Let $a' \in J$ be the multiplication of this suffix. Since $J$ is a monoid, it follows that $ea' \in J$ and therefore also $ea \in J$. Since $ea \in J$ holds for every letter that appears in $w$, it follows that $ev \in J$, and therefore also $eve \in J$.  This means that $eve$ is in the $\Hh$-class of $e$, and therefore $e=eve$ by aperiodicity (which is part of the definition of \dav), thus establishing the identity.
\end{proof}

We now prove the theorem.

To prove the implication \ref{fo2:fo2}$\Rightarrow$\ref{fo2:da}, we show  that for every language definable in \fotwo, its syntactic monoid belongs to \dav.  By  Lemma~\ref{lem:da-identity} and unravelling  the definition of the syntactic monoid, it is enough to show that for every $w_1,w_2, v, w \in \Sigma^*$  and  $n \in \set{0,1,\ldots}$,  if $v \higman w$ then 
  the words
  \begin{align*}
  w_1w ^n w^n w_2 \qquad  w_1 w^n v w^n  w_2 
  \end{align*}
  satisfy the same \fotwo sentences of quantifier rank at most $n$. This is shown using 
  a simple \ef argument.

For the  implication  \ref{fo2:da}$\Rightarrow$\ref{fo2:fo2}, we use the following lemma. 
\begin{lemma}
  Let $M$ be a monoid in \dav, and let $a_1,a_2 \in M$. Then 
  \begin{align*}
  w \in M^* \qquad \mapsto \qquad 
  \myunderbrace{a_1 \cdot \text{(multiplication of $w$)}\cdot a_2}{$ \in M$}
  \end{align*}
  is a colouring definable in \fotwo, which means that for every $c \in M$, the inverse image of $c$ under the colouring is a language that is definable in \fotwo.
\end{lemma}
If we apply the above lemma to $a_1$ and $a_2$ being the monoid identity, we conclude that the multiplication operation  is   definable in \fotwo. This implies that every language recognised by the monoid  is definable in \fotwo, thus proving the implication \ref{fo2:fo2} $\Leftarrow$ \ref{fo2:da} in the theorem. It remains to prove the lemma.
\begin{proof}
 Induction on the following parameters, ordered lexicographically:
 \begin{enumerate}
   \item size of $M$;
   \item number of elements that properly extend $a_1$ in the prefix ordering;
   \item number of elements that properly extend $a_2$ in the suffix ordering.
 \end{enumerate}
 The induction base is when $M$ has one element, in which case the colouring in the lemma is constant, and therefore definable in \fotwo. 

 Let us also prove another variant of the induction base, namely when the induction parameters (2) and (3) are zero, which means that $a_1$ is maximal in the prefix ordering and $a_2$ is maximal in the suffix ordering.  It follows that
 \begin{align*}
 \text{$\Hh$-class of $a_1aa_2$} \ =  \ \text{$\Hh$-class of $a_1ba_2$} \qquad \text{for all $a,b \in M$}.
 \end{align*}
 Since \dav implies aperiodicity, which implies $\Hh$-triviality,  the colouring in the statement of the lemma is constant, and therefore definable in \fotwo.

It remains to prove the induction step. Because of the two kinds of induction base that were considered above, we can assume that one of the parameters (2) or (3) is nonzero. By symmetry, assume that $a_1$ is not maximal in the prefix ordering.

  \begin{claim} For every $a \in M$, the following is a sub-monoid of $M$: 
    \begin{align*}
      \myunderbrace{\set{b \in M : a b \text{ is prefix equivalent to $a$}}}{we call this set the prefix stabiliser of $a$}.
    \end{align*}
    \end{claim}
\begin{proof}   The prefix stabiliser clearly   contains the monoid identity. It remains to show that it is closed under multiplication. Let $b,c$ be in the prefix stabiliser of $a$.  Using the definition of the prefix stabiliser, it is easy to construct a word  $w \in M^*$, such that $bc \higman w$ and   $aw = a$.  By Lemma~\ref{lem:da-identity}, it follows that 
  \begin{align*}
  a = aw = aw^\momega = aw^\momega bc w^\momega = abc w^\momega,
  \end{align*}
  which establishes that $bc$ is in the prefix stabiliser of $a$.
\end{proof}

Let $N \subseteq M$ be the prefix stabiliser of $a_1$; our assumption says that $N$ is a proper subset of $M$, and by the above claim it is also a sub-monoid.
  We decompose a word  $w \in M^*$ into three parts, as explained in the following picture: 
  \mypic{37} There is an \fotwo formula which selects the central position. Since  all labels in the left part are from $N$,   we can use the 
   induction assumption on a smaller monoid to prove that  the colouring
  \begin{align*}
  w \qquad \mapsto \qquad  \text{multiplication of left part}
  \end{align*}
  is definable in \fotwo. (When using the induction assumption, we restrict all quantifiers of the formulas from the induction assumption so that they quantify over positions in the left part.)  Let $c$ be the multiplication of the prefix up to and including the central position; as we have shown above, this multiplication can be computed in \fotwo. By definition of the central position, we know that $a_1$ is a proper prefix of $a_1c$, and therefore we can use the induction assumption to prove that 
  \begin{align*}
  w \qquad \mapsto \qquad a_1 c \cdot \text{(multiplication of right part)} \cdot a_2
  \end{align*}
  is a colouring definable in \fotwo.  The conclusion of the lemma follows.
\end{proof}

\exercisehead


\mikexercise{\label{eq:da-identity-simplified} Show that a  monoid belongs to \dav if and only if it satisfies the identity 
\begin{align*}
  (ab)^\momega = (ab)^\momega a (ab)^\momega \qquad \text{for all $a,b$.}
  \end{align*}
}{}
\mikexercise{Show that \fotwo has the same expressive power as \ltlunary, which is the extension of \ltlf with the following past operator:
\begin{align*}
  w, x \models \finally^{-1} \varphi \quad \eqdef \quad \exists y\ y < x \ \land \  w,y \models \varphi.
\end{align*}
}{}

\mikexercise{\label{ex:sigma-two} Define the \emph{syntactic ordering} on the syntactic monoid, which depends on the accepting set $F$, as follows:
\begin{align*}
a \le b \quad \eqdef \quad \forall x,y\in M\ xay \in F \Rightarrow xby \in F.
\end{align*}
Show that a language can be defined by a first-order sentence of the form 
\begin{align*}
\underbrace{\exists x_1  \cdots \exists x_n \forall y_1  \cdots \forall y_m \quad \overbrace{\varphi(x_1,\ldots,x_n,y_1,\ldots,y_m)}^{\text{quantifier-free}}}
_{\text{such a formula is called an $\exists^*\forall^*$-sentence}}
\end{align*}
if and only if 
\begin{align*}
w^\momega \le w^\momega v w^\momega \qquad \text{for all }\myunderbrace{v \higman w}{Higman ordering}
\end{align*}
Hint\footnote{
  An effective characterisation of $\exists^*\forall^*$-sentence was first given in 
  \incite[Theorem 3.]{arfi87}
The proof was simplified in 
\incite[Theorem 5.8]{PinWeil97a}
  The solution which uses Exercise~\ref{ex:typed-regular-expressions} is based on~\cite{PinWeil97a}. Characterisations of fragments 
of first-order logic such as $\exists^*\forall^*$ are widely studied, see 
\incite{placeZeitoun2019}
}: use Exercise~\ref{ex:typed-regular-expressions}.
}{}

\mikexercise{Show that $L$  is definable in \fotwo if and only both $L$ and its complement can be defined using $\exists^* \forall^*$-sentences. }{
Both $L$ and its complement have the same syntactic monoid, and the syntactic orderings are opposites of each other. Therefore, by Exercise~\ref{ex:sigma-two}, we see that $L$ and its complement can be defined by $\exists^* \forall^*$-sentences if and only if the syntactic monoid satisfies
\begin{align*}
  w^\momega = w^\momega v w^\momega \qquad \text{for all }\myunderbrace{v \higman w}{Higman ordering}.
  \end{align*}
We will show that the above identity is equivalent to \dav. }

\mikexercise{\label{ex:unambiguous-concatenation}We say that a regular expression 
\begin{align*}
\Sigma_0^* a_1 \Sigma_1^* \cdots \Sigma_{n-1}^* a_n \Sigma_n^*
\end{align*}
is \emph{unambiguous} if every word $w$ admits at most one factorisation
\begin{align*}
w = w_0 a_1 w_1 \cdots w_{n-1}a_n w_n \qquad \text{where $w_i \in \Sigma_i^*$ for all $i \in \set{1,\ldots,n}$.}
\end{align*}
Show that a language is a finite disjoint union of unambiguous expressions if and only if its  syntactic monoid of $L$ is in \dav\footnote{This exercise is based on
\incite{Schtzenberger1976SurLP}}.
}{}

 \chapter{Infinite words}

In this chapter, we study infinite words. 

    In Section~\ref{sec:buchi-determinisation},  we begin with the classical model of infinite words, namely $\omega$-words. In an $\omega$-word,
     the positions are ordered like the natural numbers. We show how the structure of finite semigroups described by Green's relations can be applied to prove  McNaughton's Theorem about  determinisation of $\omega$-automata. 

In Section~\ref{sec:countable-words}, we move to more general infinite words, where the positions can be any countable linear order, e.g.~the rational numbers. For this kind of infinite words, we define a suitable generalisation of semigroups, and show that it has the same expressive power as monadic second-order logic.

\section{Determinisation of \buchi automata for \texorpdfstring{$\omega$-words}{ω-words}}
\label{sec:buchi-determinisation}
An \emph{$\omega$-word} is defined to be a function from the natural numbers to some alphabet $\Sigma$. We write $\Sigma^\omega$ for the set of all $\omega$-words over  alphabet $\Sigma$. 
To recognise properties of $\omega$-words, we use \buchi automata. These have the same syntax as nondeterministic automata on finite words, but they are used to accept or reject $\omega$-words.

\begin{definition}
    [\buchi automata] The syntax of a \emph{nondeterministic \buchi automaton} is the same as the syntax of a nondeterministic finite automaton for finite words, namely it consists of:
    \begin{align*}
    \myunderbrace{Q}{states} \qquad
    \myunderbrace{{\color{white}Q}\Sigma{\color{white}Q}}{input\\ \scriptsize alphabet} \qquad
    \myunderbrace{I,F \subseteq Q}{inital and\\ \scriptsize final states} \quad 
    \myunderbrace{\delta \subseteq Q \times \Sigma \times Q}{transition relation}.
    \end{align*}
    An $\omega$-word over the input alphabet is accepted by the automaton if there exists a run which begins in an initial state, and which satisfies the \emph{\buchi condition}: some accepting state appears infinitely often in the run.
    A \emph{deterministic  \buchi automaton} is the special case when 
     there is one initial state, and the transition relation is a function from $Q \times \Sigma$ to $Q$.
\end{definition}

The literature on automata for $\omega$-words has other acceptance conditions, which will not be used in this book. One example is the \emph{Muller condition}, where the accepting set is a family of subsets of states, and a run is accepting if the set of states used infinitely often is a subset that belongs to the accepting family. Another example is the \emph{parity condition}: there is a linear order on the states, and a subset of accepting states, and a run is accepting if the maximal state used infinitely often is accepting. 

The following example shows that deterministic \buchi automata are weaker than than nondeterministic ones. 
\begin{myexample}
    Consider the language of $\omega$-words over alphabet $\set{a,b}$ where letter $a$ appears finitely often. This language is recognised by a nondeterministic \buchi automaton as in the following picture:
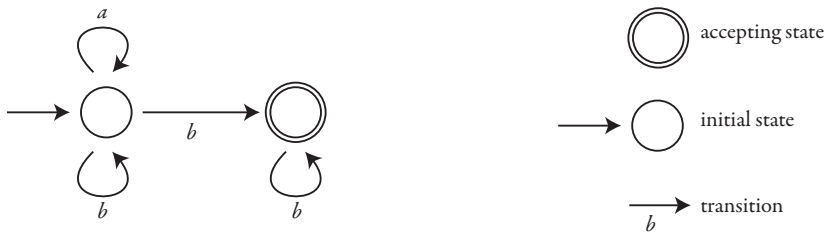
\mypic{38}
    The idea is that the automaton nondeterministically guesses some position which will not be followed by any $a$ letters; this guess corresponds to the horizontal transition with label $b$ in the picture.

    This  language is not recognised by any deterministic \buchi automaton.  Toward a contradiction, imagine a hypothetical deterministic \buchi automaton which recognises the language. Run this automaton on $b^\omega$. Since $a$ appears finitely often in this $\omega$-word, the corresponding run (unique by determinism) must use an accepting state in some finite prefix. Extend that finite prefix by appending $ab^\omega$. Again, the word must be accepted, so an accepting  state must be eventually visited after the first $a$. By repeating this argument, we get a word which has infinitely many $a$'s and where the (unique) run of the deterministic automaton sees accepting states infinitely often; a contradiction. 
\end{myexample}

The above shows that languages recognised by deterministic \buchi automata are not closed under Boolean combinations. This turns out to be the only limitation of the model, as shown in the following theorem. 
\begin{theorem} \label{thm:buchi-determinisation} The following formalisms describe the same languages of $\omega$-words:
    \begin{itemize}
        \item nondeterministic B\"uchi automata;
        \item Boolean combinations of deterministic B\"uchi automata\footnote{A  Boolean combination of deterministic \buchi automata is the same thing as a deterministic automaton with the Muller condition. Therefore, the  theorem is the  same McNaughton's Theorem,
        \incite[p.~524]{McNaughton66}
        which says that nondeterministic \buchi automata can be determinised  into deterministic Muller automata.
        }.
    \end{itemize}
\end{theorem}
A language is called \emph{$\omega$-regular} if it satisfies either of the two equivalent conditions in the above theorem. The $\omega$-regular languages are closed under Boolean combination thanks to the deterministic characterisation.   The original application of \buchi automata was \buchi's proof\footcitelong{Buchi62} that they recognise exactly the same languages of $\omega$-words as monadic second-order logic; this application is a simple corollary of Theorem~\ref{thm:buchi-determinisation}, see Exercise~\ref{ex:buchi-mso}.

The easier bottom-up implication in Theorem~\ref{thm:buchi-determinisation} follows from the following lemma. 
\begin{lemma}
    Languages recognised by nondeterministic \buchi automata are closed under union and intersection, and contain all languages recognised by  deterministic \buchi automata and their complements.
\end{lemma}
\begin{proof}
    Closure under union is immediate for nondeterministic automata. Consider now the intersection of two nondeterministic \buchi automata $\Aa$ and $\Bb$. A nondeterministic \buchi automaton $\Cc$ for the intersection is defined as follows. Take two copies of the product automaton $\Aa \times \Bb$. The accepting states are
    \begin{align*}
    \myunderbrace{(\text{accepting state}, \text{any state})}{first copy}
    \quad \cup \quad 
    \myunderbrace{(\text{any state}, \text{accepting state})}{second copy}
    \end{align*}
    Whenever $\Cc$ sees an accepting state as described above, it switches to the other copy. An accepting run of $\Cc$  must see accepting states of both copies infinitely often, and hence it recognises the intersection of the languages of $\Aa$ and $\Bb$.

    Since deterministic \buchi automata are a special case of nondeterministic ones, it remains to show that complements of deterministic \buchi automata can be simulated by nondeterministic \buchi automata.    The complement of the  language of a deterministic \buchi automaton consists of those words where final states are seen finitely often in the unique run. This can be checked by a nondeterministic \buchi automaton, which nondeterministically guesses the moment where accepting states of the original automaton will no longer be seen. Here is an example. Suppose that we want to complement the deterministic \buchi automaton
    \mypic{109}
    which checks that infinitely often, the number of $a$'s is equal to the number of $b$'s modulo 3. The nondeterministic \buchi automaton for the complement looks like this:
    \mypic{108}
    In the  simulating nondeterministic automaton, the initial states are inherited in the first copy, and the accepting states are all states in the second copy (which correspond to non-accept states in the original deterministic automaton). The above picture uses $\varepsilon$-transitions, which can be easily eliminated.
\end{proof}

We are left with the harder top-down implication in the theorem, which says that every nondeterministic \buchi automaton can be simulated by a Boolean combination of deterministic \buchi automata. 
There are several combinatorial  proofs for the determinisation result in harder implication\footnote{Apart from McNaughton's original proof from~\cite{McNaughton66}, another well-known construction is given in 
\incite[Theorem 1.]{Safra88}
Another approach, which is based on a construction of Muller and Schupp, is described in
\incite[Section 1.]{bojanczyk_automata_2018}
}. In this section, we present 
an algebraic proof, which leverages the  structural theory of finite semigroups described earlier in this book.

Let $\Aa$ be a nondeterministic \buchi automaton, with states $Q$ and input alphabet $\Sigma$. 
The rest of this section is devoted to finding a Boolean combination of deterministic \buchi automata that is equivalent to $\Aa$. 
For an $\omega$-word, define its \emph{$\omega$-type} to be the set of states from which the word is accepted. We also define the type for finite words, but here we need to store a bit more information. 
For a run of the automaton over a finite word, define the \emph{profile} of the run to be  the triple $(q,i,p)$ where $q$ is the source state of the run, $p$ is the target state of the run, and  
\begin{align*}
i = \begin{cases}
    0 & \text{if the run does not use any accepting state}\\
    1 & \text{if the run uses some accepting state}.
\end{cases}
\end{align*}
Here is a picture of a run with its profile:
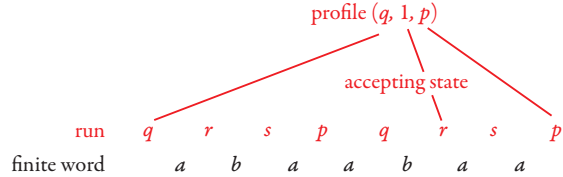
\mypic{41}
Define the \emph{type} of a finite word $w \in \Sigma^+$ to be the set of profiles of runs over this word. It is not hard to see that the function
\begin{align*}
w \in \Sigma^+ \quad \mapsto \quad \text{type of $w$}\ \in \underbrace{\powerset(Q \times \set{0,1} \times Q)}_S
\end{align*}
is a semigroup homomorphism, with a naturally defined semigroup structure on $S$. 

The following lemma shows that  types for finite and $\omega$-words   are compatible with each other.
\begin{lemma}\label{lem:type-omega-congruence} 
     If  $w_i \in \Sigma^+$ and $v_i \in \Sigma^+$ have the same type  for every $i \in \set{1,2,\ldots}$, then   $w_1 w_2 \cdots \in \Sigma^\omega$ and $v_1 v_2 \cdots \in \Sigma^\omega$ have the same $\omega$-type. 
\end{lemma}
\begin{proof}
    By substituting parts of an accepting run, while preserving the \buchi condition.
\end{proof}
Thanks to the above lemma, it makes sense to talk about the $\omega$-type  of a word  $w \in S^\omega$ built out of types; this is the $\omega$-type of some (equivalently, every) $\omega$-word that is obtained by concatenating $\omega$-many finite words with the respective finite types. In particular, it makes sense to say whether or not a word $w \in S^\omega$ is accepted by $\Aa$, since this information is stored in the type.  A special case of this notation is $ae^\omega$, where $a,e \in S$, which is the $\omega$-type of the  $\omega$-word that begins with letter $a$ and has all other letters equal to $e$.  The importance of this special case is explained by the following lemma  about factorisations of $\omega$-words\footnote{This lemma was first observed by \buchi in 
\cite[Lemma 1]{Buchi62}
where it was used to prove that nondeterministic \buchi automata are closed under complementation, without passing through a deterministic model.
}

\begin{lemma}\label{lem:ramsey-buchi}
    For every $w \in S^\omega$ there exist $a,e \in S$, such that $e$ is an idempotent, $ae=a$, and there is a   factorisation 
    \begin{align*}
    w \qquad = \qquad \myoverbrace{w_0}{type $a$} \quad
    \myoverbrace{w_1}{type $e$} \quad
    \myoverbrace{w_2}{type $e$} \quad
    \myoverbrace{w_3}{type $e$} \cdots
    \end{align*}
\end{lemma}
\begin{proof} 
    Define a \emph{cut} in $w$ to be  the space between two positions. Consider an undirected edge-labelled graph, defined as follows. Vertices are cuts. For every two distinct cuts, there is an undirected edge,   labelled  by the type of the finite word that  connects the two cuts. By Ramsey's Theorem A, see Exercise~\ref{ex:ramsey}, there exists a type $a \in S$ and an infinite set  $X$ of vertices, such  every two distinct vertices from $X$ are connected by an edge with label $e$.  Define the decomposition from the lemma to be the result of cutting $w$ along all cuts from $X$.  By assumption on $X$, every word $w_i$ with $i > 0$ has type $e$. Idempotence of $e$ follows from
    \begin{align*}
    \overbrace{
        \underbrace{w_i}_{e}
        \underbrace{w_{i+1}}_{e}
    }^e.
    \end{align*}
Finally, we can assure that $ae=a$ by joining the first two groups.
\end{proof}

A corollary of Lemmas~\ref{lem:type-omega-congruence} and~\ref{lem:ramsey-buchi} is that $w \in L$ if and only if 
\begin{itemize}
    \item[(*)] there is a factorisation as in Lemma~\ref{lem:ramsey-buchi} such that  $ae^\omega \in L$.
\end{itemize}
So far, we are doing the same argument as in \buchi's original complementation proof from~\cite{Buchi62}. In his proof,  \buchi observed that variant of (*) with $ae^\omega \not \in L$, which characterises the complement of $L$, can be expressed by a nondeterministic \buchi automaton, and therefore nondeterministic \buchi automata are closed under complementation. 

This is the place where  we diverge from \buchi's proof, since we are interested in determinisation, while \buchi was interested in complementation.  For determinisation, more insight into the structure of finite semigroups will be helpful. Since it is immediately  not clear how to express condition (*) using a deterministic \buchi automaton, we will  reformulate it. In the reformulation,  we say that a pair $(a,b) \in S^2$ appears infinitely often in an $\omega$-word $w \in \Sigma^\omega$ if  for every  $n \in \set{1,2,\ldots}$ one can find a factorisation 
\begin{align*}
w = \myunderbrace{x}{type $a$}
\myunderbrace{y}{type $b$}
z
\end{align*}
such that $x$ has length at least $n$. 

\begin{lemma}\label{lem:det-buchi-char}
    An $\omega$-word $w \in \Sigma^\omega$ is accepted by $\Aa$ if and only if
    \begin{itemize}
        \item[(**)]there exist $a,e \in S$, with $e$ idempotent, $ae=a$, and $ae^\omega \in L$,  such that both conditions below are satisfied:
        \begin{enumerate}
            \item \label{buchi:infix-yes} $(a,e)$ appears infinitely often; and
            \item \label{buchi:infix-no} if $(b,c)$ appears infinitely often, then $c$ is an infix of $e$. 
        \end{enumerate}
    \end{itemize} 
\end{lemma}
\begin{proof} The top-down implication, which says that every word accepted by $\Aa$ must satisfy (**), is an immediate  consequence of Lemma~\ref{lem:ramsey-buchi}. We are left with the bottom-up implication. Suppose that $w$ satisfies (**), as witnessed by    $a,e \in S$.
By condition~\ref{buchi:infix-yes}, there is a decomposition
\begin{align*}
    w  = w_1 v_1 w_2 v_2 w_3 v_3 \cdots
    \end{align*}
    such that for every $i \in \set{1,2,\ldots}$ the word $v_i$ has type $e$ and the prefix ending in $w_i$ has  type $a$. Let $a_i$ be the type of $w_i$. The $\omega$-type of $w$ is equal to 
\begin{align*}
 a_1 e a_2 e a_3 e \cdots.
\end{align*}
By condition~\ref{buchi:infix-no}, there is some $n$ such that  
\begin{align*}
g_i \eqdef e a_i e
\end{align*}
is an infix of $e$ for all $i \ge n$. Since $g_i$ is begins and ends with $e$, it follows that $g$ is in the $\Hh$-class of $e$ for all $i > n$. Since this  $\Hh$-class, call it $G$, contains the idempotent $e$, it must be a group by the $\Hh$-class lemma. We now complete the proof of the lemma as follows:
   \begin{eqnarray*}
            \text{$\omega$-type of $w$}  \equalbecause{$a_1 e a_2 \cdots a_{n-1} e_{n-1} =a$}\\
            a ea_n ea_{n+1} e a_{n+2} e \cdots \equalbecause{$e$ is idempotent and Lemma~\ref{lem:type-omega-congruence}}\\
            a ea_n eea_{n+1} ee a_{n+2} ee \cdots\equalbecause{definition of $g_i$}\\
            a g_{n} g_{n+1} g_{n+2} \cdots \equalbecause{by Lemma~\ref{lem:ramsey-buchi}, for some $g,f \in G$}\\
            agf^\omega \equalbecause{because $e$ is the unique idempotent in $G$}\\
            age^\omega 
             \equalbecause{some power of $g$ is the  idempotent $e$}\\
            ag^\omega \equalbecause{for the same reason}\\
            ae^\omega
        \end{eqnarray*}
        and therefore $w$ must belong to $L$.
\end{proof}

To finish the determinisation construction in Theorem~\ref{thm:buchi-determinisation}, it remains to show that condition (**) from the above lemma is a finite Boolean combination of languages recognised by deterministic \buchi automata. This will follow from the following lemma.

\begin{lemma}
    For every  $a,e \in S$ the property ``$(a,e)$ appears infinitely often'' is recognised by a deterministic B\"uchi automaton.
\end{lemma}
\begin{proof}
    Let $L \subseteq \Sigma^*$ be the set of words which can be  decomposed as
    \begin{align*}
    w = \myunderbrace{u}{type $b$} \myunderbrace{v}{type $e$} \qquad \text{for some $b \in S$ such that $aeb=a$}.
    \end{align*}
    This is easily seen to be a regular language,
    and hence it is recognised by some finite deterministic automaton $\Dd$.  The deterministic \buchi automaton $\Bb$  recognising the property in the statement of the  lemma is defined as follows. Its space is the disjoint union of the set of types $S$ and the states of  $\Dd$. 
    The initial state is the type in $S$ of the empty word. The automaton $\Bb$  begins to read input letters, keeping in its state the   type of the prefix read so far in its state, until the prefix has  type $ae$. Then it switches to the initial state $q_0$ of the automaton $\Dd$. For states of $\Dd$, the  state update function of $\Bb$ is defined as follows:
    \begin{align*}
    \delta_\Bb(q,\sigma)  \mapsto \begin{cases}
        \delta_\Aa(q,\sigma) & \text{if $q$ is not accepting in $\Dd$}\\
        \delta_\Aa(q_0,\sigma) & \text{otherwise}.
    \end{cases}
    \end{align*}
    The \buchi accepting states of $\Bb$ are the same as in $\Dd$. 
\end{proof}

This completes the proof of Theorem~\ref{thm:buchi-determinisation}.

\paragraph*{Semigroups for $\omega$-words.}
There is an implicit algebraic structure in the proof of Theorem~\ref{thm:buchi-determinisation}, which is formalised in the following definition.  
\begin{definition}\label{def:omega-semigroups}
    An $\omega$-semigroup\footcitelong[Section 7]{PerrinPin95} consists  of:
    \begin{itemize}
        \item two sets $\finsort{S}$ and $\omegasort{S}$, called the \emph{finite sort} and the \emph{$\omega$-sort}, respectively.
        \item a finite multiplication operation $\mu_{+} : (\finsort S)^+ \to \finsort S$,  which is associative  in the sense of  semigroups;
        \item an $\omega$-multiplication operation $\mu_\omega :  (\finsort S)^\omega \to \omegasort{S}$,  
        which is associative in the following sense:
        \begin{align*}
        \mu_\omega(w_1 w_2  \cdots ) = \mu_\omega(\mu_+(w_1) \mu_+(w_2) \cdots) \qquad \text{for every }w_1,w_2,\ldots \in S^+.
        \end{align*}
    \end{itemize}
\end{definition}
An example of an $\omega$-semigroup is the  automaton types that were used in the proof of Theorem~\ref{thm:buchi-determinisation}. Another example is the \emph{free $\omega$-semigroup over  a set $\Sigma$}, where the finite sort is $\Sigma^+$, the $\omega$-sort is $\Sigma^\omega$, and the two multiplication operations are defined in the natural way. The same proof as in Theorem~\ref{thm:buchi-determinisation} shows that a language is $\omega$-regular if and only if it is recognised by a homomorphism into an $\omega$-semigroup which is finite (on both sorts). This is discussed in more detail in some of the exercises at the end of this section. 

The  associativity axiom on $\omega$-multiplication can be represented using a commuting diagram, in the same spirit as for Lemma~\ref{lem:commuting-diagram-semigroup}:
   \begin{align*}
    \xymatrix@C=5cm
    {((\finsort S)^+)^\omega \ar[r]^{\text{$\omega$-multiplication in free $\omega$-semigroup over $\finsort S$}} \ar[d]_{(\mu_+)^\omega} &
     (\finsort S)^\omega \ar[d]^{\mu_\omega} \\ 
     (\finsort S)^\omega \ar[r]_{\mu_\omega}
     & \omegasort S}
 \end{align*}
 In the above diagram,  $(\mu_+)^\omega$ denotes the coordinate-wise lifting of $\mu_+$ to $\omega$-words of finite words.

\exercisehead

\mikexercise{\label{ex:ramsey} 
Prove the following result, called  \emph{Ramsey's Theorem A \footcitelong[Theorem A]{Ramsey29}}. Consider an infinite undirected graph, where every two 
distinct vertices are a connected by an edge that is labelled by one of finitely many colours. Then the graph contains an infinite monochromatic clique,  which means that  there exists a colour $e$ and an infinite set $X$ of vertices, such that every two distinct  vertices from $X$ are connected by an edge with colour $e$.  
}{
    Let $V_0$ be all the vertices in the graph.
Choose some vertex $v_1 \in V_0$. Since there are infinitely many vertices and finitely many colours, there must be some colour $c_1$ and an infinite set $V_1 \subseteq V_0 - \set v$ of vertices such that all vertices from $X_1$ are connected to $v_1$ by edges of colour $c_1.$  Repeat the same process for the graph restricted to vertices $V_1$, and keep on repeating it infinitely often, thus creating an infinite sequence
\begin{align*}
V_0, v_1, c_1, V_1, v_2, c_2, \cdots
\end{align*}
By construction, for every $i < j$ the edge connecting $v_i$ with $v_j$ has colour $c_i$. Some colour must appear infinitely often in $c_1,c_2,\ldots$ and if we restrict the sequence to coordinates $i$ with $c_i=c$, then we get the monochromatic clique in the conclusion of the Ramsey Theorem.
}

\mikexercise{ \label{ex:ultimately-periodic-nonempty}
We say that an $\omega$-word is \emph{ultimately periodic} if it has the form $wu^\omega$, for some finite words $w,u \in \Sigma^\omega$. Show that every nonempty $\omega$-regular language contains an ultimately periodic $\omega$-word.
  }{
    Take a nonempty nondeterministic \buchi automaton.
    If the automaton accepts some word, then there is an accepting state $q$, such that $q$ is reachable from some initial state (via some finite word $w$), and $q$ is also reachable from itself (via some finite word $u$). Then word $wu^\omega$ is accepted. }

\mikexercise{\label{ex:ultimately-periodic-dense}Show that two $\omega$-regular languages are equal  if and only if they contain the same ultimately periodic $\omega$-words.}{
    Suppose, toward a contradiction, that $L$ and $K$ are different languages that  are recognised by nondeterministic \buchi automata. Assume that $L - K$ is nonempty. Thanks to Theorem~\ref{thm:buchi-determinisation}, this difference is recognised by a Boolean combination of deterministic \buchi automata, and therefore by Theorem~\ref{thm:buchi-determinisation} it is recognised by a nondeterministic \buchi automaton. By Exercise~\ref{ex:ultimately-periodic-nonempty}, the difference contains some ultimately periodic word.
}

\mikexercise{Show that an $\omega$-word $w$ is ultimately periodic if and only if $\set w$ is an $\omega$-regular language.}{The language $\set{wu^\omega}$ is easily seen to be recognised by a nondeterministic \buchi automaton, where the number of states is the total length of the words $w$ and $u$. For the converse implication, we use Exercise~\ref{ex:ultimately-periodic-dense} to conclude that if a language recognised by a nondeterministic \buchi automaton contains a unique word, then that word is ultimately periodic.}

\mikexercise{\label{ex:buchi-mso} To an $\omega$-word we associate an ordered model, in the same way as for finite words. Show that a language is \mso definable (using the ordered model) if and only if it is $\omega$-regular. }{}

\mikexercise{Define an $\omega$-term  to be any tree as in the following picture:  
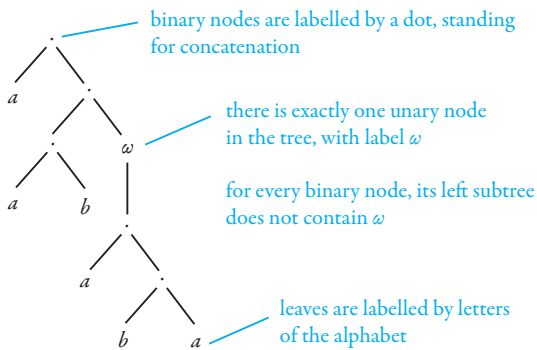
\mypic{39}
Every $\omega$-term represents some ultimately periodic $\omega$-word, but several $\omega$-terms might represent the same ultimately periodic $\omega$-word. 
Show that two $\omega$-terms represent the same ultimately periodic $\omega$-word if and only if one can be transformed into the other using the equations:
\begin{align*}
(xy)z = x(yz) \qquad 
(xy)^\omega = x(yx)^\omega \qquad 
\myunderbrace{(x^n)^\omega = x^\omega}{for every $n \in \set{1,2,\ldots}$}
\end{align*}
where $x,y,z$ stand for $\omega$-terms. }
{Using the equation $(xy)z = x(yz)$, we can write finite words without indicating the order in which the multiplication operation is applied.  We say that an $\omega$-term is  in \emph{normal form} if has the shape $wu^\omega$ such that:
\begin{enumerate}
    \item $u$ is a prime word, which means that there is no decomposition of the form $u = v^n$ for some $n \in \set{2,3,\ldots}$;
    \item there is no suffix of $w$ which is also a suffix of $u$.
\end{enumerate}
Using the equation $(x^n)^\omega = x^\omega$, we can convert every $\omega$-term  into one which satisfies condition (1). Using the equation $(xy)^\omega = x(yx)^\omega$, we can convert every pair into one which satisfies condition (2), without violating condition (1). Therefore, using the identities we can convert every $\omega$-term into one that is in normal form. For every ultimately periodic word, there is a unique $\omega$-term in normal form that generates it. 

}

\mikexercise{\label{ex:equivalences-for-omega-words} Let $L \subseteq \Sigma^\omega$. Consider the following equivalence relations on $\Sigma^+$.
\begin{itemize}
    \item Right equivalence is defined by 
    \begin{align*}
        w \sim w' \quad \eqdef \quad    wv \in L \iff w'v \in L \ \text{for every $v \in \Sigma^\omega$.}
        \end{align*}
    \item Two-sided congruence is defined by 
        \begin{align*}
            w \sim w' \quad \eqdef \quad    uwv \in L \iff uw'v \in L \ \text{for every $u \in \Sigma^*, v \in \Sigma^\omega$.}
            \end{align*}
    \item Arnold congruence is defined by 
            \begin{align*}
                w \sim w' \quad \eqdef \quad   \land 
                \begin{cases}u(wv)^\omega \in L \iff u(w'v)^\omega \in L & \text{for every $u,v \in \Sigma^*$.}\\
                    uwv \in L \iff uw'v \in L & \text{for every $u \in \Sigma^*, v \in \Sigma^\omega$.}
                \end{cases}
                \end{align*}         
\end{itemize}
Show that the latter two, but not necessarily the first one,  are  semigroup congruences, i.e.~they satisfy
\begin{align*}
    \bigwedge_{i \in \set{1,2}} w_i \sim w'_i \qquad \text{implies} \qquad w_1 w_2 \sim w'_1 w'_2.
    \end{align*}
}{}
\mikexercise{Consider the equivalence relations defined in Exercise~\ref{ex:equivalences-for-omega-words}.
Prove that the arrows in the following diagram are true implications, and provide counter-examples the missing arrows:
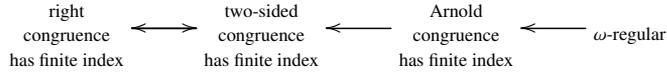
\begin{align*}
    \xymatrix{
        \txt{\scriptsize right\\
        \scriptsize congruence\\
        \scriptsize has finite index} 
        \ar[r]&
        \txt{\scriptsize two-sided \\
        \scriptsize congruence\\
        \scriptsize has finite index}
        \ar[l]
         &
        \txt{\scriptsize Arnold \\
        \scriptsize congruence\\
        \scriptsize has finite index} 
        \ar[l]&
        \txt{\scriptsize $\omega$-regular}
        \ar[l]
    }
\end{align*}
}{}

\mikexercise{Define the \emph{Arnold semigroup} of a language $L \subseteq \Sigma^\omega$ to be the quotient of $\Sigma^+$ under Arnold congruence. Let $L \subseteq \Sigma^\omega$ be a $\omega$-regular.   Show that $L$ is definable in first-order logic if and only if its Arnold semigroup is aperiodic.}{}

\mikexercise{The temporal logic \ltlf can also be used to define languages of $\omega$-words.  Let $L \subseteq \Sigma^\omega$ be a $\omega$-regular. Show that  $L$ is definable in \ltl if and only if its Arnold semigroup is suffix-trivial.
}{}

\mikexercise{Show an $\omega$-regular language where the Arnold semigroup is infix trivial, but which cannot be defined by a Boolean combination of $\exists^*$-sentences.}{ 
    Consider the language $L \subseteq \set{0,1}^\omega$ of words which have $1$ infinitely often. This Arnold semigroup of this language is $(\set{0,1},\max)$ and is therefore infix trivial. This language 
}

\mikexercise{\label{ex:safety-automata}Define a \emph{safety automaton} to be an automaton on $\omega$-words with the following acceptance condition: all states in the run are  accepting. Show that deterministic and nondeterministic safety automata recognise the same languages. 
}{}

\mikexercise{\label{ex:safety-automata-algebraic} Show that an $\omega$-regular language of $\omega$-words is recognised by a safety automaton (deterministic or nondeterministic, does not matter by Exercise~\ref{ex:safety-automata}) if and only if 
\begin{align*}
 uw^\momega v \in L \ \iff \ u(w^\momega)^\omega   \in L \qquad \text{for every $u,w \in \Sigma^+$ and $v \in \Sigma^\omega$,}
\end{align*}
where $\momega \in \set{1,2,\ldots}$ is the exponent obtained from the Idempotent Power Lemma as applied to the Arnold semigroup of $L$.
}{}

\mikexercise{\label{ex:topology} \label{ex:safety-topology} For a finite alphabet $\Sigma$, we can view $\Sigma^\omega$ as metric space, where the distance between two different $\omega$-words is defined to be
\begin{align*}
    \frac 1
    {2^{\text{(length of longest common prefix)}}}
\end{align*}
This is indeed a distance, i.e.~it satisfies the triangle inequality. 
Let $L \subseteq \Sigma^\omega$ be $\omega$-regular. Show that $L$ is recognised by a  safety automaton if and only if it is a closed set with respect to this distance.
}{}

\mikexercise{\label{ex:clopen-omega-word} Find a condition on the Arnold semigroup of an $\omega$-regular language which characterises the clopen languages (i.e.~languages which are both closed and open with respect to the distance from Exercise~\ref{ex:topology})}{The clopen languages are the definite languages, and the condition is therefore $a^\momega = a^\momega b$.}

\mikexercise{We use the topology from Exercise~\ref{ex:topology}. Define a $G_\delta$ set to be any countable intersection of open sets. Show that every $\omega$-regular language is a finite Boolean combination of $G_\delta$ sets. 
}{}

\mikexercise{ Let $L \subseteq \Sigma^\omega$ be an $\omega$-regular language, and define  $\momega$  as in  Exercise~\ref{ex:safety-automata}. Show that $L$ is recognised by a deterministic \buchi automaton if and only if:
    \begin{align*}
        u (wv^\momega)^\momega v^\omega  \in L \ \Rightarrow \ u(wv^\momega)^\omega   \in L \qquad \text{for every $u,w,v \in \Sigma^+$.}
       \end{align*}
}{}
\mikexercise{Let $L \subseteq \Sigma^\omega$. Define 
an $\omega$-congruence to be any equivalence relation $\sim$ on $\Sigma^+$ which is a semigroup congruence and which satisfies
    \begin{align}\label{eq:semi-congruence}
\bigwedge_{i \in \set{1,2,\ldots}} w_i \sim w'_i \qquad \text{implies} \qquad w_1 w_2 \cdots \in L \iff w'_1 w'_2 \cdots \in L.
\end{align}
Show that a language is $\omega$-regular  if and only if it has an $\omega$-congruence of finite index.
}{}

\mikexercise{Define \emph{semi-$\omega$-congruence} for a language $L \subseteq \Sigma^\omega$ to be an equivalence relation on finite words which satisfies~\eqref{eq:semi-congruence}, but which is not necessarily a semigroup congruence. Show that if there is a semi-$\omega$-congruence of finite index, then there is an $\omega$-congruence of finite index. }{}

\mikexercise{We say that $\sim$ is the \emph{syntactic $\omega$-congruence} of $L \subseteq \Sigma^\omega$ if it is an $\omega$-congruence, and every other $\omega$-congruence for $L$ refines $\sim$. Show that if a language is $\omega$-regular, then  it has a syntactic $\omega$-congruence, which is equal to the Arnold congruence. }{}

\mikexercise{ 
Show a language of $\omega$-words which does not have a syntactic $\omega$-congruence. }{}

\section{Countable words and \texorpdfstring{$\cc$-semigroups}{circle-semigroups}}
\label{sec:countable-words}
In this section, we move to $\cc$-words. These are  words where the set of  positions is a countable linear order. The positions could be some finite linear order, as in finite words, or  the natural numbers, as in $\omega$-words, but  some dense set, like the  rational numbers.  One advantage of $\cc$-words, as compared to $\omega$-words, is that they can be concatenated, which is  useful when defining the corresponding generalisation of semigroups.

For finite words, as well as for $\omega$-words, the approach via semigroups can be seen as an alternative to existing automata models. This is no longer the case for $\cc$-words. There  is no known corresponding automaton model, and therefore $\cc$-semigroups are the only known model of recognisability.

\begin{definition}[$\cc$-words]
    A \emph{$\Sigma$-labelled linear order} consists of a set $X$ of \emph{positions}, equipped with a total order and a labelling of type $X \to \Sigma$. 
Two such objects are considered \emph{isomorphic} if there is a bijection between their positions, which preserves the order and labelling. Define a   \emph{$\cc$-word over  $\Sigma$}  to be any isomorphism class of countable\footnote{Why do we assume countability? It turns out that  the decidability results described in this section breaks down for uncountable linear orders. In fact, the \mso theory of the order of real numbers $(\mathbb R, <)$  is undecidable, as shown 
\incite[Theorem 7.]{shelahMonadicTheoryOrder1975}
The  description of  $\cc$-semigroups in this section is based on~\cite{shelahMonadicTheoryOrder1975} and
\incite{carton_colcombet_puppis_2018}
} $\Sigma$-labelled linear orders. We write $\Sigma^\cc$ for the set of $\cc$-words\footnote{Formally speaking, this is not a set, because the linear orders form a class an not a set. However, without loss of generality we can use some fixed countably infinite set, e.g.~the natural numbers, for the positions (but the order need not be the same as in the natural numbers). Under this restriction, the labelled linear orders become a set, and no isomorphism types are lost. For this reason, we can refer to $\Sigma^\cc$ as a set. The same issue and the same solution appears in other places in this book, and we do not mention it explicitly any more.}. 
\end{definition}

Every finite word is a $\cc$-word, likewise for every $\omega$-word. Another example is labelled  countable ordinals, e.g.~any $\cc$-word where the positions are~$\omega + \omega$. Here is a more fancy example, which uses a dense set of positions. 

 \begin{example}[Shuffles] 
    A classical exercise on linear orders is that the rational numbers are the unique -- up to isomorphism -- countable linear order which is dense and has no endpoints (i.e.~neither a least nor greatest element).   The same argument, see below, shows that for every countable $\Sigma$ there is a unique $\cc$-word over $\Sigma$ which has no endpoints, and which satisfies
    \begin{align*}
    \bigwedge_{a \in \Sigma}\quad  
    \myunderbrace{\forall x \ \forall y \  \exists z  \quad x < z < y \land a(z)}{label $a$ is dense}.
    \end{align*}
    We use the name \emph{shuffle of $\Sigma$} for the above $\cc$-word. Shuffles will play an important role in semigroups for $\cc$-words. 
    
    In case the reader is not familiar with back-and-forth arguments, we explain why the shuffle is unique. 
    Define a \emph{finite partial isomorphism} between two $\cc$-words to be a bijection between two  finite subsets  of their  positions which respects the order and labels. Because shuffles have no endpoints and all labels are dense, we conclude the following property:
    \begin{itemize}
        \item[(*)] If $f$ is finite partial isomorphism between two shuffles, and $x$ is a position in the first (respectively, second shuffle), then $f$ can be extended to a finite partial isomorphism whose domain (respectively, co-domain) contains $x$.
    \end{itemize}
    Using the above property and countability, for every two  shuffles one can define  a sequence 
    \begin{align*}
     f_0 \subseteq  f_1 \subseteq f_2 \subseteq \cdots
    \end{align*}
    of finite partial isomorphisms such that every position is eventually covered by some $f_n$. The limit (set union) of these finite partial isomorphisms is an isomorphism between the two shuffles. 
\end{example}
 
We now define the generalisation of semigroups for $\cc$-words.  We use the approach to associativity via commuting diagrams that was described in Lemma~\ref{lem:commuting-diagram-semigroup}.  Recall from that lemma  that a semigroup multiplication on a set $S$  could be defined as any operation $\mu : S^+ \to S$ which makes the following diagram commute:
\begin{align*}
    \xymatrix{
         S \ar[dr]^{\text{identity}} \ar[d]_{\txt{\scriptsize view a letter as \\ \scriptsize  a one-letter word}}\\
         S^+ \ar[r]_\mu & S
      }
      \qquad
    \xymatrix@C=4cm
    {(S^+)^+ \ar[r]^{\text{multiplication in the free semigroup}} \ar[d]_{\mu^+} &
     S^+ \ar[d]^{\mu} \\ 
     S^+ \ar[r]_\mu 
     & S}
 \end{align*}
 For $\cc$-semigroups, we take the same approach: we begin by defining a free multiplication operation (which corresponds to multiplication in the free $\cc$-semigroup), and then define other $\cc$-semigroups in terms of that. For a set $S$,  define \emph{free multiplication}  to be the operation $(S^\cc)^\cc \to S^\cc$ which replaces each position in the input $\cc$-word with the $\cc$-word that is in its label (a formal definition uses a lexicographic product of labelled linear orders).

 \begin{definition}\label{def:cc-semigroup}
     A \emph{$\cc$-semigroup} consists of an underlying set $S$  equipped with  a multiplication operation $\mu : S^\cc \to S$, which is  associative in the sense that the following two diagrams commute:
     \begin{align*}
        \xymatrix{
            S \ar[dr]^{\text{identity}} \ar[d]_{\txt{\scriptsize view a letter as \\ \scriptsize  a one-letter $\cc$-word}}\\
            S^\cc \ar[r]_\mu & S
         } \qquad \xymatrix@C=2cm
   {(S^\cc)^\cc \ar[r]^{\text{free multiplication}} \ar[d]_{\mu^\cc} &
    S^\cc \ar[d]^{\mu} \\ 
    S^\cc \ar[r]_\mu 
    & S}
\end{align*}
In the above diagram, $\mu^\cc$ denotes the coordinate-wise lifting of $\mu$ to $\cc$-words of $\cc$-words.
 \end{definition}

 \begin{myexample}
     The \emph{free $\circ$-semigroup} over alphabet $\Sigma$ has $\Sigma^\circ$ as its underlying set, and its multiplication operation is free multiplication. To check that this multiplication operation is associative, one needs to prove that the following diagram commutes:
     \begin{align*}
        \xymatrix@C=3.5cm
        {((\Sigma^\cc)^\cc)^\cc \ar[r]^{\text{free multiplication  for alphabet $\Sigma^\cc$}} \ar[d]_{\txt{\scriptsize (free multiplication \\ \scriptsize for alphabet $\Sigma$)$^\cc$}} &
         (\Sigma^\cc)^\cc \ar[d]^{{\txt{\scriptsize free multiplication\\ \scriptsize 
          for alphabet $\Sigma$}}} \\ 
         (\Sigma^\cc)^\cc \ar[r]_{\text{free multiplication for alphabet $\Sigma$}} 
         & \Sigma^\cc}
     \end{align*}
     To prove this formally, one uses the formal definition of free multiplication, in terms of lexicographic products of linear orders (see Example~\ref{ex:monads-for-linear-orders}). This $\cc$-semigroup is called \emph{free} for the usual reasons; a more formal description of these usual reasons will appear later in the book, when discussing monads. 
 \end{myexample}

 \begin{myexample}
     Recall the semigroups of size two that were discussed in Example~\ref{example:semigroups}:
     \begin{align*}
        \myunderbrace{(\set{0,1}, +)}{addition mod 2} \quad (\set{0,1}, \min) \quad \myunderbrace{(\set{0,1}, \pi_1)}{$(a,b) \mapsto a$} 
        \quad \myunderbrace{(\set{0,1}, \pi_2)}{$(a,b) \mapsto b$}
        \quad (\set{0,1}, (a,b) \mapsto 1)
        \end{align*}
        Which ones can be extended to $\cc$-semigroups in at least one way?
    
        The first example, i.e.~the two-element group, cannot be extended in any way, because the multiplication $a$ of the  $\omega$-word $1^\omega$ would need satisfy
    \begin{align*}
    a = \mu(1^\omega) =  \mu(\mu(1) \mu(1^\omega)) = \mu(1a) = 1 + a.
    \end{align*}
    The remaining semigroups can be extended to $\cc$-semigroups. As we will see in Example~\ref{example-cc-well-founded}, the extensions are not necessarily unique.
 \end{myexample}

We use  $\cc$-semigroups   to recognise languages of $\cc$-words. Define a \emph{homomorphism of $\cc$-semigroups} to be a function $h$ which makes the following diagram commute:
 \begin{align*}
    \xymatrix{
       S^\cc \ar[r]^{h^\cc} \ar[d]_{\text{multiplication in $S$}} &
       T^\cc \ar[d]^{\text{multiplication in $T$}}\\
       S \ar[r]_{h} & T
    }
    \end{align*}
Like for semigroups, homomorphisms of $\cc$-semigroup can be described in terms of compositional functions.  Suppose that $S$ is a $\cc$-semigroup and  $T$ is a set, which is not yet known to have the structure of a $\cc$-semigroup. We say that  a  function $h : S \to T$ is \emph{compositional} if there exists a function $\mu : T^\cc \to T$ which makes the following diagram commute 
\begin{align*}
    \xymatrix{
       S^\cc \ar[r]^{h^\cc} \ar[d]_{\text{multiplication in $S$}} &
       T^\cc \ar[d]^{\mu}\\
       S \ar[r]_{h} & T
    }
    \end{align*}
Using the same proof as for Lemma~\ref{lem:compositional-monoid}, one shows that if $h$ is a compositional and surjective, then $\mu$ is necessarily associative, thus turning  $T$ into a $\cc$-semigroup, and furthermore $h$ is a homomorphism. A generalised version of this result, which works not just for $\cc$-semigroups but also for a wider class of algebraic structures, will be proved in Part II of this book about monads. 

We say that a language $L \subseteq \Sigma^\cc$ is \emph{recognised} by a $\cc$-semigroup $S$ if there is a homomorphism $h : \Sigma^\cc \to S$ which recognises it, i.e.
\begin{align*}
h(w)=h(w') \quad \text{implies} \quad w \in L \iff w' \in L \qquad \text{for every $w,w' \in L$.}
\end{align*}
We are mainly interested in languages recognised by finite $\cc$-semigroups, i.e.~$\cc$-semigroups where the underlying set is finite. Note that it is not immediately clear how to present the multiplication operation of a finite $\cc$-semigroup in a finite way; this question will be addressed later in this section.

\begin{myexample}\label{example-cc-well-founded}
    Consider un-labelled countable linear orders, which can be viewed as  $\cc$-words over a one letter alphabet $\set a$. 
    Consider the function 
    \begin{align*}
    h : \set a^\cc \to \set{0,1}
    \end{align*}
    which sends well-founded $\cc$-words to $1$, and the remaining $\cc$-words to $0$. 
    We claim that $h$ compositional (and therefore the language of well-founded $\cc$-words is recognised by a finite $\cc$-semigroup). Indeed,  take some  $v \in (\set{a}^\cc)^\cc$ which gives $w \in \set a^\cc$ under free multiplication. To prove compositionality, need to show that $h^\cc(v)$ uniquely determines $h(w)$. This is because $h(w)=1$ if and only if the positions of $v$ are well-founded, and every such a position is labelled by a well-founded order. All of this information can be recovered from $h^\cc(v)$. 
    The compositional function $h$ induces an underlying structure of a $\cc$-semigroup on $\set{0,1}$. When restricted to finite multiplications, this $\cc$-semigroup is the same as $(\set{0,1},\min)$. Note that a symmetric $\cc$-semigroup can be constructed, for orders which are well-founded after reversing. The symmetric $\cc$-semigroup also coincides with $(\set{0,1},\min)$ on finite words.
\end{myexample}

\begin{myexample}\label{ex:semigroup-before}
    Consider the language $L \subseteq \set{a,b,1}^\cc$, which contains $\cc$-words where some position with label $a$ is to the left of some position with label $b$. Consider the following function  
    \begin{align*}
    w \in \set{a,b,1}^\cc \quad \mapsto \quad 
    \begin{cases}
        0 & \text{if $w \in L$}\\
        1 & \text{if all letters are $1$}\\
        b & \text{if all letters are $b$ or $1$, and there is some $b$}\\
        ba & \text{if both $b$ and $a$ appear, but $w \not \in L$}\\
        a & \text{otherwise}
    \end{cases}
    \end{align*}
    This function is easily seen to be compositional, and therefore its image is a $\cc$-semigroup. The element $0$ is absorbing, and $1$ is a monoid identity. The language $L$ is therefore recognised by the corresponding $\cc$-semigroup.
   \end{myexample}

\subsection{Monadic second-order logic on $\cc$-words} As usual in this book, we are interested in properties of $\cc$-words that can be defined using \mso. Part of the appeal of \mso is the ease with which it can be applied to different setups (such as finite or infinite words, trees, graphs, etc.) and $\cc$-words are no exception. It is immediately clar what \mso for $\cc$-words should mean. Define the \emph{ordered model} of a $\cc$-word  in the same way as for finite words: the universe is the positions, and the relations and their meaning are the same as for finite words. We say that a language $L \subseteq \Sigma^\cc$ is definable in \mso if there is an \mso sentence $\varphi$, using the vocabulary of the ordered model, such that
\begin{align*}
w \in L \quad \iff \quad \text{the ordered model of $w$ satisfies $\varphi$} \qquad \text{for every $w \in \Sigma^\cc$.}
\end{align*}

\begin{myexample}
        Consider the language of well-founded $\circ$-words that was discussed in Example~\ref{example-cc-well-founded}. This language is definable in \mso, by simply writing in \mso the definition of well-foundedness:
        \begin{align*}
        \myunderbrace{\forall X}{for every\\ \scriptsize set of \\ \scriptsize positions}  \qquad \
        \myunderbrace{(\exists x \in X)}{which is nonempty} \Rightarrow \myunderbrace{(\exists x  \in X   \ \forall y \in X \  x \le y)))}{there is a least position}.
        \end{align*}
        Another example is the $\circ$-words which contain a sub-order that is dense:
        \begin{align*}
            \myunderbrace{\exists X}{exists a\\ \scriptsize set of \\ \scriptsize positions}  \qquad \
            \myunderbrace{(\exists x \in X)}{which is nonempty} \land \myunderbrace{(\forall x \in X \ \forall y \in Y\ x < y  \Rightarrow  \ \exists z \in X\  x < z < y)))}{and dense in itself}.
            \end{align*}
        An $\cc$-word which violates the second property, i.e.~it does not have any dense sub-order, is called \emph{scattered}.
\end{myexample}

Once we have built up all the necessary ideas in the Trakhtenbrot-\buchi-Elgot Theorem for finite words, it is very easy to get the extension for $\circ$-words. The same proof as for finite words (using a powerset construction on $\cc$-semigroups) gives the following result. 
\begin{theorem}\label{thm:mso-to-cc}
    If a language $L \subseteq \Sigma^\cc$ is definable in \mso, then it is recognised by a finite $\cc$-semigroup. 
\end{theorem}
The  above theorem seems all too easy. Is there a catch? Yes: the theorem alone does not give an algorithm for deciding if an \mso definable language is empty. In the case of finite words, we could remark that all of the constructions used in the proof (products and powersets) are effective, with finite semigroups represented by their multiplication tables. But, so far, we do not have any finite representation of $\cc$-semigroups yet, and therefore we cannot talk about effectivity. Such finite representations, and their application to deciding \mso, will be developed in the next section. 

\label{page:undecidability-of-mso}
To see the difficulty in getting finite representations, consider uncountable words. Theorem~\ref{thm:mso-to-cc} remains true for uncountable words. However,  satisfiability of \mso sentences over uncountable words (e.g.~the reals)  is undecidable\footcitelong[Theorem 7.]{shelahMonadicTheoryOrder1975}. This means that for uncountable words the constructions in the lemma cannot be made effective. Hence, countability will play a crucial role in finding finite representations.

Another interesting question is about the converse of the theorem: can one define in \mso every language that is recognised by a finite $\cc$-semigroup? For finite words and 
 $\omega$-words, the answer was ``obviously yes'', because one can use \mso to formalise the acceptance by an automaton. Since we have no automata for $\cc$-words, the question is harder. However, the answer is still ``yes'', and it will be given in Section~\ref{sec:mso-countable-words}.

\exercisehead

\mikexercise{
    Give a formula of \mso which is true in some uncountable well-founded linear order, but is false in all countable well-founded linear orders.
}{
    The property is that there exists some set of positions $X$ which: (a) has order type $\omega$; and (b) is co-final, which means that every element of the order is to the left of some element of $X$.  Condition (a) is expressed in \mso by saying that $X$ is the least set that contains some position $x \in X$, and is closed under taking successors inside $X$. 
}

\mikexercise{Find two countable ordinals (viewed as $\cc$-words over a one letter alphabet), which have the same \mso theory. }{}  

\mikexercise{We write $\omegaop$ for the reverse of $\omega$. An $(\omegaop + \omega)$-word is   a $\cc$-word where the underlying order is the same as for the integers.
Show that the following problem is decidable: given an \mso sentence, decide if it is true in some bi-infinite word.   
}{}

\mikexercise{We say that a $(\omegaop + \omega)$-word $v$ is \emph{recurrent} if every finite word $w \in \Sigma^+$ appears as an infix in every prefix of $v$ and in every suffix of $v$. Show that all recurrent $(\omegaop + \omega)$-words have the same \mso theory.
}{}

\mikexercise{\label{ex:cc-contextual} Let $\Sigma$ be an alphabet, and let $x \not \in \Sigma$ be a fresh letter.  For $w \in \Sigma^\cc$ and $u \in (\Sigma \cup \set{x})^\cc$,  define $u[x:=w] \in \Sigma^\cc$ to be the result of substituting each occurrence of variable $x$ in $u$ by the argument $w$. For a language $L \subseteq \Sigma^\cc$, define \emph{contextual equivalence} to be the equivalence relation on $\Sigma^\cc$ defined by 
\begin{align*}
w \sim w'  \quad \text{iff} \quad  u[x:=w]\in L \iff u[x:=w'] \in L \text{ for every }u \in (\Sigma \cup \set x)^\cc.
\end{align*}
Show that $\sim$ is a $\cc$-congruence (which means that the function that maps $w$ to its equivalence class is compositional) for every language recognised by some finite $\cc$-semigroup. 

}{}

\mikexercise{Give an example of a language $L \subseteq \Sigma^\cc$ where contextual equivalence is not a $\cc$-congruence. }{}

\mikexercise{\label{ex:synt-cc}  Show that every language recognised by a finite $\cc$-semigroup has syntactic $\cc$-semigroup, but there are some languages (not recognised by finite $\cc$-semigroups), which do not have a syntactic $\cc$-semigroup. }{}

\mikexercise{\label{ex:yield} Consider a binary tree (every node has either zero or two children, and we distinguish left and right children), where leaves are labelled by an alphabet $\Sigma$.  The tree  might have infinite branches. Define the \emph{yield} of such a tree to be the $\cc$-word where the positions are leaves of the tree, the labels are inherited from the tree, and the ordering on leaves is lexicographic (for every node, its left subtree is before its right subtree). Show that every $\cc$-word can be obtained as the yield of some tree.  }{}

\mikexercise{\label{ex:equi-decidable-cc-rat} Show that the following problems are equi-decidable:
\begin{itemize}
    \item given an \mso sentence, decide if it
    is true in some $\cc$-word $w \in \Sigma^\cc$
    \item given an \mso sentence, decide if its true in $(\Rat,<)$.
\end{itemize}
 }{}

 \mikexercise
 {\label{ex:rabin-cc} Assume Rabin's Theorem, which says that the \mso theory of the complete binary tree
  \begin{align*}
  (\set{0,1}^*, \myunderbrace{x=y0}{left\\ \scriptsize child}, \myunderbrace{x=y1}{right\\ \scriptsize child})
  \end{align*}
 is decidable. Show that the problems from Exercise~\ref{ex:equi-decidable-cc-rat} are decidable. (We will also prove this in the next section, without assuming Rabin's theorem.) }{}

\section{Finite representation of 
\texorpdfstring{$\cc$-semigroups}{circle-semigroups}}
 \label{sec:shelah-operations}
   The multiplication operation in a finite semigroup can be seen as an operation of type $S^+ \to S$, or as a binary operation of type $S^2 \to S$. The binary operation has the advantage that a finite semigroup can be represented in a finite way, by giving a multiplication table of quadratic size.  In this section, we show that a similar finite representation is also possible  for $\cc$-semigroups.  Apart from binary multiplication, we will use two types of $\omega$-iteration -- one forward and one backward -- and a shuffle operation (which inputs a set of elements, and not a tuple of fixed length).

\begin{definition}
    [\lale operations] For a $\cc$-semigroup, define its \emph{\lale operations\footcitelong[p.~109.]{lauchli1966elementary}} to be the following four operations (with their types written in red).
    \begin{align*}
    \myunderbrace{ab}{binary\\ \scriptsize multiplication\\
    \shelahtype{S^2 \to S}} \qquad  \qquad
    \myunderbrace{a^\omega}{multiplication\\ \scriptsize of $aaa \cdots$ \\ \shelahtype{S\to S}}
    \qquad \qquad
    \myunderbrace{a^{\omegaop}}{multiplication \\ \scriptsize of $\cdots aaa$ \\ \shelahtype{S\to S}}
    \qquad\qquad
    \myunderbrace{\set{a_1,\ldots,a_n}^\eta}{multiplication of the\\ \scriptsize shuffle of $a_1,\ldots,a_n$
    \\ \shelahtype{\powerset S\to S}}
    \end{align*}    
\end{definition}

The main result of this section is the following theorem, which says that the role played by the \lale operations in a finite $\cc$-semigroup is the same as the role played by  binary multiplication in a semigroup.
A difference with respect to semigroups is that the \lale operations are complete only for \emph{finite} $\cc$-semigroups, see Exercise~\ref{ex:lale-incomplete-infinite}. 
 
 \begin{theorem}\label{thm:shelah-operations}
    The multiplication operation in a  finite $\cc$-semigroup is uniquely determined by its \lale operations.
 \end{theorem}

 Another way of stating the theorem is that if $S$ is  a finite set  equipped with  the \lale operations, then there is at most one way of extending these operations to an associative multiplication $S^\cc \to S$. We say at most one instead of exactly one, because the \lale operations need to satisfy certain associativity axioms, such as:
\begin{align*}
aa^\omega = a^\omega  \qquad (ab)^\omega = a(ba)^\omega  \qquad \set{a_1,\ldots,a_n}^\eta = \set{\set{a_1,\ldots,a_n}^\eta}^\eta
\end{align*}
Because the full list of axioms and its completeness proof are both long, we do not consider them here\footnote{They  can be found in
\incite[Section 7.]{bloomEsik2005}
}. This will not be a big issue,  because we will only consider  multiplication operations that arise from compositional functions -- e.g.~the multiplication operation on \mso types of given quantifier rank $k$ -- and such  multiplication operations are guaranteed to be associative. 

\exercisehead
\mikexercise{\label{ex:lale-incomplete-infinite} Find two  infinite $\cc$-semigroups which have the same underlying set and the same \lale operations, but  different multiplication operations. 
 }{
The underlying set is $\bot + \set{a}^\cc$. 
The first multiplication operation is
\begin{align*}
\mu_1(w) = \begin{cases}
    \bot & \text{ if $w$ contains at least one $\bot$}\\
    \text{free multiplication of $w$} & \text{otherwise}.
\end{cases}
\end{align*}
The second multiplication operation $\mu_2$ is defined in the same way, except that $\mu_2(w)=\bot$ also holds  for all $\cc$-words $w$ which are not regular. Here a regular $\cc$-words is one that can be generated using finitely many application of  the \lale operations to finite words.
}
\subsection{Proof of Theorem~\ref{thm:shelah-operations}}
\label{sec:proof-of-shelah-generators}
The key idea in the proof of Theorem~\ref{thm:shelah-operations} is that the \lale operations are enough to generate all sub-algebras, as stated in the following lemma. 

 \begin{lemma}\label{lem:shelah-subalgebra}
     Let $S$ be a finite $\cc$-semigroup, and let $\Sigma \subseteq S$. Then 
\begin{align*}
\myunderbrace{\set{ \text{multiplication of $w$} : w \in \Sigma^\cc}}{this is called the  \emph{sub-algebra generated by $\Sigma$}} \subseteq S
\end{align*}
is equal to the smallest subset of $S$ which contains  $\Sigma$ and  is closed under the \lale operations.
 \end{lemma}

 Before proving the lemma, we use it to prove Theorem~\ref{thm:shelah-operations}.
 \begin{proof}
     [Proof of Theorem~\ref{thm:shelah-operations}, assuming Lemma~\ref{lem:shelah-subalgebra}.] Suppose that  $S_1$ and $S_2$ are two $\cc$-semigroups, which have the same underlying set, and where the multiplication operations  agree on the \lale operations. We will show that the multiplication operations are the same.  Consider the product $\cc$-semigroup $S_1\times S_2$, defined in the usual coordinate-wise way. Apply Lemma~\ref{lem:shelah-subalgebra}  to the diagonal 
     \begin{align*}
     \Sigma = \set{(a,a) : a \in S} \subseteq S_1 \times S_2.
     \end{align*}
     Since the \lale operations agree for $S_1$ and $S_2$, it follows from the  lemma that the sub-algebra generated by $\Sigma$ is  also  the diagonal, which shows that the  multiplication operations of $S_1$ and $S_2$ are equal.
 \end{proof}

 The rest of Section~\ref{sec:proof-of-shelah-generators} is devoted to proving Lemma~\ref{lem:shelah-subalgebra}.
 Define $L \subseteq \Sigma^\cc$ to be the $\cc$-words  whose multiplication  can be obtained from $\Sigma$ by applying the \lale operations.  To prove  Lemma~\ref{lem:shelah-subalgebra}, we need to show  $L = \Sigma^\cc$.  This will follow immediately  from the following lemma. 

\begin{lemma} \label{lem:shelah-induction} Let $S$ be a finite $\cc$-semigroup, and let $L \subseteq S^\cc$ be such that:
    \begin{enumerate}
        \item \label{shelah-ind:bin} if $w_1,w_2 \in L$ then $w_1 w_2 \in L$;
        \item \label{shelah-ind:omega} if $w_1,w_2,\ldots \in L$ have the same multiplication, then $w_1 w_2 \cdots \in L$;
        \item \label{shelah-ind:omega-op} if $w_1,w_2,\ldots \in L$ have the same multiplication, then $\cdots w_2 w_1 \in L$;
        \item \label{shelah-ind:shuffle} if $w \in L^\cc$ is such that $\mu^{\cc}(w)$ is a shuffle, then its free multiplication is in $L$.
    \end{enumerate}
    If $L$ contains all letters in a subset $\Sigma \subseteq S$, then $L$ also contains $\Sigma^\cc$.
\end{lemma}

\begin{proof} 
    We begin with  some notation for $\cc$-words. Define an \emph{interval} in a $\cc$-word to be any set of positions $X$  that is connected in the following sense:
    \begin{align*}
    \forall x \in X \ \forall y \in Y \ \forall z \quad  x < z < y \Rightarrow y \in X.
    \end{align*}
    An \emph{infix} of a $\cc$-word is defined to be any $\cc$-word that is  obtained by restricting the positions  to some interval.   For example, the rational numbers -- viewed as a $\cc$-word $w$ over a one letter alphabet $\set{a}^\cc$ --  have uncountably many intervals, but five possible infixes, namely $a$, $w$, $aw$, $wa$ and $awa$.

    We now proceed with the proof of the lemma. Suppose that  $L$ has all of the closure properties in the  assumption of the lemma. We say that   $w \in \Sigma^\cc$ is  \emph{simple} if not only $w$, but also all of its infixes are in $L$.  We will show that  every $\cc$-word in   $\Sigma^\cc$ is simple, thus proving $L=\Sigma^\cc$. For the sake of contradiction, suppose that $w \in \Sigma^\cc$ is not simple. 
       Define $\sim$ to be the binary relation on positions in $w$, which identifies positions if they are equal, or  the infix corresponding to the interval 
       \begin{align*}
        \myunderbrace{\set{z : x < z \le y}}{an open-closed interval}
       \end{align*}
       is  simple (where $x$ is the smaller position and $y$ is the bigger position).
       
    \begin{claim}
        The relation  $\sim$ is an equivalence relation,  every equivalence class is an interval, and this interval induces a simple $\cc$-word.
    \end{claim}
    \begin{proof}
        The relation $\sim$ is symmetric and reflexive by definition. Transitivity holds because  simple words are closed under binary concatenation. This establishes that $\sim$ is an equivalence relation. Because  simple $\cc$-words are closed under infixes by definition,   every equivalence class of $\sim$ is an interval. 
        
        It remains to show  that every (infix induced by an) equivalence class is  simple. Here we use countability and items~\ref{shelah-ind:bin}--\ref{shelah-ind:omega-op} from the assumption of the lemma.  Consider an equivalence class $X$.  Choose some  position $x \in X$. We will show that both intervals
        \begin{align*}
         \myunderbrace{\set{ y \in X : y > x}}{$Y$} \qquad \myunderbrace{\set{y \in X : y \le x}}{$Z$}
        \end{align*}
        are simple, and therefore $X$ itself is simple thanks to binary concatenation. We only consider the first interval $Y$, and $Z$ is treated with a symmetric argument. 

        If $Y$ has a last position, then it is simple  by definition of $\sim$. If there is no last position, than thanks to  countability we can decompose $Y$ as a union
        \begin{align*}
        Y = Y_0 \cup Y_1 \cup Y_2 \cup \cdots,
        \end{align*}
        of consecutive open-closed intervals.  By definition of $\sim$, each interval $Y_n$ is simple. Since $L$ is closed under binary concatenation by assumption~\ref{shelah-ind:bin}, also every finite union of consecutive intervals $Y_n$ is simple. Therefore, by  the Ramsey Theorem, we can assume without loss of generality that all of the intervals $Y_1,Y_2,\ldots$ (but not necessarily $Y_0$) have the same multiplication  $a \in S$. It follows that $Y_1 \cup Y_2 \cup \cdots$ is simple, thanks to assumption~\ref{shelah-ind:omega} about closure of $L$ under $\omega$-concatenation of $\cc$-words with same multiplication. Finally, $Y_0$ can be added using binary concatenation, thus proving that $Y$ is simple.
    \end{proof}
    Since the equivalence classes of $\sim$ are intervals, they can be viewed  as an ordered set, with the order inherited from the original order on positions in $w$.  Because  simple words are closed under binary concatenation, the order on equivalence classes is dense, since otherwise two consecutive equivalence classes would need to be merged into a single one.  Define $w_\sim \in S^\cc$ to be the result of replacing every equivalence class of $\sim$ by its multiplication in $S$.  By assumption that $w$ is not simple, $\sim$ has more than one equivalence class, and therefore the positions of $w_\sim$ are an infinite dense linear order. 

       \begin{claim}\label{claim:rational-dense} Some infix of  $w_\sim$ is  a shuffle.
    \end{claim}
    \begin{proof}
     Take some   $a\in S$. If there is some infinite infix of $w_\sim$ where no position is labelled by  $a$, then we can continue working in that infix (its positions  are still an infinite dense linear order). Otherwise, positions with label $a$ are dense. By iterating this argument for all finitely many elements of $S$, we  find an infinite infix where every $a \in S$ either does not appear at all, or is dense. This infix is a shuffle.  
    \end{proof}
       
     By the closure of $L$ under shuffles, the free multiplication  of the  infix from the above claim is simple. It follows that the corresponding interval should have been a single equivalence class of $\sim$, contradicting the assumption.
\end{proof}

\subsection{Decidability of MSO}
\label{sec:mso-shelah-operations}
Thanks to Theorem~\ref{thm:shelah-operations}, a finite $\cc$-semigroup can be represented in a finite way, by giving its underlying set and the multiplication tables for its \lale operations. We will use this representation to give decision procedure for \mso on $\cc$-words.

Recall the proof of Theorem~\ref{thm:mso-to-cc}, which showed that every \mso definable language is recognised by a finite $\cc$-semigroup. We will show that the constructions in the proof can be made effective, with finite $\cc$-semigroups being represented using  the \lale operations.  In the proof of Theorem~\ref{thm:mso-to-cc}, we inductively transformed the \mso formula into a recognising $\cc$-semigroup,  starting from  $\cc$-semigroups corresponding to the atomic relations in \mso, and then  by applying products $S_1 \times S_2$ and powersets $\powerset S$. It is not hard to find representations (in terms of \lale operations) for the  $\cc$-semigroups that correspond to the atomic relations. It is also easy to see that given representations of $\cc$-semigroups $S_1$ and $S_2$, one can compute a representation of the product $\cc$-semigroup $S_1 \times S_2$, because the \lale operations work coordinate-wise. The interesting case is the powerset construction, which is treated in the following lemma.



\newcommand{\omegaset}[1]{#1^\omega \red{\subseteq S^\cc}}
\newcommand{\omegaword}[1]{#1^\omega \red{\in (\powerset S)^\cc}}
\newcommand{\omegaelement}[1]{#1^\omega \red{\in \powerset S}}
\begin{lemma}\label{lem:compute-multiplication-tables}
    Given a representation (using \lale operations) of a finite $\cc$-semigroup $S$, one can compute a representation of the   powerset $\cc$-semigroup $\powerset S$. 
\end{lemma}
\begin{proof}
    In the proof, we  use lower-case letters $a,b,c$ for elements of $S$, and we use  upper-case letters $A,B,C$ for elements of the powerset $\powerset S$. We only show how to compute the multiplication table for the shuffle operation
    \begin{align*}
        \set{A_1,\ldots,A_n} \quad \mapsto \quad \set{A_1,\ldots,A_n}^\eta
        \end{align*}
        in the powerset $\cc$-semigroup $\powerset S$.
    The remaining \lale operations are treated in a similar way.

    By definition of the powerset $\cc$-semigroup, an element belongs to the set $\set{A_1,\ldots,A_n}^\eta$ if and only if it can be obtained as follows: take the $\cc$-word 
    \begin{align}\label{eq:compute-shuffle-word}
         \text{shuffle of }\set{A_1,\ldots,A_n} \quad \in (\powerset S)^\cc,
       \end{align}
       choose for each position an element of its label, and then apply the multiplication operation of $S$. In other words, $a$ belongs to $\set{A_1,\ldots,A_n}^\eta$ if and only if  there  exists a word 
        \begin{align*}
        v \in (S \times \powerset S)^\eta
        \end{align*}
        which satisfies the following properties:
        \begin{itemize}
            \item[(a)] after projecting $v$ to the first coordinate and multiplying in $S$, the result is $a$;
            \item[(b)] in every letter of $v$, the first coordinate belongs to the second coordinate;
            \item[(c)] after projecting $v$ to the second coordinate, the result is~\eqref{eq:compute-shuffle-word}.
        \end{itemize}
         We will construct a homomorphism
        \begin{align*}
        h : (S \times \powerset S)^\cc \to T
        \end{align*}
        that recognises the set of $\cc$-words which satisfies conditions (b) and (c) above. The homomorphism $h$ maps a $\cc$-word $v$ to the following information: (i) is condition (b) satisfied; (ii) is $v$ an infix of some $\cc$-word that satisfies (c); (iii) if $v$ is a single letter, then what is the letter; and (iv) does $v$ have a first/last position. The function $h$ defined this way is compositional -- and therefore it is a homomorphism -- and  the \lale operations on its image $T$ can be computed. Also, the accepting set $F \subseteq T$ can be computed, it consists of elements where the answers to questions (i) and (ii) are both ``yes'', the answer to question (iii) is ``not a single letter'', and the answer to question (iv) is ``there is neither a first nor last position''. 

        By the above discussion, the set $\set{A_1,\ldots,A_n}^\eta$ that we want to compute consists of those elements $a \in S$ that can be obtained by taking some $v \in h^{-1}(F)$, projecting to the first coordinate, and then applying the multiplication operation of $S$. Here is alternative description of this set: take the subalgebra of $S \times T$ that is generated by 
        \begin{align}\label{eq:compute-powerset-generators}
        \set{(b,h(b,B)) : b \in B \subseteq S} \quad \subseteq S \times T,
        \end{align}
        keep only the pairs from this subalgebra where the second coordinate belongs to the accepting set $F$, and then project these pairs to the first coordinate. The alternative description  can be computed, because we can compute a representation of  the product $\cc$-semigroup $S \times T$, and we can compute  subalgebras  by saturating with respect to the \lale operations thanks to Theorem~\ref{thm:shelah-operations}. 
\end{proof}

Using the above lemma, we can  deduce decidability of \mso over $\cc$-words. 

\begin{theorem}\label{thm:decidable-mso-cc-words}
    The following problem is decidable:
    \begin{description}
        \item[Input.] An \mso sentence $\varphi$, which defines a language $L \subseteq \Sigma^\cc$.
        \item[Question.]  Is the language $L$ nonempty?
    \end{description}
\end{theorem}
\begin{proof}
    By induction on formula size, we compute for each formula of \mso (possibly with free variables), a homomorphism into a finite $\cc$-semigroup that recognises its language, together with an accepting subset of the $\cc$-semigroup. The $\cc$-semigroup is represented using the \lale operations, and the homomorphism is represented by its images for the letters of the alphabet. In the induction step, we use Lemma~\ref{lem:compute-multiplication-tables} to compute a finite representation of a powerset $\cc$-semigroup.
\end{proof}

\exercisehead

\mikexercise{\label{ex:singleton-shuffle} Let $\Sigma$ be a finite alphabet, and let $w$ be the shuffle of all letters in $\Sigma$. Show a finite $\cc$-semigroup which recognises 
the singleton language $\set{w}$. 
}{
    The elements of the   $\cc$-semigroup  are:    
    \begin{align*}
        S = \Sigma \cup \set 0 \cup  \set{\text{open,closed}} \times \set{\text{open,closed}}.
    \end{align*}
The idea is that we remember the following information: the entire word if  it is just a single letter, $0$ if the word  cannot be extended to $w$, and otherwise  we only remember if there are left and right endpoints. This is formalised by the following function:
    \begin{align*}
    v \mapsto \begin{cases}
        a & \text{if $v$ consists of just the letter $a$}\\ 
        \text{(open, open)} & \text{if $v=w$}\\
        \text{(closed, open)} & \text{if $v=aw$ for some $a \in \Sigma$}\\
        \text{(open, closed)} & \text{if $v=wa$ for some $a \in \Sigma$}\\
        \text{(closed, closed)} & \text{if $v=awb$ for some $a,b \in \Sigma$}\\
        0 & \text{otherwise}.
    \end{cases}
    \end{align*}
It is not hard to see that this function is compositional, and therefore  $S$  can be seen as $\cc$-semigroup. The \lale operations of $S$ are defined as follows. The binary multiplication returns $0$ for all arguments, with the following exception:
\begin{align*}
(x,\myunderbrace{y)(x'}{one of these must not be ``closed''},y') \mapsto (x,y') \qquad \text{for }x,y,x',y' \in \set{\text{open, closed}}.
\end{align*}
The $\omega$-power returns $0$ for all arguments, with the following exception:
\begin{align*}
(\myunderbrace{x,y}{one of these must not be ``closed''})^\omega = (x,\text{open})  \qquad \text{for }x,y \in \set{\text{open, closed}}.
\end{align*}
Reverse $\omega$-power is defined symmetrically. Finally, the shuffle power is
\begin{align*}
X^\eta = \begin{cases}
    0 & \text{if $0 \in X$ or $X \subsetneq \Sigma$}\\
    \text{(open,open)} & \text{otherwise}.
\end{cases}
\end{align*}
}

\mikexercise{\label{ex:regular-cc-word}A $\cc$-word $w$ is called \emph{regular} if the singleton language $\set w$ is recognised by a finite $\cc$-semigroup. Show that $w$ is regular if and only if  it can  be constructed from the letters by using the \lale operations.
}{}

\mikexercise{\label{ex:regular-cc-word-dens} Show that every  nonempty \mso definable language $L \subseteq \Sigma^\cc$ contains some regular $\cc$-word.
}{}

\mikexercise{Show that if $w$ is a regular $\cc$-word, then $\set w$ is \mso definable (without invoking Theorem~\ref{thm:mso-complete-cc-words}).}{}

\mikexercise{Show that for every finite alphabet $\Sigma$ there exists a $\cc$-word $w \in \Sigma^\cc$ such that 
\begin{align*}
    h(wvw)=h(w)  \qquad \text{for every} 
    \myunderbrace{h: \Sigma^\cc \to S}{homomorphism into\\
    \scriptsize a finite $\cc$-semigroup} \text {and }v \in \Sigma^\cc.
\end{align*}
}{Shuffle of all regular $\cc$-words in $\Sigma^\cc$.}

\mikexercise{   \label{ex:zero-one-law} 
For a countable linear order $X$, let $\set{a,b}^X \subseteq \set{a,b}^\cc$ be the set of  $\cc$-words with  with positions $X$. We can equip this set with a probabilistic measure,  where for each position $x \in X$, the label is selected independently, with $a$ and $b$  both having probability half. We say that $X$ has a zero-one law if for every \mso definable language $L$,  
the probability of $\varphi \cap \set{a,b}^X$ is either zero or one. For which of the following $X = \Nat, \Int, \Rat$ is there a zero-one law?
}{}

\mikexercise{A countable linear order can be viewed as a $\cc$-word over a one-letter alphabet. Among these, we can distinguish the countable linear orders that are regular, i.e.~generated by the \lale operations, see Exercise~\ref{ex:regular-cc-word}.  
    Give an algorithm, which inputs a 
an countable linear order that is regular in the above sense, and decides if it has a zero-one law (in the sense of Exercise~\ref{ex:zero-one-law}).
}{}

\mikexercise{Show that every \mso definable language of $\cc$-words belongs  to the least class of languages which:
        \begin{itemize}
            \item contains the following  two languages over alphabet $\set{a,b,c}$:
            \begin{align*}
        \myunderbrace{\exists x a(x)}{some $a$} \qquad 
        \myunderbrace{\exists x \ \exists y \  a(x) \land b(y) \land x < y}{ $a$ before $b$}
            \end{align*}
            \item is closed under  Boolean combinations;
            \item is closed under images and inverse images of letter-to-letter homomorphisms.
        \end{itemize}
    }{}





\mikexercise{We say that a binary tree (possibly infinite) is \emph{regular} if it has finitely many non-isomorphic sub-trees. Show that a $\cc$-word is regular (in the sense  of Exercise~\ref{ex:regular-cc-word}) if and only if it is the yield (in the sense of Exercise~\ref{ex:yield}) of some regular tree.
}{}

\mikexercise[91]{Consider the embedding ordering (Higman ordering)  $w \higman v$ on $\cc$-words. Show that for every $\cc$-words $w$ there is a regular $\cc$-word $v$ such that  $w \higman v$ and $v \higman w$.  Hint: use Lemma~\ref{lem:shelah-induction}. }{}

\mikexercise{
Suppose that we are given a language $L \subseteq \Sigma^\cc$, represented by a finite $\cc$-semigroup $S$,  a homomorphism $h : \Sigma^\cc \to S$, and  an accepting set $F \subseteq S$.  
Give a algorithm which computes the syntactic $\cc$-semigroup (which exists by Exercise~\ref{ex:synt-cc}).

}{}

\mikexercise{
\label{ex:half-eilenberg-cc}    
Let $\Ll$ be a class of languages, such that $\Ll$ satisfies the following conditions:
\begin{itemize}
    \item every language in $\Ll$ is recognised by a finite $\cc$-semigroup;
    \item $\Ll$ is  closed under Boolean combinations;
    \item $\Ll$ is closed under inverse images of homomorphisms $h : \Sigma^\cc \to \Gamma^\cc$;
    \item  Let $L \subseteq \Sigma^\cc$ be a language in  $\Ll$. For every  $w,w_1,\ldots,w_n \in \Sigma^\cc$, $\Ll$ contains the inverse image of $L$ under the following operations: 
    \begin{align*}
    v \mapsto wv \quad v \mapsto vw \quad  v\mapsto v^\omega \quad v \mapsto v^\omegaop \quad  v \mapsto \text{shuffle of }\set{w_1,\ldots,w_n,v}.
    \end{align*}
\end{itemize}
Show that if $L$ belongs to $\Ll$, then the same is true for every language recognised by its syntactic $\cc$-semigroup.
}{}

\mikexercise{\label{ex:ltlf-cc}
    Let $\Sigma$ be an alphabet and let $c \not \in \Sigma$ be a fresh letter. We say that  $L \subseteq \Sigma^\cc$ is definable in \ltlf if there is a formula of \ltlf which defines the language $cL$, see  Exercise~\ref{ex:ltlf-fresh-letter}. Give an algorithm which inputs the finite syntactic $\cc$-semigroup of a language  $L \subseteq \Sigma^\cc$, and answers if the language  is definable in \ltlf. Hint:  the $\cc$-semigroup must be suffix trivial,  but this is  not sufficient. }{
Condition X says that the $\cc$-semigroup is suffix trivial and satisfies the following identities 
    \begin{align*}
    a^\momega = a^\omega a= aa^{\omegaop} \qquad a_1\set{a_1,\ldots,a_n}^\eta a_n = (a_1 \cdots a_n)^\momega
    \end{align*}
    }

\mikexercise{Give an algorithm which inputs the finite syntactic $\cc$-semigroup of a language  $L \subseteq \Sigma^\cc$, and answers if the language  is definable in two-variable first-order logic \fotwo. Hint:  the $\cc$-semigroup must be in \dav,  but this is  not sufficient.}{
Condition X is the following identity: 
\begin{align*}
(a_1 \cdots a_n)^\momega = (a_1 \cdots a_n)^\momega w (a_1 \cdots a_n)^\omega  \qquad \text{for every }w \in \set{a_1,\ldots,a_n}^\cc.
\end{align*}
This identity can be effectively checked, since the possible multiplications of $\cc$-words $w$ can be computed using Lemma~\ref{lem:shelah-subalgebra}.
}

\mikexercise{Show that aperiodicity is not sufficient for first-order definability for $\cc$-words: give an example of a language $L \subseteq \Sigma^\cc$ that is recognised by a finite aperiodic $\cc$-semigroup, but which is not definable in first-order logic.  }{}

\section{From \texorpdfstring{$\cc$-semigroups}{circle-semigroups} to  MSO}
\label{sec:mso-countable-words}
In Theorem~\ref{thm:mso-to-cc} we  have shown that if a language of $\cc$-words is definable in \mso, then it is recognised by a finite $\cc$-semigroup. We now show that the converse implication is also true. 

\begin{theorem}\label{thm:mso-complete-cc-words} If a  language of $\cc$-words is recognised by a finite $\cc$-semigroup, then it is  definable in \mso\footnote{This theorem was first shown in 
    \incite[Theorem 5.1.]{carton_colcombet_puppis_2018}
The proof presented here is different, and it is based on the proof in 
\incite[p.~192]{Schutzenberger65} 
which shows that every aperiodic monoid recognises a star-free language. We use the different proof because, after suitable modifications, it allows us to characterise star-free languages of $\cc$-words, see Exercise~\ref{ex:star-free-cc-effective}.}.
\end{theorem}

As mentioned before in this chapter, the theorem would be easy if there was an automaton model, which would assign states to positions, and where the acceptance condition could be formalised in \mso. Unfortunately, no such automaton model is known. Therefore, we need a different proof for the theorem. The rest of Section~\ref{sec:mso-countable-words} is devoted to such a proof.

We begin by defining regular expressions for $\cc$-words. For a finite family $\Ll$ of languages  of $\cc$-words, define the  shuffle of $\Ll$ to be the   $\cc$-words which can be partitioned into intervals so that: (a) every interval induces a word from $L$ for some $L \in \Ll$; (b) the order type on the intervals is that of the rational numbers; and (c) for every $L \in \Ll$, the intervals from $L$ are dense. 
\begin{lemma}\label{lem:shuffle-expressions}
    Languages definable in \mso are closed under Boolean combinations and the following kinds of concatenation:
    \begin{align*}
     LK \qquad  L^+ \qquad L^\omega \qquad L^\omegaop \qquad \text{shuffle of } \myunderbrace{\Ll}{a finite family \\ \scriptsize of languages}
    \end{align*}
\end{lemma}
\begin{proof}
    For the Boolean operations, there is nothing to do, since Boolean operations are part of the logical syntax. For the remaining operations, which are all variants of concatenation, we observe that \mso can quantify over factorisations, as described below. 
    
    Define a \emph{factorisation} of a $\cc$-word  to be a partition of its positions into intervals, which are called \emph{blocks}. For a factorisation,   define a \emph{compatible colouring} to be any colouring of positions that uses  two colours, such that all blocks are monochromatic, and such that for every two distinct blocks with the same colour, there is a block between them  with a different colour. A compatible colouring always exists (there could be uncountably many choices). A factorisation can be recovered from any compatible colouring: two positions are in the same block if and only if the interval connecting them is monochromatic. A compatible colouring can be represented using a single set -- namely the positions with one of the two colours. This representation can be formalised by an  \mso formula $\varphi(x,y,X)$ which says that positions $x$ and $y$ are in the same block of the factorisation (i.e.~they have the same colour and they are not separated by any position with a different colour).

    Using the above representation, we show closure of \mso under the concatenations in the lemma. For $LK$, we simply say that there exists  a factorisation with two blocks, where the first block is in $L$ and the second block is in $K$. (To say that a block is in $L$ or $K$, we observe that \mso sentences can be relativised to a given interval.) For $L^+$, we say that there exists a factorisation with finitely many blocks, where all blocks are in $L$. Here is how we express  that there are finitely many blocks: there are first and last blocks, and there is no proper subset of positions that  contains the first block and is closed under adding successor blocks. For $L^\omega$, we do the same, except that there is no last block. For $L^\omegaop$, we use a symmetric approach. For the shuffle, we say that the blocks are dense and there is no first or last block.
\end{proof}

In the proof of Theorem~\ref{thm:mso-complete-cc-words}, we will only use the closure properties of \mso from the  above lemma. In particular, it will follow that every language recognised by a finite $\cc$-semigroup can be defined by a regular expression which uses single letters and the closure operations from the lemma.

To prove Theorem~\ref{thm:mso-complete-cc-words}, we will show that the multiplication operation of every finite $\cc$-semigroup can be defined in \mso, in the following sense. 
Let $S$ be a finite $\cc$-semigroup. We will show that  for every  $a \in S$, the language 
\begin{align*}
L_a = \set{ w \in S^\cc : w  \text{ has multiplication $a$}}
\end{align*}
is \mso definable. This will immediately imply that every language recognised by a homomorphism into $S$ is \mso definable, thus proving the theorem.

The proof is by induction on the  position of $a$ in the infix ordering. Fix for the rest of this section an infix class $J \subseteq S$.  We partition $S$ into two parts: 
\begin{align*}
    \myunderbrace{\text{easy elements}}{proper prefixes of $J$} \quad \cup \quad  \myunderbrace{\text{hard elements}}{the rest}.
\end{align*}
The induction hypothesis says $L_a$ is \mso definable for every easy  $a \in S$. We will prove the same thing for every  $a \in J$.

We begin with an observation about smooth multiplications, which follows from the Ramsey argument that was used in Theorem~\ref{thm:buchi-determinisation}. We say that $w \in S^\cc$ is \emph{$J$-smooth} if the multiplication of every \emph{finite} infix  $w$ is  in $J$. This is a lifting to infinite words of the notion of smoothness that was used in Section~\ref{sec:fact-for}. By the same proof as in  Claim~\ref{claim:smooth} from that section, a $\cc$-word is $J$-smooth if and only if all of its infixes of length at most two are $J$-smooth.   The following lemma describes the multiplication of certain $J$-smooth words. 

\begin{lemma}\label{lem:smooth-limit}  Let   $e \in J$ be idempotent, and let  $w \in J^\cc$ be $J$-smooth. If $w$ is an $\omega$-word, then its  multiplication is $ae^\omega$, where $a$ depends only on  $e$ and the prefix class of  the first letter in $w$. If $w$ is an   $(\omegaop + \omega)$-word, i.e.~its positions are ordered like the integers,  then  its multiplication is $e^\omegaop e^\omega$. 
    \end{lemma}
    Since the lemma is true for every choice of idempotent $e \in J$, it follows that $e^\omegaop e^\omega$ does not depend on the choice of $e$. In particular, 
    \begin{align*}
        e^\omegaop e^\omega = f^\omegaop f^\omega 
    \end{align*}
    holds for every two idempotents $e,f$ in the same infix class.
    \begin{proof}   The main observation is the following claim. 

        \begin{claim}\label{claim:group-omega}
           If $w$ is an $\omega$-word that is  $J$-smooth and has first letter $e$, then its multiplication is  $e^\omega$. 
        \end{claim}
        \begin{proof}
            By  Lemma~\ref{lem:ramsey-buchi}, the  multiplication of $w$ is equal to  $af^\omega$, for some  $a,f$. Since $w$ is   $J$-smooth,    $a$ and $f$ belong to $J$. Since the first letter of $w$ is $e$, we have $ea=a$.  Since $f$ is infix equivalent to $e$, it admits a decomposition  $f=xeey$. Therefore 
            \begin{align*}
            af^\omega = ea(xeey)^\omega = \underbrace{eaxe}_g(\underbrace{eyxe}_h)^\omega.
            \end{align*}
            We now continue as in the proof of Lemma~\ref{lem:det-buchi-char}: because $g,h,e$ are in the same group, then $g^\omega = e^\omega = h^\omega$, and therefore $gh^\omega =e^\omega$.
        \end{proof}

    The claim immediately proves the lemma. Indeed, consider a $J$-smooth $\omega$-word with first letter $b$. The first letter admits a decomposition $b=aex$ for some $a,x \in J$, and furthemore $a$ depends only on the prefix class of $b$.  By the above claim,  the multiplication of every $J$-smooth $\omega$-word that begins with $b$  is equal to $ae^\omega$. A similar argument works when the positions are ordered as the integers: every $J$-smooth $(\omegaop + \omega)$-word has the same multiplication as a smooth $(\omegaop + \omega)$-word with an infix $ee$, and the latter has multiplication $e^\omegaop e^\omega$ thanks to the claim and its symmetric version for $\omegaop$.
    \end{proof}

    We say that a colouring
    $\lambda : S^\cc \to C$ which uses a fintie set $C$ of colours 
     is \mso definable on a subset $L \subseteq S^\cc$ if there exists an \mso definable colouring that agrees with $\lambda$ on inputs from $L$.
    The strategy for the rest of the proof is as follows.
    Define $L_J \subseteq S^\cc$ to be the $\cc$-words that have multiplication in $J$. We first show in Lemma~\ref{lemma:approximate-prefix-class} that the colouring
    \begin{align*}
        w \in S^\cc \quad \mapsto \quad \text{prefix class of the multiplication of $w$}
        \end{align*}
        is \mso definable on $L_J$. Next, in Lemma~\ref{lem:product-approximated},  we use this result about prefixes and a symmetric one for suffixes  to show  that the multiplication operation is \mso definable  on $L_J$. Finally, in Lemma~\ref{lem:lj-definable} we show  that the language  $L_J$  is definable in \mso. We  then conclude as follows: a $\cc$-word has multiplication $a \in J$ if and only if it belongs to $L_J$, and the colouring from Lemma~\ref{lem:product-approximated}  maps it to $a$.  It remains to prove the lemmas.

\begin{lemma}\label{lemma:approximate-prefix-class}
    The following colouring is \mso definable on $L_J$: 
    \begin{align*}
    w \in S^\cc \quad \mapsto \quad \text{prefix class of the multiplication of $w$}.
    \end{align*}
\end{lemma}
\begin{proof}
    We write $H \subseteq S^\cc$ for the $\cc$-words which multiply to a  hard element, i.e.~an element that is at least as big as $J$ in the infix ordering, or incomparable with $J$. This language is definable in \mso, as the complement of the language of $\cc$-words that multiply to an easy element, which is definable by induction assumption. 
    For an interval in $w$, define its multiplication to be the multiplication of the infix of $w$ that is induced by the interval. An interval is called \emph{easy} if its multiplication is easy, otherwise it is called \emph{hard}. The family of easy intervals is closed under subsets. By the induction assumption, we can check in \mso if an interval is easy or hard.  An interval is called \emph{almost easy} if it all of its proper sub-intervals are easy. 
    \begin{claim}\label{claim:almost-easy}
        The multiplication operation is \mso definable on almost easy intervals.
    \end{claim}
    \begin{proof}
        If there is a last position, then the multiplication can be easily computed: remove the last position, compute the multiplication, and then add the last position. Otherwise, if there is no last position,  then we can use Lemma~\ref{lem:ramsey-buchi} to see that  an almost easy interval has multiplication $b \in S$ if and only if it belongs to 
        \begin{align*}
        L_a (L_e)^\omega \qquad \text{for some easy $a,e$ such that $ae^\omega = a$.}
        \end{align*}
        The above condition can be formalised in \mso thanks to the induction assumption and Lemma~\ref{lem:shuffle-expressions}.
    \end{proof}

    Define a \emph{prefix interval} to be a   interval that is downward closed in the ordering of positions. We will compute in \mso the prefix class of some nonempty hard prefix interval; if the $\cc$-word is in $L_J$ then this hard prefix has the same  prefix class as  $w$. 
    Define $X$ to be the union of all easy prefix intervals; if there is no easy prefix interval then this union is empty. This union is an almost easy interval, and therefore its multiplication, call it $a$, can be computed using Claim~\ref{claim:almost-easy}. If $a$ is hard, then we are done, since $a \in J$ by the assumption that $w \in L_J$, and therefore thanks to the Egg-box lemma we can conclude that the prefix class for the  multiplication of $w$ is the same as for $a$.

    Suppose now that $a$ is not hard. Define $Y$ to be the suffix interval which is the complement of $X$.   By definition of $a$, if $b$ is the multiplication of some prefix of $Y$, then $ab$ is hard, and therefore by the same argument as in the previous paragraph, the prefix class of the multiplication of $w$ is the same as that for $ab$. Therefore, it remains to compute in \mso the multiplication of some (does not matter which one) prefix of $Y$.  If $Y$ has a first position, then we can simply use the letter in that position. Otherwise, by the Ramsey Theorem, $Y$ can be decomposed as a union of consecutive intervals 
    \begin{align*}
    Y = \cdots \cup Y_2 \cup Y_1 \cup Y_0
    \end{align*}
    such that all of the intervals $Y_1,Y_2,\ldots$ have the same multiplication, call it $c$.
    If $c$ is easy, which can be defined in \mso thanks to the induction assumption and the closure properties from Lemma~\ref{lem:shuffle-expressions}, we know that $Y$ has a prefix which multiplies to $c^\omegaop$. Otherwise, $c$ is hard, and therefore by Lemma~\ref{lem:smooth-limit}, we know that $Y$ has a prefix which multiplies to $e^\omegaop$, where $e$ is some arbitrarily chosen idempotent from $J$.
\end{proof}
In the above lemma, we have shown how to compute in \mso the prefix class of a $\cc$-word, conditionally under the assumption that its multiplication is in $J$. A symmetric argument works for suffix classes. Now we use that result to compute the actual value, still conditionally under the assumption that the multiplication is in $J$.
\begin{lemma}\label{lem:product-approximated}
    The multiplication operation of $S$ is \mso definable on $L_J$.
\end{lemma}
\begin{proof}
    Let $w \in L_J$. We use the terminology about intervals from the proof of Lemma~\ref{lemma:approximate-prefix-class} .
    
    \begin{claim}
        There exists a factorisation $w = w_1 w_2 w_3$ such that: 
        \begin{itemize}
            \item $w_1$ is either empty or in $H^\omega$;
            \item $w_2$ is a finite concatenation of almost easy $\cc$-words;
            \item $w_3$ is either empty or in $H^{\omegaop}$.
        \end{itemize}
    \end{claim}
    Here is a picture of the factorisation, in the case when $w_1$ and $w_3$ are nonempty:
    \mypic{40}
    \begin{proof}
        Define a  \emph{limit prefix} of $w$ to be any prefix interval which induces a $\cc$-word in $H^\omega$. 
        Limit prefixes are closed under (possibly infinite) unions. If there is a limit prefix, then there is a maximal one, namely the  union of all limit prefixes. Define $w_1 \in H^\omega$ to be the maximal limit prefix of $w$ (if no limit prefix exists, then $w_1$ is empty).
        Remove  the prefix $w_1$, and to the remaining part of the word apply a symmetric process, yielding a suffix $w_3 \in H^\omegaop$ and a remaining part $w_2$.  This is the factorisation in the statement of the claim. 
        
        It remains to show that $w_2$ is a finite concatenation of almost easy $\cc$-words.
        By construction, the remaining part $w_2$ does not have any prefix in $H^\omega$, nor does it have any suffix in $H^\omegaop$.       Take the union of all easy prefixes of $w_2$ (this union is nonempty, because there must be some nonempty easy prefix of $w_2$ thanks to the assumption that $w_2$ has no  suffix in $H^\omegaop$), and cut it off. After repeating this process a finite number of times, we must exhaust all of $w_2$, since otherwise there would be a prefix in $H^\omega$.  Therefore, $w_2$ is a finite concatenation of almost easy $\cc$-words.
    \end{proof}
    
    Let $w_1,w_2,w_3$ be as in the above claim. By Lemma~\ref{lem:smooth-limit}, the multiplication of $w_1$ is uniquely determined by its prefix class (under the assumption that the entire $\cc$-word belongs to $L_J$). Therefore, thanks to Lemma~\ref{lemma:approximate-prefix-class}, we can compute in \mso the  multiplication of the  $w_1$. Symmetrically,  we can compute the multiplication of $w_3$. It remains to compute the multiplication of~$w_2$.  This is done in the following claim. 
    
    \begin{claim}
        If a $\cc$-word is a finite concatenation almost easy $\cc$-words, then its  multiplication  can be computed in \mso. 
    \end{claim}
    \begin{proof}
        By the Kleene theorem about regular expressions being equivalent to finite automata, the set of finite concatenations of almost easy intervals can be described using a regular expression, where the atomic expressions describe almost easy words of that multiply to a given element. Such a regular expression can be formalised in \mso thanks to Lemma~\ref{lem:shuffle-expressions}
    \end{proof}
    
\end{proof}

\begin{lemma}\label{lem:lj-definable}
    The language $L_J$ is \mso definable.   
\end{lemma}
\begin{proof}
    Define $I \subseteq S$ to be the hard elements which are not in $J$. This is an ideal in the $\cc$-semigroup $S$, i.e.~if $w \in S^\cc$  has at least one letter in $I$, then its multiplication is in $I$. Define $L_I$ to be the $\cc$-words that multiply to an element of  $I$. Again, this is an ideal, this time in the free $\cc$-semigroup $S^\cc$. We will show how to define $L_I$ in \mso; it will follow that $L_J$ is \mso definable as 
    \begin{align*}
    L_J = H - L_I.
    \end{align*}
    The key is the following characterisation of $L_I$. Define an \emph{error} to be a $\cc$-word in $S^\cc$ which satisfies at least one of the following  conditions:
    \begin{itemize}
        \item binary error: belongs to   $L_a L_b$ for some $a,b \in S-I$ such that $ab \in I$; 
        \item $\omega$-error: belongs to   $(L_a)^\omega$, for some $a \in S-I$ such that $a^\omega  \in I$; 
        \item  $\omegaop$-error: belongs to $ (L_a)^\omegaop$, for some $a \in S-I$ such that $a^\omegaop  \in I$; 
        \item  shuffle error: is in the shuffle of   $\set{L_a}_{a \in A}$ for some $A \subseteq S-I$ such that $A^\eta \in I$.
    \end{itemize}
    Note that in the above definition, we can use languages $L_a$ for $a \in J$. These languages are not yet known to be definable in \mso, but they are conditionally definable in the sense used by Lemma~\ref{lem:product-approximated}.

    \begin{claim}\label{lem:bad-patterns} 
        A $\cc$-word belongs to $L_I$ if and only if  it has an error  infix.
    \end{claim}
    \begin{proof}
        Clearly every error is in $L_I$, and since $L_I$ is an ideal, it follows that   every $\cc$-word with an error infix is in $L_I$. We are left with the converse implication: every $\cc$-word in $L_I$ contains an error infix. 
        To prove this   implication, we will show that the language 
    \begin{align*}
    L = \set{w \in S^\cc : \text{if $w \in L_I$  then $w$  has an error infix}} 
    \end{align*}
    satisfies the assumptions of Lemma~\ref{lem:shelah-induction}, with $\lambda$ being the multiplication operation in $S$. The conclusion of Lemma~\ref{lem:shelah-induction} will then say that  $L$  is equal to $S^\cc$, thus showing that every $\cc$-word in $L_I$ has an error infix.
    
    The first assumption of Lemma~\ref{lem:shelah-induction} says that $L$ is closed under binary concatenation. Suppose that $u,v \in L$. We need to show that $uv \in L$. Suppose that $uv \in L_I$. If $u \in L_I$, then it has an error infix by assumption on $u \in L$, and therefore also $uv$ has an error infix. We argue similarly if $v \in L_I$. Finally, if both $u,v$ multiply to elements in $S-I$, then $uv$ is a binary  error. 

    The remaining assumptions of Lemma~\ref{lem:shelah-induction} are checked the same way. 
    \end{proof}

    As remarked before Claim~\ref{lem:bad-patterns}, the definition of errors refers to languages $L_a$ with $a \in J$, which are not yet known to be definable in \mso. We deal with this issue now. 
    By Lemma~\ref{lem:product-approximated}, for every $a \in J$ there an  \mso definable language which contains all  $\cc$-words that have multiplication $a$, and does not contain any $\cc$-words that have  multiplication in $J- \set a$. By removing the  $\cc$-words with easy multiplications from that language, we get an \mso definable language $K_a$ with
    \begin{align*}
    L_a \subseteq K_a \subseteq L_a \cup L_I.
    \end{align*}
     Define a \emph{weak error} in the same way as an error, except that $K_a$ is used instead of $L_a$ for $a \in J$. Since $K_a$ is obtained from $L_a$ by adding some words from the ideal $L_I$, it follows from Claim~\ref{lem:bad-patterns} that a $\cc$-word is in $L_I$ if and only if it has an infix that is a weak error. Finally, weak errors can be defined by an expression which uses \mso definable languages and the closure operators from Lemma~\ref{lem:shuffle-expressions}, and therefore weak errors are \mso definable. It follows that $L_I$ is \mso definable, and therefore $L_J$ is \mso definable.
\end{proof}

As we have already remarked when describing the proof strategy, the  above lemma completes the proof of the induction step in Theorem~\ref{thm:mso-complete-cc-words}. Indeed, a $\cc$-word has multiplication  $a \in J$  if and only if it belongs to $L_J$ and it is assigned $a$ by the colouring from Lemma~\ref{lem:product-approximated}.

\exercisehead


\mikexercise{The syntax of  {star-free expression} for $\cc$-words is the same as for finite words, except that the complementation operation is interpreted as $\Sigma^\cc - L$ instead of $\Sigma^* - L$. Define a $\cc$-star-free language to be a language $L \subseteq \Sigma^\cc$ that is defined by a star-free expression. Show that if $L$ is $\cc$-star-free, then its syntactic $\cc$-semigroup is aperiodic, but the converse implication fails.
}{}

\mikexercise{What is the modification for $\cc$-star-free expressions that is needed to get  first-order logic (over the ordered model)?}{
    Instead of concatenation $LK$, we need to use marked concatenation $LaK$ where $a$ is a letter.
}

 \mikexercise{
    Show that if $L \subseteq \Sigma^\cc$ is $\cc$-star-free, then the same is true for every language recognised by  its syntactic $\cc$-semigroup.}{}

\mikexercise{\label{ex:star-free-omega-closed} Show that if $L \subseteq \Sigma^\cc$ is $\cc$-star-free, then the same is true for $L^\omega$. 
}{}

\mikexercise{\label{ex:star-free-constructions}  Show that if $S$ is aperiodic, then the constructions from Lemmas~\ref{lemma:approximate-prefix-class} and~\ref{lem:product-approximated} can be done using $\cc$-star-free expressions. }{}

\mikexercise{\label{ex:star-free-cc-effective} Show that  $L \subseteq \Sigma^\cc$ is $\cc$-star-free if and only if its syntactic $\cc$-semigroup is finite, aperiodic and satisfies\footnote{This exercise is based on 
\incite[Theorem 2, item 2.]{colcombetSreejith2015}
}:
\begin{align*}
e^\omegaop = e = e^\omega \quad \Rightarrow \quad e = \set e^\eta \qquad \text{for every idempotent $e$.}
\end{align*}
Hint: use Exercises~\ref{ex:star-free-omega-closed} and~\ref{ex:star-free-constructions}.
}{}
\mikexercise{Show that languages of $\cc$-words definable in first-order logic (in the ordered model) are not closed under concatenation $LK$. }{The singleton languages $\set{a^\omega}$ and $\set{a^\omegaop}$ are definable in first-order logic, but their concatenation is not.}

\mikexercise{\label{ex:cc-regular-associative} We say that a multiplication operation $\mu : S^\cc \to S$ is regular-associative if it satisfies the associativity condition from Definition~\ref{def:cc-semigroup}, but with the diagrams restricted so that  only 
\begin{align*}
S^\bullet = \set{w \in S^\cc : \text{$w$ is regular}}
\end{align*}
is used instead of $S^\cc$. Show that if $S$  finite and $\mu : S^\cc \to S$ is \mso definable and regular-associative, then $\mu$ is associative. 
}{}

\mikexercise{
Show that if $S$ is finite and $\mu : S^\bullet \to S$ is regular associative, then it can be extended uniquely to an associative multiplication $\bar \mu : S^\cc \to S$. Hint: the \mso formulas  defined in the   proof of Theorem~\ref{thm:mso-complete-cc-words} depend only on the \lale operations of the $\cc$-semigroup $S$.
}{
    The \mso formulas  defined in the   proof of Theorem~\ref{thm:mso-complete-cc-words} depend only on the \lale operations of the $\cc$-semigroup $S$. Therefore, the construction in the theorem allows us to define, for  every $\mu : S^\bullet \to S$, an \mso definable colouring  $\lambda : S^\cc \to S$. (Furthermore, if $\mu$ arose from an associative multiplication operation by restricting it to $S^\bullet$, then $\lambda$ is equal to that multiplication operation.) 

    We claim that if $\mu$ was regular-associative, then $\lambda$ agrees with $\mu$ on regular inputs. To prove this claim, we simply follow the proof of Theorem~\ref{thm:mso-complete-cc-words}, but we use regular $\cc$-words instead of all $\cc$-words. 

    Once we know that $\lambda$ agrees with $\mu$ on regular inputs, and $\mu$ is regular-associative, we can apply Exercise~\ref{ex:cc-regular-associative}.
 }

 \part{Monads}

  \chapter{Monads}
\label{sec:monads}
As discussed in Chapter~\ref{chapter:semigroups}, instead of viewing a semigroup as having a binary multiplication operation, one could think of a semigroup as a set  $S$ equipped with a  multiplication operation $\mu : S^+ \to S$, which is associative in the sense that the following two diagrams commute:
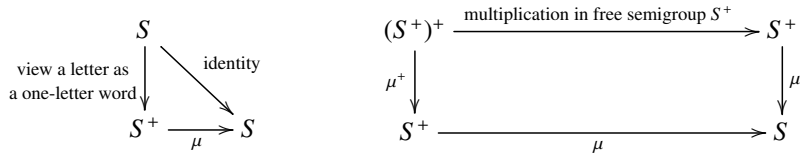
\begin{align*}
    \xymatrix{
       S \ar[dr]^{\text{identity}} \ar[d]_{\txt{\scriptsize view a letter as \\ \scriptsize  a one-letter word}}\\
       S^+ \ar[r]_\mu & S
    }
    \qquad \qquad
       \xymatrix@C=4cm
       {(S^+)^+ \ar[r]^{\text{multiplication in free semigroup $S^+$}} \ar[d]_{\mu^+} &
        S^+ \ar[d]^{\mu} \\ 
        S^+ \ar[r]_{\mu}
        &  S}
    \end{align*}
The same is true for monoids, with $*$ used instead of $+$, and  for $\cc$-semigroups, with $\cc$ used instead of $+$. In this chapter, we examine the common pattern behind these constructions, which is that they are the Eilenberg-Moore algebras for the monads of $+$-words, $*$-words and $\cc$-words, respectively.

From the perspective of this book, the idea behind monads is the following. Instead of first defining  not necessarily free algebras (e.g.~semigroups) and then defining free algebras (e.g.~the  free semigroup) as a special case, an opposite approach is used. We begin with the free algebra (which is  the monad), and then other, not necessarily free, algebras are defined as a derived notion (which is  the Eilenberg-Moore algebras of the monad). This opposite approach is  useful for less standard algebras such as graphs, where axiomatising the not necessarily free algebras is possible but tedious and not intuitive, while the free algebra is very natural, because it consists of graphs with a certain substitution structure.

\section{Monads and their Eilenberg-Moore algebras}
This section, presents the basic definitions for monads and their algebras. These notions make sense for arbitrary categories. However,  for simplicity we use the category of sets and functions, because this is where most of our examples live. In later chapters we will consider multi-sorted sets (e.g.~sets with sorts $\set{+,\omega}$ for $\omega$-semigroups, or sets with sorts $\set{0,1,\ldots}$ for hypergraphs), but multi-sorted sets  is as far as we go with respect to the choice of categories.

\begin{definition}[Monad]
     A \emph{monad} in the category of sets\footnote{The same definition can be applied to any other category, by using ``object'' instead of ``set'', and  ``morphism'' instead of ``function''.
     }  consists of the following ingredients:
\begin{itemize}
    \item \emph{Structures:}  for every set  $X$, a set   $\monad X$;
    \item \emph{Substitution: }  for every function $f : X \to Y$, a   function $\monad f : \monad X \to \monad Y$;
    \item \emph{Unit and free multiplication:} for every set  $X$, two functions
    \begin{align*}
    \myunderbrace{\unit X : X \to \monad X}{the unit of $X$}
    \qquad 
    \myunderbrace{\mult_X : \monad \monad X \to \monad X}{free  multiplication on $X$}.
    \end{align*}
\end{itemize}
These ingredients are subject to six axioms 
\begin{align*}
\myunderbrace{\text{\eqref{eq:monad-axiom-functorial}}}{$\monad$ is a functor}
\qquad \qquad 
\myunderbrace{\text{\eqref{eq:monad-axiom-naturality-unit} 
\quad \eqref{eq:monad-axiom-naturality-mult}}}{unit and multiplication\\
\scriptsize are natural transformations}
\qquad \qquad 
\myunderbrace{\text{\eqref{eq:monad-axiom-associative-unit-1}\quad \eqref{eq:monad-axiom-associative-mult} 
\quad \eqref{eq:monad-axiom-associative-unit-2}}}{$\monad X$ with free  multiplication\\
\scriptsize is an Eilenberg-Moore algebra,
\\ \scriptsize and one more associativity axiom}
\end{align*}
which will be described later in this section. 
\end{definition}


Before describing the monad axioms, we discuss some examples, and define Eilenberg-Moore algebras.  The purpose of the monad axioms  is to ensure that  Eilenberg-Moore algebras are well-behaved, and therefore it is easier to  see the monad axioms after the definition of  Eilenberg-Moore algebras. But even before that, we begin with  an example of the monad of finite words, where the Eilenberg-Moore algebras are monoids, to  illustrates what we want to do with monads.

\begin{example}[Monad of finite words]\label{ex:monad-of-finite-words}
    The monad of finite words is defined as follows. The structures are defined by   $\monad X = X^*$. For a function $f : X \to Y$, its corresponding substitution  
    \begin{align*}
    \myunderbrace{\monad f : \monad X \to \monad Y}{can also be written as \\ \scriptsize $f^* : X^* \to Y^*$  \\ \scriptsize for this particular monad}
    \end{align*}
    is defined by applying $f$ to every letter in the input word.  For a set $X$, the unit operation of type
    \begin{align*}
    \myunderbrace{X \to \monad X}{$X \to X^*$ \\ \scriptsize for this particular monad},
    \end{align*}
    maps a  letter  to the one-letter word  consisting of this letter. Free  multiplication, which is a function of type 
    \begin{align*}
        \myunderbrace{\monad \monad X \to \monad X}{$(X^*)^* \to X^*$ \\ \scriptsize for this particular monad}
        \end{align*}    
    flattens a word of words into a  word. 
 The monad of $\cc$-words is defined  in the same way, except that it uses $\cc$ instead of $*$.
\end{example}

For this book, the key notion for monads is Eilenberg-Moore algebras. The idea is that $\monad X$ describes the free algebra, while the Eilenberg-Moore algebras are the algebras that are not necessarily free.

\begin{definition}[Eilenberg-Moore algebras]
     An Eilenberg-Moore algebra in a monad $\monad$, also called a \emph{$\monad$-algebra},  consists of an \emph{underlying set} $A$ and a \emph{multiplication operation} $\mu : \monad A \to A$, subject to the following associativity axioms:
$$
    
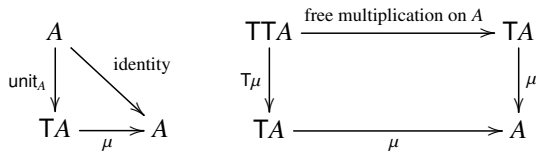
\xymatrix{
        A \ar[dr]^{\text{identity}} \ar[d]_{\unit A
        }\\
        \monad A \ar[r]_\mu & A
     } \qquad 
\xymatrix @R=2pc @C=6pc { \monad \monad A \ar[r]^{\text{free multiplication on $A$}} \ar[d]_{\monad \mu} & \monad A \ar[d]^{\mu} \\
\monad A \ar[r]_{\mu} & A
}$$
\end{definition}
The above definition also makes sense for categories other than the category of sets, with $A$ being an object in the category and $\mu$ being a morphism. In the definition, the reader will recognise, of course, the diagramatic definitions of semigroups, monoids and $\cc$-semigroups. Also, the diagramatic definition of $\omega$-semigroups will fall under the scope of the above definition, if we think about $\omega$-semigroups as living in the category of sets with two sorts $\set{+,\omega}$.

By abuse of notation, we  use the same letter to denote a $\monad$-algebra and its underlying set, assuming that the multiplication operation is clear from the context. Also, if the monad $\monad$ is clear from the context, we will say algebra instead of $\monad$-algebra.

\begin{myexample}[Group monad]
    The \emph{free group over a set $X$} is defined to be 
    \begin{align*}
    (\myunderbrace{\blue X + \red X}{two copies of $X$,\\ \scriptsize one blue, and one red})^*
    \end{align*}
     modulo the identities
\begin{align*}
\blue x \red x = \red x \blue x = \varepsilon \qquad \text{for every $x \in X$,}
\end{align*}
where $\varepsilon$ represents the empty word, while $\blue x$ and $\red x$ represent the blue and red copies of $x$.
The identities can be applied in any context, for example
    \begin{align*}
      \blue z \red x \qquad 
      \blue z   \red y \blue  y \red  x  \qquad    \blue z  \blue z \red  z \red x \red y  \blue  y,
    \end{align*}
    represent the same element of the free group.
    Define $\monad$ to be the monad where $\monad X$ is the free group over $X$, and  the remaining  monad structure is defined similarly  as for finite words, except that we have the two copies of the alphabet, and the identities. The unit operation maps an element to its blue copy.
    
    An algebra over this monad is the same thing as a group.
Indeed, if $G$ is an  algebra over this monad, with multiplication $\mu$, then the group structure is recovered as follows:
\begin{align*}
\myunderbrace{1 \eqdef \mu(\varepsilon)}{group identity}
\qquad
\myunderbrace{x^{-1} \eqdef  \mu(\red x)}{group inverse}
\qquad 
\myunderbrace{x \cdot y \eqdef  \mu(\blue{xy})}{group operation}
\end{align*}
The axioms of a group are easily checked, e.g.~the axiom $x \cdot x^{-1}$ is proved as follows:
\begin{eqnarray*}
    x \cdot x^{-1} \equalbecause{definition of inverse} \\
    x \cdot \mu(\red x) \equalbecause{unit followed by multiplication is the identity, i.e.~axiom~\ref{eq:monad-axiom-associative-unit-1}} \\
    \mu(\blue x) \cdot \mu(\red x) \equalbecause{definition of the group operation}\\
    \mu(\mu(\blue x)\mu(\red x)) \equalbecause{associativity of multiplication, i.e.~axiom~\ref{eq:monad-axiom-associative-mult}}\\
    \mu(\blue x \red x) \equalbecause{equality in the free group}\\
    \mu(\varepsilon) \equalbecause{definition of group identity}\\
    1 & & 
\end{eqnarray*}
For the converse, we observe that for every group $G$, its group multiplication can be extended uniquely to an operation of type $\monad G \to G$, and the resulting operation will be associative in the sense required by Eilenberg-Moore algebras. 
\end{myexample}

\subsection{Axioms of a monad}
\label{sec:axioms-of-a-monad}
Having described some intuition behind monads and their Eilenberg-Moore algebras, we now describe the axioms of a monad.

\paragraph*{Functoriality.} The first group of axioms says that the first two ingredients (the structures and substitutions) of a monad are a functor in the sense of category theory. This means that substitutions  preserve the identity and composition of functions. Preserving the identity means that if  we apply $\monad$ to the identity function on $X$, then the  result is the identity function on $\monad X$. Preserving composition means that the composition of substitutions is the same as the substitution of their composition, i.e.~for every functions $f : X \to Y$ and $g : Y \to Z$, the following diagram commutes 
\begin{align}
    \label{eq:monad-axiom-functorial}
     \xymatrix @R=2pc { \monad X \ar[r]^{\monad f} \ar[dr]_{\monad (g \circ f)} & \monad Y \ar[d]^{\monad g} \\& \monad Z
}
\end{align}
 
\paragraph*{Naturality.}   The naturality  axioms say that for every function $f :  X \to Y$, the following diagrams commute. 
\begin{align}
\label{eq:monad-axiom-naturality-unit}
     \xymatrix @R=2pc { X \ar[r]^{f} \ar[d]_{\unit X} & Y \ar[d]^{\unit Y} \\
\monad X \ar[r]_{\monad f}& \monad Y
} 
\end{align}
\begin{align}
\label{eq:monad-axiom-naturality-mult}
\xymatrix @R=2pc { \monad \monad X \ar[r]^{\monad \monad f} \ar[d]_{\text{free multiplication on $X$}} & \monad \monad Y \ar[d]^{\text{free multiplication on $Y$}} \\
\monad X \ar[r]_{\monad f}& \monad Y
}
\end{align}
In the language of category theory, this means that the unit and free multiplication are natural transformations. Also, as we will see later on, the second naturality axiom (naturality of free multiplication) says that the substitution  $\monad f$ is a homomorphism between the free algebras $\monad X$ and $\monad Y$. 

\paragraph*{Associativity.} We now turn to the most important monad axioms, which ensure the  Eilenberg-Moore algebras are well behaved.  The main associativity axiom says  that for every set $X$, the set $\monad X$ equipped with free multiplication  on $X$ is a $\monad$-algebra (we call this the free $\monad$-algebra over $X$, or simply free algebra if the monad is clear from the context). By unravelling the definitions, this means that the following two diagrams commute:
\begin{align}
    \label{eq:monad-axiom-associative-unit-1}
    \xymatrix@C=7pc{
        \monad X \ar[dr]^{\text{identity}} \ar[d]_{\unit {\monad X}
        }\\
        \monad \monad X \ar[r]_{\text{free multiplication on $X$}} & \monad X
     } 
\end{align}
\begin{align}
    \label{eq:monad-axiom-associative-mult}
\xymatrix @R=2pc @C=7pc { \monad \monad \monad X \ar[r]^{\text{free multiplication on $\monad X$}} \ar[d]_{\monad \text{(free multiplication on $X$)}} & \monad \monad X \ar[d]^{\text{free multiplication on $X$}} \\
\monad \monad X \ar[r]_{\text{free multiplication on $X$}} & \monad X
}
\end{align}
Apart from the above two, there is one more associativity axiom, namely:
\begin{align}
    \label{eq:monad-axiom-associative-unit-2}
    \xymatrix@C=7pc{
        \monad X \ar[dr]^{\text{identity}} \ar[d]_{\monad(\unit { X})
        }\\
        \monad \monad X \ar[r]_{\text{free multiplication on $X$}} & \monad X
     } 
\end{align}

This completes the axioms of a monad, and the definition of a monad.

\exercisehead
\mikexercise{\label{ex:lift-binary-relation} Consider a monad $\monad$ in the category of sets. For a binary relation $R$ on a set $X$, define 
\begin{align*}
    R^\monad \subseteq (\monad X) \times (\monad X)
\end{align*}
to be the binary relation on $\monad X$ that is defined by
\begin{align*}
R^\monad =    \set{((\monad 
\pi_1)(t), (\monad \pi_2)(t)) : t \in \monad R} \quad\text{where $\pi_i : X \times X \to X$ is the $i$-th projection.}
\end{align*}
Does transitivity of $R$ imply transitivity of $R^\monad$?
}{}

\subsection{Homomorphisms and recognisable languages }
\label{sec:monad-homomorphism}
A homomorphism between two $\monad$-algebras is any function between their underlying sets which is consistent with  the multiplication operation, as formalised in the following definition. 
\begin{definition}[Homomorphism] Let $\monad$ be a monad. A \emph{$\monad$-homomorphism} is a  function $h : A \to B$ on the  underlying sets of two $\monad$-algebras $A$ and $B$,  which makes the following diagram commute 
    \begin{align*}
        \vcenter{\xymatrix  @C=3pc { \monad A \ar[r]^{\monad h} \ar[d]_{ \text{multiplication in $A$}} & \monad B \ar[d]^{\text{multiplication in $B$}} \\
        A \ar[r]_{h} & B
        }}
        \end{align*} 
\end{definition}
When the monad is clear from the context, we simply write homomorphism, instead of $\monad$-homomorphism.
Again, the reader will recognise the notion of homomorphism for semigroups, monoids and $\cc$-semigroups. 

In the rest of this section, we describe some basic properties of homomorphisms.

\begin{lemma}\label{lem:composition-of-homomorphisms}
    Homomorphisms are closed under composition.
\end{lemma}
\begin{proof}
    Consider two homomorphisms 
    \begin{align*}
    \xymatrix{
        A \ar[r]^g & B \ar[r]^h & C.
    }
    \end{align*}
    Saying that the composition $h \circ g$ is a homomorphism is the same as saying that the perimeter of the following diagram commutes:
    \begin{align*}
        \xymatrix@C=6pc{ \monad A \ar[dd]_{\text{multiplication in $A$}} \ar[rr]^{\monad (h \circ g)} \ar[dr]_{\monad  g} & & \monad C  \ar[dd]^{\text{multiplication in $C$}}\\
        & \monad B \ar[ur]_{\monad h} \ar[d]_{\text{ multiplication in $B$}}  \\
        A   \ar[r]^{g} & B \ar[r]^{h} & C}
      \end{align*}
      The upper triangular face commutes because of the functoriality axioms (substitutions are compatible with composition). The left and right  triangular faces commute by assumption that $g$ and $h$ are homomorphisms.
\end{proof}

Recall that we defined the \emph{free algebra over a set $X$} to be $\monad X$ equipped with  the free multiplication operation of type $\monad \monad X \to \monad X$. The monad axioms say that this is indeed an algebra. It is called free because of the universal property given in the following lemma.

\begin{lemma}[Free Algebra Lemma]
    For every set $X$, the free algebra $\monad X$ has the following universal property:
    \begin{align}\label{eq:free-monad-algebra}
        \red{\forall} \blue{\exists!} \qquad
        \xymatrix@C=4cm{ 
           X
           \ar@[red][r]^{\red{\text{function on sets}}}
           \ar[dr]_-{\unit X} 
        & \red {\text{algebra $A$}}\\
         & {\monad X}
         \ar@[blue][u]_{\blue {\text{homomorphism of algebras}}} }
      \end{align}
\end{lemma}
\begin{proof}
We begin by showing that there is at least one blue homomorphism $\blue h$ for each red function $\red f$; later we show that this homomorphisms is unique.  Let $\mu : \monad \red A \to \red A$ be  multiplication in the algebra $\red A$. 
Define $\blue h$ to be the composition of the following functions:
\begin{align*}
    \xymatrix{
        \monad X \ar[r]^{\monad \red f} &
        \monad \red A \ar[r]^{\mu} &
        \red A.
    }
\end{align*}
The axiom on naturality of free multiplication says that $\monad \red f$ is a homomorphism from the free algebra $\monad X$ to the free algebra $\monad \red A$. The associativity axiom in the definition of an Eilenberg-Moore algebra says that  multiplication $\mu$ is a homomorphism from the free algebra $\monad \red A$ to the algebra $\red A$. Thanks to Lemma~\ref{lem:composition-of-homomorphisms}, $\blue h$ is a homomorphism, as the composition of two homomorphisms $\monad \red f$ and $\mu$. 

We now show uniqueness -- every  homomorphism $\blue h$ which makes the diagram must be equal to the one described above.  Consider the following diagram:
$$ 
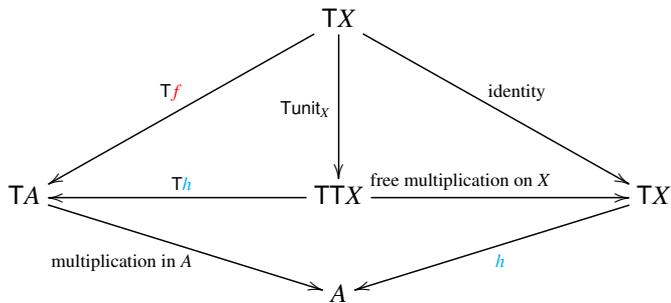
\xymatrix@C=8pc{
& \monad  X \ar[ddl]_{\monad \red f} \ar[ddr]^{\text{identity}} \ar[dd]_{\monad \unit  X}\\
\\
\monad A \ar[dr]_{\text{multiplication in $A$} \qquad} & \monad \monad  X \ar[l]_{\monad \blue h} \ar[r]^-{\text{free multiplication on $X$ \qquad \qquad}}  & \monad  X \ar[dl]^{\blue h} \\
& A}$$
The upper left triangular  face commutes by applying  $\monad$ to the assumption that $\blue h$ extends $\red f$. (Applying $\monad$  preserves commuting diagrams, because of the functoriality axioms.)  The upper right triangular face commutes by the first associativity axiom. The lower four-sided face commutes, because it says that  $\blue h$ is a homomorphism. Therefore, the perimeter of the diagram commutes. The perimeter says that  says that $\blue h$ must be equal to $\monad f$ followed by multiplication in $A$, and therefore $\blue h$ is unique. 
\end{proof}

\paragraph*{Compositional functions.} Fix a monad $\monad$. Suppose that $A$ is an algebra, while $B$ is a set, which is not (yet) equipped with a multiplication operation. We say that a function $h : A \to B$ on sets is \emph{compositional} if 
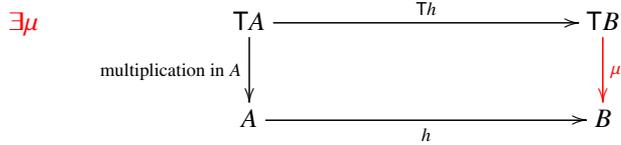
\begin{align*}
     \red{\exists \mu} \qquad
    \xymatrix@C=4cm{ 
       \monad A  
       \ar[r]^{\monad h}
       \ar[d]_{\text{multiplication in $A$}}
       & \monad B
       \ar@[red][d]^{\red \mu}\\
       A 
       \ar[r]_h
    & B}    
\end{align*}
This is the same notion of compositionality as was used for monoids, semigroups and $\cc$-semigroups in part I of the book. For the same reason as before, surjective compositional functions are equivalent to surjective  homomorphisms, as stated in the following lemma.

\begin{lemma}\label{lem:monad-compositional}
    If $A$ is an algebra, $B$ is a set, and $h : A \to B$ is compositional and surjective, then there is a (unique) multiplication operation on $B$ which turns it into an algebra and $h$ into a homomorphism.
\end{lemma}
\begin{proof}
    The multiplication operation -- no surprises here -- is $\mu$ from the definition of a compositional function. 
    The diagram in the  definition of a compositional function is the same diagram as in the definition of a homomorphism, and therefore if  $B$ equipped with $\mu$ is an algebra, then $h$ is a homomorphism. It remains to show that $B$ equipped with $\mu$ is indeed an algebra. We only prove the more interesting of the two  associativity diagrams, namely the one with a rectangular diagram.

    We first observe that $\monad$ preserves surjectivity of functions\footnote{This part of the argument is true for the category of sets, and also for multi-sorted sets, but fails in general, since functors do not need to preserve epimorphisms in general categories.}. Indeed, if a function  $h : A \to B$ is surjective, then it has a one-sided inverse, i.e.~a function  $h^{-1} : B \to A$ such that $h \circ h^{-1}$ is the identity on $B$. By the functoriality axioms,  $\monad h^{-1}$ is a one-sided inverse for $\monad h$, and therefore $\monad h$ is also surjective. This argument justifies the  surjectivity  annotation (double-headed arrows) in the  following diagram.
    
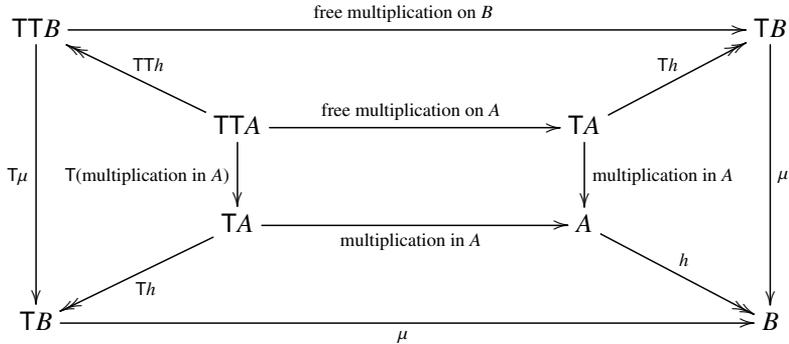
\begin{align*}
		\xymatrix@C=1.8cm{
        \monad \monad B 
        \ar[ddd]_{\monad \mu} 
        \ar[rrrr]^{\text{free multiplication on $B$}}  
        & & & & 
        \monad B\ar[ddd]^{\mu}
        \\	
        & 
        \monad \monad A 
        \ar@{->>}[ul]_{\monad \monad h} 
        \ar[d]_{\monad (\text{multiplication in $A$})}
        \ar[rr]^{\text{free multiplication on $A$}} 
        & &
        \monad A 
        \ar[d]^{\text{multiplication in $A$}} 
        \ar@{->>}[ur]^{\monad h}
        \\
        & \monad A 
        \ar[rr]_{\text{multiplication in $A$}} 
        \ar@{->>}[dl]^{\monad h}
        & &
        A 
        \ar@{->>}[dr]^h
        \\
        \monad B 
        \ar[rrrr]_{\mu}
        & & & &
        B
		}
    \end{align*}
    The central  rectangular face commutes by the assumption that $A$ is an algebra. 
	The upper trapezoid face commutes by naturality of free multiplication.  The right and lower trapezoid faces commute by definition of a compositional function, and the left trapezoid face commutes by the same definition with  $\monad$ applied to it. It follows that all paths that begin in $\monad \monad A$ and end in $B$ denote the same function. Since $h$ is surjective, it follows that the  perimeter of the diagram commutes. This proves the second of the associativity diagrams in the definition of an  Eilenberg-Moore algebra.
\end{proof}

\paragraph*{Recognisable colourings and languages.}
In this book, we are most interested in the Eilenberg-Moore algebras as recognisers of languages.  A language is a subset $L$ of a free algebra $\monad \Sigma$. (Typically we are interested in the case where the alphabet $\Sigma$ is finite, but this assumption does not seem to play a role in the results that we care about, so we omit it.)   A language is called \emph{recognisable}\footnote{
    The definition of recognisable languages for monads appears first in 
    \incite[Section 11.]{eilenbergAutomataGeneralAlgebras}
    The above paper uses Lawvere theories, which correspond to finitary monads (see Section~\ref{sec:finitary-monads}). The main result of~\cite{eilenbergAutomataGeneralAlgebras}, Theorem III,  concerns free Lawvere theories, which correspond to the  monads described in Example~\ref{ex:term-monad}, and says that recognisable languages for such monads can be described using least fix-points. With the exception of Example~\ref{ex:term-monad}, none of the monads studied in this book are free.  
} if it is recognised by a finite algebra, as explained in the following definition (which uses a slightly more general notion of language, called colourings).

\begin{definition}[Recognisable colourings]\label{def:recognisable-monad}
    Fix a monad $\monad$. An \emph{algebra colouring} is defined to be any function from an algebra to  a  set of colours\footnote{For some monads, it would be more useful to deviate from this definition. For example, in the monad from  Example~\ref{ex:monad-algebra-over-field} that deals with vector spaces, a more useful notion of colouring is a linear map to the underlying field. Therefore, one could think of a parametrised  notion of recognisability, where  the notion of ``algebra colouring'' is taken as a parameter. Nevertheless, for all monads that are studied in more detail in this book, Definition~\ref{def:recognisable-monad} is good enough.}. A  \emph{finite algebra} is an algebra where the underlying set is finite\footnote{Like for algebra colourings, sometimes this notion of finite algebra is not the right one. In the monad from Example~ref{ex:monad-algebra-over-field}, the more useful notion is that a finite algebra is one where the underlying set is a vector space of finite dimension. Again, one could think of the notion of ``finite algebra'' as being a parameter.}. An algebra  colouring $L : A \to U$ is called \emph{recognisable} if it factors through a homomorphism into a finite algebra, as expressed in the following diagram:
    \begin{align*}
        \red{\exists }  \qquad
        \xymatrix@C=2cm{ 
           A
           \ar[r]^L 
           \ar@[red][dr]_-{ \txt{\scriptsize \red{homomorphism } \qquad \\ \scriptsize \red{into a  finite algebra }}} 
        & U \\
         & \red B
         \ar@[red][u]_{\red {\text{algebra colouring}}} }
      \end{align*}
      Note that a recognisable colouring will necessarily use finitely many colours.
\end{definition}

A {language} can be viewed as the  special case of an algebra 
  where the algebra is a free algebra  and there are  two colours ``yes'' and ``no''. For languages, we prefer set notation, e.g.~we can talk about the complement of a language, or use Boolean operations for  languages.  The above definition is easily seen to coincide with the notions of recognisability for semigroups, monoids and $\cc$-semigroups that were discussed in the first part of this book. In the next section, we give more examples.

\exercisehead



\mikexercise{
\label{ex:powerset-algebra}    
For an algebra $A$ with multiplication operation $\mu : \monad A \to A$, define its powerset as follows: the underlying set is the powerset $\powerset A$, and multiplication is defined by 
\begin{align*}
        t \in \monad \powerset A \quad \mapsto \quad \set{ \mu(s) : s \in^\monad t},
\end{align*}
where $\in^\monad$ is defined as in Exercise~\ref{ex:lift-binary-relation}. Show an example of a monad $\monad$ where this construction does not yield an algebra.
}{}

\mikexercise{\label{ex:reco-closed-under-homo} Does the group monad satisfy the following implication:
\begin{itemize}
    \item[(*)] If $L \subseteq \monad \Sigma$ is recognisable, and $h : \monad \Sigma \to \monad \Gamma$ is a homomorphism, then $h(L)$ is recognisable.
\end{itemize}
 What about surjective homomorphisms?
}{}

\mikexercise{Consider the 
implication in the previous exercise. Show that even if we restrict $h$ to functions of the form $\monad f$ for some surjective $ f : \Sigma \to \Gamma$, then the implication can still be false in some monads.
}{}

\section{A zillion examples}
\label{sec:zillion}
Monads have an abundance of interesting examples. This section is devoted to   a collection of such examples, with an emphasis on the  algebras arising from the monads, and the languages  recognised by the finite algebras.

\subsection{Monads for words}

 We begin with several examples of monads that study words, both finite and infinite. We have already discussed finite words in Example~\ref{ex:monad-of-finite-words}. The following example discusses infinite words, up to a fixed cardinality.
    
\newcommand{\monadkappa}{\monad_{\!\kappa}}
\begin{myexample}[Chains]\label{ex:monads-for-linear-orders}
    Define a \emph{chain}  over a set $X$ to be a  linear order with positions labelled by $X$, modulo isomorphism of labelled linear orders. 
    For an infinite cardinal  $\kappa$, define a monad  $\monadkappa$   as follows. The set $\monadkappa X$ consists of chains over $X$, which have cardinality strictly less than $\kappa$. For example, if $\kappa$ is the first infinite cardinal $\aleph_0$ then the monad describes finite words, and if $\kappa$ is the first uncountable cardinal then the monad describes $\cc$-words. 
    The monad structure is defined in the same way as for finite words and  $\cc$-words. Nevertheless, we give  a more exact description below. 
    
    For a function $f$, the corresponding substitution $\monadkappa f$ is defined by applying $f$ to the labels in the input chain and leaving the positions and ordering unchanged. The unit maps a letter  to the chain with one position labelled by that letter. 
    The free multiplication operation is defined using lexicographic products, as follows.
    Suppose that $w \in \monadkappa \monadkappa X$.  The positions in the free multiplication of $w$  are pairs $(i,j)$ such that $i$ is a position of $w$ and $j$ is a position in the label of position $i$ in the chain $w$, call this label   $w(i) \in \monadkappa X$. The label of such a position is inherited from $j$, and the ordering is lexicographic. The cardinality of the resulting chain is at most   $\kappa$, since every infinite cardinal satisfies $\kappa = \kappa^2$.
    
This is a monad.
    We only prove one of the monad axioms, namely
    \begin{align*}
\xymatrix @R=2pc @C=7pc { \monadkappa \monadkappa \monadkappa X \ar[r]^{\text{free multiplication on $\monadkappa X$}} \ar[d]_{\monadkappa \text{(free multiplication on $X$)}} & \monadkappa \monadkappa X \ar[d]^{\text{free multiplication on $X$}} \\
\monadkappa \monadkappa X \ar[r]_{\text{free multiplication on $X$}} & \monadkappa X
}
    \end{align*}
    Let $w \in \monadkappa \monadkappa \monadkappa X$. If, in the diagram above, we first go right and then down, then the resulting linear order will have positions of the form $((i,j),k)$, where $i$ is a position in $w$, $j$ is a position in $w(i)$, and $k$ is a position in $w(i)(j)$. If, in the diagram, we first go down and then right, then we get positions of the form $(i,(j,k))$, where $i,j,k$ satisfy the same conditions as above. In both cases, the tuples of positions are ordered lexicographically, and the label is inherited from $k$. Therefore 
    \begin{align*}
    ((i,j),k) \mapsto (i,(j,k))
    \end{align*}
    is an isomorphism of labelled linear orders, and hence the two outcomes are equal as chains.

    If we take $\kappa$ to be the first cardinal bigger than the continuum cardinal $\mathfrak c$,  then  $\monadkappa$ describes chains of cardinality at most $\mathfrak c$. In this case, we have the following phenomenon. 
    Recall the  powerset construction that was described in Section~\ref{sec:shelah-operations}. This construction also makes sense for chains of size at most $\mathfrak c$. Define $\algclass$ to be the least class of $\monadkappa$-algebras which contains the syntactic algebra of the language ``every $a$ is before every $b$'', and which is closed under  products and the powerset construction.  Using the same proof as in the Trakhtenbrot-\buchi-Elgot Theorem and in   Theorem~\ref{thm:mso-to-cc}, one can show that every \mso definable language $L \subseteq \monadkappa \Sigma$ is recognised by an algebra from $\algclass$. As we have mentioned on page~\ref{page:undecidability-of-mso},  satisfiability for \mso over the reals is undecidable, and therefore there  is no finite way of representing algebras from $\algclass$. This means that the powerset construction over finite $\monadkappa$-algebras is not computable.
\end{myexample}

In the above example, we consider all chains of given cardinality. One can also consider subclasses of chains, subject to some condition on the underlying linear order, as described in the following example.

\begin{myexample}
    Consider a set $\Xx$ of linear orders which is closed under free multiplication as defined in the previous example, when  viewed as chains over a one letter alphabet. If we restrict  the monad from the previous example to chains where the underlying linear order is in $\Xx$, then we also get a monad. This construction yields the following monads (in all cases, we assume some fixed upper bound on $\kappa$ on the cardinality, e.g.~we can require countability):
\begin{itemize}
    \item well-founded words (the class of well-founded linear orders);
    \item scattered words (the class of scattered orders, i.e.~those into which one cannot embed the rational numbers)\footnote{Algebras for the monad of countable scattered words are studied in 
    \incite{rispalComplementationRationalSets2005}
    };
    \item dense words (the class which contains two orders: a singleton order for units, and the rational numbers).
\end{itemize}
\end{myexample}

\begin{myexample}[$\omega$-semigroups]
    We now describe a monad that corresponds to $\omega$-semigroups, see Definition~\ref{def:omega-semigroups}. Since an $\omega$-semigroup has two sorts, we leave the category of sets, and use instead the category  
    \begin{align*}
      \setcat^{\set{+,\omega}}
    \end{align*}
    of sets with two sorts $+$ and $\omega$. 
    An object in this category is a set, where every element is assigned  exactly one of two sorts, called $+$ and $\omega$. We use the name \emph{sorted set} for the  objects, for the purpose of this example.   A morphism in this category is any sort-preserving function between sorted sets.   We also use the following notation for sorts:
    \begin{align*}
    \myunderbrace{\sorted X}{a sorted set}  \quad =\quad \myunderbrace{\sorted X[+]}{elements\\ \scriptsize of sort $+$ }\quad \cup \quad \myunderbrace{\sorted X[\omega]}{elements\\ \scriptsize of sort $\omega$}.
    \end{align*}
     Define a monad $\sorted \monad$ over this category as follows. For a sorted set $\sorted X$, the sorted set  $\sorted{\monad X}$ is defined by:
    \begin{align*}
    (\sorted{\monad X})[+] = (\sorted X[+])^+ \qquad \cup \qquad
    (\sorted{\monad X})[\omega] = (\sorted X[+])^* (\sorted X[\omega]) \cup (\sorted X[+])^\omega.
    \end{align*}
    For a morphism $\sorted{ f : X \to Y}$, the substitution morphism $\sorted{ \monad f}$ is defined in the natural way, by applying $\sorted f$ to every letter. The unit and free multiplication are defined in the natural way as well. (An element of sort $+$ in $\sorted \monad \sorted \monad \sorted X$ is simply a finite nonempty word of finite nonempty words over $\sorted X[+]$, and we can use free multiplication from the monad of finite nonempty words. On sort $\omega$, there are more cases to consider, but the definition is natural as well. ) An Eilenberg-Moore algebra over this monad is the same thing as an $\omega$-semigroup, as defined at the end of Section~\ref{sec:buchi-determinisation}.
\end{myexample}

\subsection{Other monads}
We now present two monads -- finite multisets and finite sets -- which can be viewed as finite words modulo some equalities. Because these monads arise by imposing equalities on finite words,  their Eilenberg-Moore algebras for these monads are going to be special cases of monoids.
\begin{myexample}[Finite multisets]
    Define $\monad X$ to be the finite multisets over $X$. We write finite multisets using red brackets like this
    \begin{align*}
    \multiset{x,x,y,y,y,z}.
    \end{align*}
    A multiset is finite if it has finitely many elements, and each element appears finitely many times.  Functions are lifted to multisets point-wise, e.g.
    \begin{align*}
    \multiset{x_1,\ldots,x_n} \qquad \stackrel{\monad f} \mapsto \qquad 
    \multiset{f(x_1),\ldots,f(x_n)}.
    \end{align*}
    Another perspective on finite multisets is that they are finite words modulo commutativity $xy=yx$. 
    The unit is $x \mapsto \multiset x$, and free multiplication is simply removing nested brackets, e.g.
    \begin{align*}
    \multiset{\multiset{x,y},\multiset z} \mapsto \multiset{x,y,z}.
    \end{align*}
    This is a monad. An algebra over this monad is the same thing as commutative monoid. Recognisable languages over this monad are the same things are regular languages -- in the usual sense -- which are commutative, see Exercise~\ref{ex:commutative-regular-languages}.
    
    If we lift the restriction on finite supports, then we do not get a monad. The problem is with the substitutions: if  $f : X \to \set a$ is the constant function with an infinite domain, then there is no way to define 
    \begin{align*}
    (\monad f) \multiset{\myunderbrace{x_1,x_2,\ldots}{infinitely many distinct elements}}.
    \end{align*}
    The problem is that the output multiset should contain $a$ infinitely many times. To overcome this problem, we could allow multisets with infinitely many copies of an element. 
\end{myexample}

\begin{myexample}[Idempotent finite words]
    Define $\monad X$ to be finite words $X^*$, modulo the equation $ww=w$. For example,
    \begin{align*}
    abcababc = abc(ab)^2c = abcabc = (abc)^2 = abc = (ab)^2c = ababc.
    \end{align*}
    The remaining ingredients of the monad are defined in the natural way. An algebra over this monad is the same thing as an idempotent monoid, i.e.~a monoid where all elements are idempotent.  Green and Rees show  that if $X$ is a   finite set, then  $\monad X$ is finite\footcitelong[p.~35]{green_rees_1952}. It follows that for every finite alphabet, there are finitely many  languages over this alphabet, and all of them  are recognisable.
\end{myexample}
\begin{myexample}[Powersets]
    The \emph{powerset monad}, and its variant the \emph{finite powerset monad},  are defined in the same way as the multiset monad, except that we use sets (or finite sets) instead of multisets. The substitutions are defined via images (in the language of category theory, we use the co-variant powerset functor, as opposed to the contra-variant powerset functor, which uses inverse images):
    \begin{align*}
      A \subseteq X \qquad \stackrel{\monad f}\mapsto \qquad \set{f(x) : x \in A } \subseteq Y.
    \end{align*} 
    Algebras over the finite powerset monad are the same thing as  monoids that are commutative and idempotent. If  $X$ is a  finite set, then both powerset monads generate  finite sets; and therefore all languages over finite alphabets are recognisable. 
    \end{myexample}

    \begin{myexample}
        [Terms] \label{ex:term-monad} Fix a ranked set $\Sigma$, i.e.~a set where every element has an associated arity in $\set{0,1,\ldots}$. For example, we could have 
        \begin{align*}
        \Sigma = \set{\myunderbrace{a}{arity 2}, \myunderbrace{b}{arity 1}, \myunderbrace{c}{arity 0} }.
        \end{align*}
        Based on $\Sigma$, we define a monad $\monad_\Sigma$ as follows.
        Define $\monad_\Sigma X$ to be the terms over $\Sigma$ with variables $X$, i.e.~an element of $\monad_\Sigma X$ is a tree that looks like this:
        \mypic{47}
        The unit operation maps $x \in X$ to a term which consists only of $x$. The substitution $\monad_\Sigma f$ is defined by applying $f$ to the variables and leaving the remaining part of the term unchanged. Finally, free multiplication replaces each variable with the corresponding term. It is a simple exercise to check that an algebra over the monad $\monad_\Sigma$  is the same thing as an \emph{algebra of type $\Sigma$}, in the sense of universal algebra, i.e.~it consists of an underlying set equipped,  with one operation for every letter in $\Sigma$\footcitelong[Definition 1.3]{sankappanavar1981course}. In the terminology of automata theory, both of these notions are the same as {deterministic bottom-up tree automata over finite trees}, where $\Sigma$ is the input alphabet\footcitelong[Section 2]{thatcherGeneralizedFiniteAutomata1968}. From the above observation it follows that a language $L \subseteq \monad X$ is recognisable  in the sense of Definition~\ref{def:recognisable-monad}  if and only if it is a regular tree language in the sense of automata theory\footnote{This monad describes finite trees. Finding an algebraic account for languages of infinite trees remains an open problem. This problem is discussed in the following papers:
        \incite{DBLP:journals/corr/abs-1808-03559}
\incite{DBLP:journals/lmcs/BojanczykK19}
        } , where the  input alphabet is  obtained from $\Sigma$ by adding one letter of arity $0$ for each element of $X$. If the ranked set $\Sigma$ contains only letters of arity exactly one, then a $\monad$-algebra can be seen as a deterministic word automaton with input alphabet $\Sigma$, without distinguished initial and final states.
    \end{myexample}

\begin{myexample}[Vector spaces]\label{ex:monad-of-vector-spaces}
    In this example, we discuss vector spaces over some field. For the sake of concreteness, we use the field of rational numbers.
    Define $\monad X$ to be the vector space, over the  field of rational numbers, where the basis is $X$. In other words, elements of $\monad X$ are finite linear combinations of elements from $X$ with rational coefficients. For example, 
    \begin{align*}
    3x + 7y - 0.5z \in \monad \set{x,y,z}.
    \end{align*}
    The action of $\monad$ on functions is defined by 
\begin{align*}
q_1 x_1 + \cdots + q_n  x_n \qquad \stackrel{\monad f} \mapsto \qquad q_1 f(x_1) + \cdots + q_n f(x_n).
\end{align*}
The unit operation maps $x \in X$ to the corresponding basis vector, and free multiplication is defined in the natural way, as illustrated in the following example:
\begin{align*}
    3(4x + 0.5y) - 0.2(5x -0.1y)  \quad \mapsto \quad 12 x + 1.5y - x + 0.02y = 11x + 1.52y.
\end{align*}
An algebra $A$ over this monad, with multiplication $\mu$, is also equipped with the structure of a vector space, because we can add elements
\begin{align*}
a + b \eqdef \mu(a+b)
\end{align*}
and multiply them by scalars $q$ from the field of rational numbers:
\begin{align*}
qa \eqdef \mu(qa).
\end{align*}
If $B \subseteq A$ is a basis for the vector space $A$, then the algebra $A$ is isomorphic to $\monad B$. Therefore,  over this monad,  every algebra is isomorphic to a free algebra. 
\end{myexample}

\newcommand{\tvec}{\monad_{\mathrm{vec}}}
\begin{myexample}[Algebra over a field]
    \label{ex:monad-algebra-over-field}
    Define $\monad X$ to be finite linear combinations of words in $X^*$, with rational coefficients. For example, 
    \begin{align*}
    2xyx + -2xx + 0.5xyz \in \monad \set{x,y,z}.
    \end{align*}
    We can view elements of this monad as polynomials with non-commuting variables.
    In other words, $\monad X = \tvec (X^*)$, where $\tvec$ is the monad of vector spaces from Example~\ref{ex:monad-of-vector-spaces} and $X^*$ is the monad of finite words\footnote{This is an example of a composite monad that arises via a distributive law of two monads. This type of construction was first described in
    \incite[Chapter on distributive laws]{appelgate1969seminar}
     }.
    On functions, the monad acts as follows
    \begin{align*}
        q_1 x_1 + \cdots + q_n  x_n \qquad \stackrel{\monad f} \mapsto \qquad q_1 f^*(x_1) + \cdots + q_n f^*(x_n),
        \end{align*}
        where $f^*$ is the substitutions in the monad of finite words. The unit maps $x$ to the linear combination which has the one-letter word $x$ with coefficient $1$. Free multiplication is defined like for polynomials, but the variables are non-commuting, e.g.:
        \begin{align*}
            3(4x -2y)(2xy+yy)   \quad \mapsto \quad 24xxy + \myunderbrace{12xyy - 12yxy}{this is not 0} - 6yyy.
        \end{align*}
     Every algebra over this monad has the structure of a vector space over the rationals, but there is more structure (e.g.~one can multiply two elements of the algebra)\footnote{Algebras over this monad are  known as ``algebras over the field of rational numbers'', but we avoid this terminology due to the over-loading of  ``algebra over''.}. 

     What is a recognisable colouring over this monad? In the  context of this monad (and also  the simpler monad of vector spaces from Example~\ref{ex:monad-of-vector-spaces}), it is more useful to work with different notions of ``finite algebra'' and ``algebra colouring'': instead of finite algebras, one should consider finite dimensional algebras (i.e.~those  where the underlying vector space has finite dimension), and instead of algebra colourings one should consider  linear maps to vector spaces. Under these adapted definitions, the algebra colourings recognised by finite algebras are exactly those which are recognised by  weighted automata, see Exercise~\ref{ex:weighted-automata}. 
\end{myexample}

\exercisehead

\mikexercise{\label{ex:term-modulo-equivalences} Consider the monad $\monad_\Sigma$ from Example~\ref{ex:term-monad}, where $\Sigma$ is some ranked set (possibly infinite). Let $X$ be some possibly infinite set of variables, and consider a set
\begin{align*}
\myunderbrace{\Ee \subseteq (\monad_\Sigma X) \times (\monad_\Sigma X)}{elements of this set will be called identities}.
\end{align*}
For a set $Y$, define $\sim$ to be the least congruence on $\monad_\Sigma Y$ that satisfies 
\begin{align*}
(\monad_\Sigma f)(t_1) \sim (\monad_\Sigma f)(t_2) \qquad \text{for every $(t_1,t_2)\in \Ee$ and $f : X \to Y$.}
\end{align*}
(This congruence can be obtained by intersecting all congruences with the above property.) Define a new monad as follows: $\monad Y$ is equal to $\monad_\Sigma Y$ modulo $\sim$, and the remaining components of the monad are defined in the natural way. Show that this is a monad.
}{}

\begin{figure}[]
    \centering
    
\begin{align*}
    \xymatrix{
        \smonad X 
        \ar[r]^{\delta_X}
        \ar[d]_{\smonad f}
        &
        \monad X 
        \ar[d]^{\monad f}
        \\
        \smonad Y 
        \ar[r]_{\delta_Y}
        &
        \monad Y
    }
    \end{align*}
    \begin{align*}
        \xymatrix @R=2pc 
    { X 
    \ar[d]_{\text{unit in $\smonad$}} 
    \ar[dr]^{\text{unit in $\monad$}}
    \\
    \smonad X \ar[r]_{\delta_X}& \monad X
    }
    \end{align*}
    \begin{align*}
    \xymatrix @R=2pc 
    { \smonad \smonad X 
    \ar[r]^{\smonad \delta_x}
    \ar[d]_{\text{free multiplication in $\smonad$}} 
    &
    \smonad \monad X 
    \ar[r]^{\delta_{\monad X}}
    &
    \monad \monad X
    \ar[d]^{\text{free multiplication  in $\monad$}}
    \\
    \smonad X \ar[rr]_{\delta_X}& & \monad X
    } 
    \end{align*}
    \caption{These three diagrams should commute for all sets $X$ and all functions $f : X \to Y$. The top diagram says that 
    $\set{\delta_X}_X$ is a natural transformation, while the bottom two diagrams say that it is compatible with unit and free multiplication.
    }
    \label{fig:monad-morphism}
\end{figure}
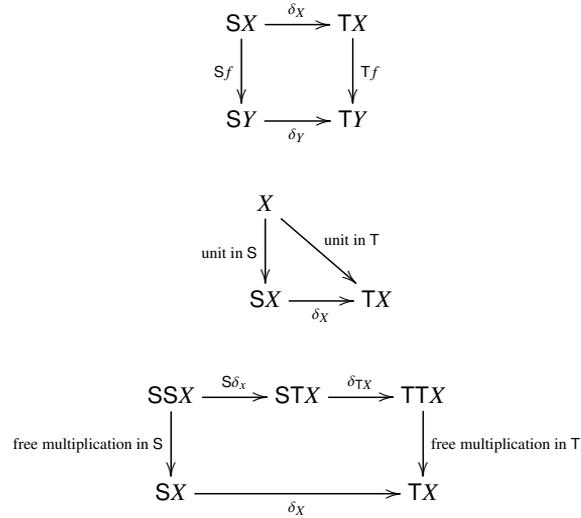

\mikexercise{For monads $\smonad$ and $\monad$, define a monad morphism from $\smonad$ to $\monad$ to be a family of 
functions
\begin{align*}
\set{ \delta_X : \smonad X \to \monad X}_{\text{$X$ is a set}}
\end{align*}
which is subject to the axioms in Figure~\ref{fig:monad-morphism}. Using the monads from  Section~\ref{sec:zillion}, give  five examples of monad morphisms, and five examples of pairs of monads which do not allow a monad morphism.
}{}

\mikexercise{\label{ex:monad-regular-elements} We say that $w \in \monad \Sigma$  is \emph{regular}  if   $\set w$ is a recognisable language. Find a monad where there are no regular elements. (Hint: it appears in this section.) }{
Groups. Consider a  homomorphism $h$ from the free group $\monad \Sigma$ to some finite group $G$. If $n \in \set{1,2,\ldots}$ is the order of the group $G$, then 
\begin{align*}
h(w) = h(w^{n+1}) \qquad \text{for every }w \in \monad \Sigma,
\end{align*}
and therefore $h$ cannot recognise any singleton language.

}
\mikexercise{What is the monad for rings (commutative and non-commutative)? Semirings? }{}

\mikexercise{Consider the following monad $\monad$. The set $\monad X$ is the set of $\omega$-words $X^\omega$, and the substitution $\monad f : \monad X \to \monad Y$ is defined coordinate-wise. The unit is $x \mapsto x^\omega$, and free multiplication  is defined by 
\begin{align*}
    w \in \monad \monad X \qquad \mapsto \qquad (i \mapsto 
    \myunderbrace{(w[i])[i]}{$i$-th letter of\\ 
    \scriptsize of $i$-th letter of $w$}).
\end{align*}
Show that this is a monad. Also, show that a language  $L \subseteq \monad \Sigma$ is recognisable if and only if it is clopen in the sense of Exercise~\ref{ex:clopen-omega-word}.
}{}

\mikexercise{\label{ex:monad-of-countable-well-founded-chains} Let $\monad$ be the monad of countable well founded chains. Show that a finite algebra with universe $S$  is uniquely determined by the operations
\begin{align*}
    \myunderbrace{ab}{binary\\ \scriptsize multiplication\\
    \shelahtype{S^2 \to S}} \qquad 
    \myunderbrace{a^\omega}{multiplication of\\ \scriptsize $aaa \cdots$ \\ \shelahtype{S\to S}}
    \end{align*}  
}{}

\mikexercise{\label{ex:aperiodic-countable-well-founded-chains} Consider 
the monad from Exercise~\ref{ex:monad-of-countable-well-founded-chains}. Show that a language $L \subseteq \monad \Sigma$ is definable in first-order logic (in the ordered model) if and only if it is recognised by a finite $\monad$-algebra $S$ where the underlying semigroup is aperiodic. 
}{}

\mikexercise{\label{ex:aperiodic-countable-well-founded-chains-regular-expressions} Consider 
the monad from Excercise~\ref{ex:monad-of-countable-well-founded-chains}. Consider regular expressions defined by the usual operators, plus $L^\omega$. Show that these expressions do not describe all recognisable languages. 
}
{The set of all countable well-founded orders, i.e.~the full language $\monad \set a$,  is not described by any regular expression (of the form in the exercise). The issue is that the regular expressions have bounded nesting of $\omega$, which is not true for countable well-founded orders. More precisely, define words as follows:
\begin{align*}
w_0 = a \qquad w_{n+1}= (w_n)^\omega.
\end{align*}
By induction on $n$, one shows that there is no regular expression with $\omega$-height $n$ which contains $w_{n+1}$. Here, the $\omega$-height is the nesting depth of $\omega$ (the other operators are not counted).
}

\mikexercise{Consider the monad and regular expressions from Exercise~\ref{ex:aperiodic-countable-well-founded-chains-regular-expressions}. Give an effective condition on finite algebras which corresponds exactly to the Boolean combinations of regular expressions. }
{There is a ranking function $r : A \to \set{0,1,\ldots}$ such that $r(a)=0$ if and only if $a$ is an absorbing zero element, and 
\begin{align*}
r(a^\omega) < r(a) \qquad \text{for all $a$ with nonzero rank.}
\end{align*}
}

\mikexercise{\label{ex:monad-of-countable-scattered-chains}  Let $\monad$ be the monad of countable scattered chains. Show that a finite algebra with universe $S$  is uniquely determined by the operations
\begin{align*}
    \myunderbrace{ab}{binary\\ \scriptsize multiplication\\
    \shelahtype{S^2 \to S}} \qquad 
    \myunderbrace{a^\omega}{multiplication of\\ \scriptsize $aaa \cdots$ \\ \shelahtype{S\to S}}
    \qquad
    \myunderbrace{a^{\omegaop}}{multiplication of\\ \scriptsize $\cdots aaa$ \\ \shelahtype{S\to S}}
    \end{align*}  
}{}

\mikexercise{Consider the monads from Examples~\ref{ex:monad-of-countable-well-founded-chains} and~\ref{ex:monad-of-countable-scattered-chains}. In which of these monads is first-order logic (over ordered models) equivalent to star-free expressions?
}{}

\mikexercise{Consider the monad from Example~\ref{ex:monad-of-countable-scattered-chains}. Which class of languages corresponds to aperiodicity (of the semigroup underlying the $\monad$-algebra)?
}{}

\mikexercise{\label{ex:weighted-automata} A weighted automaton over the rationals consists of:
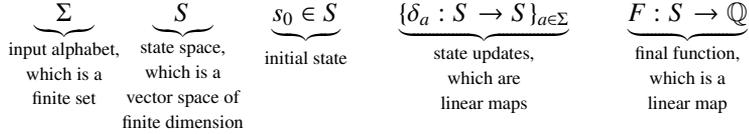
\begin{align*}
\myunderbrace{\Sigma}{input alphabet,\\
\scriptsize which is a \\
\scriptsize finite set}
\qquad 
\myunderbrace{S}{state space,\\
\scriptsize which is a \\
\scriptsize vector space of\\
\scriptsize finite dimension}
\qquad 
\myunderbrace{s_0 \in S}{initial state}
\qquad 
\myunderbrace{\set{\delta_a : S \to S}_{a \in \Sigma}}{state updates,\\
\scriptsize which are \\
\scriptsize linear maps}
\qquad 
\myunderbrace{F : S \to \Rat}{final function,\\
\scriptsize which is a  \\
\scriptsize linear map}
\end{align*}
The semantics of this automaton is a function of type $\Sigma^* \to \Rat$ defined as follows. Given an input word, do the following: start with the initial state, apply the state update for the first letter, then the state update for the second letter, and so on for all letters, and at the end apply the final function.  The semantics of a weighted automaton can also be naturally extended to finite  linear combinations of words over $\Sigma$, i.e.~to elements of the $\monad \Sigma$ as in Example~\ref{ex:monad-algebra-over-field}.

Show that  $L : \monad \Sigma \to \Rat$ is recognised by a weighted automaton if  and only if it is recognised by a finite dimensional $\monad$-algebra.
}{For the left-to-right implication, suppose that $\Aa$ is a weighted automaton, with state space $S$.  Define $A$ to be the set of linear maps $S \to S$, viewed as matrices. This is a $\monad$-algebra, with the multiplication defined in the natural way (using composition of linear maps and multiplying them by rational scalars).  It has finite dimension, namely the square of the dimension of $S$. Consider the homomorphism 
\begin{align*}
h : \monad \Sigma \to A,
\end{align*}
which is the unique extension of $a \in \Sigma \mapsto \delta_a$. This homomorphism recognises the language $L$, because 
\begin{align*}
L(w) =   \myunderbrace{F(h(w)(\text{initial state of $\Aa$}))}{a linear map, assuming that the argument is $h(w)$}
\end{align*}

We now prove the right-to-left implication. Suppose that $L : \monad \Sigma \to \Rat$ is recognised by  a homomorphism of $\monad$-algebras 
\begin{align*}
\xymatrix{
    \monad \Sigma \ar[dr]_{L}\ar[r]^{h}& A\ar[d]^{F}\\
    & \Rat
}
\end{align*}
A weighted automaton is defined as follows. The state space is $A$. The initial state is the multiplication of the empty empty word in $A$. The final function is $F$.
}

\section{Syntactic algebras}
\label{sec:syntactic-algebras}
In this section, we show that if an algebra colouring is recognisable, then it has a syntactic homomorphism, i.e.~a recognising homomorphism that stores the minimal amount of information\footnote{The results of this section, with the exception of Section~\ref{sec:finitary-monads}, are based on 
\incite[Part I.]{bojanczykRecognisableLanguagesMonads2015}
}.

\begin{definition}
    [Syntactic homomorphism] Fix a monad $\monad$ in the category of sets\footnote{This definition also makes sense for monads in other categories, assuming that one  one interprets ``surjective functions'' as ``epimorphisms''. However, the results in this book about the existence of syntactic homomorphisms will depend on the category of sets, and its generalisation to sorted sets.}.
    \label{sec:synt-hom} The \emph{syntactic homomorphism} of an algebra colouring $L : A \to U$ is any   surjective  homomorphism 
    \begin{align*}
    h :A \to  B
    \end{align*}
     which recognises $L$ and which  is minimal in the sense explained in the following quantified diagram
    \begin{align*}
      \red{\forall }  \blue{\exists! } \qquad
      
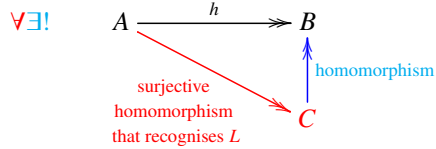
\xymatrix@C=2cm{ 
         A
         \ar@{->>}[r]^h 
         \ar@{->>}@[red][dr]_-{ \txt{\scriptsize \red{surjective} \qquad \\ \scriptsize \red{homomorphism }\\ 
         \scriptsize{\red{that recognises $L$}}}} 
      & B \\
       & \red C
       \ar@[blue]@{->>}[u]_{\blue {\text{homomorphism}}} }
    \end{align*}
\end{definition}

The algebra used by the syntactic  homomorphism is called the \emph{syntactic algebra}.  The syntactic algebra, if it exists,  is unique up to isomorphism of algebras. Also the syntactic homomorphism is unique in the following sense: every two syntactic homomorphisms will have the same kernel (equivalence relation on $A$ that identifies two elements with the same homomorphic image). This unique kernel is called the \emph{syntactic congruence} of $L$. 

We are mainly interested in the case where the algebra colouring describes a language, i.e.~$A$ is a free algebra, and there are two colours ``yes'' and ``no''.

There are two main results in this section. The first one, Theorem~\ref{thm:syntactic-homomorphism-monad}, says that if an algebra colouring  is recognisable,   then it has a syntactic homomorphism.  In general, colourings that are not recognisable need not have syntactic homomorphisms. The second one, Theorem~\ref{thm:plotkin}, says that a monad is finitary (roughly speaking, this means that all structures described by the monad are finite) if and only if every (not necessarily recognisable) algebra colouring has a syntactic homomorphism.  To illustrate these theorems, we begin with an example of an algebra colouring that does not have a syntactic homomorphism. In light of Theorems~\ref{thm:syntactic-homomorphism-monad} and~\ref{thm:plotkin}, the example uses a colouring that  is  not recognisable and a  monad that is not finitary.
\begin{myexample}\label{ex:no-syntactic-algebra} 
    Consider the monad of $\cc$-words. (As we will see later, this monad is not finitary.) Define $L$ to be the set of  $\cc$-words over a one letter alphabet $\set{a}$ which contain every finite word as an infix. More formally, 
    \begin{align*}
    L=  \set{ w \in \set{a}^\cc : \text{$a^n$ is an infix of $w$ for every $n \in \set{0,1,\ldots}$}}.
    \end{align*}
     This language is not recognisable, because all finite words must have different images under any recognising homomorphism (we leave this as an exercise for the reader).  We will show that $L$ does not have a syntactic homomorphism.
    For $n \in \set{1,2,\ldots}$, define 
    \begin{align*}
    w_n = \text{(shuffle of $\set a$)} \cdot a^n \cdot \text{(shuffle of $\set a$)}.
    \end{align*}
    Define  $h_n$ to be the function 
        \begin{align*}
        w \in \set{a}^\cc \quad \mapsto \quad 
        \begin{cases}
            w_n & \text{if $w=w_{n+1}$}\\
            w & \text{otherwise.}
        \end{cases}
        \end{align*}
        This function is  compositional for every $n$, and therefore it can be viewed as a homomorphism.
    If there would be a syntactic homomorphism, then it would need to factor through $h_n$. Since $h_n$ gives the same result for $w_n$ and $w_{n+1}$, therefore the same would have to be true for the syntactic homomorphism.   Therefore, the syntactic homomorphism $h$, if it existed, would need to give the same result for all $\cc$-words $w_1,w_2,\ldots.$
        By associativity, we would have 
    \begin{align*}
    h(\myunderbrace{w_1 w_1 w_1 \cdots}{$\not \in L$})=h(\myunderbrace{w_1 w_2 w_3 \cdots}{$\in L$}).
    \end{align*}
     Therefore the syntactic homomorphism does not exist.
\end{myexample}

\subsection{Terms and congruences}
\label{sec:terms}
To construct the syntactic homomorphism, we will use classical notions from universal algebra, such as terms and congruences, adapted to the monad setting. These notions and their basic properties are described below.

\paragraph*{Terms.} If $X$ is a set, possibly infinite, then  a \emph{term over variables $X$} is defined simply to be any element of $\monad X$. Given an algebra $A$, a term $t \in \monad X$ is interpreted as the following operation
    \begin{align*}
\myunderbrace{\eta \in A^X \quad \mapsto \quad  \text{multiply $(\monad \eta) (t)$ in $A$}}
{we write $\semterm t A$ for this operation}.
    \end{align*}
    Arguments of $t^A$ are called  \emph{variable valuations}. An operation of the form $t^A$, for some $t$, is called a \emph{term operation\footnote{
        If we are working in a category other than the category of sets, then the term operation $t^A : A^X \to A$ is not necessarily a morphism in the category. For example, in  a category of sorted sets with at least two sorts, the set $A^X$ is not a sorted set.
    }} in the algebra $A$. We distinguish between a term (which can be viewed as syntax) and the term operation that it generates in a given algebra (which can be viewed as semantics). If a term uses a finite set  of  $n$ variables with some implicit ordering, then  we write 
    \begin{align*}
        \semterm t A (a_1,\ldots,a_n)  
    \end{align*}
    for the result of applying $\semterm t A$ to the variable  valuation which maps the $i$-th variable to $a_i$. 
    
\begin{myexample}
    Consider the monad of finite words. The word $xy$ is a term, and the term operation induced it in an algebra (which is the same thing as a monoid) is binary multiplication. The operation induced by the term $\varepsilon$, which has an empty set of variables, is the constant that represents the monoid identity. Another example of a term operation is squaring, which is given by the term $xx$. A non-example is the idempotent power operation $a \mapsto a^\momega$. This is not a term operation, because the number $\momega$ depends on the algebra at hand (also, this number does not  exist in some infinite algebras).
    
    In the monad of $\cc$-chains, the \lale  operations of $\omega$-power and $\omegaop$-power are term operations which arise from the univariate terms  $x^\omega$ and $x^{\omegaop}$. To model  shuffling, we use an infinite family of terms, with the $n$-th one being the shuffle of $\set{x_1,\ldots,x_n}$.
\end{myexample}

Term operations commute with homomorphisms, as shown below.
\begin{lemma}\label{lem:terms-commute-with-homomorphisms}
    If $h : A \to B$ is a homomorphism, and $t \in \monad X$ is a term, then the following diagram commutes:
    \begin{align*}
    \xymatrix@C=1.5cm{
        A^X
        \ar[r]^{\eta \mapsto h \circ \eta}
        \ar[d]_{t^A}
         &
        B^X
        \ar[d]^{t^B}
        \\
        A \ar[r]_h & 
        B
    }.
    \end{align*}
\end{lemma}
\begin{proof} Consider a term $t \in \monad X$ and a valuation $\eta \in A^X$. We show below that if we start in a valuation $\eta \in A^X$, and follow the down-right and the right-down paths in the diagram from the statement, then we get the same element of $B$.
    \begin{align*}
    \text{down-right path applied to $\eta$} & =\\
    h(t^A(\eta)) 
    \equalbecause{definition of $t^A$}\\
    h( \text{(multiplication in $A$)}(\monad \eta))
    \equalbecause{$h$ is a homomorphism}\\
    \text{(multiplication in $B$)}(h \circ \monad \eta))
    \equalbecause{$\monad$ is a functor}\\
    \text{(multiplication in $B$)}(\monad(h \circ  \eta))
    \equalbecause{definition of $t^B$}\\
    t^B(h \circ \eta) & =\\
    \text{right-down path applied to $\eta$.} 
    \end{align*}
\end{proof}

\paragraph*{Congruences.} Define a  \emph{congruence} in an algebra $A$ to be an equivalence on the underlying set that satisfies any of the equivalent conditions in the following lemma.
\begin{lemma}
    \label{lem:congruence-definition}
    Let $A$ be an algebra. For every equivalence relation $\sim$ on its underlying set, the following conditions are equivalent:
    \begin{enumerate}
        \item \label{congruence:compositional} the function which maps $a \in A$ to its equivalence class is compositional;
         \item \label{congruence:kernel-hom}$\sim$ is the kernel of some homomorphism from $A$ to some algebra $B$;
        \item \label{congruence:commutes-with-terms} $\sim$ commutes with  every term operation,  which means that:
        \begin{align*}
            \myunderbrace{\eta_1 \sim \eta_2}{$\eta_1(x) \sim \eta_2(x)$ \\ \scriptsize for every $x \in X$}\quad \Rightarrow \quad t^A (\eta_1) \sim t^A (\eta_2)  \qquad \text{for every $t \in \monad X$ and $\eta_1,\eta_2 \in A^X$.}
            \end{align*}
    \end{enumerate}
\end{lemma}
\begin{proof}
    The implication~\ref{congruence:compositional}$\Rightarrow$ \ref{congruence:kernel-hom} follows from Lemma~\ref{lem:monad-compositional} which says that compositional functions are the same as homomorphisms. The implication \ref{congruence:kernel-hom}$\Rightarrow$\ref{congruence:commutes-with-terms} follows from Lemma~\ref{lem:terms-commute-with-homomorphisms}, which says that term operations commute with homomorphisms. The implication~\ref{congruence:commutes-with-terms}$\Rightarrow$\ref{congruence:compositional} follows from the definition of compositional functions. 
\end{proof}
 By condition~\ref{congruence:kernel-hom}, every congruence induces a  quotient algebra, where the universe is equivalence classes.

    \exercisehead

    \mikexercise{\label{ex:term-basis} Fix a monad in the category of sets.  Consider a set of terms $\basis$. We say that $\basis$ is a \emph{term basis} if for every finite algebra $A$ and subset $\Gamma \subseteq A$, the sub-algebra generated by $\Gamma$ is equal to the least subset of $A$ that contains $\Gamma$ and which is closed under applying term operations corresponding to terms  from $\basis$.  Show that if $\basis$ is a term basis, then a finite algebra $A$ is uniquely determined by its $\basis$-multiplication tables, which  is the family of term operations $\set{t^A}_{t \in \basis}$.
}{}

\mikexercise{Let $\monad$ be a monad which has a finite term basis $\basis$.  Show that given the $\basis$-multiplication tables in an algebra $A$, one  can compute the $\basis$-multiplication tables of the powerset algebra $\powerset A$ (as defined in Exercise~\ref{ex:powerset-algebra}, assuming that $\powerset A$ is indeed an algebra). 
}{}
\mikexercise{Find a notion of \emph{computable term basis} which generalises the previous exercise so as to capture the \lale operations in the monad of $\cc$-words.
}{}

\mikexercise{
    \label{ex:term-basis-congruence}
    Consider a monad with a term basis $\basis$. Show that $\sim$ is an equivalence relation on the underlying set of a finite algebra, then $\sim$ is a congruence if and only if it commutes with all term operations from the term basis $\basis$. 
}{
    Let $\basis$ be the term basis, let  $A$ be a finite algebra, and let $\sim$ be an equivalence relation on its universe. Define 
    \begin{align*}
    B \subseteq A \times A
    \end{align*}
    to be the subalgebra of $A \times A$ that is generated by $\sim$ when viewed as a set of pairs. By definition of subalgebras,  $B$ is the set of values of the form 
    \begin{align*}
    t^{A \times A}(\eta),
    \end{align*}
    where $t$ ranges over terms of the form $t \in \monad X$ for some set of variables $X$, and $\eta$ ranges over valuations from $X$ to the the set of pairs $\sim$.   For $i \in \set{1,2}$, let 
    \begin{align*}
    \pi_i : A \times A \to A
    \end{align*}
    be the projection onto the $i$-th coordinate. Since this is a homomorphism, and term operations commute with homomorphisms thanks to Lemma~\ref{lem:terms-commute-with-homomorphisms}, we have 
    \begin{align*}
    \pi_i(t^{A \times A})(\eta) = t^A(\pi_i \circ \eta).
    \end{align*}
    Therefore, $B$ is equal to the set of pairs
    \begin{align*}
    (t^A(\pi_1 \circ \eta), t^A(\pi_2 \circ \eta))
    \end{align*}
    where the ranges of  $t$ and $\eta$ are the same as when describing $B$. Since every pair of valuations
    \begin{align*}
    \eta_1, \eta_2 : X \to A
    \end{align*}
    which are pointwise equivalent under $\sim$ can be described as $\pi_1 \circ \eta$ and $\pi_2 \circ \eta$ for some $\eta$, it follows that 
    \begin{align*}
    B = \set{(t^A(\eta_1),t^A(\eta_2)): t \in \monad X \text{ and } \eta_1 \sim \eta_2 \in A^X}.
    \end{align*}
    If we use the definition of congruences as equivalence relations that commute with all term operations, then $\sim$ is a congruence if and only if $B \subseteq \sim$.   
    
    Since $\basis$ is a term basis, it follows that $B$ is also the smallest subset of $A \times A$ which contains all pairs from $\sim$ and which is closed under applying all term operations from $\basis$. By assumption that $\sim$ commutes all term operations from the basis, it follows that $B \subseteq \sim$. 
}

\mikexercise{\label{ex:vectorial-terms} Define a \emph{vectorial term} to be any function
\begin{align*}
f : 
\myunderbrace Y{output \\ \scriptsize variables}
\to \monad \myunderbrace{X}{input \\ \scriptsize variables}.
\end{align*}
For  an algebra $A$ and a vectorial term $f$ as above, define
\begin{align*}
\myunderbrace{f^A : A^X \to A^Y}{such a function is called \\ \scriptsize a \emph{vectorial term operation}}
\end{align*}
to be the function which maps  $\eta \in A^X$ to the  following function:
\begin{align*}
    \xymatrix{
        Y 
        \ar[r]^{f}
        &
        \monad X
        \ar[r]^{\monad \eta}
         &
         \monad A
        \ar[r]^{\mu}
        &
    A.
    }
    \end{align*}
Show that  vectorial term operations are closed under composition.
}
{
    Closure under composition is proved in the following claim.

\begin{claim} \label{claim:Kleisli-functor} 
    For every vectorial terms 
    \begin{align*}
    g : Z \to \monad Y \qquad f : Y \to \monad X
    \end{align*}
    there is a vectorial term 
    \begin{align*}
    \myunderbrace{g \cdot f  : Z \to \monad X}
    {this is called the \emph{Kleisli composition} of $g$ and $f$}
    \end{align*}
    such that the following diagram commutes for every algebra $A$:
    \begin{align*}
    \xymatrix{
        A^X 
        \ar[r]^{f^A}
        \ar[dr]_{(g \cdot f)^A}
        &
        A^Y 
        \ar[d]^{g^A}
        \\
        & A^Z
    }
    \end{align*}
    
\end{claim}
\begin{proof}
    The Kleisli composition $g \cdot f$ from the statement of the claim is defined to be the composition of the following functions:
    \begin{align*}
    \xymatrix{
        Z \ar[r]^g
        &
        \monad Y
        \ar[r]^{\monad f}
        &
        \monad \monad X
        \ar[rr]^{\text{free multiplication}}
        && 
        \monad X
    }   
    \end{align*}
Consider a valuation $\eta : Z \to A$. To prove that applying $(g \cdot f)^A$ to $\eta$ gives the same result as applying first $f^A$ and then $g^A$, as required in the statement of the claim, consider the following diagram. 
\begin{align*}
    \xymatrix@C=1.7cm
    {
        X 
        \ar[r]^g
        \ar[drr]_{g \cdot f}
         &
         \monad Y
         \ar[r]^{\monad f}
         \ar@/^3.0pc/[rrr]^{\monad (f^A(\eta))}
        &
        \monad \monad X 
        \ar[r]^{\monad \monad \eta} 
        \ar[d]^{\txt{\scriptsize free \\ \scriptsize
        multiplication}}
        &
        \monad \monad A 
        \ar[r]^{\monad \mu} 
        \ar[d]^{\txt{\scriptsize free \\ \scriptsize multiplication}}
        &
        \monad A 
        \ar[d]^{\mu}\\
        &
        & 
        \monad X
        \ar[r]_{\monad \eta} 
        &
        \monad A
        \ar[r]_\mu
         & 
        A
    }
\end{align*}
The bottom-most path from $X$ to $A$ describes the result of applying $(g \cdot f)^A$ to the valuation $\eta$, while the top-most path describes the result of applying first $f^A$ and then $g^A$. Therefore, to prove the lemma, it is enough to show that the diagram commutes.  The top face in the diagram (incident with the curving arrow) describes the definition of $f^A$, with the functor $\monad$ applied to it. The left-most face is the definition of Kleisli composition.  The left rectangular face  is naturality of free multiplication, which is axiom~\eqref{eq:monad-axiom-naturality-mult} of monads. The right rectangular face is the associativity axiom for Eilenberg Moore algebras. 
\end{proof}
}

    \subsection{Syntactic homomorphisms for  recognisable  colourings}
    In this section, we prove  the first result about syntactic homomorphisms, which says that they always exist for  algebra colourings that are recognisable. 
    \begin{theorem}\label{thm:syntactic-homomorphism-monad} 
        Let $\monad$ be a monad in the category of sets. Every recognisable algebra colouring has a  syntactic homomorphism.   
    \end{theorem}

The proof is based on congruences. The main result is (a strengthening of) the observation that congruences of finite index in a given algebra, ordered by inclusion when viewed as sets of pairs, form a lattice. This means that   every two congruences of finite index have a least upper bound and greatest lower bound. For the greatest lower bound (which is called the  \emph{meet} in the terminology of  lattices), the observation is straightforward: if $\sim_1$ and $\sim_2$ are congruences, not necessarily of finite index, then their intersection (when viewed as a set of pairs) clearly commutes with all term operations, and is therefore a congruence. Since every lower bound must be contained in the intersection, it follows that the intersection is the greatest lower bound.

The least upper bound (which is called the \emph{join} in the terminology of lattices) is more interesting. Here, we use the assumption on finite index (it is enough that one of the congruences has finite index), see Exercise~\ref{ex:finiteness-assumption-in-join-needed} for why this assumption is needed.

\begin{lemma}\label{lem:join-congruences}
    Let $A$ be an algebra. If $\sim_1$ and $\sim_2$ are congruences and $\sim_1$ has finite index,  then they have a join congruence, i.e.~a least upper bound among all congruences in $A$.
\end{lemma}

\begin{proof}
    Define the join to be the transitive closure of the union of $\sim_1$ and $\sim_2$.
        In other words, two elements of the algebra are related by the join if one can be reached from the other using a finite number of steps which use either one of the congruences $\sim_1$ or $\sim_2$.  This is the same as the join in the lattice of equivalence relations, and therefore  every congruence that contains both $\sim_1$ and $\sim_2$ must contain the join defined above. It remains to prove that this join is in fact a congruence, and not just an equivalence relation. To prove this, we 
     use the definition of congruences which says that they  commutes with all term operations in the algebra. We begin with the special case of term operations with  finitely many variables.

        \begin{claim}\label{claim:join-commutes-finitely-many-variables}
            The join commutes with all term operations  that have finitely many variables. 
        \end{claim}
        \begin{proof} Consider a term operation 
            \begin{align*}
            t^A(x_1,\ldots,x_n)
            \end{align*}
            with finitely many variables. We need to show that  every two outputs 
            \begin{align*}
            t^A(a_1,\ldots,a_n) \qquad \text{and} \qquad  t^A(b_1,\ldots,b_n)
            \end{align*}
            are equivalent under the join, assuming that the inputs are pairwise equivalent under the join. This is proved using a finite number of steps, where in each step we use commutation of $t^A$ with either $\sim_1$ or $\sim_2$.
        \end{proof}
        Using the above claim, and the assumption that $\sim_1$ has finite index, we show that the join commutes with all term operations, even those with infinitely many variables.  Consider a term $t \in \monad X$ and two valuations 
        \begin{align*}
        \eta_1,\eta_2 : X \to A
        \end{align*}
        which are pointwise equivalent with respect to the join. We need to show that applying $t^A$ to both of these valuations gives outputs that are equivalent with respect to the join. We begin with the special case when the two valuations have finite images.
        \begin{claim}\label{claim:special-case-of-finite-image}
            If  $\eta_1,\eta_2$ have finite image $B \subseteq A$, then $t^A(\eta_1) = t^A(\eta_2)$. 
        \end{claim}
        \begin{proof}
            Define $Y$ to be the set of pairs in $B \times B$ which are equivalent under $\sim$ and let $\eta : X \to Y$  be the function which maps an variable from $X$ to the pair of its images under $\eta_1$ and $\eta_2$. Each valuation $\eta_i$ can be decomposed as first applying $\eta$, and then taking the $i$-th projection, as  in the following diagram:
            
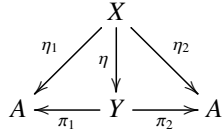
\begin{align*}
            \xymatrix{
                & X 
                \ar[dl]_{\eta_1}
                \ar[dr]^{\eta_2}
                \ar[d]_{\eta}
                \\
                A &
                Y 
                \ar[l]^{\pi_1}
                \ar[r]_{\pi_2}
                &
                A
            }
            \end{align*}
            Applying the term operation $t^A$ to a valuation $\eta_i$ is the same as applying the term operation $(\monad \eta(t))^A$ to the valuation $\pi_i$, which is proved by chasing the following diagram:
            
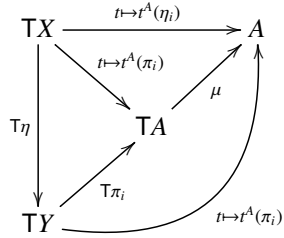
\begin{align*}
            \xymatrix{
                \monad X 
                \ar[rr]^{t \mapsto t^A(\eta_i)}
                \ar[dd]_{\monad \eta} 
                \ar[dr]^{t \mapsto t^A(\pi_i)}
                & &
                A \\
                &
                \monad A
                \ar[ur]_\mu \\
                \monad Y
                 \ar[ur]_{\monad \pi_i}
                \ar@/_3pc/[uurr]_{t \mapsto t^A(\pi_i)}
            }
            \end{align*}
            Since the  term operation $(\monad \eta(t))^A$ uses a finite set of variables $Y$, and since the valuations $\pi_1$ and $\pi_2$ are equivalent under $\sim$, the statement of this claim follows from Claim~\ref{claim:join-commutes-finitely-many-variables}.
        \end{proof}
        Using the above claim, and the assumption that $\sim_1$ has finite index, we conclude the proof of the lemma. 
        Choose a function 
        \begin{align*}
        \alpha : A \to  A
        \end{align*}
        that maps every element of $A$ to some chosen element in its equivalence class under $\sim_1$. This function has finite image, because $\sim_1$ has finite index. We now conclude the proof of the lemma as follows:
        \begin{align*}
        t^A(\eta_1) 
        \equivalentbecause{\sim_1}{$\sim_1$ is a congruence}\\
        t^A(\alpha \circ \eta_1)
        \equivalentbecause{\sim}{$\alpha$ has finite image and Claim~\ref{claim:special-case-of-finite-image}}
        \\
        t^A(\alpha \circ \eta_2)
        \equivalentbecause{\sim_1}{$\sim_1$ is a congruence}
        \\
        t^A(\eta_2)
        \end{align*}
\end{proof}
Using the above lemma, we complete the proof of Theorem~\ref{thm:syntactic-homomorphism-monad}.
\begin{proof}[Proof of Theorem~\ref{thm:syntactic-homomorphism-monad}]
    Consider algebra colouring $L : A \to U$ that is recognisable. Let $\Cc$ be the  set of congruences in $A$ that recognise the colouring in the following sense
\begin{align*}
a \sim b \qquad \text{implies} \qquad L(a)=L(b).
\end{align*}
Because the colouring is recognisable, the there is at least one congruence  $\approx$ of finite index in $
\Cc$.  Define $\Cc_\approx \subseteq \Cc$ to be the congruences which contain $\approx$. This is a finite set, since every congruence in $\Cc_\approx$ is obtained by merging some of the finitely many equivalence classes in $\approx$. By Lemma~\ref{lem:join-congruences}, $\Cc_\approx$ is a finite lattice, and therefore it has a greatest element, call it $\sim$. Again by Lemma~\ref{lem:join-congruences}, every congruence in $\Cc$ has an upper bound in $\Cc_\approx$, and therefore $\sim$ is the greatest element also of $\Cc$. We will prove that the 
quotient homomorphism
\begin{align*}
    h : A \to A_{/\sim},
    \end{align*}
is the syntactic homomorphism of $L$.

By translating maximality of $\sim$ into  the language of homomorphisms, it follows that  that every surjective  homomorphism $g : A \to B$ that recognises $L$ must factor through $h$,
 i.e.~must be  some function $f$ such that $g=f \circ h$.  The last thing to show is that $f$ is in fact a homomorphism, and not just any function on the underlying sets of the algebras. This is shown in the following claim, with $C$ being the quotient $A_{/\sim}$.

\begin{claim}\label{claim:function-that-factors-is-homomorphism}
    Let $A,B,C$ be algebras, let $g,h$ be  surjective homomorphisms, and let $f$ be a function which makes the following diagram commute. 
    
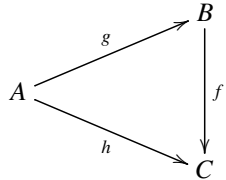
\begin{align*}
        \xymatrix@R=0.6cm@C=2cm{
            & B \ar[dd]^f\\
            A \ar[ur]^g \ar[dr]_h \\
            & C
        }
        \end{align*}
     Then $f$ is a homomorphism.
\end{claim}
\begin{proof}
    Consider the following diagram:
    
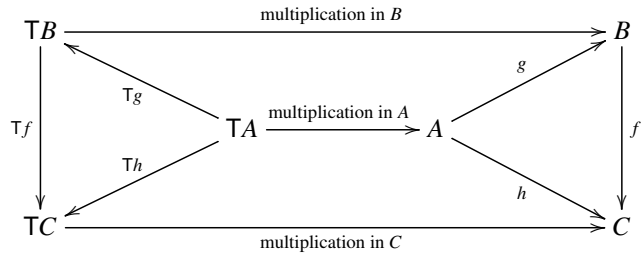
\begin{align*}
    \xymatrix@C=2cm{
        \monad B 
        \ar[rrr]^{\text{multiplication in $B$}}
        \ar[dd]_{\monad f}
        &&& 
         B
         \ar[dd]^f\\
        & 
        \monad A 
        \ar[ul]^{\monad g}
        \ar[dl]_{\monad h}
        \ar[r]^{\text{multiplication in $A$}}
        &
        A 
        \ar[ur]^g
        \ar[dr]_h
        \\
        \monad C
        \ar[rrr]_{\text{multiplication in $C$}}
        &&&
        C
    } 
    \end{align*}
    The upper and lower faces commute because $g$ and $h$ are homomorphisms, and the left and right faces commute by definition of $f$. Since all arrows in the diagram are surjective, it follows that the perimeter of the diagram commutes, which means that $f$ is a homomorphism.
\end{proof}
\end{proof}


\exercisehead

\mikexercise{
    \label{ex:finiteness-assumption-in-join-needed}
    Show that the assumption on finite index of $\sim_1$ in Lemma~\ref{lem:join-congruences} is needed.}{}
\mikexercise{
    Consider the monad from Example~\ref{ex:monad-algebra-over-field}.  Show that if   $\lambda : \monad \Sigma \to \Rat$ is recognised by a weighted automaton, then the same is true for the syntactic homomorphism.
}{}

\mikexercise{This exercise can be seen as a variant 
of Moore's algorithm for computing the syntactic congruence.  Consider a monad in the category of sets, together with a term basis, see Exercise~\ref{ex:term-basis}. 
Consider an algebra colouring $\lambda : A \to U$ where $A$ is a finite algebra.  Show that the syntactic congruence of $\lambda$ is the greatest (coarsest) equivalence relation on $A$ recognises $\lambda$ and which  is stable under all term operations from the term basis. 
}{}{}

\subsection{Finitary monads}
\label{sec:finitary-monads}
In the monad of finite words, every language --  not just  recognisable ones -- has a syntactic homomorphism. For example, in the monad of finite words, the syntactic homomorphism of the non-recognisable language ``the number of $a$ letters is equal to the number of $b$ letters''
maps a word to the difference (number of $a$'s $-$ number of $b$'s).
In the monad of $\cc$-words, some non-recognisable languages do not have syntactic homomorphisms, as witnessed by Example~\ref{ex:no-syntactic-algebra}. What is the difference?

The difference, as will be shown in Theorem~\ref{thm:plotkin} below, is that every finite word uses only a finite subset of the alphabet, which is no longer true for  $\cc$-words. This is made precise by the following definition.

\begin{definition}[Finitary elements and monads] Let $\monad$ be a monad in the category of sets.
    We say that an element $t \in \monad X$ is \emph{finitary} if
\begin{align*}
t = (\monad f)(t) \qquad \text{for some $f : X \to X$ with finite image.}
\end{align*}
We say that $\monad$ is  finitary if for every $X$, all elements of $\monad X$ are finitary.
\end{definition}
 
For example, the monad of finite words is finitary, while the monads of $\cc$-words is not. 
The following theorem shows that finitary monads are exactly those monads where all algebra colourings have syntactic homomorphisms. 

\begin{theorem}\label{thm:plotkin} Let $\monad$ be a monad in the category of sets. Then $\monad$ is finitary if and only if every algebra colouring has a syntactic homomorphism\footnote{This theorem is unpublished work of Gordon Plotkin and the author.}.
\end{theorem}
\begin{proof}
    For the left-to-right implication, we use the same proof as for Theorem~\ref{thm:syntactic-homomorphism-monad}. Define the join of a possibly infinite set of congruences to be the transitive closure of their union. By the same reasoning as in Claim~\ref{claim:join-commutes-finitely-many-variables}, the join commutes with all term operations that have finitely many variables. Because the monad is finitary, all term operations are like this, and therefore the join commutes with all term operations, and is therefore a congruence. (We have thus shown that for finitary monads, the congruences in an algebra form a complete lattice.) If we now take the join of all congruences that recognise a given algebra colouring, then we get the syntactic congruence, and the quotient homomorphism is the syntactic homomorphism.

    We now prove the converse implication. Suppose that every algebra colouring has a syntactic homomorphism.
    Fix some set $X$. We will show that all elements of $\monad X$ are finitary. Let $\red X$ be a disjoint copy of $X$.  For a finite subset $Y \subseteq X$, define 
    \begin{align*}
    f_Y : \red X +  X \to \red X +  X
    \end{align*}
    to be the function which maps each element to itself, with the exception of the red copies of elements from $Y$, which are mapped to their corresponding black copies. Define $\sim$ to be the equivalence relation on $\monad (X + \red X)$ which identifies two elements if, for some finite $Y \subseteq X$, they have the same image under $\monad f_Y$.  

    \begin{claim}
        $\sim$ is a congruence on $\monad(X + \red X)$.
    \end{claim}
    \begin{proof}
        We first argue that $\sim$ is  an equivalence relation.  Transitivity argued as follows:
    \begin{align*}
    w_1\myunderbrace{ \sim}{as witnessed \\ 
    \scriptsize by $Y_1 \subseteq X$} w_2  
    \quad \text{and} \quad 
    w_2\myunderbrace{ \sim}{as witnessed \\ 
    \scriptsize by $Y_2 \subseteq X$} w_3
    \qquad \text{implies} \qquad
    w_1\myunderbrace{ \sim}{as witnessed \\ 
    \scriptsize by $Y_1 \cup Y_2 \subseteq X$} w_3.
    \end{align*}
    Consider the syntactic homomorphism of $\sim$, which exists by the assumption  that syntactic congruences exist. Let $\approx$ be the kernel of the syntactic homomorphism, which means that (a) $\approx$ is a congruence that is contained in  $\sim$ when viewed as a set of pairs, and (b)  $\approx$ contains  every  congruence that is contained in $\sim$.
    We will show that $\sim$ is actually equal to $\approx$, and therefore $\sim$ is a congruence. In light of (a), it is enough to show that $\sim$ is contained in $\approx$. Indeed, suppose that two elements are equivalent under $\sim$. By definition, this means that they have the same image under $\monad f_Y$ for some finite $Y$. Since $\monad f_Y$ is a  homomorphism that   recognises $\sim$, it follows by (b) that the two elements are equivalent under $\approx$. 
\end{proof}
    
    Consider the functions 
    \begin{align*}
\red f, f :  X \to \red X +  X
        \end{align*}
        such that $\red f$ maps each argument to its red copy, and $f$ is the identity.
For every  $x \in X$, its unit is mapped by $\monad \red f$ and $\monad f$ to elements which
 have the same image under $\monad f_{\set x}$, and therefore are equivalent under $\sim$.  Since $\sim$ is a congruence, and all units in $\monad X$ are mapped by $\monad \red f$ and $\monad  f$ to elements equivalent under $\sim$,  it follows that for every $w \in \monad X$, its images under $\monad \red f$ and $\monad  f$ are equivalent under $\sim$.  By definition of $\sim$, this means that for every $w \in \monad X$  there must be some finite $Y \subseteq X$ such that 
\begin{align}\label{eq:plotkin-Y}
(\monad (f_Y \circ \red f))(w) =
(\monad  (f_Y \circ f))(w)
\end{align} 
We now complete the proof that every element of $\monad X$ is finitary. 
Let $w \in \monad X$, and let $Y$ be such that the above equivalence holds.  Choose an element  $y \in Y$ and consider the function
\begin{align*}
    g :   \red  X + X \to X  
    \end{align*}
     which is the identity on $X$ and maps all red letters to $y$.   
We have 
\begin{eqnarray*}
    w \equalbecause{because $g \circ  f$ is the identity on $X$}\\ 
    (\monad (g \circ  f ))(w) \equalbecause{because $f_{ Y}$ is the identity on black letters}\\
    (\monad (g \circ f_{ Y} \circ   f ))(w) \equalbecause{by~\eqref{eq:plotkin-Y}}\\
    (\monad (g \circ f_{Y} \circ  \red f ))(w ).
\end{eqnarray*}
The image of the function $g \circ f_{\red Y} \circ \red f$ is contained in $Y$, and therefore we have established that $w$ is finitary. 
\end{proof}

\exercisehead

\mikexercise{Give an example of a monad which is not finitary, but where every language $L \subseteq \monad \Sigma$ with a finite alpahbet $\Sigma$ has a syntactic homomorphism.}{ The (not necessarily finite) powerset monad.}


\mikexercise{Let $\monad$ be a monad in the category of sets. Show that if every algebra colouring with two colours has a syntatic homomorphism, then every algebra colouring with an arbitrary number of colours has a syntactic homomorphism.}{}

\mikexercise{Give an example of a monad $\monad$ which is not finitary, but such that all elements of $\monad X$ are finitary for countable $X$.}{}

\mikexercise{
Let $S$ be a finite set of sort names, and consider the category 
\begin{align*}
\mathsf{Set}^S
\end{align*}
of $S$-sorted sets with sort-preserving functions.  Prove  Theorem~\ref{thm:syntactic-homomorphism-monad}  for monads over this category.  }{}

\mikexercise{Consider a category of sorted sets, as in the previous exercise, but with infinitely many sort names. Define a finite algebra to be one that is finite on every sort. Show that Theorem~\ref{thm:syntactic-homomorphism-monad} fails.
}{}

\mikexercise{Show that a monad  in the category of sets is finitary if and only if it arises as a result of the construction described in Exercise~\ref{ex:term-modulo-equivalences}.}{}

\mikexercise{Recall the notion of \emph{regular elements} from Exercise~\ref{ex:monad-regular-elements}. Show that if $t$ is regular, then it is finitary. 
}{}

\mikexercise{Assume that the regular elements, as considered in the previous exercise, are closed under free multiplication in the following sense: if $t \in \monad X$ is a regular term operation, and $\eta : X \to \monad Y$ is a valuation of its variables that uses only regular elements, then $t^{\monad Y}(\eta)$ is a regular element. Under these assumptions, define a monad of regular elements.
}{}

\section{The Eilenberg Variety Theorem}
\label{sec:eilenberg}
In Chapter~\ref{chap:logics}, we proved several theorems of the kind
\begin{align*}
\text{class of languages} \qquad 
\sim 
\qquad
\text{class of semigroups}.
\end{align*}
For example, a language of finite words is definable in first-order logic if and only if it is recognised by an aperiodic semigroup. In this section  we prove that every class of languages with good closure properties will correspond to a class of algebras with good closure properties. The theorem was originally proved by Eilenberg for monoids\footcitelong[Theorem 13.2]{Eilenberg76}, but with some extra care one can make the proof work in the abstract setting of monads.

\subsection{Unary polynomials}
Before stating and proving the theorem, we describe unary polynomials, which are used in the definition of language varieties. 
For an algebra $A$, define a \emph{unary polynomial}\footnote{The terminology of ``terms'' and ``polynomials'' comes from universal algebra, see
\incite[Definition 13.3.]{sankappanavar1981course}
This terminology  can be explained -- or at least more easily remembered --  as follows. Consider the ring of the reals
$
(\mathbb R, +, -, \times, 0, 1).
$
A term  operation in this ring can only use the constants $0$ and $1$ which are given in the ring as an algebra, and therefore term operations correspond to polynomials with integer coefficients. If we want to get all polynomials, we need to  allow the terms to use arbitrary elements of $\mathbb R$ as constants.
} to be any function of the form 
\begin{align*}
a \in A \quad \mapsto \quad t^A(a,c_1,\ldots,c_n) \in A,
\end{align*}
which is obtained for some choice of $n \in \set{0,1,\ldots}$, some  term $t$ with $n+1$ variables\footnote{A more principled definition, which allows more variables and infinitely many constants, is discussed in Exercise~\ref{ex:vectorial-polynomials}. Since we use unary polynomials mainly for finite algebras, the more elementary definition given here is enough. } and some parameters $c_1,\ldots,c_n \in A$. 
\begin{lemma}
    In every algebra,  unary polynomials are closed under composition.
\end{lemma}
\begin{proof}
    Consider two unary polynomials 
    \begin{eqnarray*}
        a \in A &\quad \mapsto \quad & t^A(a,c_1,\ldots,c_n)\\
a \in A  & \quad \mapsto \quad & s^A(a,d_1,\ldots,d_m).
        \end{eqnarray*}
    To prove that the composition of the above two unary polynomials is also a unary polynomial, we will show that there is a term $u$ with $1+m+n$ variables which satisfies the following equality:
\begin{align}\label{eq:the-term-u}
 \myunderbrace{t^A(s^A(a,d_1,\ldots,d_m),c_1,\ldots,c_n)}
 {a composition of two unary polynomials}
 = 
 \myunderbrace{u^A(a,d_1,\ldots,d_m,c_1,\ldots,c_n)}
 {a single unary polynomial}.
\end{align}
For the purposes of this proof, we treat a number such as $n$ as a set which has $n$ elements.
Define the term $u$ to be the result of applying the term operation  
\begin{align*}
    t^{\monad(1+m +n)} : (\monad(1+m+n))^{1+n} \to \monad(1+m+n)
\end{align*}
to the valuation $f$ 
which maps the variable in $1$ to the term $s$ (seen as a term over a larger set of variables that does not use the last $n$ variables), and which maps the variables in $n$ to their corresponding units.  The equality~\eqref{eq:the-term-u} follows from the following claim, in the case where the variables $X$ are  $1 + n$, the variables   $Y$ are $1 + m + n$, and the valuation $\eta \in A^Y$ is 
\begin{align*}
 (a,d_1,\ldots,d_m,c_1,\ldots,c_n) \in A^{1+m+n}.
\end{align*}

        \begin{claim}\label{claim:kleisli}
            For every $f : X \to \monad Y$ and  $t \in \monad X$, the term 
            \begin{align*}
            u \eqdef t^{\monad Y}(f)
            \end{align*}
            makes the following diagram commute for  every algebra $A$:
            \begin{align*}
            \xymatrix
            {
                A^Y 
                \ar[dr]^{u^A}
                \ar[d]_{\eta \in A^Y\  \mapsto\  x\in X \  \mapsto\  (f(x))^A(\eta)}
                \\
                A^X \ar[r]_{t^A} 
                &
                A
            }
            \end{align*}
        \end{claim}
        \begin{proof}
            This claim is the same as Exercise~\ref{ex:vectorial-terms}.
            Let $\eta \in A^Y$ be a valuation. If we apply the function in the vertical arrow from the diagram to $\eta$, then we get the valuation $\rho \in A^X$ that is the composition of the following functions:
            \begin{align*}
                 \xymatrix{
                     X \ar[r]^{f} &
                     \monad Y 
                     \ar@/^2pc/[rr]^{t \mapsto t^A(\eta)}
                     \ar[r]^{\monad \eta}
                     &
                     \monad A
                     \ar[r]^{\mu}
                     &
                     A,
                 }
            \end{align*}
             where $\mu$ is the multiplication operation of the algebra $A$. The diagram in the statement of the claim says that 
             \begin{align*}
             t^A(\rho) = u^A(\eta).
             \end{align*}
             To prove this, consider the following diagram:
            \begin{align*}
            \xymatrix@C=2cm{
                \monad X 
                \ar[rr]^{\monad \rho}
                \ar[d]_{\monad f}
                &&
                \monad{A}
                \ar[dd]^{\mu}
                \\
                \monad \monad Y
                \ar[r]^{\monad \monad \eta}
                \ar[d]_{\text{free multiplication}} 
                &
                \monad \monad A
                \ar[d]_{\text{free multiplication}} 
                \ar[ur]^{\monad \eta}
                \\
                \monad Y 
                \ar[r]_{\monad \mu}
                &
                \monad A
                \ar[r]_{\mu}
                &
                A
            }
            \end{align*}
            The upper-left face commutes by definition of $\rho$. The lower-left face commutes by naturality of free multiplication, and the lower-right face commutes by associativity of $\mu$. Therefore, the entire diagram commutes. If we apply the  top-most path  from $\monad X$ to $A$ in the diagram to the term $t \in \monad X$, then we get the result $t^A(\rho)$, while if we apply the bottom-most path to the same term, then we get the result $u^A(\eta)$. Since the diagram commutes, these results are equal, thus proving the claim. 
        \end{proof}
\end{proof}
Another result about unary polynomials that will be used in the proof of the Eilenberg Variety Theorem is the following characterisation of congruences in finite algebras. The finiteness assumption is important, see Exercise~\ref{ex:finitness-important-for-cong-poly}.
\begin{lemma}\label{lem:congruences-in-terms-of-polynomials}
    An equivalence relation $\sim$ in a finite algebra $A$ is a congruence if and only if it commutes with all unary polynomials, in the sense that
    \begin{align*}
    a \sim b \quad\text{implies} \quad f(a) \sim f(b) 
    \end{align*}
    holds for every $a,b \in A$ and every unary polynomial $f : A \to A$.
\end{lemma}
\begin{proof}
    The  left-to-right implication is immediate, and does not need the assumption on finiteness of the algebra. If $\sim$ is a congruence, then it  commutes with all term operations, and so it must also commute with  unary polynomials, which are term operations with some arguments fixed. 
    
    The   right-to-left implication is proved similarly to Lemma~\ref{lem:join-congruences} about joins of congruences. Suppose that $\sim$ commutes with all unary polynomials. By the same argument as in Claim~\ref{claim:join-commutes-finitely-many-variables}, where arguments are replaced one by one in finitely many steps, it follows that $\sim$ commutes with all term operations that have finitely many variables. Since the algebra $A$ is finite, all valuations for term operations have finite image, and therefore we can use  Claim~\ref{claim:special-case-of-finite-image} to prove that $\sim$ commutes with all term operations, and therefore it is a congruence. 
\end{proof}

\exercisehead

\mikexercise{\label{ex:finitness-important-for-cong-poly}Show that the finiteness assumption in Lemma~\ref{lem:congruences-in-terms-of-polynomials} is  needed. Hint: use Example~\ref{ex:no-syntactic-algebra}.
}{
    Consider the language $L$ from Example~\ref{ex:no-syntactic-algebra}. We show that contextual equivalence for this language (viewed as a colouring with values ``yes'' and ``no'') is not a congruence. 
    A unary polynomial  in the free algebra $\set a^\cc$  corresponds to  a $\cc$-word over alphabet $\set{a,x}$. From the point of view of $L$, there are two kinds of unary polynomials. If $f$ is a unary polynomial such that the corresponding $\cc$-word over $\set{a,x}$ contains $a^n$ as an infix for every $n$, then   $f(w) \in L$  for every $w \in \set a ^\cc$. Otherwise, in the corresponding word there is some bound $n_0\in \set{0,1,\ldots}$ such that words $a^n$ with $n>n_0$ do not appear as infixes; in this case  $f(w) \in L\iff w \in L$. These observations imply that contextual equivalence for $L$ has two equivalence classes: namely $L$ and its complement. In particular,  contextual equivalence is not a congruence, since otherwise  $L$ would be recognisable.
}

\mikexercise{\label{ex:vectorial-polynomials}
Define a \emph{vectorial polynomial} in an algebra $A$ to be any operation 
\begin{align*}
f : A^X \to A^Y,
\end{align*}
for some sets $X$ and $Y$, 
which arises as follows: (a) transform an input valuation $A^X$ to a larger valuation $A^{X+Z}$ by mapping the variables from $Z$ to some fixed constants; and then (b) apply a vectorial term operation $A^{X+Z} \to A^Y$, as described in Exercise~\ref{ex:vectorial-terms}. Show that vectorial polynomials are closed under composition.

}{}{}

\subsection{Varieties}
The classes with good closure properties will be called varieties, in analogy with the  varieties that appear in Birkhoff's theorem from universal algebra.  
In this section, we define varieties, and give several examples of them. 
There will be two kinds of varieties: for algebras and for languages.  We begin with the algebras. In the following definition, a \emph{quotient} of an algebra is any image of that algebra under a surjective homomorphism. In other words, a quotient is a quotient under some congruence.

\begin{definition}
    [Algebra variety] \label{def:algebra-variety} Fix a $\monad$  in the category of sets.  An \emph{algebra  variety} is a class $\algclass$ of finite $\monad$-algebras  with the following closure properties:
    \begin{itemize}
        \item \emph{Quotients.} If $\algclass$ contains $A$, then it contains every quotient of $A$.
        \item \emph{Sub-algebras.} If $\algclass$ contains $A$, then it contains  every sub-algebra of $A$.
        \item \emph{Products.} If $\algclass$ contains $A$ and $B$, then it contains $A \times B$.
    \end{itemize}
    \end{definition}

\begin{myexample}\label{ex:algebra-pseudovariety}
    Consider the monad of finite words, where algebras are monoids. Examples algebra varieties include: finite groups, finite aperiodic monoids,  finite infix trivial monoids, or finite prefix trivial monoids.
\end{myexample}

\begin{myexample}
    Here is a non-example. Consider the monad of nonempty finite words, where algebras are semigroups. The class of monoids (i.e.~semigroups which have an identity element) is not an algebra variety, because it is not closed under sub-algebras.     
\end{myexample}

\begin{myexample}
    \label{ex:identities}
    Consider a monad $\monad$. Define an \emph{identity} to be a  pair of terms $s,t \in \monad X$ over a common set of variables $X$. An algebra $A$ is said to satisfy the identity if 
    \begin{align*}
    s^A(\eta) = t^A(\eta) \qquad \text{for every }\eta \in A^X.
    \end{align*}
    The class of finite algebras that satisfy a given identity (more generally, all identities in a given set of identities) is easily seen to be an algebra variety.  For example, the algebra variety of commutative semigroups arises from the identity
    \begin{align*}
    xy = yx
    \end{align*}   
    in the monad of nonempty finite words.
Some algebra varieties do not arise this way. For example, the varieties discussed in Example~\ref{ex:algebra-pseudovariety} do not arise from (even possibly infinite sets of) identities. Identities will be discussed in more detail in Section~\ref{sec:identities}.
\end{myexample}

We now describe language varieties. In Eilenberg's original formulation, this is a class of regular languages that is closed under Boolean combinations, inverse images of homomorphisms, and inverse images of operations of the form 
\begin{align*}
w \in \Sigma^+ \mapsto v_1 w v_2 \in \Sigma^+ \qquad \text{for  fixed $v_1,v_2 \in \Sigma^*$.}
\end{align*}
In the more abstract setting of monads, the role of these operations will be played by unary polynomials, as described in the following definition.

In the following definition, by recognisable languages we mean recognisable subsets of free algebras.
\begin{definition}
    [Language variety] Let $\monad$ be a monad in the category of sets.  A \emph{language variety} is a class $\langclass$ of recognisable languages with the following closure properties:
    \begin{itemize}
        \item \emph{Boolean combinations.} $\langclass$ is closed under Boolean combinations, including complementation.
        \item \emph{Inverses of homomorphisms.} If $h : \monad \Sigma \to \monad \Gamma$ is a homomorphism of free algebras, then $\langclass$ is closed under inverse images of $h$. 
        \item \emph{Inverses of unary polynomials.} If $f : \monad \Sigma \to \monad \Sigma$ is a unary polynomial in a free algebra $\monad \Sigma$, then $\langclass$ is closed under inverse images of $f$. 
    \end{itemize}
\end{definition}

\begin{myexample}
    Consider the monad of finite words, where algebras are monoids. We will that  languages definable in first-order logic are a language variety. Closure under Boolean combinations is immediate, because we are dealing with a logic. Closure under inverse images of homomorphism or unary polynomials can be proved using \ef games: if $f$ is either a homomorphism or a unary polynomial, then a strategy copying argument shows  that
     \begin{align*}
     \txt{\small Duplicator wins the \\ \small  $k$ round game on $w$ and $w'$} \qquad \text{\small implies} \qquad 
     \txt{\small  Duplicator wins the \\ \small  $k$ round game on $f(w)$ and $f(w').$}
     \end{align*}
     This implies that first-order definable languages are closed under inverse images of homomorphisms and unary polynomials.
     The same is true for first-order logic on $\cc$-words.
\end{myexample}

\begin{myexample}
    Consider again the monad of finite words, where algebras are monoids. The definite languages from Example~\ref{ex:definite} are not a variety, because the class of definite languages is not closed under inverse images of  the homomorphisms. Indeed, the language 
    \begin{align*}
    \myunderbrace{a\set{a,b}^* \subseteq \set{a,b}^*}{words that being with $a$}
    \end{align*}
    is definite. If we take the inverse image under the homomorphism 
    \begin{align*}
        h: \set{a,b,c}^* \to \set{a,b}^*,
    \end{align*}
    which erases the $c$ letters, then we get the language 
    \begin{align*}
    \myunderbrace{c^* a \set{a,b,c}^* \subseteq \set{a,b,c}^*,}{words that begin with $a$ if $c$ is erased}
    \end{align*}
    which is not definite. The problem is with homomorphism that erase letters. If we would consider the same class of languages but in the monad of nonempty finite words, where algebras are semigroups, then we would get a variety. 
\end{myexample}

\exercisehead

\mikexercise{ Consider the monad of finite words. Show that a class of languages $\langclass$ is a variety if and only if it is closed under Boolean combinations, inverse images under homomorphisms, and inverse images of unary polynomials of the form:
\begin{align*}
     w \mapsto v_1 w v_2  \qquad \text{for every choice of parameters $v_1,v_2 \in \Sigma^*$}.
\end{align*}
}{}

\mikexercise{ Consider the monad of finite words. Show that there are uncountably many algebra varieties.
In particular, for some algebra varieties, the membership problem $A \stackrel ? \in \algclass$ is undecidable.
}{
Let $P$ be a set of prime numbers. Consider the class of finite groups, where the order is a number where all divisors are from $P$. This is an algebra variety.
}

\mikexercise{ Consider the monad of $\cc$-words. Show that a class of languages $\langclass$ is a variety if and only if it is closed under Boolean combinations, inverse images under homomorphisms, and inverse images of unary polynomials of the following forms:
\begin{align*}
     w \mapsto v_1 w v_2  & & \text{for every choice of parameters $v_1,v_2 \in \Sigma^\cc$}\\
     w \mapsto w^\omega \\
     w \mapsto w^\omegaop \\
     w \mapsto \text{shuffle of }\set{w,v_1,\ldots,v_n}
     && \text{for every choice of parameters $v_1,\ldots,v_n \in \Sigma^\cc$.}
\end{align*}
}{}

\mikexercise{\label{ex:potthof} Consider the  monad $\monad $ from Example~\ref{ex:term-monad}, where $\monad X$ describes terms over a fixed ranked set $\Sigma$ with variables $X$. We view term $t \in \monad \Gamma$ as a model, where the elements are the nodes of the corresponding tree, there is a binary ancestor relation $x \le y$, and for every $\sigma \in \Sigma + \Gamma$ there is a unary relation $\sigma(x)$ which selects nodes with label $\sigma$. Show that the class of languages definable in first-order logic is not a variety. Hint: read the exercises in Chapter~\ref{sec:forest-algebra}.   }{}

\subsection{Algebra varieties are the same as language varieties }
In this section we prove that the two notions of variety are equivalent.
\begin{theorem}[Eilenberg Variety Theorem] \label{thm:eilenberg}
    Let $\monad$ be a monad in the category of sets.
    Then the maps in the following diagram are mutually inverse bijections.
    
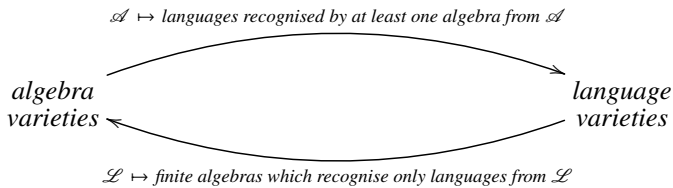
\begin{align*}
    \xymatrix@C=6cm{
        \txt{algebra \\
        varieties}
        \ar@/^2pc/[r]^{\algclass \ \mapsto \ \text{languages recognised by at least one algebra from $\algclass$}}
        &
        \txt{language \\
        varieties}
        \ar@/^2pc/[l]^{\langclass \ \mapsto \ \text{finite algebras which recognise only  languages from $\langclass$}}
    }
    \end{align*}
\end{theorem}

\begin{proof}
Let us write $\langmap$ for the left-to-right map, and $\algmap$ for the right-to-left map. 
  We first show that each of these two maps take varieties to varieties, and then we show that the two maps are mutual inverses. 

\begin{enumerate}
    \item  We first show that if the input to $\langmap$ satisfies a weaker assumption than being an algebra variety, namely it is closed under products, then the output   $\langmap \algclass$ is a language variety.  
    
    We begin with Boolean combinations. If $L$ is recognised by an algebra $A \in \algclass$, then its complement is recognised by the same algebra. If furthermore  $K$ is recognised by $B \in \algclass$, then  $L \cup K$ and $L \cap K$ are both recognised by the product $A \times B$, which belongs to $\algclass$ by closure under products.  
    
    Consider now the inverse images. Let  $L$ be a language that is  recognised by a homomorphism
    \begin{align*}
    h : \monad \Sigma \to A \in \algclass.
    \end{align*}
    We need to show that $\langmap \algclass$ contains all inverse images of $L$ under homomorphisms and unary polynomials.  Consider first the homomorphisms: let  $g : \monad \Gamma \to \monad \Sigma$ be a homomorphism, and consider  the inverse image of $L$ under $g$, which can be written as $L \circ g$ if we view $L$ as a function with outputs ``yes'' and ``no''. This inverse image is recognised by the homomorphism  $h \circ g$, which uses the algebra $A$, and therefore it belongs to $\langmap \algclass$. The  same kind argument applies to unary polynomials. Consider a unary polynomial  $f : \monad \Sigma \to \monad \Sigma$. As we have remarked in the proof of  Lemma~\ref{lem:congruences-in-terms-of-polynomials}, congruences  commute with unary polynomials, which means that $h \circ f = f \circ h$, and therefore  $h$ also recognises the inverse image $L \circ f$.
    \item  Similarly to the first step, we show that if the input to $\algmap$ satisfies a weaker condition than being a language variety, namely it is closed under under unions and intersections, then the output is an algebra variety\footnote{The first two steps of this proof establish that the maps $\langmap$ and $\algmap$ form  what is known as a Galois connection, between 
\begin{itemize}
    \item classes of finite algebras closed under products; and
    \item classes of recognisable languages closed under unions and intersections.
\end{itemize}
In the terminology of Galois connections, the  varieties of both kinds are the closed sets, with respect to this Galois connection.
}.   Every language recognised by a sub-algebra of $A$ is also recognised by $A$, and the same is true for quotients, and therefore $\algmap \langclass$ is closed under sub-algebras and quotients of $A$. Consider now products. Suppose that a language $L$ is recognised by a homomorphism
\begin{align*}
h : \monad \Sigma \to A \times B \qquad \text{with }A,B \in \algmap \langclass.
\end{align*}
For every $a \in A$, the inverse image 
\begin{align*}
L_a = h^{-1}( \set a \times B)
\end{align*}
is recognised by the homomorphism
\begin{align*}
    h_A : \monad \Sigma \to A,
\end{align*}
which is the 
composition of $h$ with the projection to $A$. Since the latter homomorphism has domain $A$, it follows that $L_a \in \langclass$.  For similar reasons, if $b \in B$ then $\langclass$ contains the language 
\begin{align*}
L_b = h^{-1}(A \times \set b).
\end{align*}
The intersection $L_a \cap L_b$ is the inverse image under $h$ of the pair  $(a,b)$.  Every  language recognised by $h$ is a finite union of such languages; and therefore it belongs to $\langclass$ by closure under unions and intersections.

\item We now show that the maps $\algmap$ and $\langmap$ are mutual inverses. We first show that  every algebra variety $\algclass$ satisfies
\begin{align*}
\algclass = \algmap \langmap \algclass,
\end{align*}
with the dual equality being proved in the next step. The above equality  is the same as showing that $A \in \algclass$ if and only if 
\begin{itemize}
    \item[(*)] every language recognised by $A$ is recognised by some algebra in $\algclass$.
\end{itemize}
Clearly every algebra  $A \in \algclass$ satisfies (*). We now prove the converse implication. Suppose that an algebra $A$ satisfies (*).  The multiplication operation 
\begin{align*}
\mu : \monad A \to A
\end{align*}
in the algebra $A$ is a homomorphism  from the free algebra $\monad A$ to $A$. By the assumption that $A$ satisfies (*),  every language recognised by this homomorphism is recognised by some algebra from  $\algclass$. In particular, for every $a \in A$ the language $\mu^{-1}(a)$ is  recognised by some homomorphism 
\begin{align*}
h_a : \monad A \to B_a \in \algclass.
\end{align*}
Consider the product homomorphism
\begin{align*}
h : \monad A \to  \prod_{a \in A}B_a \qquad t \mapsto (h_a(t))_{a \in A}.
\end{align*}
Define $B$ to be the image of $h$. The algebra $B$ is a sub-algebra of  a product of algebras from $\algclass$, and therefore it belongs to  $\algclass$. From now on,  we view $h$ as surjective homomorphism onto its image $B$. This  homomorphism  recognises all languages $\mu^{-1}(a)$, and therefore $\mu$ factors through $h$, i.e.~there is some function $f$ which makes the following diagram commute:
\begin{align*}
    \xymatrix{
        \monad A \ar[r]^\mu \ar[dr]_{h} & A \\
& B \ar[u]_f 
    }
\end{align*}
By Lemma~\ref{claim:function-that-factors-is-homomorphism}, $f$ is not just a function but also  a homomorphism of algebras. This means that $A$ is the image of $B$ under a surjective homomorphism. In other words, $A$ is a quotient of $B$, and therefore $A \in \algclass$. 

\item In the final step, we  show that every language variety $\langclass$ satisfies
\begin{align*}
\langclass = \langmap \algmap \langclass.
\end{align*}
This is the same as showing that $L \in \langclass$ if and only if 
\begin{itemize}
    \item[(*)] $L$ is recognised by an algebra that only recognises languages from $\langclass$.
\end{itemize}
Clearly (*) implies $L \in \langclass$, so we focus on the converse implication. Suppose that $L \in \langclass$, and its syntactic homomorphism, which exists by Theorem~\ref{thm:syntactic-homomorphism-monad}, is 
\begin{align*}
h : \monad \Sigma \to A.
\end{align*}
To prove (*), we will show that all languages recognised by the syntactic algebra $A$ belong to  $\langclass$.

\begin{claim}\label{claim:contextual-congruence}
    Let $F : A \to \set{\text{``yes'', ``no''}}$ be the accepting set in the syntactic algebra, which means that $L$ is equal to $F \circ h$. Then  two elements of $A$ are equal  if and only if they have the same values under  $F \circ f$ for every unary polynomial $ f : A \to A$.
\end{claim}
\begin{proof}
    Let $\sim$  be the equivalence relation on $A$  which identifies two elements that have the same image under $F \circ f$ for every unary polynomial $f$.  Because unary polynomials are closed under composition, it follows that $\sim$ commutes with all unary polynomials, and therefore it is a congruence by Lemma~\ref{lem:congruences-in-terms-of-polynomials}. 
     Because the identity is a special case of a unary polynomial, elements that are equivalent under $\sim$ have the same value under  $F$. This means that the quotient homomorphism of $\sim$ recognises $F$, and therefore $\sim$ must be the identity since otherwise $A$ would not be the syntactic algebra of $L$. 
\end{proof}
The following claim shows that unary polynomials in $A$ can be pulled back, along the homomorphism $h$, to unary polynomials in $\monad \Sigma$. 
\begin{claim}\label{claim:homomorphisms-on-unary-polynomials}
    For every unary polynomial $f : A \to A$ there is a unary polynomial $f^h : \monad \Sigma \to \monad \Sigma$ which makes the following diagram commute: 
    \begin{align*}
    \xymatrix{
        \monad \Sigma 
        \ar[r]^{f^h}
        \ar[d]_h &
        \monad \Sigma 
        \ar[d]^{h}\\
        A 
        \ar[r]_f &
        A 
    }
    \end{align*}
\end{claim}
\begin{proof}
    Consider a unary polynomial $f : A \to A$ of the form 
    \begin{align*}
    a \in A \quad  \mapsto \quad  t^A(a,c_1,\ldots,c_n).
    \end{align*}
    Because the syntactic homomorphism is surjective, for each $i \in \set{1,\ldots,n}$ there must be some $s_i \in \monad \Sigma$ which is mapped to $c_i$ by $h$. Since term operations commute with homomorphisms by Lemma~\ref{lem:terms-commute-with-homomorphisms}, the diagram in the claim commutes if we choose $f^h$ to be 
    \begin{align*}
    s \in \monad \Sigma \quad  \mapsto \quad  t^{\monad \Sigma}(s,s_1,\ldots,s_n).
    \end{align*}
\end{proof}
We are now ready to show that $\langclass$ contains all languages recognised by the syntactic algebra $A$.

We first show that $\langclass$ contains all languages recognised by the syntactic homomorphism $h : \monad \Sigma \to A$, and then we generalise this result to other homomorphisms into $A$. 
By Claim~\ref{claim:contextual-congruence} and finiteness of the algebra $A$, there is a finite set 
\begin{align*}
\Xx \subseteq A \to A
\end{align*}
of unary polynomials in the algebra $A$ such that two elements are equal if and only if they  have the same values for all functions from the set
\begin{align*}
\set{F \circ f : f \in \Xx}.
\end{align*}
Putting this together with Claim~\ref{claim:homomorphisms-on-unary-polynomials}, it follows that  two elements of $\monad \Sigma$ have the same image under $h$ if and only if they belong to the same sets from the finite family
\begin{align*}
    \set{\myunderbrace{L \circ f_h}{a language that belongs to $\langclass$ \\ 
    \scriptsize as the inverse image of $L$\\
    \scriptsize under the unary polynomial $f_h$} : f \in F}.
    \end{align*}
In other words, every inverse image $h^{-1}(a)$ is a finite Boolean combination of languages from the above family, and therefore it belongs to $\langclass$ by closure under Boolean combinations. This in turn means that all languages recognised by $h$ are in $\langclass$.

We now prove that not only does $\langclass$ contain every language recognised by the syntactic homomorphism $h$, as we have already shown,  but it also contains every language recognised by a homomorphism 
\begin{align*}
g : \monad \Gamma \to A
\end{align*}
which uses the same target algebra of the syntactic homomorphism. By surjectivity of the syntactic homomorphism and the universal property of the free algebra $\monad \Gamma$, we can choose some homomorphism $f$ which makes the following diagram commute
\begin{align*}
\xymatrix{
    \monad \Gamma \ar[dr]^g \ar[d]_f \\
    \monad \Sigma \ar[r]_h & A\\
}
\end{align*}
By the above diagram, every language recognised by $g$ is an inverse image, under $f$, of some language recognised by $h$. Since we have already proved that every language recognised by $h$ is in $\langclass$, and $\langclass$ is closed under inverse images of homomorphisms such as $f$, we see that every language recognised by $g$ is in $\langclass$.
\end{enumerate}
\end{proof}

\exercisehead

\mikexercise{
Let $S$ be a finite set, and consider the category 
\begin{align*}
\mathsf{Set}^S
\end{align*}
of $S$-sorted sets with sort-preserving functions. 
State and prove the Eilenberg Variety Theorem for monads over this category.}{}

\mikexercise{
\label{ex:weighted-varieties}    
Consider the monad 
from Example~\ref{ex:monad-algebra-over-field}, which corresponds to weighted automata. We adapt to varieties to the weighted setting as follows.  Define an algebra variety to be class of finite-dimensional algebras which is closed under sub-algebras, quotients and products. Define a language variety to be a class $\langclass$ of linear maps $\monad \Sigma \to \Rat$, recognised by finite-dimensional algebras, which is closed under inverse images of homomorphisms and polynomials, and which is closed under combinations in the following sense: if $\langclass$ contains 
\begin{align*}
\set{\lambda_i : \monad \Sigma \to U_i}_{i \in \set{1,2}},
\end{align*}
and $f : \Rat^2 \to \Rat$ is a linear map, then $\langclass$ contains also 
\begin{align*}
w \mapsto f(\lambda_1(w),\lambda_2(w)).
\end{align*}
Show that the Eilenberg Variety Theorem holds for varieties understood in this way. 
}{}

\mikexercise{
    Consider the weighted varieties from the previous example. 
    What is the weighted analogue of star-free languages? Hint: consider the concatenation of two linear maps
    \begin{align*}
    \lambda_1,\lambda_2 : \monad \Sigma \to U
    \end{align*}
to be the linear map which is defined as follows on $\Sigma^*$
    \begin{align*}
    (\lambda_1 \cdot \lambda_2) (a_1 \cdots a_n) = \sum_{i \in \set{0,\ldots,n}} \lambda_1(a_1 \cdots a_i) \cdot \lambda_2(a_{i+1} \cdots a_n),
    \end{align*}
and which is extended to $\monad \Sigma$ by linearity.
    }{}

\mikexercise{\label{ex:contextual-equivalence} For an algebra colouring $L : A \to U$, define \emph{contextual equivalence} to be the equivalence relation on $A$ which identifies two elements of $A$ if they have the same image under $L \circ f$ for every unary polynomial $f : A\to A$. Show that if $L$ is recognisable, then contextual equivalence is the syntactic congruence of $L$.
}{}{}

\mikexercise{Show that contextual equivalence, as defined in the previous exercise, need not be a congruence for algebra colourings that are not recognisable.}{Consider the language $L$ from Example~\ref{ex:no-syntactic-algebra}. We show that contextual equivalence for this language (viewed as a colouring with values ``yes'' and ``no'') is not a congruence. 
A unary polynomial  in the free algebra $\set a^\cc$  corresponds to  a $\cc$-word over alphabet $\set{a,x}$. From the point of view of $L$, there are two kinds of unary polynomials. If $f$ is a unary polynomial such that the corresponding $\cc$-word over $\set{a,x}$ contains $a^n$ as an infix for every $n$, then   $f(w) \in L$  for every $w \in \set a ^\cc$. Otherwise, in the corresponding word there is some bound $n_0\in \set{0,1,\ldots}$ such that words $a^n$ with $n>n_0$ do not appear as infixes; in this case  $f(w) \in L\iff w \in L$. These observations imply that contextual equivalence for $L$ has two equivalence classes: namely $L$ and its complement. In particular,  contextual equivalence is not a congruence, since otherwise  $L$ would be recognisable.}

\section{Identities and Birkhoff's Theorem}
\label{sec:identities}
In this section we return to the identities that were described in Example~\ref{ex:identities}.  Recall that an  \emph{identity} is a pair of  terms over a common set of variables. 
We say that an algebra $A$ satisfies an identity consisting of terms $s,t \in \monad X$ if 
\begin{align*}
h(s)=h(t) \qquad \text{for every homomorphism $h : \monad X \to A$}.
\end{align*}
This is equivalent to the definition given in Example~\ref{ex:identities}, which said that  an algebra $A$ satisfies the identity if the two term operations $s^A$ and $t^A$ are equal.

Below we present two theorems about classes of algebras that can be defined using identities.  There will be two theorems,  one for not necessarily finite algebras, and one for finite algebras.

\subsection*{Identities for varieties not necessarily finite algebras}
We begin with the first theorem about identities, which is a monad variant of Birkhoff's  Theorem from universal algebra\footcitelong[Theorem 10.]{Birkhoff35}. The theorem says that a class of not necessarily finite algebras can be described by identities if and only if it  is a \emph{Birkhoff  variety}, which means that it is closed under images of surjective homomorphisms, subalgebras and (not necessarily finite) products. Traditionally, Birkhoff  varieties are called simply algebra varieties, but the latter name has already been used in  this book for  classes of finite algebras that are described in Definition~\ref{def:algebra-variety}. To avoid confusion, for the purposes of this section where the two kinds of algebra varieties are used, we use the name \emph{Eilenberg  variety} for the varieties of finite algebras.

We say that a class of not necessarily finite algebras is \emph{defined by a set of identities $\eqclass$} if the algebras in the class are exactly those that satisfy all identities from $\eqclass$. 

\begin{theorem}[Birkhoff]\label{thm:birkhoff}
    Let $\monad$ be a monad in the category of sets. 
    A class $\algclass$ of not necessarily finite algebras is a Birkhoff variety if and only if it can be defined by some set of identities.
\end{theorem}

We will prove a slightly stronger result, which establishes a duality between algebras and identities. Under this duality, Birkhoff varieties will correspond to sets of identities that are closed under consequences, as described below.

A \emph{consequence} of  a set of identities $\eqclass$ is defined to be  any identity that is satisfied in every algebra that  satisfies all identities from $\eqclass$. We say that a set of identities  is \emph{closed} if it contains all of its consequences.  In symbols, a set of identities $\eqclass$ is closed if it satisfies the following:
\begin{align*}
    \myunderbrace{\forall e}{for every\\ \scriptsize identity} \ (\myunderbrace{\forall A}{for every \\ \scriptsize algebra} \myunderbrace{A \models \eqclass \Rightarrow A \models e}{if $A$ satisfies all 
    \\ \scriptsize identities 
     in $\eqclass$,\\ \scriptsize then it satisfies $e$}) \Rightarrow e \in \eqclass.
\end{align*}

By taking the right and then left arrow in the following theorem, we immediately get the Birkhoff's Theorem. 

\begin{theorem} \label{thm:birkhoff-duality}
    Let $\monad$ be a monad in the category of sets.
    Then the maps in the following diagram are mutually inverse bijections.
    
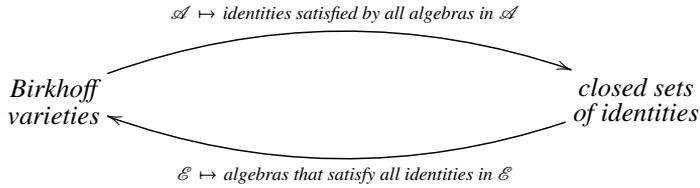
\begin{align*}
    \xymatrix@C=6cm{
        \txt{Birkhoff\\ 
        varieties}
        \ar@/^2pc/[r]^{\algclass \ \mapsto \ \text{identities satisfied by all algebras in  $\algclass$}}
        &
        \txt{closed sets\\ of identities}
        \ar@/^2pc/[l]^{\eqclass \ \mapsto \ \text{algebras that satisfy all identities in  $\eqclass$}}
    }
    \end{align*}
\end{theorem}
\begin{proof}
    Let us write $\eqmap$ for the left-to-right map in the diagram from the theorem, and $\algmap$ for the right-to-left map. Almost by definition, applying the map $\eqmap$ to any set of algebras will produce a set of identities that is closed. It is also not hard to see that the map $\algmap$ produces Birkhoff varieties, because algebras satisfying a given identity are closed under surjective homomorphic images, subalgebras and possibly infinite products. 
    It remains to show that the maps are mutually inverse, which corresponds to the following two equalities:
    \begin{align*}
    \myunderbrace{\eqclass  = \eqmap \algmap \eqclass}
    {for every closed set of identities $\eqclass$}
    \hspace{2cm}
    \text{and}
    \hspace{2cm}
    \myunderbrace{\algclass  = \algmap \eqmap  \algclass}
    {for every Birkhoff variety $\eqclass$}.
    \end{align*}
    The first equality says that an identity belongs to $\eqclass$ if and only if it is satisfied by all algebras that satisfy all identities in $\eqclass$. This equality is simply the definition of a closed set of identities, and the equalities holds for closed sets. We are left with the second equality, which says that an algebra belongs to $\algclass$ if and only if (*) it satisfies all identities that are true in all algebras from $\algclass$. Clearly every algebra from $\algclass$ satisfies (*). The converse implication follows from the following lemma, and the closure properties of Birkhoff varieties.

\begin{lemma}\label{lem:hsp}
     Let $\algclass$  set of algebras. 
    If an algebra $B$ satisfies all identities that are true in all algebras from $\algclass$, then $B$ is a homomorphic image of a subalgebra of a (possibly infinite) product of algebras from $\algclass$. 
\end{lemma}
\begin{proof} The key observation is that satisfying an identity can be interpreted in terms of homomorphisms, in the following way: an identity consisting of two terms $\monad B$ is true in all algebras from $\algclass$ if and only if the two terms have the same image  under  every homomorphism
\begin{align*}
h : \monad B \to A \in \algclass.
\end{align*}
Take the product of all possible homomorphisms $h$ as above, and  restrict the resulting homomorphism to its image,  yielding  a surjective   homomorphism
\begin{align*}
H : \monad B \to A \in \myunderbrace{\text{subalgebras}}{because we restricted to the image} \text{of products of algebras in $\algclass$,}
\end{align*}
such that two terms in $\monad B$ are an identity true in all algebras from $\algclass$ if and only if they have the same image under $H$.   If two terms  $s,t \in \monad B$ have the same image under $H$, then they form an  identity that is true in all algebras from $\algclass$, and therefore also an identity that is true in the algebra $B$, by assumption on $B$. If an identity is true in $B$, then the two terms in the identity must have the same result under the multiplication operation of the algebra $B$, since the latter is an example of a homomorphism of type $\monad B \to B$. Summing up, we have shown that if two terms in $\monad B$ have the same result under $H$, then they have the same multiplication. 
This means that the multiplication operation of the algebra $B$ factors through the surjective homomorphism $H$. 
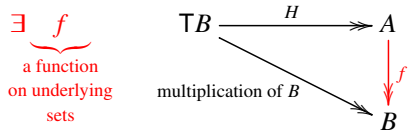
\begin{align*}
    \red{\exists 
    \myunderbrace{f}{a function \\
    \scriptsize on underlying \\
    \scriptsize sets}}
      \qquad
    \xymatrix@C=2cm{ 
       \monad B
       \ar@{->>}[r]^H 
       \ar@{->>}[dr]_-{ \text{ multiplication of $B$\qquad}}
    & A 
    \ar@[red]@{->>}[d]^{\red f} 
    \\
     & B
      }
  \end{align*}
By Claim~\ref{claim:function-that-factors-is-homomorphism}, the function $\red f$ is actually a homomorphism. Therefore, $B$ is  is the image of $A$ under some surjective homomorphism.  
\end{proof}

\end{proof}

\subsection*{Identities for varieties of finite algebras}
We now turn to identities that characterise varieties of finite algebras, the same varieties that were use in the Eilenberg Variety Theorem. To avoid confusion with the Birkhoff varieties of possibly infinite algebras that are also discussed in this chapter, we use the name \emph{Eilenberg varieties} for varieties of finite algebras. 

\begin{theorem}[Eilenberg-\schutz]\label{thm:eilenberg-schutz}\footnote{This theorem is based on 
    \incite[Theorem 1]{eilenberg1975pseudovarieties}
    The theorem cited above differs in two ways from our Theorem~\ref{thm:eilenberg-schutz}: (a) our theorem works for any monad subject to the assumption on countably many finite algebras; (b) the characterisation in terms of identities from~\cite{eilenberg1975pseudovarieties} is different, because it gives a sequence of identities (and not sets of identities), and it requires satisfying all but finitely many identities from the sequence. Eliminating difference (b) seems to require some extra assumptions on the monad.
    }
    Let $\monad$ be a monad in the category of sets, such that there are countably many finite algebras up to isomorphism. The following conditions are equivalent for every class $\algclass$ of finite algebras:
    \begin{enumerate}
        \item \label{eil-schutz:variety} $\algclass$ is an Eilenberg variety, i.e.~it is closed under images of surjective homomorphisms, subalgebras and finite products;
        \item \label{eil-schutz:identities} there is a sequence of sets of identities
         \begin{align*}
            \myunderbrace{\eqclass_1 \supseteq \eqclass_2 \supseteq \cdots }{each $\eqclass_n$ is a set of identities}
            \end{align*}
            such that a finite algebra belongs to $\algclass$ if and only if for some $n \in \set{1,2,\ldots}$ it satisfies all identities in $\eqclass_n$. 
    \end{enumerate}
\end{theorem}
\begin{proof}
    We begin with the implication \ref{eil-schutz:identities}$\Rightarrow$\ref{eil-schutz:variety}. If an algebra satisfies all identities in $\eqclass_n$, then the same is true for all of its subalgebras and images under surjective homomorphic images. For binary products the argument is the same: if an algebra satisfies all identities in $\eqclass_n$ and another algebra satisfies all identities in $\eqclass_m$, then their product satisfies all identities in $\eqclass_{\max(m,n)}$. This establishes that every class of algebras satisfying condition~\ref{eil-schutz:identities} has the closure properties required of an Eilenberg variety.

    Consider now the converse implication~\ref{eil-schutz:variety}$\Rightarrow$\ref{eil-schutz:identities}. Define $\algclass_n$ to the first $n$ algebras from $\algclass$, with respect to the countable enumeration from the assumption in the theorem. Define $\eqclass_n$  to be the identities that are satisfied by all algebras in $\algclass_n$. We claim that an algebra $B$ belongs to $\algclass$ if and only if for some $n \in \set{1,2,\ldots}$ it satisfies all identities in $\eqclass_n$. The left-to-right implication is immediate. For the right-to-left implication, we can apply Lemma~\ref{lem:hsp} to conclude that $B$ is the homomorphic image of a subalgebra of a product of algebras from $\algclass_n$. Furthermore, if we inspect the proof of Lemma~\ref{lem:hsp}, we will see that the product is finite, because the homomorphisms from $\monad B$ to algebras in $\algclass_n$, as used in~\eqref{eq:homo-as-identity}, can be chosen in finitely many ways. 
\end{proof}

 \exercisehead
 \mikexercise{Show that Theorems~\ref{thm:birkhoff} and~\ref{thm:eilenberg-schutz} are  also true for monads in categories  of sorted sets, even with infinitely many sorts.}{}

 \mikexercise{Give an example of a monad which violates the assumption on countably many finite algebras from Theorem~\ref{thm:eilenberg-schutz}. Hint: see Section~\ref{sec:infinite-forests}.}{}

 \part{Trees and graphs}

 \chapter{Forest algebra}
\label{sec:forest-algebra}

In this chapter, we present a monad that  models  trees\footnote{
    This section is based on
    \incite{bojanczykForestAlgebras2008}
}.  The trees are  finite, node labelled, unranked (no restriction on the number children for a given node), and without a sibling order. Other kinds of trees can be modelled by other monads. 

\section{The forest monad}
\label{sec:the-forest-monad}
In fact,  the monad will represent  slightly more general objects, namely forests (multisets of trees) and contexts (which are forests with a port that is meant to be replaced by a forest or context).
The algebras are going to be two-sorted, with the sort names being ``forest'' and ``context''. For the rest of this chapter, define a \emph{two-sorted set} to be a set together with a partition into elements of forest sort and elements of context sort.  We use a convention where forest-sorted elements are written in \red{red}, context-sorted elements are written in \blue{blue}, and black is used for elements  whose sort is not known or which come from a set without sorts.  

    A \emph{forest} over a two-sorted set $\Sigma$ consists of a set of \emph{nodes};  a partial \emph{parent} from  nodes to nodes, and   a \emph{labelling} from nodes to $\Sigma$. Here is a picture:
     \mypic{65}
The parent function must be acyclic, and the labelling function must respect the following constraint: leaves (nodes that are not parents of any other node) have labels of sort ``forest'', while non-leaves have labels of sort ``context''.

We assume that  forests  are  nonempty, i.e.~there is at least one node.  
Note that there is no order on siblings in our definition of forests. The definition 

We use the usual tree terminology, such as root (a node without a parent), ancestor (transitive reflexive closure of the parent relation), child (opposite of the parent relation), descendant (opposite of ancestor) and sibling (nodes with the same parent). We assume that all roots are siblings.

Apart from forests, the forest monad will also talk about \emph{contexts}, which are  forests with an extra dangling edge that is attached to a node with a context label, as in the following picture:
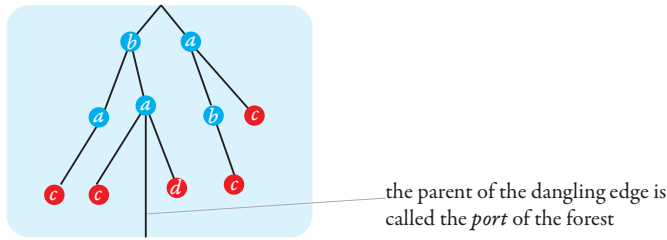
\mypic{66}

 \paragraph*{The forest monad.} We now define a monad structure on  forests and contexts.
 
 \begin{figure}
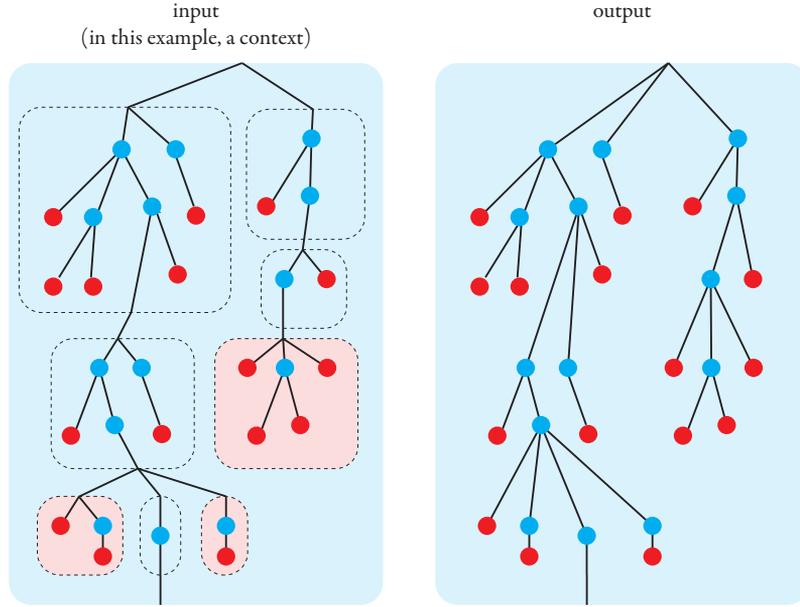

     \centering
     \mypic{69}
     \caption{Free multiplication in the forest monad}
     \label{fig:forest-multiplication}
 \end{figure}
\begin{definition}
    Define the \emph{forest} monad as follows. 
    \begin{itemize}
        \item  The underlying category  is two-sorted sets, where objects are two-sorted sets (with sorts ``forest'' and ``context'') and the morphisms are  sort-preserving functions between two-sorted sets. 
        \item For a two-sorted set $\Sigma$, the forest-sorted elements in $\fmonad \Sigma$ are forests over $\Sigma$, while the context-sorted elements  are    contexts over $\Sigma$. A sort-preserving function $f : \Sigma \to \Gamma$ is lifted to a sort-preserving function $\fmonad f : \fmonad \Sigma \to \fmonad \Gamma$ by  applying $f$ to the label of every node and leaving the rest of the structure unchanged.
        \item The unit operation maps a label $a \in \Sigma$ to the unique forest or context that has one node with  label $a$, as in the following pictures:
        \mypic{70}
        \item Free multiplication is the operation of type $\fmonad \fmonad \Sigma \to \fmonad \Sigma$ that is illustrated in Figure~\ref{fig:forest-multiplication}. More formally, the free multiplication of  $t \in \fmonad \fmonad \Sigma$ is defined as follows.  The nodes are pairs $(u,v)$ such that $u$ is a node of $t$ and $v$ is a node in the tree or context which is the label of $u$. The label is inherited from $v$, while the parent of a node $(u,v)$ is defined as follows (in the following $t_u \in \fmonad \Sigma$ is the label of node $u$ in $t$):
        \begin{align*}
            \begin{cases}
                (u,\text{$t_u$-parent of $v$}) & \text{if $v$ is not a root in $t_u$;}\\
                (\text{$t$-parent of $u$},\text{port of $t$-parent of $u$}) & \text{if $v$ is a root in $t_u$ and $u$ is not a root in $t$;}\\
                \text{undefined} & \text{otherwise}
            \end{cases}
            \end{align*}
        If $t$ is a context, then the port in the free multiplication is defined to  be the port of the context that labels the  port of $t$.
    \end{itemize}
\end{definition}
We leave it as an exercise for the reader to check that the monad axioms are satisfied by the above definition. We use the name \emph{forest algebras} for Eilenberg-Moore algebras over this monad.  

\section{Recognisable languages}
The rest of this chapter is devoted to a study of the languages recognised by forest algebras. We care mainly about languages recognised by  finite forest algebras, which are forest algebras that have finitely many elements on both sorts. We begin with some examples.

The notion of compositional function and Lemma~\ref{lem:monad-compositional} about compositional functions corresponding to homomorphisms is also true for monads in the category of sorted sets used by forest algebra. Therefore, we will mainly describe  homomorphisms using the terminology of compositional functions. 

\begin{myexample}  \label{ex:forest-exists-label}
    Let  $\Gamma \subseteq \Sigma$ be two-sorted alphabets. We can view 
    \begin{align*}
    \fmonad \Gamma \subseteq \fmonad \Sigma
    \end{align*}
    as a language, which only contains those  forests and context over alphabet $\Sigma$ where all labels are from $\Gamma$. Here is a homomorphisms into a finite algebra that recognises this homomorphisms  Consider the function  which inputs a forest or context in $\fmonad \Sigma$, and outputs the following information: (a) is it a forest or context; (b) are all labels from $\Gamma$?  This function is easily seen to compositional, and  therefore $h$ it can be viewed as a homomorphism of forest algebras.  The co-domain of the homomorphism $h$ is a forest algebra with two elements on the forest sort, and two elements on the context sort.
\end{myexample}

\begin{myexample}\label{ex:forest-counting-modulo}
    Let $\Sigma$ be a two-sorted alphabet, let $n \in \set{1,2,\ldots}$. Consider the function $h$ which inputs a forest or context in $\fmonad \Sigma$, and outputs the following information: (a) is it a forest or context; (b) what is the number of nodes modulo $n$. This function is compositional, and therefore $h$ can be viewed as a homomorphism of forest algebras. This homomorphism recognises the language of forests or contexts where the number of nodes is divisible by $n$. 
\end{myexample}

\begin{myexample}
    Consider an alphabet $\Sigma$ where all letters have context type. In this case, there are no forests over $\Sigma$, because there can be no leaves. For the same reason, every context over $\Sigma$ looks like this:
    \mypic{71}
    In other words, 
    $\fmonad \Sigma$ is empty on the forest sort, and is isomorphic to the free semigroup $\Sigma^+$ on the context sort. Since the monad structure of the free semigroup agrees with the monad structure of the forest monad, it follows that a forest algebra with an empty forest sort is the same thing as a semigroup.
\end{myexample}

\exercisehead

\mikexercise{Show that recognisable languages in the forest monad  are closed under images of (not necessarily letter-to-letter) homomorphisms 
\begin{align*}
h : \fmonad \Sigma \to \fmonad \Gamma.
\end{align*}
}
{
Consider an alphabet $\Sigma$ which has one forest-sorted letter $\red a$ and no context-sorted letters. There are no contexts over this alphabets and the forests look like this:
\begin{align*}
\red {a +  \cdots + a}
\end{align*}
Consider the homomorphism $h$ which maps $\red a$ to $\red {a+b}$, where $\red b$ is some new forest-sorted letter. The image of $\fmonad \Sigma$ under this homomorphism consists of forests where the number of $\red a$'s is equal to the number of $\red b$'s.

}

\mikexercise{Consider a variant of the forest monad, where we allow contexts where the port is a root, like in the following example:
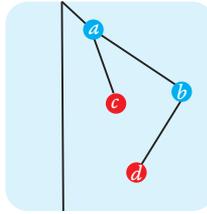
\mypic{99}
Show that in this variant, recognisable languages are not closed under images of homomorphisms, but are closed under images of letter-to-letter homomorphisms.}{}

\subsection{A finite representation}
\label{sec:basic-operations-forest-algebra}
As usual with the monad approach, one needs to explain how  algebras can be finitely represented. 
Even if the underlying sorted set is finite, the multiplication operation 
\begin{align*}
    \mu : \fmonad A \to A
\end{align*}
is in principle an infinite object. We show below a finite representation for the multiplication operation, in analogy to semigroups, where one only needs to define  multiplication  for inputs of length two. When discussing this finite representation, we use as much as possible the abstract language of monads; this will allow us to see analogies with other finite representations in this book. 

\paragraph*{A term basis.} Like for any monad, a \emph{term} in the forest monad is defined to be an element of $\fmonad X$ for some two-sorted set of variables $X$. Here is a picture of a term:
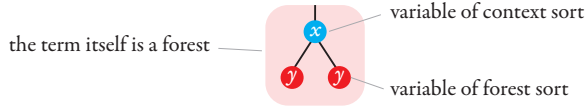
\mypic{68} A difference with respect to terms for monads in the category of sets is that in the forest monad -- which lives in the category of two-sorted sets -- the variables are sorted, which means that there are forest variables, and context variables. Also, the term itself has a sort (call this the output sort). 
 When interpreted in an algebra $A$, a term $t \in \monad X$ induces a term operation $t^A$ defined by
\begin{align*}
\eta \in A^X \quad \mapsto \quad  \text{multiplication in $A$ applied to $(\fmonad \eta)(t)$}.
\end{align*}
The input to the term operation is a sort-preserving valuation of the variables, while the output is an element of the algebra whose sort is the output  sort of the term. For example, the term
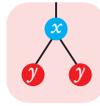
\mypic{110}
 induces a term operation which inputs a context sorted $\blue x$ and a forest-sorted $\red y$, and outputs a forest sorted element. Note that term operations are not morphisms in the category of two-sorted sets, if only because there is no clear way of assigning a sort to the input valuation.

We distinguish the following terms in forest algebra. 

\begin{definition}
    Define the \emph{basic forest algebra terms} to be the following terms:
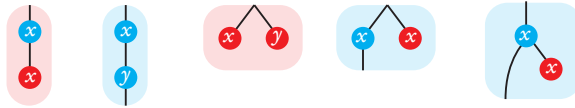
\mypic{72}
(These happen to be all  terms with exactly two nodes, modulo renaming variables.)
The \emph{basic operations} in  a forest algebra $A$ are defined to be  the  term operations that are induced in $A$ by these terms. We also use the following notation for the basic operations, listed in the order from the picture above (the colour of an operator is the colour of the output sort): 
\begin{align*}
\myunderbrace{\blue x \red \cdot \red x}
{inputs a \\
\scriptsize context $\blue x$\\
\scriptsize and forest $\red x$\\
\scriptsize and outputs\\
\scriptsize a forest}
\qquad
\myunderbrace{\blue x \blue \cdot \blue y}
{inputs a \\
\scriptsize context $\blue x$\\
\scriptsize and context $\blue y$\\
\scriptsize and outputs\\
\scriptsize a context}
\qquad 
\myunderbrace{\red x \red + \red y}
{inputs a \\
\scriptsize forest $\red x$\\
\scriptsize and forest $\red y$\\
\scriptsize and outputs\\
\scriptsize a forest}
\qquad 
\myunderbrace{\blue x \blue + \red x}
{inputs a \\
\scriptsize forest $\red x$\\
\scriptsize and context $\blue x$\\
\scriptsize and outputs\\
\scriptsize a context}
\qquad 
\myunderbrace{\blue x \blue \oplus \red x}
{inputs a \\
\scriptsize forest $\red x$\\
\scriptsize and context $\blue x$\\
\scriptsize and outputs\\
\scriptsize a context}
\end{align*}

\end{definition}

\begin{theorem}\label{thm:basic-operations-of-forest-algebra}
    The multiplication operation in a forest algebra is uniquely determined by the basic operations.
\end{theorem}
\begin{proof}
    Every forest or context can be constructed from the units by applying the basic  operations. 
\end{proof}

The forest algebra in the above theorem  does not  need to be finite. If it is  finite, then it can be finitely represented by giving the multiplication tables for the basic operations. Using this representation, we can talk about algorithms that process finite forest algebras.

\begin{figure}
    \centering
    \begin{align*}
        \begin{array}{rclrl}
    \red x \red + \red(\red y \red + \red z\red) &\red =& \red(\red x \red + \red y\red ) \red + \red z
    && \text{(F1) \scriptsize forests with $\red +$ are a semigroup}
    \\
    \red x \red + \red y &\red =& \red y \red + \red x
    && \text{(F2) \scriptsize  the forest semigroup is commutative}\\
    \blue x  \blue \cdot \blue(\blue y \blue z\blue) &\blue =& \blue(\blue x \blue y\blue) \blue \cdot \blue z 
    && \text{(F3) \scriptsize  contexts with $\blue \cdot$ are a semigroup}\\
    \blue x \red \cdot \red(\blue y \red \cdot \red z\red ) &\red =& \blue(\blue x \blue y\blue) \red \cdot \red z
    && \text{(F4) \scriptsize $\red \cdot$ is an  action of contexts on forests}\\
    \blue x \blue+ \red(\red y \red + \red z\red) &\blue =& \blue(\blue x \blue + \red y\blue) \blue + \red z
    && \text{(F5) \scriptsize $\blue +$ is an  action of forests on contexts}\\
    \blue x \blue \oplus \red(\red y \red + \red z\red) &\blue =& \blue(\blue x \blue \oplus \red y\blue) \blue \oplus \red z
    && \text{(F6) \scriptsize $\blue \oplus$ is an  action of forests on contexts}\\
    \red(\blue x \red \cdot \red y\red) \red + \red z &\red=& \blue(\blue x  \blue + \red z\blue) \red \cdot \red y  
    && \text{(F7) \scriptsize  compatibility of the  actions}\\
     \blue(\blue x \blue \oplus \red y\blue) \red \cdot \red z &\red=& \blue x \red \cdot \red{ (y\red +z)}
    && \text{(F8) \scriptsize  compatibility of the  actions}\\
    \blue(\blue x \blue \oplus \red y\blue) \blue \cdot \blue z &\blue=& \blue x \blue \cdot \blue{ (z\blue + \red y)}
    && \text{(F9) \scriptsize  compatibility of the  actions}
        \end{array}
    \end{align*}
    \caption{Axioms of forest algebra. The colour of the brackets indicates the sort of the bracket, and the colour of the equality sign indicates the sort of the compared elements. }
    \label{fig:forest-axioms}
\end{figure}

We can also give simple list of axioms forest algebra, see Figure~\ref{fig:forest-axioms}. These axioms are sound (they are satisfied by the basic operations in every forest algebra) and complete (if one gives five operations on a two sorted set $A$ that satisfy the axioms, then these operations can be extended to a  forest algebra multiplication $\mu : \fmonad A \to A$). Using this  axiomatisation, we can effectively check if a finite representation  of a forest algebra is correct, i.e.~it comes from some forest algebra.

\exercisehead
\mikexercise{Show that for every $t \in \fmonad \Sigma$ there is a decomposition 
\begin{align*}
t = f(t_1,\ldots,t_n) 
\end{align*}
such that  $f$ is a term of size at most 4 (and therefore the number of arguments $n$ is at most 4), and all arguments $t_1,\ldots,t_n$ have at most half the size (number of nodes) of $t$. 
}{
    Define a \emph{factor} in $t$ to be a set of nodes that is obtained by choosing some set of siblings, and taking all of their descendants. 
    \begin{claim}
        There must be a  factor $X$ which contains between one third and two thirds of all nodes.
    \end{claim}
    \begin{proof} Take a forest factor  $X$ which contains at least half  of all nodes, and which is inclusion-wise minimal for this property.  Let $x_1,\ldots,x_n$ be the minimal nodes in $X$, these are all siblings. Define $X_i$ to be $x_i$ and its descendants.  By minimality, each of the sets $X_1,\ldots,X_n$ has at most half of the nodes. If some $X_i$ contains at least one third of the nodes, then we are done. Otherwise,  choose the smallest $i \in \set{1,\ldots,n}$ such that
        \begin{align*}
            X_1 \cup \cdots \cup X_i
        \end{align*} 
        contains at half of all nodes; because each of the  sets $X_1,\ldots,X_i$ has at most one third of all nodes it follows that the above union has at most two thirds of all nodes.
    \end{proof}
    Let $X$ be the set from the above claim, and define $t_X \in \fmonad \Sigma$ to be the result of restricting $t$ to the nodes of $X$. If $X$ contains the port then $t_X$ is a context, otherwise $t_X$ is a forest. 
}

\mikexercise{Fix some language $L \subseteq \fmonad \Sigma$ that is recognised by a finite forest algebra. Suppose that we begin with some forest $t \in \fmonad \Sigma$ and then we receive a stream of updates and queries. Each update changes a label of some node (the set of nodes and the parent function are  not changed by updates). Each query asks if the current forest belongs to $L$. Show that one can compute in linear time a data structure (at the beginning, when the first forest $t$ is given), such that updates can be processed in logarithmic time and queries can be processed in constant time. }{}

\mikexercise{Prove completeness for the axioms (F1)--(F6).}{}

\subsection{Syntactic algebras}
\label{sec:syntactic-algebras-forest-algebra}
We now discuss syntactic algebras must necessarily exist in the forest monad. This is shown by a minor adaptation of the results from Section~\ref{sec:syntactic-algebras}.
As mentioned in Section~\ref{sec:syntactic-algebras}, syntactic homomorphisms also make sense in other categories, such as the category of two-sorted sets used by the forest monad.

In the forest monad, an algebra colouring is a sort-preserving function from the underlying two-sorted set in a forest algebra to some two-sorted set of colours. A subset $L \subseteq A$  can be seen as  special case of algebra colouring which uses four colours
\begin{align*}
\myunderbrace{\red{\set{\text{yes, no}}}}{forest sort} \cup 
\myunderbrace{\blue{\set{\text{yes, no}}}}{context sort}.
\end{align*} 
 For the category of two-sorted sets, surjective functions are those which are surjective on both sorts.

The results on existence of syntactic homomorphisms from Section~\ref{sec:syntactic-algebras} can be easily adapted to the forest monad -- more generally, to every monad in every category of sorted sets --   as explained in the following theorem and its proof.
\begin{theorem}\label{thm:sorted-syntactic} Let $\monad$ be a monad in a category of sorted sets (there could be more than two sorts, even infinitely many).
    \begin{enumerate}
        \item If $\monad$ is finitary, then every algebra colouring has a syntactic homomorphism;
        \item If the monad is not necessarily finitary, but there are finitely many sort names, then every algebra colouring recognised by a finite algebra (finite on every sort) has a syntactic homomorphism.
    \end{enumerate}
     
\end{theorem}
\begin{proof}
    For item (1) we use the same proof as in the left-to-right implication for Theorem~\ref{thm:plotkin}, while for item (2) we use the same proof as in Theorem~\ref{thm:syntactic-homomorphism-monad}. The only difference is that the variables in term operations have sorts. In item (2), the assumption on finitely many sort names is used\footnote{This assumption is indeed necessary, which can  be proved using ideas from \incite{DBLP:journals/lmcs/BojanczykK19}
    } in the extension of Lemma~\ref{lem:join-congruences} to say that there are finitely many equivalence classes of altogether of a congruence that has finitely many equivalence classes on each sort. 
    Apart from this difference, the rest of the proof is the same.
\end{proof}

In particular, since the forest monad is finitary, it follows that every language $L \subseteq \fmonad \Sigma$ in the forest monad has a syntactic algebra.  As discussed in the exercises, the syntactic algebra for a language recognised by a finite forest algebra can be computed. 

Also, the Eilenberg Variety Theorem holds for the forest monad. In the statement, the unary polynomial operations are the sorted version that is described in the proof of Theorem~\ref{thm:sorted-syntactic}, apart from this change the statement of the theorem and its proof are the same as in Section~\ref{sec:eilenberg}. More generally, the Eilenberg Variety Theorem works for every monad in every category of sorted sets, assuming that there are finitely many sorts. When generalising the proof  of the Eilenberg Variety Theorem to multi-sorted algebras, we use the assumption on finitely many sorts in step (4) of the proof, to show that there are finitely many possible unary polynomial operations in a finite algebra.

\exercisehead

\mikexercise{
Show that the syntactic algebra can be computed for a language $L \subseteq \fmonad \Sigma$ that is recognised by a finite forest algebra.  The input to the algorithm is a homomorphism
\begin{align*}
 h : \fmonad \Sigma \to A
\end{align*}
into a finite forest algebra, together with an accepting set $F \subseteq A$. The forest algebra is represented using its basic operations, as in Theorem~\ref{thm:basic-operations-of-forest-algebra}, and the homomorphism is represented by its values on the units. 
}{
    
       We assume that the homomorphism is surjective, since the image under $h$ can be computed as follows: take the images of the units, and close this subset under applying the basic operations. 
        \begin{claim}
            An equivalence relation in the underlying set of a forest algebra is a congruence if and only if it is compatible with the basic operations:
            \begin{align*}
            \red {a \sim a'} \quad \land \quad \red{b \sim b'} \qquad \text{implies} \qquad   \red{a + b \sim a' + b} 
            \end{align*}
        \end{claim}

}

\subsection{Infinite trees}
\label{sec:infinite-forests}
Define a monad $\fomega$ in the same way as the monad $\fmonad$, except that we allow the forests and contexts to be countably infinite. A node might have infinitely many children, and there might be infinite branches.  This monad is no longer finitary. 

We do not discuss this monad in more detail, apart from the following example, which shows that it is not clear what a ``finite algebra'' should be for this monad.

\begin{myexample}
    Consider  the two-sorted alphabet $\Sigma = \set{\blue a, \blue b}$. Even though there are no forest-sorted letters, it is still possible to construct an infinite forest over this alphabet, because there is no need for leaves.   We say that a language of $\omega$-words is \emph{prefix-independent} if it is stable under removing or adding a single letter as a prefix. Consider some  prefix-independent language $L$ of $\omega$-words over the alphabet $\Sigma$, not necessarily regular. For example 
    \begin{align*}
    L = \set{ \blue b^{n_1} \blue a \blue b^{n_2} \blue a \cdots : \text{the sequence $n_1,n_2,\ldots$ contains infinitely many primes}}.
    \end{align*}
    Define a \emph{branch} in a forest to be a set of nodes that is linearly ordered by the descendant relation, and which is maximal inclusion-wise for this property. An $L$-branch is a branch where the sequence of labels, starting from the root, belongs to $L$.

    Define $L' \subseteq \fomega \Sigma$ to be the set of infinite forests where  every node has at least two children and every node belongs to some $L$-branch. 
    We claim that $L'$ is recognised by a finite algebra, with at most 36 elements, regardless of the choice of $L$ (as long as it is prefix independent). Since there are uncountably many possible choices for $L$, it follows that there is no finite way of representing algebras in this monad that have at most 36 elements. In particular, ``finite on every sort'' is not a reasonable choice of ``finite algebra'' for this monad.

      Define a function $h$ from $\fomega \Sigma$ to a finite two-sorted set as follows. For forests, the function $h$ gives the  answers to the following questions:
\begin{itemize}
    \item[\red {1.}] is the forest in $L'$?
    \item[\red {2.}] are there are at least two roots?
\end{itemize}
For contexts, the function $h$ gives the  answers to the following questions:
    \begin{itemize}
        \item[\blue{1.}] is it possible to fill the port with some forest so that the result is in $L'$?
        \item[\blue{2.}] are there are at least two roots?
        \item[\blue{3.}] does the  port have a sibling?
        \item[\blue{4.}] is the context equal to the unit of $\blue a$? 
        \item[\blue{5.}] is  the context equal to the unit of  $\blue b$?
    \end{itemize}
    The red questions have at most 4 possible answers, and the blue questions have at most 32 possible answers, hence the number 36. In fact, this number can easily be reduced; for example in case of a ``no'' answer to question {\blue 1}, there is no need to store the answers for the remaining questions. 
We leave it as an exercise for the reader to check that the  function $h$ is compositional. It follows that the image of the function $h$, call it $A$, is a finite algebra for the monad $\fomega$. 
\end{myexample}

\exercisehead

\mikexercise{\label{ex:thin-monad} Consider the monad $\fomega$. We say that a forest or context in this monad is \emph{thin} if it has countably many branches. Define $\fthin \Sigma \subseteq \fomega \Sigma$ to be thin forests or contexts. Show that this is a monad.}{

}

\mikexercise{\label{ex:thin-ordinals} 
Show that a countable forest or context is thin, in the sense of Exercise~\ref{ex:thin-monad}, if and only if one can assign countable ordinal numbers to its children so that if a node is labelled by  ordinal number $\alpha$, then all of its children are labelled by ordinal numbers $\le \alpha$, and at most one child is labelled by $\alpha$. 
}{}{}

\mikexercise{Consider the monad $\fthin$ from Exercise~\ref{ex:thin-monad}. Show that a finite algebra over this monad is determined uniquely by its forest algebra operations (as in Theorem~\ref{thm:basic-operations-of-forest-algebra}) plus the following two term operations:
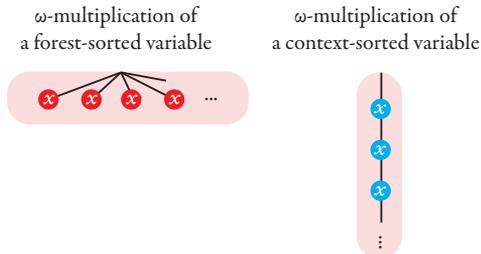
\mypic{79}

}{}

\section{Logics for forest algebra}
In Chapter~\ref{chap:logics}, we presented many examples of logics on finite words that could be characterised using structural properties of recognising monoids. In this section, we present some results of this type for forest algebra. Unfortunately, there are fewer interesting examples in the case of forest algebra, since the algebraic theory of forest languages is still not properly understood. A notable gap in our logic is first-order logic on trees, which is not known to have an algebraic characterisation, as will be discussed in Section~\ref{sec:first-order-logic-forests}.

\subsection{Monadic second-order logic} 
We begin with monadic second-order logic.  The idea is as usual: to each forest or context we associate a model, and then we  use monadic second-order  logic to describe properties of that model. There is one twist: because siblings in a forest or context are not ordered, we will need to extend \mso with modulo counting in order to make it expressively complete for all recognisable languages.

\begin{definition}
    Define the \emph{ordered model} of a forest or context as follows: the universe is the nodes, and it is equipped with the following relations:
    \begin{align*}
    \myunderbrace{x \le y}{ancestor} \qquad 
    \myunderbrace{a(x)}{$x$ has label $a \in \Sigma$} \qquad 
    \myunderbrace{port(x)}{$x$ is the port}
    \end{align*}
    The arguments to the relations are $x,y$, while the letter $a$ is a parameter. Each choice of $a \in \Sigma$ gives a different relation. 
\end{definition}

By using different logics on the ordered model, we get different classes of languages. 
We begin with monadic second-order logic.
A language $L \subseteq \fmonad \Sigma$ is called \emph{\mso definable} if it can be defined by a formula of monadic second-order logic using the ordered model.  The logic   \mso is not enough to define all recognisable languages, because it cannot count the number of nodes modulo two (or three, etc.). The problem is that there is no order on the siblings, so if we get a forest
\begin{align*}
\red {a + \cdots + a}
\end{align*}
that consists of $n$ nodes that are both roots and leaves, then we cannot use the usual trick of  selecting even-numbered nodes  to count parity. (A more formal argument will be given below.) For these reasons, we  extend \mso with modulo counting. In this extension -- called \emph{counting \mso} -- for every set variable $X$ and numbers $n \in \set{2,3,\ldots}$ and $\ell \in \set{0,\ldots,n-1}$  we can write a formula
\begin{align*}
|X| \equiv \ell \mod n
\end{align*}
which says that the size of the set $X$ is congruent to $\ell$ modulo $n$. This extension is expressively complete for the recognisable languages, as shown in the following theorem. 

\begin{theorem}\label{thm:forest-mso-with-counting}
    A language $L \subseteq \fmonad \Sigma$ is recognised by a finite forest algebra if and only if it is definable in counting \mso.  
\end{theorem}
\begin{proof} Both implications in the theorem are proved in a similar way as for finite words, so wo only give a proof sketch. 
    \paragraph*{From counting MSO to a finite forest algebra.} Same proof as for finite words: we remove the first-order variables (by coding them as singleton sets), and then we show by induction that for every formula of \mso (possibly with free variables), its corresponding language is recognised by a finite forest algebra.  In the induction steps we use products and powersets, both of which are finiteness preserving constructions for forest algebras. 

    \paragraph*{From a finite forest algebra to counting MSO.}
     Suppose that $L \subseteq \fmonad \Sigma$ is recognised by a homomorphism
    \begin{align*}
    h : \fmonad \Sigma \to A
    \end{align*}
    into a finite forest algebra.     
    The idea is the same as for finite words: the defining formula  inductively computes the value under $h$ for every subtree in the input forest or context. The induction corresponds to a bottom-up pass through the input\footnote{This idea works for objects such as finite words or forest algebra, because they have a canonical way of parsing (left-to-right for words, or bottom-up for forest algebra), which can be defined in \mso. For $\cc$-words, we do not know  any simple  parsing method, which is the reason why  the implication from finite algebras to logic in Section~\ref{sec:mso-countable-words} was hard. A similar phenomenon will appear for graphs, which will be discussed in the next chapter, which also do not have any simple definable canonical way of parsing. }. 
    Define the \emph{subtree} of a node in a tree or context as explained in the following picture:
    \mypic{83}
    In a forest, all subtrees are trees, while in a context some subtrees are trees and others are contexts.
     Define the \emph{type} of a node  to be the image under $h$ of its subtree.   The following claim shows that the type of a node can be inferred from the types of its children using counting.
    
    \begin{claim}\label{claim:forest-local-computation} For every $t \in \fmonad \Sigma$ and every node $x$ in $t$, the type of $x$  depends only on the answers to the following questions:
        \begin{itemize}
            \item what is the  label of node $x$?
            \item are there exactly $n$ children of $x$ type $a$?
            \item does $n$ divide the number of children of $x$ with type $a$? 
        \end{itemize}
        where $a$ ranges over elements of the forest algebra $A$ and $n \in \set{0,\ldots,|A|}$.
    \end{claim}
    \begin{proof}
        Let $a_1,\ldots,a_m$ be the types of the subtrees of the children of $x$. At most one of these types is a context, because there is at most one port. We only consider the case where  all of the types are forests (and hence they will be written in red below); the case when one type is a context  is treated similarly. If $\blue a$ is the label of node $x$, then the type of $x$ is equal to 
        \begin{align*}
        \blue a \red{\cdot (a_1 + \cdots + a_m)}.
        \end{align*}
        The label $\blue a$ is known, while the red part is multiplication in the forest semigroup of $A$, which is a commutative semigroup. In a commutative semigroup, the result of multiplication depends only on the number of times each argument is used. Furthermore, since the forest semigroup has size at most  $|A|$, then the number of times an argument is used  needs to be remembered only up to threshold $|A|$ and modulo some number that is at most $|A|$, see Exercise~\ref{ex:commutative-regular-languages}.
    \end{proof}
    Consider some enumeration $A = \set{a_1,\ldots,a_n}$ of the elements in the algebra. Some of these elments have forest sort and some have context sort. Using the above claim, we can write a formula 
    \begin{align*}
    \varphi(X_1,\ldots,X_n)
    \end{align*}
    of counting \mso which holds if and only if for every $i \in \set{1,\ldots,n}$, the set $X_i$ is exactly the set of nodes with type $a_i$. The formula simply checks that the types for each node are consistent with the types of its children, as described in the  claim. Finally, the image $h(t)$ of a forest or context can be computed in counting \mso by guessing the sets $X_1,\ldots,X_n$ that satisfy the formula $\varphi$ above, and then inferring $h(t)$ from the types of the root nodes (with the same argument as in the above claim). Since the image $h(t)$ determines membership  of $t$ in the language, it follows that the language itself is definable in counting \mso.
\end{proof}

The construction of an algebra in the above theorem is effective: given a sentence of \mso, we can construct a recognising homomorphism 
\begin{align*}
h : \fmonad \Sigma \to A
\end{align*}
into a finite forest algebra (and compute an accepting set $F \subseteq A$). The finite forest algebra is represented by its basic operations, as discussed in Section~\ref{sec:basic-operations-forest-algebra}, and the homomorphism is represented by its images for the units.

 The exact role of counting is explained in the following theorem.

\begin{theorem}\label{thm:forest-mso-without-counting}
    A language $L \subseteq \fmonad \Sigma$ is definable in \mso (without counting) if and only if it is recognised by a forest algebra where the forest semigroup (forests equipped with $\red +$) is aperiodic. 
\end{theorem}
A corollary of this theorem is that modulo counting is needed to define the language ``even number of nodes'', since this language cannot be defined by a forest algebra with an aperiodic forest semigroup.
\begin{proof} For the left-to-right implication, we use the same proof as in the left-to-right implication of Theorem~\ref{thm:forest-mso-with-counting}. The only difference is that in Claim~\ref{claim:forest-local-computation} we do not need modulo counting. This is because for every commutative aperiodic semigroup, the outcome of multiplication depends only on the number of times that each argument is used up to some finite threshold, without modulo counting. 

Consider now the right-to-left implication, which says that if a language is definable in \mso without counting, then it is recognised by a finite forest algebra with an aperiodic forest semigroup. Here, again, we use the same proof as in Theorem~\ref{thm:forest-mso-with-counting}, where a recognising forest algebra is constructed by starting with some atomic forest algebras, and then applying products and the powerset construction. Since we do not need the relation 
\begin{align*}
|X| \equiv \ell \mod n
\end{align*}
from the set model, all of the atomic forest algebras have forest semigroups that are aperiodic. Products clearly preserve aperiodicity of the forest semigroup, and the same is true powersets, as explained in the following lemma.

\begin{lemma}\label{lem:powerset-for-commutative-semigroups}
    If $S$ is a commutative\footnote{Commutativity is important in the proof of the  lemma. For example, in Exercise~\ref{ex:u2-powerset-generates} we showed that every finite semigroup, not necessarily aperiodic, can  be obtained by applying products and powersets to an aperiodic semigroup. } aperiodic semigroup, then the same is true for its powerset semigroup $\powerset S$.
\end{lemma}
\begin{proof}
    In this proof, we use multiplicative notation for the semigroup operation.
    By aperiodicity of $S$, there is some   $\momega \in \set{1,2,\ldots}$  such that every element  of $b \in S$ satisfies $b^\momega = b^\momega b$. To establish aperiodicity of the powerset semigroup, we will show that every element $A \subseteq S$ of the power set semigroup satisfies 
    \begin{align*}
    A^{n} = A^{n+1}
    \end{align*}
    where $n$ is the size of $S$ times $\momega +1$. We only show the inclusion $A^{n+1} \subseteq A^n$, the same proof can be used to establish the opposite inclusion.   Let 
    \begin{align*}
    a = a_1 \cdots a_{n+1} \in A^{n+1}.
    \end{align*}
    By the pigeon-hole principle and choice of $n$,  some  $b \in A$ must appear at least $\momega+1$ times in the sequence $a_1,\ldots,a_{n+1}$. By commutativity and aperiodicity of $S$, one extra occurrence of $b$ can be eliminated from the multiplication, proving  $A^{n+1} \subseteq A^{n}$.
\end{proof}
\end{proof}

\begin{corollary}
    A language is definable in \mso without counting if and only if its syntactic forest algebra is finite and has an aperiodic forest semigroup.
\end{corollary}
\begin{proof}
Aperiodicity of the forest semigroup is preserved when taking subalgebras and quotients (images under surjective homomorphisms). Since the  syntactic forest algebra can be obtained from any recognising forest algebra by taking a subalgebra and then a quotient, the result follows from Theorem~\ref{thm:forest-mso-without-counting}.
\end{proof}
Since the syntactic forest algebra can be computed for recognisable languages, it follows that given a sentence of counting \mso, we can decide if there is a  sentence of \mso which does not use counting and which is equivalent on forests and contexts.

\subsection{First-order logic}
\label{sec:first-order-logic-forests}
For finite words, the king of algebraic characterisations was  the \schutz-McNaughton-Papert-Kamp Theorem, which described the languages of finite words that can be defined in first-order logic (using the ordered model). Unfortunately, finding a generalisation of this theorem to forest algebra (or any other algebra modelling trees) remains an open problem\footnote{This is in contrast to first-order logic on trees with the child relation (and not the descendant ordering), which has an algebraic characterisation, see 
\incite[Theorem 1.]{benediktSegoufin2009}
}.  Our discussion of first-order logic in the forest monad is limited to some remarks and one example. 

As discussed in Section~\ref{sec:syntactic-algebras-forest-algebra}, 
The Eilenberg Variety Theorem works also for the forest monad. One can show that, in the forest monad, the class of languages definable in first-order logic is a language variety, see Exercise~\ref{ex:forest-fo-is-variety}. Therefore,  from the Eilenberg Variety Theorem  it follows that whether or not a language $L \subseteq \fmonad \Sigma$ is definable in first-order logic depends only on the syntactic algebra of the language. However, it is not known if the corresponding property of syntactic algebras is decidable.  Here is an example  which shows that aperiodicity -- which characterised the syntactic algebras for first-order definable languages of finite words -- is not enough for forest algebra. 

\begin{myexample} \label{ex:boolean-expressions} Consider an alphabet
    \begin{align*}
    \Sigma = \set{ 
        \myunderbrace{\blue \lor, \blue \land}{context sort}, \quad  
        \myunderbrace{\red 0, \red 1}{forest sort}
    }.
    \end{align*}
    A forest over this alphabet is the same thing as a multiset of positive Boolean formulas, as in the following picture:
    \mypic{73}
    We define the \emph{value} of a node in a forest over this alphabet to be the value of the  Boolean formula in the subtree of the node. 
    Consider the language 
    \begin{align*}
    L = \set{t \in \fmonad \Sigma : \text{$t$ is a forest where all roots have value 1}}
    \end{align*}
    If we look at the syntactic forest algebra of this language, then both the forest semigroup and the context semigroup are aperiodic (in fact, they are idempotent). Nonetheless, the language is not definable in first-order logic, see Exercise~\ref{ex:boolean-not-first-order}.
\end{myexample}

\exercisehead

\mikexercise{\label{ex:forest-fo-is-variety} Prove that first-order logic, as discussed in Section~\ref{sec:first-order-logic-forests}, is a variety in the sense of the Eilenberg variety theorem.}{}

\mikexercise{Show that the language of forests where some leaf has even depth is not definable in first-order logic.}{}

\mikexercise{\label{ex:boolean-not-first-order} Show that the language from Example~\ref{ex:boolean-expressions} is not definable in first-order logic. }{}{}

\mikexercise{\label{ex:potthoff} Consider a two-sorted alphabet 
\begin{align*}
\Sigma = \set{\blue{\text{left, right}}, \red{\text{left, right}}}.
\end{align*}
A \emph{binary tree} over this alphabet is a tree where every node is either a leaf, or it has exactly two children, with labels ``left'' and ``right'' in the appropriate sort. There are no constraints on the root label. 
Define $L \subseteq \fmonad \Sigma$ to be the set of binary trees where all leaves are at even depth. Show that this language is first-order definable. Hint: show first that there is a first-order language which separates $L$ from the set of binary trees where all leaves are at odd depth.
}{}{}

\mikexercise{Define \emph{anti-chain logic} to be the variant of \mso where set quantification is restricted to anti-chains, i.e.~sets of nodes that are pairwise incomparable with respect to the descendant relation. Show that anti-chain logic can define all recognisable languages that contain only binary trees, as defined in Example~\ref{ex:potthoff}.
}{}

\mikexercise{A unary node in a forest or context is defined to be a node with exactly one child. 
    Show that anti-chain logic with modulo counting can define every recognisable language where every element has no unary nodes.}{}

\subsection{Branch languages}
In this section, we discuss  languages which are defined only by looking at branches in a forest or context. We say that a forest or context is a  \emph{branch} if all nodes are linearly ordered by the ancestor relation. Here is a picture: 
\mypic{75}
A branch can be viewed as a word, consisting of the labels of the nodes in the branch, listed in root-to-leaf order. 
 For a forest or context $t$, define a \emph{branch of $t$} to be any branch that can be obtained from  selecting some $x$ which is either the port or a leaf, and restricting $t$ to the ancestors of $x$. The branch is a context if $x$ is the port, otherwise the branch is a forest. Here is a picture:
 \mypic{74}
The following theorem gives a characterisation  of languages that are determined by only looking at branches. 
\begin{theorem}\label{thm:path} 
    For every language    $L \subseteq \fmonad \Sigma$,  not necessarily recognisable, the following conditions are equivalent
    \begin{itemize}
        \item[(1)] membership $t \in L$ depends only on the set of branches in $t$;
        \item[(2)]  the syntactic forest algebra  of $L$ satisfies the identities 
        \begin{align*}
       \myunderbrace{\blue a \red{ \cdot (b + c)} \red = \red(\blue a \red{\cdot b}\red) \red + \red(\blue a \red{ \cdot c}\red)}{distributivity}
         \quad \text{and} \quad  
         \myunderbrace{\red{b + b  = b}}{
             idempotence of \\ \scriptsize the forest semigroup
         }
         \qquad 
         \text{for every }\myunderbrace{\blue{a}, \red{b}, \red{c}}{the colour indicates the \\ \scriptsize sort of the variable} \in A.
        \end{align*}
    \end{itemize}
    If $L$ is recognisable, then the above conditions are also equivalent to:
    \begin{itemize}
        \item[(3)] $L$ is a finite Boolean combination of languages of the form ``for some branch, the corresponding word is in $K\subseteq \Sigma^*$'', where $K$ is regular. Different languages $K$ can be used in different parts of the Boolean combination.
    \end{itemize}
\end{theorem}
\begin{proof}\ 
    \begin{itemize}
        \item (1)$\Rightarrow$(2) In the free algebra $\fmonad \Sigma$, applying the identities from condition (2)  does not affect the set of branches.  
\item (2) $\Rightarrow$ (1) For a forest or context  $t \in \fmonad \Sigma$, define its \emph{branch normal form} to be the forest or context that is the union of all branches in $t$, as described in the following picture:
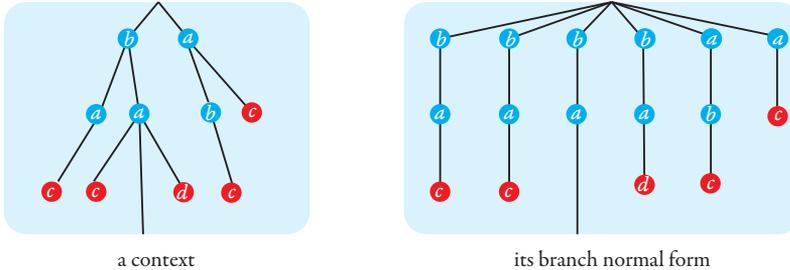
\mypic{76}
If the syntactic algebra satisfies the  distributivity identity  in the theorem, then a forest or context has the same image under the syntactic homomorphism as its branch normal form. Since the  branch normal form is determined uniquely by the multiset of branches, it follows that the image under the  syntactic homomorphism depends only on the multiset of branches\footnote{One could think that the  distributivity  identity alone (without the identity for idempotence) characterises exactly the languages where membership depends only on the multiset (and not just the set) of branches. This is not true, see Exercise~\ref{ex:multiset-of-branches}.}. Thanks to the idempotence identity, it is only the set of branches that matters for membership in the language, and therefore the language must be branch testable.
\item (3) $\Leftrightarrow$ (1) for recognisable languages.  Clearly (3) implies (1). Consider now the converse implication.  Let $h$ be the syntactic homomorphism of a recognisable language.
By condition (1) and the definition of a syntactic homomorphism,  membership $t \in L$ depends only on the set 
\begin{align*}
H(t) = \set{h(s) : \text{$s$ is a branch in $t$}}.
\end{align*}
For every $a$ in the syntactic algebra, define $K_a \subseteq \Sigma^+$ to be the words that correspond to branches which have value $a$ under the syntactic homomorphism. This language is recognised by a finite semigroup (which is easily constructed from the syntactic  forest algebra), and therefore it is regular. Finally, $a \in H(t)$ if and only if for some branch the corresponding word is in $K_a$. Therefore, $H(t)$ can be described using a finite Boolean combination of languages of the form $K_a$. 
    \end{itemize}
\end{proof}

Condition (2) in the above theorem can be effectively checked given the syntactic algebra. Since the syntactic algebra can be computed for recognisable languages, it follows that one can decide if a   recognisable language satisfies any of the conditions in the above theorem. 

\exercisehead
\mikexercise{\label{ex:multiset-of-branches} Give an example of a language $L \subseteq \fmonad \Sigma$ where membership depends only on the multiset of branches, but where the syntactic algebra violates the distributivity identity from Theorem~\ref{thm:path}.}{}

\mikexercise{Give an algorithm which decides if a recognisable  language $L \subseteq \fmonad \Sigma$ is of the form: ``for some branch, the corresponding word is in $K \subseteq \Sigma^*$'', for some regular $K$.  }{
\begin{theorem}\label{thm:universal-path}
    Let  $L \subseteq \fmonad \Sigma$ be a language, not necessarily recognisable, which contains only forests. The following conditions are equivalent:
    \begin{enumerate}
        \item \label{universal-path-1} There is some language $K \subseteq \Sigma^*$ such that 
        \begin{align*}
        L = \set{t \in \fmonad \Sigma : \text{$t$ is a forest and all of its path words belong to $K$}}.
        \end{align*}
        \item \label{universal-path-2} Every context $\blue a \in \fmonad \Sigma$ and forests $
    \red b, \red c \in \fmonad \Sigma$ satisfy the following equivalence:
     \begin{align*}
            \blue{a}\red{\cdot (b + c)} \in L \quad \text{iff} \quad \blue{a} \red{\cdot b} \in L \text{ and } \blue a \red{\cdot b} \in L.
            \end{align*}
    \end{enumerate}
    Furthermore, if $L$ is recognisable then one can check if the above conditions hold, and if they do, then $K$ can be chosen so that it is a regular word language. 
\end{theorem}
\begin{proof}
    The implication \ref{universal-path-1} $\Rightarrow$ \ref{universal-path-2} is immediate, since the same paths words appear on both sides of the equivalence.  Consider now the converse implication \ref{universal-path-1} $\Rightarrow$ \ref{universal-path-2}.
    Clearly every universal path language is closed under eliminating subtrees,  because eliminating a subtree makes the set of path words smaller.

    For the converse implication, suppose that $L$ contains only forests and is closed under eliminating subtrees. Define a \emph{linear tree} to be a forest with one leaf.  Consider the following word language 
    \begin{align*}
    K = \set{\text{unique path word of $t$} : \text{$t$ is a linear tree in $L$}} \subseteq \Sigma^*.
    \end{align*}
    Because $L$ is closed under eliminating subtrees, we have 
\begin{align*}
L = \forall K.
\end{align*}

    A root-to-leaf path in a forest $t$ can be seen as a linear tree which is obtained from $t$ by keeping only the nodes that appear on the path. 
\end{proof}
}

\subsection{Modal logic}

We finish  this section with a  discussion of tree variants for some of the temporal logics that  were discussed in Chapter~\ref{chap:logics}. When working with trees and forests, we  use the terminology of modal logic, described as follows. 

Define a \emph{Kripke model} to be a directed graph with vertices labelled by some alphabet $\Sigma$. Here is a picture of a Kripke model:
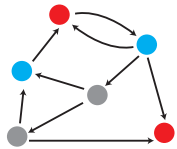
\mypic{77}
Vertices of the Kripke model are called \emph{worlds}, and the edge relation is called \emph{accessibility}. Accessibility does not need to be  transitive. In this section, we  only study Kripke models where accessibility is acyclic. 
To express properties of worlds in  Kripke models, we use modal logic,
whose formulas are constructed as follows:
\begin{align*}
\myunderbrace{a}
{the current \\ \scriptsize world has\\
\scriptsize  label $a \in \Sigma$}
\qquad 
\myunderbrace{
    \Diamond \varphi
}{some \\ \scriptsize accessible \\
\scriptsize world \\ 
\scriptsize satisfies $\varphi$}
\qquad 
\myunderbrace{
    \Box \varphi
}{every \\ 
\scriptsize accessible \\
\scriptsize world \\ 
\scriptsize satisfies $\varphi$}
\qquad 
\myunderbrace{\neg \varphi \quad \varphi \land \psi \quad \varphi \lor \psi}{Boolean combinations}.
\end{align*}
We use the following notation for the semantics of modal logic:
\begin{align*}
\myunderbrace{M}{Kripke \\ \scriptsize model}, \myunderbrace{v}{world\\ \scriptsize of $M$} \models \myunderbrace{\varphi}{formula \\ \scriptsize of modal \\ \scriptsize logic}.
\end{align*}
We  use modal logic to define properties of forests, by assigning a Kripke model to each forest, as explained in the following picture:
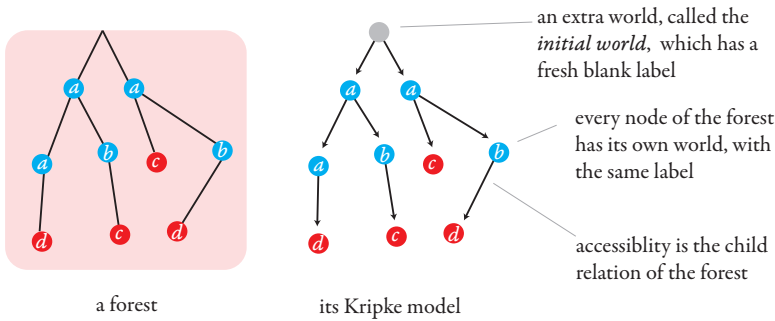
\mypic{78}
One could also assign a Kripke model to a context, by doing the same construction, except with a special marker for the port node. We choose not to do this, without any deeper reasons,  and therefore in what follows we only discuss languages that contain only forests.

\begin{definition}[Forest languages definable in modal logic]
    We say that a formula of modal logic is true in a forest if it is true in the initial world of its Kripke model. 
    A  forest language is called \emph{definable in modal logic} if there is a formula of modal logic that is true in exactly the forests from the language.
\end{definition}

The following theorem characterises  modal logic in terms of two identities. A corollary of the theorem is that one can decide if a language is definable in  modal logic, because it suffices to check if the identities hold in the syntactic algebra of a language. 

\begin{theorem}\label{thm:modal-logic-child}
    Let $L \subseteq \fmonad \Sigma$ be a language that contains only forests. Then  $L$  is definable in modal logic if and only if its syntactic forest algebra  is finite and satisfies the identities
    \begin{align*}
        \red{a + a = a}  \qquad 
        \blue{c^{\momega}}\red { \cdot a } \red =
    \blue{c^{\momega}}\red { \cdot b } \qquad 
    \text{for all $\red a, \red b, \blue c$}
    ,
    \end{align*}
    where $\blue \momega \in \set{1,2,\ldots}$ is the idempotent exponent of the context semigroup.
\end{theorem}
\begin{proof} The rough idea is that the identities say that the membership in the language is invariant under bisimulation (the first identity) and depends only on nodes at constant depth (the second identity). These are exactly the properties that characterise modal logic.  A more detailed proof is given below.

    Define the \emph{modal rank} of a formula to be the nesting depth of the modal operators $\Diamond$ and $\Box$. Here are some examples:
    \begin{align*}
      \myunderbrace{a}{modal rank 0} \qquad 
      \myunderbrace{(\Diamond a) \land (\Diamond b)}{modal rank 1} \qquad 
      \myunderbrace{(\Diamond (a \land \Box b) \land (\Diamond b)}{modal rank 2}.
    \end{align*}
    When the alphabet is finite and fixed, then there are finitely many formulas of given modal rank, up to logical equivalence. This is because, up to logical equivalence, there are finitely many Boolean combinations of formulas from a given set.
    To prove the theorem, we use a slightly more refined result, in the following claim, which characterises  the expressive power of modal logic of given modal rank. 
    \begin{claim}\label{claim:modal-logic-fixed-depth}
        A forest language can be defined by a formula of modal rank $n \in \set{0,1,\ldots}$ if and only if its syntactic algebra satisfies the identities
    \begin{align*}
            \red{a + a = a } 
\qquad
            \blue{c^{n}}\red { \cdot a } \red =
        \blue{c^{n}}\red { \cdot b } 
        \end{align*}
    \end{claim}
    \begin{proof}
        We say that a Kripke model is tree-shaped  if the accessibility relation gives a finite tree, with edges directed away from the root (this is the case for the Kripke models that we assign to forests).  We say that two tree-shaped Kripke models are \emph{bisimilar} if one can be transformed into the other by applying the identity $\red {a+a=a}$, i.e.~duplicating or de-duplicating identical sibling subtrees\footnote{For tree-shaped Kripke models this notion coincides with the usual notion of bisimulation for general Kripke models.}. For $n \in \set{0,1,\ldots}$, we say that two  tree-shaped  Kripke models are $n$-bisimilar if,  after removing all worlds that are separated by more than $n$ edges from the root,  they  are bisimilar.  By induction on $n$ one shows that every equivalence class of $n$-bisimilarity can be defined by a formula of modal logic with modal rank $n$; and conversely formulas of modal rank $n$ are invariant under $n$-bisimilarity. The identities in the statement of the claim say that the forest language is invariant under $n$-bisimilarity, and hence the claim follows.  
    \end{proof}
    The theorem follows immediately from the above claim. Indeed, if the identities in the theorem are satisfied, then the language can be defined by a formula of modal logic with modal rank $\blue \momega$. Conversely, if the language is defined by a formula of nesting depth $n$, then membership in the language is not affected by nodes which are more than $n$ edges away from the root, and therefore the syntactic algebra must satisfy
    \begin{eqnarray*}
        \blue{c^{\momega}} \red a  \equalbecause{because $\blue{c^\momega}$ is idempotent}\\
         \blue{c^{\momega n}} \red a \equalbecause{Claim~\ref{claim:modal-logic-fixed-depth}}\\
          \blue{c^{\momega n}} \red b  
          \equalbecause{because $\blue{c^\momega}$ is idempotent}\\ \blue{c^{\momega}} \red b.    
    \end{eqnarray*}
\end{proof}

\paragraph*{Transitive modal logic.}  A formula of  modal logic as discussed above can only talk about nodes that are at some constant distance from the root. We now discuss a variant of modal logic which can talk about arbitrarily deep nodes. The formulas stay the same, only the interpretation of forests as Kripke structures changes.

For a forest, define its  \emph{transitive Kripke model}  in the same way as the Kripke model, except that the accessibility relation now describes the transitive closure of the child relation. In other words, accessibility now represents the proper descendant relation. 

\begin{definition}[Forest languages definable in  transitive modal logic]
    A   language that contains only forests  is called \emph{definable in transitive modal logic}\footnote{In the terminology of  temporal logic, this logic is also called $\mathsf{EF}$, which refers to the ``exists finally'' operator of (branching time) temporal logic. Theorem~\ref{thm:modal-logic-transitive} is 
    \incite[Theorem 5.3]{bojanczykForestAlgebras2008}.
    } if there is a formula of modal logic that is true in (the initial world) of exactly the forests from the language.
\end{definition}

The following theorem characterises transitive modal logic in terms of two identities. A corollary of the theorem is that one can decide if a language is definable in transitive modal logic.
\begin{theorem}\label{thm:modal-logic-transitive}
    Let $L \subseteq \fmonad \Sigma$ be a language that contains only forests. Then  $L$  is definable in transitive modal logic if and only if its syntactic forest algebra is finite and  satisfies the identities
        \begin{align*}
            \red{a + a = a } 
        \qquad 
            \blue{c}\red { \cdot a } \red =
        \red (\blue{c}\red { \cdot a) } \red{+ a} \qquad\text{for all }\red a, \blue c.
        \end{align*}
\end{theorem}
\begin{proof}
    It is easy to see that the identities  must be true in the syntactic algebra of every  language  definable in transitive modal logic. The first identity says that the language must be invariant under bisimulation, which is clearly true for transitive modal logic. For the second identity, we observe that going from $\blue c \red{\cdot a}$ to $\red{(\blue c \cdot a)+a }$ does not affect the transitive Kripke model, up to bisimulation.

    The rest of this proof is devoted to the right-to-left implication. Let      \begin{align*}
    h : \fmonad \Sigma \to A
    \end{align*}
    be a homomorphism into an algebra $A$ that satisfies the identities.  By induction on the size of $A$, we will show that for every forest-sorted $\red a \in A$, the inverse image $h^{-1}(\red a)$ is definable in transitive temporal logic (we say that such $\red a$ is definable in the rest of the proof).  This immediately yields the right-to-left implication. 
    In the rest of the proof, we define the \emph{type} of an element of $\fmonad \Sigma$ to be its  image under $h$. 

    In the proof, we  use  a reachability ordering on the algebra $A$  defined as follows. We say that $a \in A$ is \emph{reachable} from $b \in A$, denoted by $a \ge b$, if there is some $t \in \fmonad A$ which uses $b$ at least once, and which gives $a$ under the multiplication operation of $A$.  (Reachability can be seen as the  forest algebra variant of the infix relation for semigroups.)  Reachability is easily seen to be a pre-order, i.e.~it is transitive and reflexive.   We draw the reachability ordering in red when comparing forest-sorted elements. 
     Thanks to the identities in the theorem, reachability is  anti-symmetric when restricted to the forest sort, as explained in the following claim. 
    \begin{claim}\label{claim:forest-ef-antisymmetric}
        If  forest-sorted $\red a, \red b \in A$ are reachable from each other, then $\red {a = b}$.
    \end{claim}
    \begin{proof}
        For forest-sorted $\red a, \red b \in A$, reachability $\red{ a \ge  b}$ is equivalent to 
    \begin{align*}
    \myunderbrace{\red {a = \blue c\cdot \red b}}{for some context-sorted $\blue c$} \qquad \text{or} \qquad 
    \myunderbrace{\red{ a =  c +b}}{for some forest-sorted $\red c$} \qquad \text{or} \qquad 
    \red{ a = b}.
\end{align*} In the presence of the identities from the assumption of the theorem, all three conditions above imply  $\red {a =a+b}$. For the same reason, $\red{a \le b}$ implies  $\red{b=a+b}$, and therefore $\red{a=b}$.
    \end{proof}

In every finite forest algebra there is a maximal forest-sorted element with respect to reachability, because every two forests can be combined using $\red +$, into a forest that is bigger than both of them in the reachability ordering. 
By Claim~\ref{claim:forest-ef-antisymmetric}, the maximal element is unique. Fix the maximal element $\red a$ for the rest of the proof.  

The following claim uses the induction assumption on algebra size to give a sufficient condition for definability. 

    \begin{claim}\label{claim:sub-maximal-definable}
        Let $\red b\in A$ be non-maximal and forest-sorted. Every  forest-sorted $ \red c \in A$ that is not reachable from $\red b$ is definable. 
    \end{claim}
    \begin{proof} Let $I$ be the set of elements in $A$ that are reachable from $\red b$.  This is an ideal, which means that if  $t \in \fmonad A$ contains at least one letter from $I$, then its multiplication is in $I$. Furthermore, this ideal contains at least two forest-sorted elements, by assumption that $\red b$ is non-maximal. Define $\sim$ to be the equivalence relation on $A$ which identifies two elements if they are equal, or both have the same sort and belong to the ideal $I$. Because $I$ is an ideal, $\sim$ is a  congruence. Because $\red b$ is non-maximal, the congruence is non-trivial, every forest-sorted equivalence class is definable thanks to the  induction assumption on algebra size. Because $\red c$ is not reachable from $\red b$, it does not not belong to the ideal $I$, and thus its equivalence class  consists of $\red c$ only, and hence $\red c$ is definable. \end{proof}

    We  use the above claim to show that, with at most two exceptions, all forest-sorted elements of $A$ are definable. 
    Call an element $\red b$ \emph{sub-maximal} if it is not maximal and  $\red c \red > \red b$ implies that $c$ is maximal.   If $\red c$ is neither maximal nor sub-maximal then  it is definable by the above claim, because  there is some sub-maximal $\red b$ from which $\red c$ is not reachable. For the same reason, if there are at least two sub-maximal elements, then all sub-maximal elements are definable.  In particular, if there are at least two sub-maximal elements, then all elements are definable: the non-maximal elements are all definable, and the maximal element is definable as the complement of the  remaining elements. 

    We are left with the case when there is exactly one sub-maximal element, call it $\red b$. We will show that the maximal element $\red a$ is definable, and therefore $\red b$ is definable (as the complement of the remaining elements).  Define the \emph{descendant forest} of a node  $x$ to be the forest that is obtained by keeping only the proper descendants of $x$.  Since we do not allow empty forests, the descendant forest is defined only when $x$ is not a leaf.   
    \begin{claim}\label{claim:ef-induction}
        A forest has type $\red a$ if and only if it contains a node $x$, with label $\sigma \in \Sigma$, such that  one of the following conditions hold:
        \begin{enumerate}
            \item $x$ is not a leaf, its  descendant  forest  has type $\red {d < b}$, and  $h(\sigma) \cdot \red d  \red = \red a$;  or
            \item $x$ is not a leaf, its descendant forest has type  $\red {d \ge b}$, and $h(\sigma) \cdot \red b  \red = \red a$; or
            \item the node $x$ is a leaf and $h(\sigma) = \red a$.
        \end{enumerate}
    \end{claim}
    \begin{proof}
        To prove the bottom-up implication, we observe  that each of the conditions (1, 2, 3) implies that the subtree of $x$ has type $\red a$; by  maximality of $\red a$ the entire forest must then also have type  $\red a$.  For conditions (1, 3) the observation is immediate. For condition  (2), there are two cases to consider: either  the descendant forest has type $\red a$ and the subtree of $x$ has type $\red a$ by maximality, or the descendant forest  has type $\red b$ and   the subtree of $x$ has type  $\red a$ by $h(\sigma) \cdot \red b  \red = \red a$. 
        
        We now prove the top-down implication. 
        We show that if every node in a forest violates all of the conditions   (1, 2, 3), then the forest has type  $\red{ \le b}$. This is proved by induction on the number of nodes. If the forest has only one node, then we use condition (3). The second case is when  the forest has at least two trees, i.e.~it can be decomposed as 
           $ \red {t = t_1 + t_2}.$
        By induction assumption, both $\red{t_1}$ and $\red{t_2}$ have types $\red{ b_1,b_2 \le \red b}$. By the identity  in the statement of the  theorem, we get 
        \begin{align*}
        \red {b =  b +  b_1 + b_2},
        \end{align*}
        which implies that $\red{b_1 + b_2 \le b}$.  The final case is when $\red t$ is a tree, whose root $x$ has label $\blue \sigma$ and descendant forest $\red s$.  By induction assumption, $\red s$ has a type $\red d \red \le \red b$.  If $\red {d < b}$ then we use condition (1) to infer that $\red t$ has type $\red \le \red b$, otherwise we use condition (2).
    \end{proof}
    To finish the proof of the theorem, it remains to show that the conditions in  the above claim can be expressed using transitive modal logic.  Condition (3) can easily be checked. In condition (1), the element $\red d$ is definable because it is neither maximal nor sub-maximal. Therefore, 
     there is a formula of modal logic which is true in the world corresponding to a node  $x$ (in the descendant Kripke model) if and only  if the descendant forest of $x$ has type $\red d$. Therefore, there is a formula of modal logic which is true in the world corresponding to $x$ if and only if it satisfies (1). For similar reasons, we can define condition (2), since the union of the languages for $\red a$ and $\red b$ is definable, by taking the complement of the remaining definable languages. 
\end{proof}

\exercisehead

\mikexercise{Consider transitive modal logic for the monad $\fomega$ of infinite forests and contexts, as discussed in Section~\ref{sec:forest-algebra}. Show that for infinite forests, the equations from Theorem~\ref{thm:modal-logic-transitive} are sound (i.e.~if a language is definable in transitive modal logic, then the syntactic algebra satisfies the equations) but not complete (i.e.~the converse implication to soundness fails). }{}

  \chapter{Hypergraphs of bounded treewidth}

In this chapter, we study  algebras for graphs. Although in principle the algebras can describe arbitrary graphs, the more interesting results will assume bounded treewidth.

\section{Graphs, logic, and treewidth}
We begin by discussing graphs, but later we will move to a slightly more general notion, called hypergraphs,  which will provide the necessary structure to define a monad. In this chapter, the graphs and hypergraphs are assumed to be finite. 

\begin{definition}[Graph]\label{def:graphs}
    A \emph{graph} consists of a set of a finite set of vertices, together with a binary symmetric edge relation. 
\end{definition}
Here is a picture of a graph, with dots representing vertices and lines representing edges:
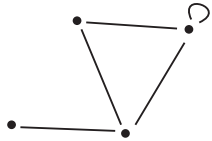
\mypic{49}

We  use logic, mainly \mso, to define properties of graphs, with graphs represented as models according to the following definition.
\begin{definition}[Graph languages definable in \mso]\label{def:incidence-model}
    Define the \emph{incidence model} of a graph as follows. The universe is the disjoint union of the vertices and the edges, and there is a  binary \emph{incidence} relation, which is interpreted as 
    \begin{align*}
    \set{(v,e) : \text{vertex $v$ is incident with edge $e$}}.
    \end{align*}
\end{definition}
The two kinds of elements in the universe of the incidence model -- vertices and edges -- can be distinguished using first-order logic: an edge is an element of the universe that is incident to some vertex, the remaining elements of the universe are vertices.

Monadic second-order logic over the incidence model defined above, which is the main logic of interest in this chapter, is sometimes called  \msotwo. A related logic is monadic second-order logic over a representation of graphs where the universe consists only of the vertices, and there is a binary relation for the edges. The related logic is sometimes called \msoone. The difference is that \msotwo can quantify over sets of vertices and edges, while \msoone can only quantify over sets of vertices. (For first-order logic, the two ways of representing graphs as models does not affect the expressive power, since on first-order quantification over edges can be replaced by two first-order quantifications over vertices.)   
The difference between \msoone and \msotwo is explained in the following example.
\begin{myexample}\label{example:grid}
    A \emph{clique} is a   graph where every two vertices are connected by an edge.
A  \emph{rectangular grid} is a graph that looks like this:
    \mypic{63}
    Both cliques and rectangular grids can be defined both in  \msoone and in \msotwo. Consider now the set of graphs which are {cliques} of prime size. A clique has prime size if and only if it satisfies the following property: (*) one cannot remove edges so as to get a rectangular grid which has at least two rows and at least two columns. Property (*) can  be directly expressed in \msotwo, but it cannot be expressed in \msoone, see Exercise~\ref{ex:no-cliques-of-prime-size}.
\end{myexample}

In this chapter, we are mainly interested in monadic second-order logic. First-order logic can only define properties that are local\footnote{The notion of locality is made precise by the Gaifman Theorem, see
\incite[Theorem 2.5.1]{ebbinghausFlumFinite}
}, e.g.~the existence of  a cycle of length three:
\begin{align*}
\exists u\ \exists v \ \exists v  \qquad \ E(u,v)   \land E(v,w) \land E(w,u).
\end{align*}
A classical example of a property that is non-local, and therefore cannot be expressed in first-order logic, is graph connectivity.  Using an \ef argument, one can show that a sentence of first-order logic cannot distinguish between a large cycle and a disjoint union of two large cycles:
\mypic{50}
On the other hand, connectivity  can be expressed in monadic second-order logic, already in the \msoone model,  as witnessed by the following sentence
\begin{align*}
\myunderbrace{\exists X}{there is a set\\ \scriptsize of vertices,}
\qquad 
    \myunderbrace{(\exists v\ v \in X)\ \land (\exists v \  v \not \in X) }{which is neither empty nor full,} \ \land
    \myunderbrace{
        (\forall v \ \forall w  \ E(v,w) \land v \in X  \Rightarrow w \in X).}{and which is closed under taking edges}            
\end{align*}

Already first-order logic is undecidable on graphs, in the following sense: it is undecidable whether or not a sentence of first-order logic is true in some graph. This undecidability is explained in the following example.
\begin{myexample}\label{ex:grids-first-order-logic}
    Consider directed graphs with coloured vertices and edges. These extra features can be easily encoded, using  first-order logic, in the undirected and unlabelled graphs that are discussed in this section, see the exercises.  A computation of a Turing machine can be visualised as a coloured rectangular grid, where each vertex represents a tape cell in a given moment of the computation, as in the  following picture:
    \mypic{51}
    By formalising the definition of a computation of a Turing machine,  one can write a sentence of first-order logic, which is true in a connected graph if and only if it represents an accepting computation of a given Turing machine. From this, one can deduce that the  halting problem reduces to satisfiability for first-order logic over finite graphs, see Exercise~\ref{ex:halting-grid}.  
\end{myexample}

We will no longer discuss first-order logic for graphs. Also, from now on, when talking about \mso, we mean the \msotwo variant that uses the incidence model from Definition~\ref{def:incidence-model}\footnote{This difference is not so important in the context of this chapter. This because we will be mainly interested in graphs of bounded treewidth, and for bounded treewidth the logics \msoone and \msotwo models are equivalent, see Exercise~\ref{ex:msoone-msotwo-bounded-tw}}.

\exercisehead

\mikexercise{\label{ex:no-edges} Show that for graphs without edges, \msoone and \msotwo has the same expressive power as first-order logic.}{}

\mikexercise{\label{ex:no-cliques-of-prime-size} Show that the set of cliques of prime size from Example~\ref{example:grid} cannot be defined in \msoone. }{}

\mikexercise{\label{ex:grids-exercise} For the purposes of this exercise, we consider directed graphs with two types of edges, blue and red. For such a graph, the associated model has the vertices as the universe, and two binary predicates for the red and blue edges. 
Show that rectangular  grids, as described in Example~\ref{ex:grids-first-order-logic}, can be defined in first-order logic. We assume that the input graph has one connected component.
}{}

\mikexercise{\label{ex:halting-grid}Show that the following problem is undecidable: given a sentence of first-order logic, decide if it is true in at least one finite graph.}{Take the sentence from Exercise~\ref{example:grid}, and add a further requirement that there is at least one corner. }

\mikexercise{ Unlike for the rest of this chapter, this exercise and the next one consider possibly infinite graphs.
    Consider two decision problems: (a) is a sentence of first-order logic true in at least one  finite graph; (b) is a sentence of first-order logic true in at least one  possibly infinite graph. Show that (a) is recursively enumerable (there is a Turing machine that accepts yes-instances in finite time, and does not halt on no-instances), while (b) is a co-recursively enumerable (there is a Turing machine that does not halt on yes-instances, and rejects no-instances in finite time).}{}

    \mikexercise{Show that for every $k \in \set{1,2,\ldots}$ the following property of graphs is definable in \mso using the incidence model: ``the graph is connected, infinite, and has degree at most $k$''. Show that ``the graph is connected and infinite'' is not definable. }{}

\mikexercise{Show that the existence of a Hamiltonian cycle cannot be expressed in \mso,  using the \msoone representation of graphs as models.}{\begin{myexample}[Hamiltonian cycles]\label{ex:hamiltonian}
    Using \mso over the incidence model, we can say that a graph has a Hamiltonian cycle, i.e.~a  cycle that visits every vertex exactly once. A Hamiltonian cycle can be seen as a set of edges such that: (a) every vertex in the graph is incident to exactly two edges  one incoming and one outgoing edge from the set; (b) the subgraph induced by the chosen edges is  connected. 
    The incidence model is important here. In the \msoone model, where quantification over sets of edges is not allowed, one cannot express the existence of a Hamiltonian path\footcitelong[Proposition 5.13]{courcelleGraphStructureMonadic2012}.
\end{myexample}
}

\mikexercise{Show that the existence of an Euler cycle (every edge is visited exactly once) cannot be expressed in \mso, using the \msoone representation of graphs as  models.}{}

\mikexercise{Consider graphs which allow parallel edges (i.e.~multiple edges connecting the same two vertices). The incidence model makes sense for such graphs as well. For $\ell \in \set{1,2,\ldots}$ define the $\ell$-reduction of a graph to be the result the following operation: for each pair of vertices $v,w$ we only keep the first $\ell$ edges that go from $v$ to $w$. Show that for every \mso sentence $\varphi$ there is some $\ell$ such that $\varphi$ is true in a graph (with parallel edges) if and only if it is true in its $\ell$-reduction.
}{}

\subsection{Treewidth}
The undecidability problems described in Example~\ref{ex:grids-first-order-logic} are avoided if we consider graphs that are similar to trees. The notion of similarity that we care about is treewidth, as defined below\footnote{For an introduction to treewidth, including a brief history, see 
\incite[Section 12.]{diestel}
}.
\begin{definition}[Tree decompositions]\label{def:tree-decomposition}
    A \emph{tree decomposition}  consists of:
    \begin{itemize}
        \item a graph, called the \emph{underlying graph};
        \item a set of \emph{nodes}, equipped with a tree ordering (i.e.~there is a least node called the root, and for every node $x$, the set of nodes $<x$ is totally ordered);
        \item for each node, an associated nonempty set of vertices called its \emph{bag}.
    \end{itemize}
    These should satisfy the following constraints:
        \begin{enumerate}
             \item  \label{tree-decomp:cover}  every edge in the underlying graph is covered by some bag, i.e.~there is some bag that contains  both endpoints of the edge;
            \item   \label{tree-decomp:introduce} every vertex $v$ of the underlying graph is \emph{introduced} in exactly one node, which means there is exactly one node $x$ such that $v$ is in the bag of $x$ and either $x$ is  the root or $v$ is not in the bag of the parent of $x$.\end{enumerate}
\end{definition}
Here is a picture of a tree decomposition:
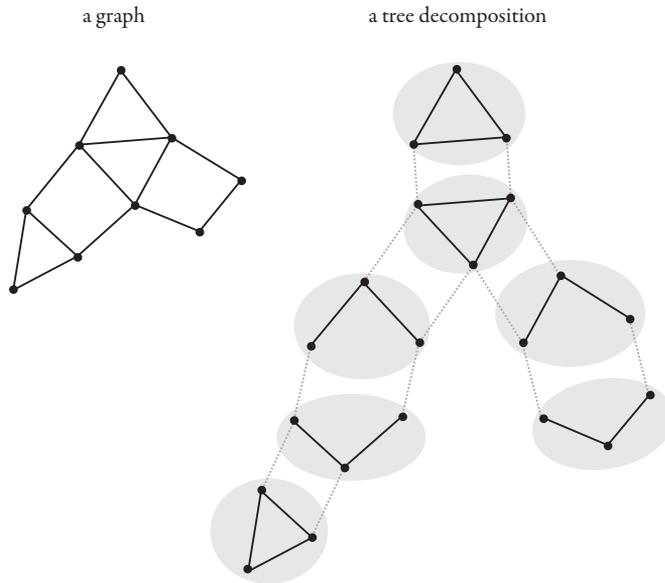
\mypic{112}
In the picture, the gray circles are bags, and the dotted lines connect appearances of the same vertex in several bags.
The \emph{width} of a decomposition is defined to be the maximal size of bags, minus one. For example, the tree decomposition in the above picture has width two, because its maximal bag size is three. The \emph{treewidth} of a graph is the minimal width of a tree decomposition for the graph. 

The reason for the minus one in the definition of width is so that trees, where the bags in the natural have tree decomposition have size two, get assigned treewidth one. This is illustrated in the following picture:
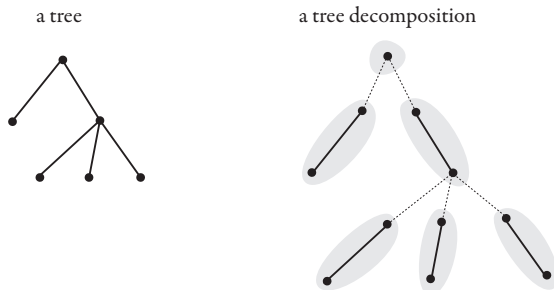
\mypic{55}
 Forests (i.e.~disjoint unions of trees) are the only graphs of  treewidth one. 

Cycles have treewidth two, as illustrated in the following example:
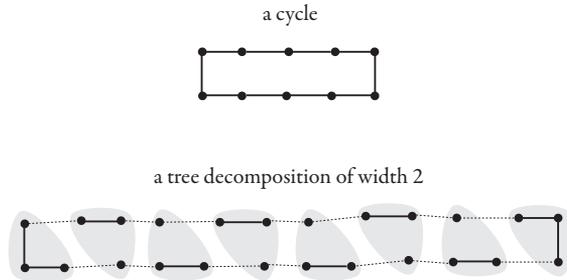
\mypic{56}
The tree decomposition in the above picture is a \emph{path decomposition}, i.e.~every node in the tree decomposition has at most one child. Path decompositions will play an important role in Section~\ref{sec:definabl-tree-decompositions}.

If a graph has $k+1$ vertices, then it has treewidth at most $k$, since one can always use a trivial tree decomposition where all vertices of the graph are in the same bag. For cliques, the trivial tree decomposition is optimal, as explained in the following example.

    \begin{myexample}
         We show that for cliques, every tree  decomposition must have a bag which contains all vertices. Consider a tree decomposition of a clique. If two vertices in a tree decomposition are connected by an edge, then the nodes which introduce these two vertices must be related by the ancestor relation (if they would be unrelated, then there could be no bag that contains both vertices). Therefore, in a tree  decomposition of a clique,  the nodes that introduce the clique vertices  must be linearly ordered by the ancestor relation. The maximal, i.e.~furthest from the root, node in this linear order must have all vertices of the clique in its bag.
    \end{myexample}

    Another example of graphs with unbounded treewidth is rectangular grids, see the exercises. In fact, the Grid Theorem\footnote{
        For a recent paper about the Grid Theorem, see
        \incite{ChuzhoyT19}
    }, which is stated but not proved in the exercises, says that a class of graphs has unbounded treewidth if and only if it contains all rectangular grids as minors. 
    We will show later in this chapter that for  every $k \in \set{1,2,\ldots}$, the class of graphs of treewidth at most $k$ has a decidable \mso theory. In the exercises we also discuss a corollary of the Grid Theorem, which says that decidability of \mso for bounded treewidth is  optimal:  if the \mso theory of a class of graphs is decidable, then the class has bounded treewidth.

\exercisehead

\mikexercise{
    We say that a graph $G$ is a \emph{minor} of a graph $H$ if one can find a family of disjoint vertex sets 
    \begin{align*}
    \set{X_v \subseteq \text{vertices of $H$}}_{v \in \text{vertices of $G$}}
    \end{align*}
    such that every set of vertices $U$ in $G$ satisfies:
    \begin{align*}
    \myunderbrace{\text{$U$ is connected in $G$}}
    {a subset of vertices is connected if \\
    \scriptsize the induced subgraph is connected} \qquad \text{implies} \qquad (\bigcup_{v \in U} X_v) \text{ is connected in $H$.}
    \end{align*}
    (It is enough to check the implication for sets $U$ with at most two vertices; and we assume that one vertex sets are connected).
    Show that if $L$ is a property of  graphs that is definable in \mso using the incidence model (we use the incidence model for the remaining exercises), then the same is true for  ``some minor satisfies $L$''.

}{}

\mikexercise{The Grid Theorem says that if a class of  graphs has unbounded treewidth, then for every $n \in \set{1,2,\ldots}$ there is some graph in the class which has an $n\times n$ grid as a minor. Using the Grid Theorem, prove that if a class of  graphs has decidable \mso theory, then it has bounded treewidth.
}{}

\mikexercise{Show a class of graphs that has undecidable \mso theory and bounded treewidth. }{ Graphs where the number of vertices belongs to some undecidable set of numbers.}
\mikexercise{Show that the $n \times n$ grid has treewidth at least $n$.}{}

\mikexercise{Show that for every language $L \subseteq \set a^*$ recognised in linear time by a (possibly nondeterministic) Turing machine, the language 
\begin{align*}
\set{ G : \text{$G$ is an $n \times n$ grid such that $a^n \in L$}}
\end{align*}
is definable in \mso.}{}

\mikexercise{Show that if a graph has treewidth at most $k$, then one can choose an orientation of its edges so that every vertex has at most $k+1$ outgoing edges. }{Take a tree decomposition, and direct each edge  so that it points toward the vertex that is introduced closer to the root. Under this orientation, the edges that point away from a vertex $v$ must necessarily point to vertices that are in the bag when $v$ is introduced, and there are at most $k+1$ such vertices (including $v$).}
\mikexercise{\label{ex:msoone-msotwo-bounded-tw} Show that for every $k \in \set{1,2,\ldots}$, the logics \msoone and \msotwo have the same expressive power for graphs of bounded treewidth.}{}{}

\section{The hypergraph monad and Courcelle's Theorem}
In this section, we introduce algebras for graphs. These algebras are defined in terms of  a monad  that describes graphs\footnote{This monad is based on the hyperedge replacement algebras of Courcelle. A discussion of these algebras can be found in
\incite[Section 2.3.]{courcelleGraphStructureMonadic2012}
The presentation of  hyperedge replacement that uses monads is based on
\incite{bojanczykTwoMonadsGraphs2018}
.}. In order to define the monad, we will need to add more structure to graphs, namely labels, directed hyperedges (i.e.~edges that connect a number of vertices that is not necessarily two), and distinguished vertices called ports.  We use the name \emph{hypergraph} for graphs with such extra structure. 

  Like any monad, the hypergraph monad will allow us to talk about  algebras, homomorphisms, terms, recognisable languages, syntactic algebras, etc. The main result of this section is  Courcelle's Theorem, which says that every graph property definable in \mso is necessarily recognisable. In the next Section~\ref{sec:definabl-tree-decompositions}, we prove a converse to Courcelle's Theorem, which  says that for bounded treewidth, recognisability implies definability in \mso.



\newcommand{\incfun}{\ranked{\text{incident}}}
\newcommand{\rankstar}[1]{{\color{black}#1}^{\ranked *}}

We begin with a formal definition of hypergraphs.
\begin{definition}
    A \emph{hypergraph} consists of:
    \begin{itemize}
        \item A set $V$ of \emph{vertices}.
        \item A set $E$ of \emph{hyperedges}. Each hyperedge has an associated arity in $\set{0,1,\ldots}$.
        \item A  set $\Sigma$  of \emph{labels}. Each label has an associated arity in $\set{0,1,\ldots}$.
        \item A non-repeating sequence of distinguished vertices called \emph{ports};
        \item For each hyperedge $e$, an associated  label in $\Sigma$ of same  arity, and a non-repeating sequence of \emph{incident vertices} whose length is the arity of $e$.
    \end{itemize}
\end{definition}
In the end, we care mainly about hypergraphs that have no ports, i.e.~the sequence of ports is empty, but the ports will appear when decomposing hypergraphs into parts.  We use the name  \emph{non-port vertices} for  vertices that are not ports.
 For a hyperedge $e$ of arity $n$, we write 
$e[1],\ldots,e[n]$ for the sequence of incident vertices, and use the name \emph{incidence list} for this sequence.
In this chapter, all hypergraphs are assumed to be finite, which means that there are finitely many vertices and hyperedges.  
We draw hypergraphs  like this:
    \mypic{42}
To avoid clutter in the pictures,  we  skip the gray numbers on the edges and the numbers of the ports, in situations where  they are not important for the picture or implicit from the context. 


        A  graph can be represented  as a  hypergraph. The representing hypergraph has no ports, and the vertices are the same as in the graph. Each edge of the graph is represented by two binary hyperedges (with some fixed label), one in each direction.   Here is a picture:
        \mypic{48}
        Directed graphs can be represented in the same way, but with the hyperedges not necessarily using both opposing directions. 
        \paragraph*{The hypergraph monad.} We now describe the monad structure of hypergraphs. The main idea behind free multiplication is that  a hyperedge can be replaced by a hypergraph of matching arity\footnote{This is the reason why Courcelle uses the name \emph{hyperedge replacement} for the corresponding algebras. }.  This replacement, which will be the free multiplication in the monad, is illustrated in the Figure~\ref{ex:grids-first-order-logic}.

\begin{figure}[]
    \centering
    \mypic{43}
\mypic{44}
    \caption{Free multiplication in the hypergraph monad.}
    \label{fig:free-multiplication-hmonad}
\end{figure}

\begin{definition}
    [Hypergraph monad] The hypergraph monad, denoted by $\hmonad$, is defined as follows.
    \begin{itemize}
        \item The underlying category is the \emph{category of ranked sets}
        \begin{align*}
        \mathsf{Set}^{\set{0,1,\ldots}},
        \end{align*}
        which is the category of sorted sets, where the sort names are natural numbers.
        Objects  in this category are \emph{ranked sets}, i.e.~sets where every element has an associated arity in $\set{0,1,\ldots}$. Morphisms are arity-preserving functions between ranked sets.

        \item For a ranked set $\rSigma$, the ranked set  $\hmonad \rSigma$ consists of finite hypergraphs labelled by $\rSigma$, modulo isomorphism. The arity of a hypergraph is  the number of ports.
        \item For a  function $\ranked{f : \Sigma \to \Gamma}$,  the function $\ranked{\hmonad f : \hmonad \Sigma \to \hmonad \Gamma}$
applies $\ranked f$ to the  labels, without changing the rest of the hypergraph structure.
\item The unit operation in the monad associates to every  letter $a \in \rSigma$ of arity $n$  a  hypergraph  which has $n$ ports, no other vertices, and one    hyperedge labelled by $a$ which is incident to all ports (in increasing order).  Here is a picture:
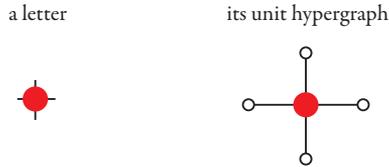
\mypic{64}
\item Let $G \in \hmonad \hmonad \rSigma$ be a hypergraph labelled by hypergraphs. Its free multiplication  is defined as follows. The vertices are vertices of $G$, plus pairs $(e,v)$ such that $e$ is a hyperedge of $G$ and $v$ is a non-port vertex in the hypergraph $G_e$ that is the label of the hyperedge $e$.  The  hyperedges are pairs $(e,f)$, where $e$ is a hyperedge of $G$ and $f$ is a hyperedge in $G_e$. The arities and labels of hyperedges are inherited from the second coordinate, while the incidence lists are defined by 
\begin{align*}
(e,f)[i] = \begin{cases}
    f[i] & \text{if $f[i]$ is a non-port vertex}\\
    e[j] & \text{if $f[i]$ is the $j$-th port.}
\end{cases}
\end{align*}
\end{itemize}
\end{definition}

We leave it as an exercise for the reader to check that the above definition satisfies the monad axioms. This completes the definition of the hypergraph monad. 

The hypergraph monad generalises the forest monad, as shown in the following example.
\begin{myexample}
     A forest can be represented as a hypergraph of arity one, as explained in the following picture:
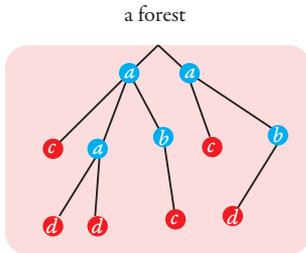
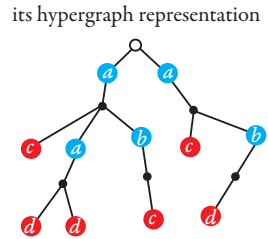
\mypic{80}    
Nodes of forest type in the forest are represented by hyperedges of arity one, while nodes of context type are represented by hyperedges of arity two.

A context can be represented as a hypergraph of arity two:
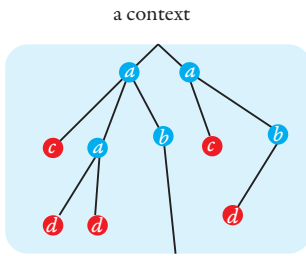
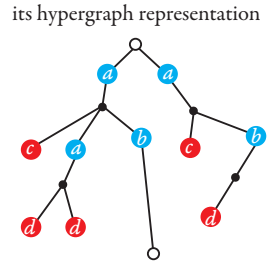
\mypic{81}
This representation is consistent with the monad structures of the forest monad and the context monad. Therefore, we can think of the forest monad as being a sub-monad of the hypergraph monad (when we identify the forest sort with arity 1, and the context sort with arity 2). In particular, from every algebra of the hypergraph monad we can extract an algebra of the forest monad.
\end{myexample}

The rest of this section is devoted to discussing the algebraic notions that arise from the hypergraph monad, such as   algebras, homomorphisms, recognisability, and terms. 

\exercisehead

\mikexercise{
Show that $\hmonad$ satisfies the monad axioms.
}{
Let us focus on associativity of free multiplication:
\begin{align*}
\xymatrix @R=2pc @C=7pc { \hmonad \hmonad \hmonad \rSigma \ar[r]^{\text{free multiplication on $\hmonad \rSigma$}} \ar[d]_{\hmonad \text{(free multiplication on $\rSigma$)}} & \hmonad \hmonad \rSigma \ar[d]^{\text{free multiplication on $\rSigma$}} \\
\hmonad \hmonad \rSigma \ar[r]_{\text{free multiplication on $\rSigma$}} & \hmonad \rSigma
}
\end{align*}
Let $G \in \hmonad \hmonad \hmonad \rSigma$. If we apply to $G$ the functions on the down-right path  in the above diagram, then the resulting  hypergraph will have hyperedges  of the form
\begin{align*}
(e,(f,g)),
\end{align*}
where $e$ is a hyperedge of $G$, $f$ is a hyperedge in the label of $e$, and $g$ is a hyperedge in the label of $f$. If we apply the functions on the right-down path, then the resulting   hypergraph will have hyperedges of the form
\begin{align*}
((e,f),g),
\end{align*}
where the conditions on $e,f,g$ are the same as above. We will show that 
\begin{align*}
(e,(f,g)) \mapsto ((e,f),g)
\end{align*}
is going to be an isomorphism between the two portd hypergraphs. 
}

\mikexercise{Show that connected hypergraphs are also a monad.}

\subsection{Recognisable languages}
In this chapter, we are mainly interested in languages recognised by algebras in the hypergraph monad. 
We use the name \emph{hypergraph algebra} for such algebras. We  are especially  interested in languages recognised by hypergraph algebras that are finite in the sense that they are finite on every arity.

\begin{myexample}\label{ex:commutative-monoid-hypergraph-algebra}
    Let $M$ be a commutative monoid.  Define a hypergraph  algebra as follows. The underlying ranked set  $A$ has a copy of $M$ on each arity, i.e.~the underlying set is 
    \begin{align*}
    M \times \myunderbrace{\set{0,1,\ldots}}{the arity}.
    \end{align*}
    The multiplication operation in the hypergraph algebra inputs  a hypergraph  in $\hmonad A$ and outputs the multiplication -- in the monoid $M$ -- of all the labels of its hyperedges. Because the monoid is commutative, the order of multiplication is not important.  The result of this multiplication is viewed as an element of the copy of $M$ that corresponds to the arity of the input hypergraph.
    It is not hard to see that this operation is associative, i.e.~it satisfies the axioms of  Eilenberg-Moore algebras.
    
    The hypergraph algebra constructed this way can be used to  recognise some simple languages of hypergraphs.   Apply the above construction to   the commutative  monoid
    \begin{align*}
   M =  (\set{0,1},\lor),
    \end{align*}
    yielding a hypergraph algebra  $A$. This  algebra  recognises the  language 
    \begin{align*}
    \set{ G \in \hmonad \Sigma : \text{some  hyperedge has label in $\Gamma$}} \qquad \text{for ranked sets $\Gamma \subseteq \Sigma$.}
    \end{align*}
    The homomorphism maps a hypergraph to $1$ if it belongs to the language, and to $0$ otherwise, with the number stored in the copy of the monoid that matches the arity of the input graph.
    Another application of this construction is recognising  the language of hypergraphs with an even number of hyperedges;  here the appropriate monoid is the two element group.
\end{myexample}

The hypergraph algebras in the above example are infinite, but finite on every arity.  This is the best we can do in the hypergraph monad, because it is impossible for a hypergraph algebra to have an underlying set that is finite altogether. The reason is that the multiplication operation $\mu : \hmonad A \to A$ in a hypergraph algebra is arity-preserving, and  $\hmonad A$ is nonempty on every arity (as witnessed by hypergraphs without hyperedges). Therefore,  the underlying set of a hypergraph algebra must be  nonempty on every arity. 
 In the following definition, and for the rest of this chapter, we assume that  ``finite hypergraph algebras'' are those which have finitely many elements for each arity. 
\begin{definition}[Recognisable language of hypergraphs]
    We say that a language $L \subseteq \hmonad \rSigma$ is \emph{recognisable}  if it is recognised by a homomorphism into a hypergraph algebra which has finitely many elements on every arity.
\end{definition}
This definition will turn out to be not restrictive enough, as far as general hypergraphs are concerned, see Example~\ref{ex:cliques-of-prime-size}. In fact, no entirely satisfactory definition of ``finite algebra'' for general hypergraphs is known, and possibly does not exist. However, for hypergraphs of  bounded treewidth, hypergraph algebras that are finite on every sort will  be a satisfactory definition that is equivalent to \mso, as we will see in Section~\ref{sec:definabl-tree-decompositions}.

 In Example~\ref{ex:commutative-monoid-hypergraph-algebra}, we already saw some examples of recognisable languages of hypergraphs. Here are some more  examples.

\begin{myexample}\label{ex:paths-in-a-hypergraph}
    Define a \emph{path} in a hypergraph to be sequence of the form
    \begin{align*}
        v_0 \stackrel{e_1}\to v_1 \stackrel{e_2}\to \cdots \stackrel{e_{n-1}}\to v_{n-1} \stackrel{e_n}\to v_n,
    \end{align*}
    where $v_0,\ldots,v_n$ are vertices  and $e_1,\ldots,e_n$ are hyperedges, such that each hyperedge $e_i$ is incident with both $v_{i-1}$ and $v_i$. Note that the notion of path does not depend on the order of the incidence lists for the hyperedges. The \emph{source} of the path is the vertex  $v_0$, its \emph{target} is  the vertex $v_n$, and we say that  the path \emph{connects} the source with the target. A hypergraph is called \emph{connected} if every vertex  can be connected to every other vertex  via a path. Define $h$ to be the function 
    which maps a hypergraph to the following information: (a) its arity; (b) is there a pair of non-connected  vertices such that at least one of them is not a port; and (c) which pairs of ports are connected.   One can check that this function is compositional, and therefore its image can be equipped with the structure of a hypergraph algebra so that $h$ is a homomorphism.  The corresponding hypergraph algebra is finite on every arity. Therefore, the language of connected hypergraphs is recognisable.
\end{myexample}

\begin{myexample}
    \label{ex:three-colouring}
    In this example, we show that the language of $k$-colourable hypergraphs is recognisable.
    Define a \emph{$k$-colouring} of  a  hypergraph  to be a function from vertices to $\set{1,\ldots,k}$ 
    such that no hyperedge has an incidence list that   uses some  colour twice. (In particular,  all  hyperedges have arity at most $k$.) Define $h$ to be the function which maps a hypergraph to the following information: (a) its arity; (b) which functions from the ports to $\set{1,\ldots,k}$ can be extended to $k$-colourings. If the hypergraph has arity zero, then (b) is just one bit of information: is there a $k$-colouring or not. This function is compositional, and has finitely many values for each arity, and therefore the language of $k$-colourable hypergraphs is recognisable. 
\end{myexample}

The following example illustrates a problem with of our notion of recognisability, which is that it allows for too many algebras, at least as long as hypergraphs of unbounded treewidth are allowed.
\begin{myexample}
    \label{ex:cliques-of-prime-size}
    We say that a hypergraph is a \emph{clique} if every two vertices are adjacent (i.e.~connected by some hyperedge).  Let 
    \begin{align*}
    P \subseteq \set{0,1,\ldots}
    \end{align*}
    be any set of natural numbers, possibly undecidable.  We will show that  the  language ``cliques with no ports, where the number of vertices is in $P$'' is recognisable.  
Let $h$ be the function which maps a hypergraph to the following information: (a) its arity; (b) is there a pair of non-adjacent vertices such that at least one of them is not a port; (c) which ports are adjacent.  If the arity is zero, then $h$ also stores: (d) is the number of vertices in $P$. This function is compositional, and has finitely many values for each arity, and therefore the language  ``cliques whose size is in $P$'' is recognisable.  \end{myexample}

As we will see later on, the problem from the above example will disappear once we restrict attention to hypergraphs of bounded treewidth. 

\exercisehead

\mikexercise{Show that every recognisable language in the hypergraph monad has a syntactic algebra. }{}

\subsection{Terms and tree decompositions.} Tree decompositions and treewidth can be naturally extended to hypergraphs, as formalised in Definition~\ref{def:treewidth-hypergraphs} below, and illustrated in the following picture:
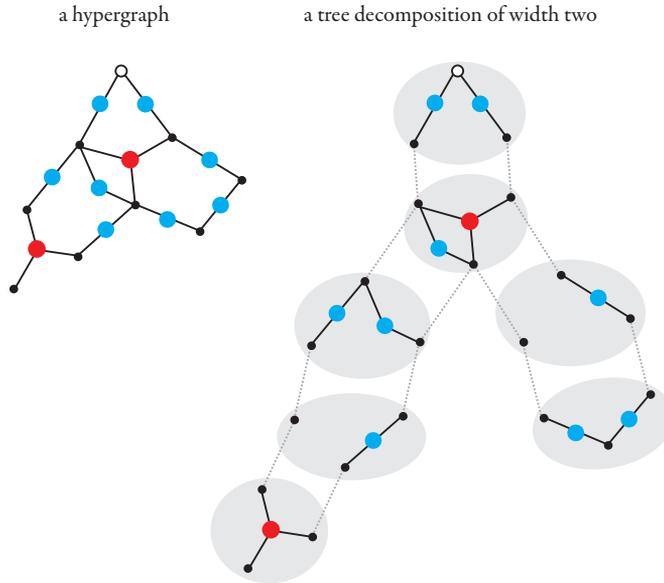
\mypic{111}
In this section we discuss an alternative perspective on treewidth, which is defined using monad terminology.

\begin{definition}[Tree decompositions for hypergraphs]
    \label{def:treewidth-hypergraphs} Tree decompositions are defined for hypergraphs in the same way as for graphs, with the following differences: (a) for every hyperedge there must be some bag which contains  its entire incidence list (we say that such a bag \emph{covers} the hyperedge); (b)  every port  of the hypergraph appears in the root bag.
\end{definition}

  As before, the width of a tree decomposition is the maximal bag size minus one, and the treewidth of a hypergraph is the minimal width of a tree decomposition.  For hypergraphs which represent graphs (i.e.~no ports, and every edge is represented by two binary hyperedges in opposing directions), the above  definition coincides with Definition~\ref{def:tree-decomposition}.

A bag in a tree decomposition can cover an unbounded number of hyperedges. This will not be a problem for our intended applications, since the properties of hypergraphs that we study will not depend in an important way on parallel hyperedges (i.e.~hyperedges with the same incidence lists).

\paragraph*{Tree decompositions as terms.}
  The algebraic structure of the hypergraph monad can be used to give an alternative description of treewidth.  
Recall the notion of terms from  Section~\ref{sec:terms}: a term  over variables $X$ is any element of $\hmonad X$. As was the case for the forest monad,  terms in the  hypergraph monad are sorted, which means that each variable used by the term has an arity, and the term itself has an arity. If $A$ is a hypergraph algebra, then a term $t \in \hmonad X$  induces a \emph{term operation} $t^A§$ and defined by
\begin{align*}
\myunderbrace{\eta \in A^X}{an arity-preserving\\
\scriptsize valuation of the variables} \qquad  \mapsto \qquad  \myunderbrace{\text{multiplication in $A$ applied to $(\hmonad \eta)(t)$.}}{an element of the hypergraph algebra $A$, \\ \scriptsize whose arity is the same as the arity of $t$}
\end{align*}
§§As was the case for forest algebra,   term operations  are in general  not   arity-preserving, if only because their inputs do not have a well-defined arity.

Since a term is a hypergraph, it has some treewidth. The following lemma shows that hypergraphs of treewidth at most $k$ are closed under applying (term operations induced by) terms of treewidth at most $k$. The  hypergraph algebra used in the lemma is the  free hypergraph algebra.
\begin{lemma}\label{lem:treewidth-closed-under-terms}
    Let $t \in \hmonad X$ be a term and let $\eta \in (\hmonad \Sigma)^X$ be a valuation of its variables. If the hypergraphs $t$ and $\set{\eta(x)}_{x \in X}$ have treewidth at most $k$, then the same is true for the result of applying the term operation $t^{\hmonad \Sigma}$  to $\eta$.
\end{lemma}
\begin{proof}
    Take a tree decomposition for the term $t$. For every hyperedge $e$ which is labelled by a variable $x$, find a node of the tree decomposition whose bag contains the incidence list of the hyperedge, remove the hyperedge,  and add a child to this node with a tree decomposition of $\eta(x)$.
\end{proof}

A corollary of the above lemma is that there is a well-defined monad for hypergraphs of treewidth at most $k$. 
This monad, call it $\hmonad_k$, uses only hypergraphs with treewidth at most $k$, with all the monad structure inherited from $\hmonad$.  The underlying category is ranked sets with arities at most $k+1$; since hypergraphs with bigger arities will have treewidth at least $k+1$. 

\paragraph*{The treewidth terms.} We now show a family of terms which can be used to generate all hypergraphs of given treewidth.  Define the \emph{treewidth terms} to be the terms\footnote{There is an  inconsistency in our use of the words ``introduce'' and ``forget''. When we say that a node in a tree decomposition introduces a vertex, we take a top-down perspective on tree decompositions. On the other hand, the name of the ``forget'' term in Figure~\ref{fig:treewidth-terms} is based on a bottom-up perspective of the same phenomenon. } from Figure~\ref{fig:treewidth-terms}. For  a hypergraph algebra, §define its \emph{treewidth $k$ operations} to be the term operations induced in the algebra by treewidth terms that have treewidth at most $k$. 
The following theorem shows that the treewidth terms can be used to generate all hypergraphs of given treewidth.

\newcommand{\treewidthterm}[3]{
    
    \begin{minipage}{7cm}
       {\it #1. } #2        
        \end{minipage}  \quad
        \begin{minipage}{3cm}
         \mypic{#3}   
        \end{minipage}\\\\
}
\noindent

\begin{figure}
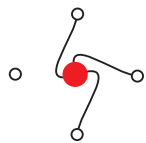

    \centering
    \begin{tabular}{l}
        \treewidthterm{Forgetting} 
        {Let $x$ be a variable of arity $k+1$.  The \emph{$k$-forgetting term} is defined  to be the hypergraph  in $\hmonad \set x$ which has  ports $\set{1,\ldots,k}$, one non-port vertex $v$, and one hyperedge  with label $x$ and incidence list $(1,\ldots,k,v)$.
        }{59}
        \treewidthterm{Fusion} 
        {Let   $x,y$ be two variables of arity $k$. The \emph{$k$-fusion term} is defined to be the hypergraph in $\hmonad \set{x,y}$ which has $k$ ports, no vertices, and two hyperedges with labels $x$ and $y$ and incidence list  $(1,\ldots,k)$.
        }{57    }
        \treewidthterm{Rearrangement} 
        {Let $f : \set{1,\ldots,k} \to \set{1,\ldots,\ell}$ be an injective function. Let $x$ be a  variable of arity $k$. The \emph{$f$-rearrangement term}  is defined to be the hypergraph in $\hmonad \set x$ which has $\ell$ ports, no vertices, and one hyperedge with label $x$ and incidence list $(f(1),\ldots,f(k))$.
        }{58}
    \end{tabular}
    \caption{ The treewidth terms. The parameters $k, \ell$ are from $\set{0,1,\ldots}$. In the pictures, the ports are ordered clockwise from top, and the same is true for the vertices incident to a hyperedge.  }
    \label{fig:treewidth-terms}
\end{figure}


\begin{theorem} \label{thm:treewidth-in-terms-of-terms} Let $k \in \set{1,2,\ldots}$. A hypergraph has treewidth  at most $k$ if and only if it can be generated (in the free hypergraph algebra) from  hypergraphs with no vertices and at most $k+1$ ports,  by applying treewidth $k$ operations.
\end{theorem}
\begin{proof}
    The right-to-left implication follows from Lemma~\ref{lem:treewidth-closed-under-terms}.
    
    Consider now the left-to-right implication.
    Every tree decomposition can easily be modified, without affecting its width, into a tree decomposition which satisfies: (*) the root bag contains the ports and no other vertices, and if a node has at least two children, then the node and all of its children have the same bag.  To ensure condition (*), we can insert an extra node on every parent-child edge which has the same bag as the parent. By a simple induction on the size number of nodes, one shows that for every width $k$ tree decomposition satisfying (*), the underlying hypergraph can be generated using the treewidth $k$ operations as in the statement of the lemma. 
\end{proof}

Using the above theorem, and the same  argument as in Theorem~\ref{thm:shelah-operations}, we get the following corollary, which gives a finite representation for algebras in the monad $\hmonad_k$ of hypergraphs of treewidth at most $k$. 

\begin{corollary}
    Let $k \in \set{1,2,\ldots}$ and consider the monad $\hmonad_k$ of hypergraphs with treewidth at most $k$. The multiplication operation in an algebra over this monad is uniquely determined by its treewidth $k$ operations.
\end{corollary}

\subsection{Courcelle's Theorem}
In this section, we prove Courcelle's Theorem\footcitelong[Theorem 4.4]{courcelleMonadicSecondorderLogic1990}, which says that all languages definable in \mso are recognisable. 
To define properties of hypergraphs in \mso, we use a hypergraph version of  the incidence model,  defined as follows. 
\begin{definition}
    [Incidence model] The incidence model of a hypergraph is defined as follows. The universe is vertices and hyperedges, 
    and it is equipped with the following relations:
    \begin{align*}
    \myunderbrace{e[i]=v}{$v$ is the $i$-th\\
    \scriptsize vertex incident to $e$}
    \qquad 
    \myunderbrace{\mathrm{port}_i(v)}{$v$ is the \\ \scriptsize $i$-th port}
    \qquad 
    \myunderbrace{a(e)}{hyperedge $e$\\
    \scriptsize has label $a$.}.
    \end{align*}
    The arguments of the relations are $e$ and $v$, while $i \in \set{1,2,\ldots}$ and $a \in \Sigma$ are parameters. Each choice of parameters gives a different relation.
\end{definition}
As was the case for forest algebra, recognisability holds also for \emph{counting \mso}, which extends \mso by allowing the following form of modulo counting: for every $n \in \set{2,3,\ldots}$ and every $\ell \in \set{0,1,\ldots,n-1}$ there is a predicate
\begin{align*}
|X| \equiv \ell \mod n,
\end{align*}
which inputs a set and says if the size of this set is congruent to $\ell$ modulo $n$.  The modulo counting predicate is a second-order predicate, since it inputs a subset of the universe, and not an element (or tuple of elements) in the universe. 
Counting \mso  is more powerful than  \mso without counting, e.g.~``the number of vertices is even'' can be defined in counting \mso but not in \mso, see Exercise~\ref{ex:no-edges}.   As we will show in Section~\ref{sec:definabl-tree-decompositions}, counting \mso is enough to describe all recognisable properties of hypergraphs, assuming bounded treewidth.

\begin{theorem}[Courcelle's Theorem]
    If a language  $L \subseteq \hmonad \Sigma$ is definable in counting \mso, over the incidence model, then it is recognisable, i.e.~recognised by a homomorphism
    into a hypergraph algebra that is finite on every arity.
\end{theorem}

We use the same construction as in previous chapters. The main step of the proof, which deals with set quantification, is presented in the following lemma.
\begin{lemma}\label{lem:hypergraph-mso-language-closure}
    The recognisable languages images  under  functions of the form\footnote{In Exercise~\ref{ex:ltl-homo-bad} we show that the assumption on letter-to-letter homomorphisms is important.} 
    \begin{align*}
    \myunderbrace{\hmonad f : \hmonad \Sigma \to \hmonad \Gamma \qquad \text{for $f : \Sigma \to \Gamma$}}
    {such functions are called letter-to-letter homomorphisms.}
    \end{align*}
\end{lemma}
\begin{proof}
    We  use a powerset construction for hypergraph algebras. Since we have already used powerset constructions before, we  take this opportunity     to discuss powerset constructions  in more detail and  generality, so that we can  think about the kinds of monads that allow a powerset construction  (these are not all monads, e.g.~the group monad does not have a powerset construction).

    For a ranked set $X$, define its \emph{powerset} to be the ranked set $\powerset X$ where elements of arity $n$ are sets of  elements from $A$ that have arity $n$. For an arity-preserving function $f : X \to Y$  on ranked sets, define 
    \begin{align*}
        \powerset f : \powerset X \to \powerset Y
    \end{align*}
    to be the arity-preserving function that maps a set to its image\footnote{In the language of category theory, $\powerset$ is the co-variant powerset functor, as opposed to the contra-variant powerset functor which uses inverse images instead of forward images.}. For a ranked set $X$, define \emph{distribution on $X$} to be the function of type 
\begin{align*}
\hmonad \powerset X \to \powerset \hmonad X
\end{align*}
which inputs a hypergraph $G$, and outputs the set of hypergraphs that can be obtained from $G$ by choosing for each edge an element of its label.  

Although trivial in the hypergraph monad, the following  claim is not true in general for all monads, e.g.~it would be false in the group monad for a naturally defined distribution.  
\begin{claim}
    \label{claim:naturality-of-distribution} Distribution is a natural transformation, which means that the following diagram commutes for every arity-preserving function $f : X \to Y$
    \begin{align*}
    \xymatrix{
        \hmonad \powerset X
        \ar[r]^{\hmonad \powerset f}
        \ar[d]_{\text{distribute on $X$}} 
        &
        \hmonad \powerset Y
        \ar[d]^{\text{distribute on $Y$}}
        \\
        \powerset \hmonad X 
        \ar[r]_{\powerset \hmonad f}
        &
        \powerset \hmonad Y
    }
    \end{align*}
\end{claim}
\begin{proof}
     The right-down path corresponds to the following procedure: for each hyperedge, choose an element of its label, and then apply $f$. The down-right path corresponds to the following procedure: for each hyperedge, take the image under $f$ of its label, and then choose an element.  The two procedures give the same result. This is true thanks to the following  property of distribution for hypergraphs: if we apply distribution on $X$ to some hypergraph, then every hypergraph in the resulting set will have the same vertices, ports and hyperedges as the original hypergraph. This property would not hold, for example, in the group monad, and  the claim would be false in the group monad\footnote{In fact, there is no powerset construction for algebras in the group monad. Nevertheless, by a proof that does not use powerset algebras, one can show that in the group monad the recognisable languages are closed under images of letter-to-letter homomorphisms.}.
\end{proof}

We  use the powerset and distribution to prove the lemma. 
    Suppose that a language $L$ is recognised by a homomorphism
    \begin{align*}
    h : \hmonad \Sigma \to A,
    \end{align*}
    and consider a letter-to-letter homomorphism
    \begin{align*}
    \hmonad f : \hmonad \Sigma \to \hmonad \Gamma.
    \end{align*}
    Define $g$ to be the composition of the following two functions:
    \begin{align*}
    \xymatrix@C=3cm{
        \hmonad \powerset \Sigma \ar[r]^{\text{distribute on $\Sigma$}} &
        \powerset \hmonad \Sigma \ar[r]^{\powerset h} &
        \powerset A.
    }
    \end{align*}
    \begin{claim}
        The function $g$ is compositional.
    \end{claim}
    \begin{proof} Consider the following diagram, with  red letters being labels of  faces:
        
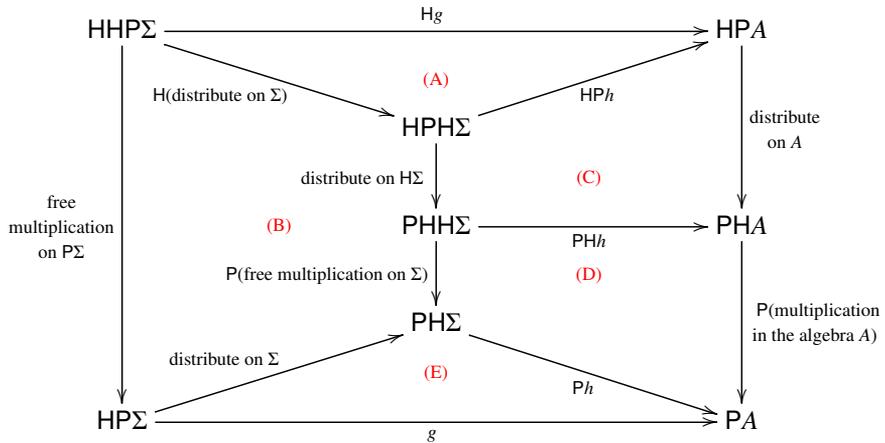
\begin{align*}
            \xymatrix
            @C=3cm{
                \hmonad \hmonad \powerset \Sigma
                \ar[rr]^{\hmonad g}
                \ar[dddd]_{\txt{\scriptsize free \\
                \scriptsize multiplication\\ \scriptsize on $\powerset \Sigma$}}
                \ar[dr]_{\txt{\scriptsize  $\hmonad$(distribute  on $\Sigma$)}\qquad}
                &&
                \hmonad \powerset A
                \ar[dd]^{\txt{\scriptsize distribute \\ 
                \scriptsize  on $A$}}
                \\
                &
                \hmonad \powerset \hmonad \Sigma 
                \facelabel u A
                \ar[ur]_{\hmonad \powerset h}
                \ar[d]_{\text{distribute on $\hmonad \Sigma$ }} &
                \\
                \facelabel r B
                & 
                \powerset \hmonad \hmonad \Sigma
                \facelabel {ur} C
                \facelabel {dr} D
                \ar[d]_{\powerset\text{(free multiplication on $\Sigma$)}}
                \ar[r]_{\powerset \hmonad h}
                & 
                \powerset \hmonad A
                \ar[dd]^{\txt{ \scriptsize $\powerset$(multiplication \\ \scriptsize in the algebra $A$)}}
                \\
               & 
                \powerset \hmonad \Sigma 
                \ar[dr]_{\powerset h}
                \facelabel d E
                &
                \\
                \hmonad \powerset \Sigma
                \ar[rr]_g
                \ar[ur]^{\text{distribute on $\Sigma$}\qquad}
                &&
                \powerset A
            }
            \end{align*}
            If we prove that the perimeter of the diagram commutes, then we will prove that $g$ is compositional (the composition of the two arrows on the right-most side will be the multiplication operation in the powerset algebra $\powerset A$). Faces \faceref A and \faceref E commute by definition of $g$. Face \faceref{C} commutes by naturality of distribution from Claim~\ref{claim:naturality-of-distribution}, and face \faceref{D} commutes because $h$ is a homomorphism. It remains to show that face \faceref{B} commutes. This again, is proved via simple check\footnote{In the language of category theory,   face \faceref{B} is the main axiom of a distributive law of a monad over a functor. }, similarly to Claim~\ref{claim:naturality-of-distribution}.
    \end{proof}
    
    Like for any compositional function, the image of $g$ can be equipped with a multiplication operation which turns $g$ into a homomorphism.  (As mentioned in the proof of the claim above, this multiplication operation is the composition of the two arrows on the right-most side side from the diagram.) We will use the homomorphism $g$ to recognise the image of $L$ under $\hmonad f$.
    Consider the following diagram:
    \begin{align*}
        \xymatrix@C=5cm{
            \hmonad \Gamma 
            \ar[r]^{\hmonad (\text{inverse image under $f$})} 
            \ar[d]_{\text{inverse image under $\hmonad f$}} &
            \hmonad \powerset \Sigma 
            \ar[dl]_{\text{distribute on $\Sigma$}\qquad}
            \ar[d]^{g}\\
            \powerset \hmonad \Sigma 
            \ar[r]_{\powerset h}
            &
            \powerset A
        }
    \end{align*}
    The top-left  face in the diagram commutes by definition of distribution, and the bottom-right  face  commutes by definition of $g$.  A hypergraph $G \in \hmonad \Gamma$ belongs to the image of the language  $L$ under the function $\hmonad f$ if and only if applying the function on the down-right path in the diagram gives a set that intersects the image $h(L)$.  Since the diagram commutes, it follows that the right-down path in the diagram recognises the image of the language $L$ under the function $\hmonad L$. The right-down path is a homomorphism, as a composition of two homomorphisms. Also, $\powerset A$ is finite on every arity, because finiteness on every arity is preserved by powersets. 
\end{proof}

The above lemma implies that recognisable languages are closed under quantification of sets of hyperedges (since a subset of the hyperedges can be seen as a colouring of hyperedges with two colours ``yes'' and ``no''). This motivates the following logic.

\begin{definition}
    Define \emph{hyperedge  counting \mso} to be the following variant of \mso. There is no first-order quantification, and set quantifiers range over sets of hyperedges. The logic allows the following relations on sets of hyperedges:
\begin{align*}
    \myunderbrace{X \subseteq Y}
    {set inclusion} \qquad     
    \myunderbrace{
        X \subseteq a
    }{every \\ \scriptsize hyperedge  \\
    \scriptsize in $X$ has \\ \scriptsize label $a \in \Sigma$}
    \qquad  
    \myunderbrace{
        i \in X[j]
    }{there exists a \\ \scriptsize hyperedge \\ \scriptsize $e \in X$ 
    such \\ \scriptsize that $e[j]$ is\\
    \scriptsize  the $i$-th port}
    \qquad  
    \myunderbrace{
        X[i] \cap Y[j] \neq \emptyset
    }{there exist hyperedges \\ \scriptsize $e \in X$ and $f \in Y$\\
    \scriptsize such that $e[i]=f[j]$}
    \qquad 
    \myunderbrace{|X| \equiv 0 \mod n.}{the number of hyperedges \\ \scriptsize in $X$  is dvisible by $n$}
\end{align*}
In the above relations, the arguments are the sets $X, Y$. The  labels $a \in \Sigma$ and numbers $i,j,n \in \set{1,2,\ldots}$ are parameters. Each choice of parameters gives a different relation.
\end{definition}

Our usual proof of the translation of \mso to algebras shows the following result, which is almost Courcelle's theorem, except that the logic is hyperedge \mso instead of counting \mso. The minor difference between the two logics -- which boils down to isolated vertices -- will be treated later on.
\begin{lemma}\label{lem:hypergraph-logic-to-language}
    If a language $L \subseteq \hmonad \Sigma$ is definable in hyperedge counting \mso, then it is recognisable.
\end{lemma}
\begin{proof}
    Same proof as for the monads for words and forests.  Consider a  formula of hyperedge counting \mso 
    \begin{align*}
    \varphi(\myunderbrace{X_1,\ldots,X_n}{the free variables represent\\ \scriptsize ses of hyperedges}),
    \end{align*}
    where $\Sigma$ is the ranked set of labels used by the underlying hypergraphs. Define the  \emph{language} of this formula  to be the set hypergraphs over an extended alphabet that consists of  $2^n$ disjoint copies of $\Sigma$. This language is defined in the same way as for words and forests: for each hyperedge, the bits from $2^n$ in its label determine which of the sets $X_1,\ldots,X_n$ contain the hyperedge. By induction on formula size, we prove that every formula has a  recognisable language. 
    The induction step is proved in the same way as for words and forests: for Boolean combinations we use homomorphisms into product algebras, while for the quantifiers we use the powerset construction from  Lemma~\ref{lem:hypergraph-mso-language-closure}.
    
    We are left with the induction base. 
    For the formulas  $X \subseteq Y$ and $X \subseteq a$, the corresponding hypergraph language is of the form ``every hyperedge has a label in $\Gamma \subseteq \Sigma$''. Such languages were shown to be recognisable in Example~\ref{ex:commutative-monoid-hypergraph-algebra}. In the same example, we showed how to count hyperedges modulo some number, thus showing recognisability of the modulo counting relation. Consider now the language which corresponds to the  relation
    \begin{align*}
    i \in X[j].
    \end{align*}
    Let $h$ be the function $h$ which maps a hypergraph to the following information: (a) its arity; (b) which ports belong to $X[j]$. This function is easily seen to be compositional, and it has finite image on every arity, and therefore it witnesses recognisability of the language corresponding to $i \in X[j]$. A similar argument works for 
\begin{align*}
 X[i] \cap Y[j] \neq \emptyset.
\end{align*}
\end{proof}

\begin{lemma}\label{lem:non-isolated}
    For every sentence of counting \mso, there is a sentence of hypergraph counting \mso which gives the same results on hypergraphs without isolated vertices.
\end{lemma}
\begin{proof}
    Let $n$ be the maximal arity of letters in the finite alphabet.  Every set $X$ of  non-isolated vertices can be represented by $n$ sets of hyperedges as 
    \begin{align*}
    X = X_1[1] \cup \cdots \cup X_n[n],
    \end{align*}
    where $X_i$ is the set of hyperedges whose $i$-th incident vertex is in $X$. Using this representation, we can quantify over sets of non-isolated vertices by using quantification over sets of hyperedges. 
\end{proof}

A corollary of Lemmas~\ref{lem:hypergraph-logic-to-language} and~\ref{lem:non-isolated} is that for every language definable in counting \mso, there is a recognisable language that agrees with it on hypergraphs without isolated vertices. To finish the proof of Courcelle's Theorem, we need to take into account the isolated vertices, which is a minor inconvenience that is left as an exercise for the reader, see Exercises~\ref{ex:isolated-and-not} and~\ref{ex:isolated-vertices}.

    This completes the proof of Courcelle's Theorem. 

    \exercisehead

\mikexercise{
    Consider graphs (not hypergraphs). Show that the existence of an Eulerian cycle can be defined in counting \mso, but not in \mso.
}{}

\mikexercise{Show that the existence of a Hamiltonian cycle cannot be defined in counting \mso with set quantification restricted to sets of vertices (and not hyperedges).}{}

\mikexercise{\label{ex:ltl-homo-bad} Show  that Lemma~\ref{lem:hypergraph-mso-language-closure} ceases to be true if we allow homomorphisms that are not necessarily letter-to-letter.}{}

\mikexercise{Show that for every \mso formula $\varphi(X)$ with one free set variable, the following problem can be solved in linear time: 
\begin{itemize}
    \item \emph{Input.} A tree decomposition $T$;
    \item \emph{Output.} The maximal size of a set of vertices $X$, such that $\varphi(X)$ is true in the underlying hypergraph.
\end{itemize}
}{}

\mikexercise{\label{ex:isolated-and-not}
For a hypergraph $G$, define two hypergraphs  $\alpha(G)$ and $\beta(G)$ as follows:
\begin{itemize}
    \item $\alpha(G)$: remove all isolated vertices;
    \item $\beta(G)$: remove all hyperedges and non-isolated vertices.
\end{itemize}
The functions $\alpha$ and $\beta$ are not arity-preserving, since the arity of $G$ is equal to the sum of arities of $\alpha(G)$ and $\beta(G)$.
Show that every sentence of counting \mso is equivalent to a finite Boolean combination of sentences of counting \mso, each of  which talks about only   $\alpha(G)$ or $\beta(G)$. 
}{}

\mikexercise{\label{ex:isolated-vertices} Recall the functions $\alpha$ and $\beta$ from the previous exercise.  Show that if $L$ is a language of hypergraphs that is definable in \mso, then the same is true for the languages
\begin{align*}
\set{G : \alpha(G) \in L } \qquad \text{and} \qquad \set{G : \beta(G) \in L }.
\end{align*}
Together with Exercise~\ref{ex:isolated-and-not}, this observation completes the proof of Courcelle's Theorem.}
{
    Consider a language defined by the logic in the 
, we observe that it can be defined by checking an ultimately periodic property of the numbers of ports and isolated vertices. For a hypergraph $G \in \hmonad \Sigma$, define $\gamma(G)$ to be the word $a^n b^m$ where $n$ is the number of ports and $m$ is the number of isolated vertices.  This  word is roughly the same thing as $\beta(G)$,  except that it has an order. Therefore, for every sentence $\varphi$ of counting \mso on graphs, one can easily find a sentence $\psi$ of counting \mso over finite words which makes the following diagram commute:
\begin{align*}
\xymatrix{
    \hmonad \Sigma \ar[r]^-\gamma
    \ar[d]_{\beta} 
    &
    a^* b^*
    \ar[d]^{\psi} \\
    \hmonad \Sigma \ar[r]_\varphi
     &
    \set{\text{yes, no}}
}
\end{align*}
For finite words, counting \mso has the same expressive power as \mso, because modulo counting can be expressed using the order of the word. Since \mso on finite words can only define regular languages, it follows that $\psi$ defines a regular language contained in $a^* b^*$. Every such regular language is defined as the intersection of $a^*b^*$ with   a finite Boolean combination of constraints of the form 
\begin{align*}
\myunderbrace{\#_\sigma = n}
{letter $\sigma \in \set{a,b}$ appears\\
\scriptsize exactly $n$ times }
\qquad \qquad 
\myunderbrace{\#_\sigma \equiv k \mod n}
{the number of \\ \scriptsize appearances of  $\sigma \in \set{a,b}$ \\
\scriptsize is congruent to $k$ modulo $p$}
\end{align*}
It follows that $L_\beta$ is a finite Boolean combination of languages of the form ``exactly $n$ ports'', ``exactly $n$ isolated vertices'', ``the number of ports is congruent to $k$ modulo $n$'', ``the number of isolated vertices is congruent to $k$ modulo $n$''. All of these languages are easily seen to be recognisable, using a construction similar to Example~\ref{ex:commutative-monoid-hypergraph-algebra}.}

\subsection{Satisfiability for bounded treewidth}
    We finish this section with an algorithm for deciding satisfiability of counting \mso, assuming bounded treewidth. Recall that already first-order logic on graphs has undecidable satisfiability, and this  undecidability carries over to the more general setting of hypergraphs and counting \mso.
    We recover decidability if we restrict attention to hypergraphs of bounded treewidth.

\begin{theorem}
    The following problem is decidable:
    \begin{itemize}
        \item {\bf Input.} A sentence of counting \mso  and $k \in \set{1,2,\ldots}$.
        \item {\bf Question.} Is the sentence true in  some hypergraph of treewidth at most $k$?
    \end{itemize}
\end{theorem}
\begin{proof} We use the proof of Courcelle's Theorem, with an emphasis on computability.
    We say that a ranked set is \emph{computable} if its elements can be represented in a finite way, and there is an algorithm which inputs an arity $k$ and either outputs the finite list of all elements of arity $k$ (if there are finitely many), or starts enumerating these elements (if there are infinitely many). The algorithm also says if there are finitely many elements of arity $k$ or not. 
    We say that a hypergraph algebra $A$ is \emph{computable} if its underlying ranked set is computable, and its multiplication operation is also computable (the inputs to the multiplication are finite hypergraphs, which can be represented in a finite way).  
    
    Free hypergraph algebras over computable alphabets are computable,  all of the hypergraph algebras that we used as recognisers for the atomic relations in the  proof of Courcelle's Theorem are computable, and computability is preserved under the products and the powerset construction. Therefore, we get the following computable strengthening of Courcelle's Theorem: given a sentence of counting \mso, which defines a property of hypergraphs over a finite alphabet $\Sigma$, we can compute a recognising homomorphism 
    \begin{align*}
    h : \hmonad \Sigma \to A
    \end{align*}
    into a computable hypergraph algebra. The hypergraph algebra, homomorphism, and accepting set are  represented by the corresponding algorithms.

    Let $A_k \subseteq A$ be the image under $h$ of all hypergraphs with treewidth at most $k$.  By Theorem~\ref{thm:treewidth-in-terms-of-terms}, $A_k$ is equal to the smallest subset of $A$ that contains the letters and which is closed under applying the treewidth $k$ operations in the hypergraph algebra $A$.  Since $A_k$ is contained in the finite part of $A$ which has arity at most $k+1$, and there are finitely many treewidth $k$ operations, it follows that $A_k$ can be computed. Finally, we check if $A_k$ contains at least one element of the accepting set. 
\end{proof}

\exercisehead
\mikexercise{Show that the following problem is decidable: given a first-order formula, decide if it is true in some rectangular grid. Here we are talking about unlabelled rectangular grids as in Example~\ref{example:grid}, and not labelled rectangular grids as in Example~\ref{ex:grids-first-order-logic}.}{}

\mikexercise{\label{ex:compute-syntactic-hypergraph-algebra} Show that if a hypergraph language $L$ has bounded treewidth and is definable in counting \mso, then its syntactic algebra is computable.}{}

\mikexercise{Show that the following problem is decidable: given $k \in \set{1,2,\ldots}$ and an \mso sentence $\varphi$, decide if $\varphi$ is true in infinitely many hypergraphs of treewidth at most $k$.}{}

\section{Definable tree decompositions}
\label{sec:definabl-tree-decompositions}

In this section\footnote{The results of this section are based in \incite{bojanczykDefinabilityEqualsRecognizability2016a}}, we show that for hypergraphs of bounded treewidth, tree decompositions can be defined in \mso. One application of this result is going to be a converse of Courcelle's Theorem for bounded treewidth: every recognisable property is definable in counting \mso for hypergraphs of bounded treewidth.

 We begin by explaining how a tree decomposition can be defined in  \mso. 
This is split into two ingredients: in Definition~\ref{def:introduction-ordering} we represent a  tree decomposition  using a  binary relation on vertices, and then in Definition~\ref{def:set-parameters} we show how such a binary relation can be defined in \mso. 

\paragraph*{The introduction ordering.} We represent a tree decomposition using  the order in which vertices of the underlying hypergraph are introduced.

\begin{definition} [Introduction ordering] \label{def:introduction-ordering} Define the \emph{introduction ordering} of a tree decomposition to be the following binary relation on vertices in the underlying hypergraph:
        \begin{align*}
            \myunderbrace{\text{
                $v$ is  introduced in the same node as, or an ancestor of, the  node    introducing  $w$.
            }}{ we say that  \emph{$v$ is introduced before vertex $w$} of the tree decomposition}
            \end{align*}
\end{definition}

The introduction ordering is a pre-order, i.e.~it is transitive and reflexive. It need not be anti-symmetric, because several vertices might be introduced in the same node. 

\begin{myexample}\label{example:isolated-vertices-introduction-ordering}
    Consider a hypergraph which has no ports or hyperedges, but only isolated vertices, like in the following picture:
    \mypic{120}
    One tree decomposition for this hypergraph has a node for each vertex, with the bag containing only that vertex, and with the nodes ordered left-to-right. Its introduction ordering is  in the following picture:
    \mypic{121}
    An alternative tree decomposition, has the same nodes and bags. However, this time we have some chosen root, and the remaining nodes are its children. Here is the introduction ordering for the alternative tree decomposition:
    \mypic{123}
In both pictures above, the introduction ordering is anti-symmetric, because each node of the tree decomposition introduces a single vertex. Here is a picture of an  introduction ordering  which is not anti-symmetric (and has two components):
    \mypic{122}
\end{myexample}

We now explain how a tree decomposition can be recovered from its introduction ordering. To do this, we  use two mild assumptions on tree decompositions (in the following, we say that a hyperedge is \emph{introduced} in node $x$ if $x$ is the least node that contains the incidence list of the hyperedge):
\begin{description}
    \item[(A)]every node introduces at least one vertex;
    \item[(B)] if a vertex $v$ is in the bag of node $x$, then it is incident to some  hyperedge that is introduced in $x$ or its descendants. 
\end{description}
Every tree decomposition can be transformed into a tree decomposition that satisfies (A) and (B), without increasing the width. In order to satisfy (A), we merge every node that does not introduce any vertices with its parent. In order to satisfy (B), we remove a vertex $v$ from all bags that violate condition (B).

If a tree decomposition satisfies (A) and (B), then it can be recovered from its introduction ordering as follows. Thanks to condition (A),  the nodes of the  tree decomposition  are equivalence  classes of vertices with respect to the equivalence ``$v$ is introduced before $w$ and vice versa'',  and the tree order on nodes is the inherited from the introduction ordering. Thanks to condition (B), a vertex $v$ is present in the bag of a node $x$ if and only if there is hyperedge that is incident to $v$ and a vertex that is introduced in $x$ or its descendants. The way that we recover a tree decomposition from its introduction ordering can be formalised  in \mso.

For the rest of Section~\ref{sec:definabl-tree-decompositions}, we only consider tree decompositions that satisfy (A) and (B). 

\paragraph*{Relations definable using set parameters.}
To represent the introduction ordering of a tree decomposition, we will use a formula of \mso that is equipped with extra set parameters, as described in the following definition. In the definition, when evaluating an \mso formula in a hypergraph, we use the incidence model from Definition~\ref{def:incidence-model}, where the universe is both vertices and hyperedges. 


\begin{definition}[Definable tree decompositions] \label{def:set-parameters}  An \emph{\mso formula with set parameters} is  an \mso formula of the form 
        \begin{align*}
        \varphi(
        \myunderbrace{    
        Y_1,\ldots,Y_n}{set variables \\ \scriptsize called\\ \scriptsize \emph{set parameters}}
        ,\quad
        \myunderbrace{    
            x_1,\ldots,x_m}{element variables \\ \scriptsize called\\ \scriptsize \emph{arguments}}
            ).
        \end{align*}
        We say that an $m$-ary relation $R$ in a hypergraph  is \emph{definable}  by $\varphi$ if 
        \begin{align*}
        \myunderbrace{\exists Y_1 \cdots \exists Y_n}
        {there is a choice \\
        \scriptsize of set parameters}\ 
        \myunderbrace{
        \forall x_1 \ \cdots \forall x_m \quad  R(x_1,\ldots,x_m) \iff \varphi(Y_1,\ldots,Y_n,x_1,\ldots,x_m).}{such that after fixing these set parameters in $\varphi$, \\ 
        \scriptsize we get exactly the relation $R$}
        \end{align*}
        We say that a tree decomposition is \emph{definable by $\varphi$} if its introduction ordering is definable by $\varphi$ in the underlying graph. We say that a set of hypergraphs has \emph{definable tree decompositions of bounded width} if there is an \mso formula $\varphi$ with set parameters, and a width $k \in \set{0,1,\ldots}$, such that every hypergraph from the set has a tree decomposition that is definable by $\varphi$ and has width at most $k$. 
\end{definition}

Since the above definition uses the  incidence model for hypergraph,  the set parameters can use hyperedges, even though the introduction ordering itself uses only vertices.  Note also that the definition uses  \mso, and not counting \mso. We will show that bounded treewidth implies definable tree decompositions of bounded width; not using counting will make the result stronger.

\begin{myexample} \label{example:independent-sets-definable-tree-decompositions} Define an \emph{independent set} to be a hypergraph that has  only vertices and no ports or hyperedges, as discussed in Example~\ref{example:isolated-vertices-introduction-ordering}. We will show that independent sets have definable tree decompositions of bounded width. There is a minor difficulty, which is that we need to avoid the path decompositions where the introduction ordering looks like this: \mypic{121} 
    The reason is that there is no single \mso formula with set parameters that can define a linear order on every   independent set, see Exercise~\ref{ex:definable-order}. The solution is to consider tree decompositions of depth two, where the introduction ordering looks like this:
    \mypic{123}
The introduction ordering for such a tree decomposition is definable by an \mso formula, which has one set parameter that describes the root.
\end{myexample}
\begin{myexample}\label{example:cycles-definable-tree-decompositions} In this example, we show that cycles have definable tree decompositions of bounded width. By a cycle, we mean a hypergraph that looks like this:
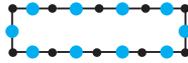
\mypic{117}
As in Example~\ref{example:independent-sets-definable-tree-decompositions}, we need careful with the choice of decomposition. Consider first  a tree decomposition that looks like this:
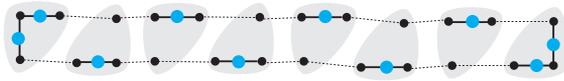
\mypic{86}
The introduction ordering for the above tree  decomposition looks like this:
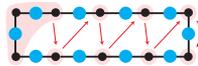
\mypic{119}
Note how the successor relation of this introduction ordering connects vertices which are far away in the cycle. For this reason, in order to define this introduction ordering, we would need an \mso with set parameters whose size would depend on the length of the cycle. 

To get definable tree decompositions for cycles, we use tree decompositions where the introduction ordering looks like this:
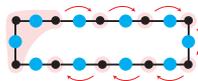
\mypic{118}
The idea behind such a tree decomposition is that that nodes of the tree decomposition correspond to a clockwise traversal of the cycle, with all bags containing the first vertex  (in the above picture, the first vertex is the bottom-left corner). The introduction ordering for this tree decomposition can be defined by an \mso formula with two set parameters, one to indicate the first vertex, and another one to indicate the direction (clockwise or not) of the traversal. 
\end{myexample}


\exercisehead

\mikexercise{\label{exercise:one-defining-formula}
    Suppose that $\varphi_1,\varphi_2$ are two \mso formulas with set parameters, with the same number of arguments. There is a single \mso formula $\varphi$ with set parameters which defines every relation definable by either $\varphi_1$ or $\varphi_2$.}
{ Let the formulas from the assumption be
    \begin{align*}
    \set{\varphi_i(Y_1,\ldots,Y_{n_i}, x_1,\ldots,x_m)}_{i \in \set{1,2}}.
    \end{align*}
     The formula from the conclusion adds an extra set parameter $Y_0$, whose emptiness determines which of the two formulas $\varphi_1$ or $\varphi_2$ should be used:
    \begin{align*}
    \myunderbrace{
        (Y_0= \emptyset \Rightarrow \varphi_1)}{if $Y_0$ is empty, then use $\varphi_1$} \quad \land \quad 
        \myunderbrace{
        (Y_0 \neq \emptyset \Rightarrow \varphi_2).}
        {otherwise use $\varphi_2$}
    \end{align*}
    Since there is no need to use separate set parameters for $\varphi_1$ and $\varphi_2$, 
    the number of parameters is one plus the maximal number of parameters in $\varphi_1$ and $\varphi_2$.
}

\mikexercise{Define \msoone to be the variant of \mso where set quantification is restricted to sets of vertices (and not hyperedges). Show that for hypergraphs of bounded treewidth, \msoone has the same expressive power as \mso. 
}{}

\mikexercise{Let $k \in \set{0,1,\ldots}$. Show that there is an \mso formula $\varphi$ with set parameters, such that for every hypergraph of treewidth at most $k$,  every unary relation (i.e.~a set of vertices and hyperedges) is definable by $\varphi$ using set parameters that contain only vertices.  }{}{}

\mikexercise{
\label{ex:definable-order}   Show that there is no \mso formula $\varphi$ with set parameters, such that  every hypergraph has  a linear order  definable by $\varphi$. Hint: consider independent sets.  }{}{}


\mikexercise{Show that a formula as in Exercise~\ref{ex:definable-order} can be found, if we want the linear order only for connected hypergraphs with  degree at most $k$ (every vertex is adjacent to at most $k$ hyperedges).}{}

\mikexercise{Show that a formula as in Exercise~\ref{ex:definable-order} can be found, if we want a  spanning forest instead of a linear order.}{}

\mikexercise{\label{ex:tree-depth} 
We say that a set of hypergraphs $L$ has \emph{bounded treedepth} if there is some $\ell$ such that every hypergraph in $L$ has a tree  decomposition of width and height at most $\ell$ (the height is the maximal depth of nodes). Without using Theorem~\ref{thm:definable-tree-decompositions}, show  that if  $L$ is recognisable and has bounded treedepth, then it is definable in counting \mso.}{}

\subsection{Bounded treewidth implies definable tree decompositions}
We are now ready to state the main result of Section~\ref{sec:definabl-tree-decompositions}.

\begin{theorem}\label{thm:definable-tree-decompositions}
    If $L \subseteq \hmonad \Sigma$   has bounded treewidth, then it has definable tree decompositions of bounded width. 
\end{theorem}

The width of the tree decompositions in the assumption and in the conclusion of the above theorem will be different. An analysis of the proof would show that if all hypergraphs in $L$  have treewidth at most $k$, then the definable  tree decompositions from the conclusion of the theorem will have width at most  doubly exponential in $k$. With more care in the proof, we could produce optimal width tree decompositions\footnote{
    Definability of  tree decompositions of optimal width is shown in
    \incite[Theorem 2.]{bojanczykOptimizingTreeDecompositions2017a}
}, but the sub-optimal width will be enough for our intended application,  which is the converse of Courcelle's Theorem that will be presented in Section~\ref{sec:application-to-recognisability}.


When defining  tree decompositions in \mso, we will not care about the labelling relation ``hyperedge $e$ has label $a \in \Sigma$''. For this reason, we will not specify the alphabet $\Sigma$ for the rest of this section. 

Here is a plan for the rest of  Section~\ref{sec:definabl-tree-decompositions}:

\begin{itemize}
    \item In Section~\ref{sec:nesting-lemma}, page~\pageref{sec:nesting-lemma}, we state and  prove the Merging Lemma, which shows how a  definable tree decomposition of definable  tree decompositions can be merged into a single definable tree decomposition.
    \item In Section~\ref{sec:bounded-pathwidth}, page~\pageref{sec:bounded-pathwidth}, we prove a  special case of Theorem~\ref{thm:definable-tree-decompositions}, which says that   hypergraphs of bounded pathwidth have  definable tree decompositions.
    \item In Section~\ref{sec:general-case-treewidth}, page~\pageref{sec:general-case-treewidth},  we complete the proof of the theorem.
    \item In Section~\ref{sec:application-to-recognisability}, page~\pageref{sec:application-to-recognisability}, we apply the theorem to get a converse of Courcelle's Theorem for hypergraph of bounded treewidth. 
\end{itemize}


Before proceeding with the proof, we define torsos.  Torsos will  be used frequently in the proof. 
\paragraph{Torsos.} Torsos are used restrict a tree decomposition to the hypergraph corresponding to some subset of the nodes.  The subsets  that we care about  are called \emph{factors}, and are  defined in the following picture (where dots indicate nodes of a tree decomposition):
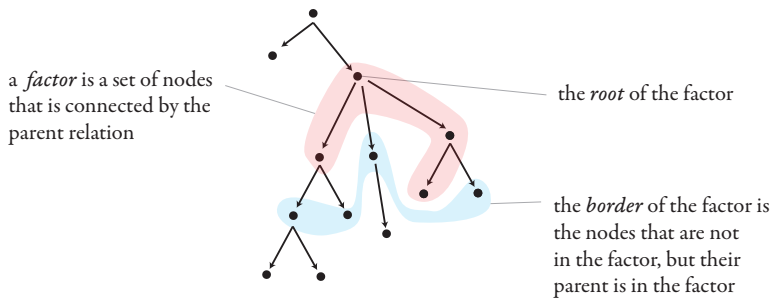
\mypic{116}
Suppose that $T$ is a tree decomposition and $X$ is a factor. We  define below a hypergraph, called the {torso of $X$ in $T$}. The torso will contain vertices and hyperedges of the underlying  hypergraph  of $T$ that appear in the factor, plus extra hyperedges corresponding to the border of the factor. Before defining torsos formally, we need to overcome one more difficulty. In the torso, we will need a linear ordering for its ports, and for  incidence lists in the  extra hyperedges (because ports are ordered, and incidence lists are also ordered). To get such linear orders, we will use a local colouring of the underlying hypergraph, as defined below.

\begin{definition}[Local colouring]
    Define a \emph{local colouring} of a tree decomposition of width $k$ to be a colouring of vertices in the underlying hypergraph with colours $\set{0,\ldots,k}$ such that in every bag, all vertices have different colours. 
\end{definition}

Every tree decomposition has a local colouring, which can be obtained in a greedy way by colouring the root bag, then colouring the bags of the children, and so on.  If a tree decomposition is equipped with a local colouring, then every bag has an implicit linear order, from the smallest colour to the biggest colour.   We also assume that the local colouring is chosen so that the implicit linear order is consistent with the ordering of the ports, i.e.~the colours of the ports are increasing.

For a node $x$ in a tree decomposition, define its \emph{adhesion} to be the vertices from the bag of $x$ that are  either ports of the underlying hypergraph, or which appear also in parent of $x$.


\begin{definition}[Torso]
    \label{def:torso} 
    Let $T$ be a tree decomposition, together with a local colouring, and let $X$ be a factor.
    The \emph{torso of $X$ in $T$}, denoted by $T/X$, is defined to be the following hypergraph. The vertices and hyperedges are those which  are introduced in nodes from $X$, plus the adhesion of the root node of the factor. 
    The ports are the adhesion factor's root, ordered according to the local colouring.  Furthermore, for  every node $x$ in the border of $X$, we add  a hyperedge (called a \emph{border hyperedge}) whose incidence list is the adhesion of $x$, ordered  according to the local colouring.
\end{definition}

If a factor is a subtree, i.e.~its border is empty, then the torso will have no border hyperedges.
For the rest of Section~\ref{sec:definabl-tree-decompositions}, we assume that every tree decomposition comes with an implicit local colouring. This way, we can simply talking about torsos in a tree decomposition, without indicating explicitly the local colouring which is needed to define the torsos.

\subsection{The Merging Lemma}
\label{sec:nesting-lemma}
In this section we state and prove the Merging Lemma, which  is based on the following simple idea. Suppose that we have an ``external''  tree decomposition, possibly of unbounded width,   where every bag has an accompanying ``internal'' tree decomposition,  of  width at most $k$.   We will show that these tree decompositions can be merged, in an \mso definable way, into a single tree decomposition of width at most $k$. 
 This lemma will be used several times in the proof, with the internal tree decompositions typically  obtained by applying some kind of induction assumption. 
 
 The internal tree decompositions are formalised using  torsos for factors with one node. Torsos for  larger factors will be used later, in Section~\ref{sec:general-case-treewidth}.

\begin{lemma}[Merging Lemma] Let $k \in \set{1,2,\ldots}$.
    Let $T$ be a tree decomposition (call it external) such that for every  every node $x$, the torso $T/\set x$ has a tree decomposition $T_x$ (call it internal) of width at most $k$.  Then:
    \begin{enumerate}
        \item\label{merging:width} The underlying hypergraph of $T$ has a tree decomposition of width at most~$k$.
        \item \label{merging:definable} Suppose that  $\psi, \varphi$ are \mso formulas with set parameters, such that  the external tree decomposition are definable by $\psi$ and all internal tree decompositions are definable  by $\varphi$. Then the tree decomposition from~\ref{merging:width}  is definable by an \mso formula with set parameters, which depends only on $\psi,\varphi$ and~$k$, and which does  not depend on the external and internal tree decompositions.
    \end{enumerate}
\end{lemma}
\begin{proof}   
    Here is a picture of the external and internal tree decompositions.  
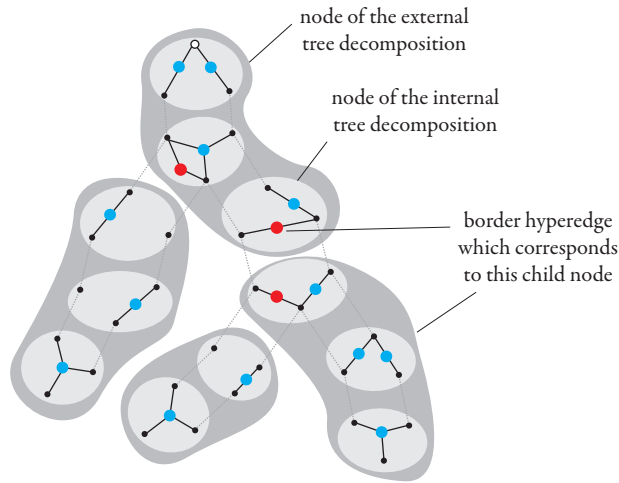
\mypic{113}\label{picture:merging-lemma}
In the picture above, red circles are used for  border hyperedges of the torsos, and blue circles are used  for the remaining hyperedges. Note that the red border hyperedges are only present in the torsos, and not  in the underlying hypergraph of the external tree decomposition.

To prove item~\ref{merging:width}, we define a \emph{merged tree decomposition} of width $k$ as follows.  First, take the disjoint union of the internal tree decompositions, which gives a forest. Next,  convert this forest into a tree, by selecting parents for roots in the following way. Suppose that  $x$ is a root  node in this forest, which corresponds to the root node of some internal tree decomposition $T_y$. If $y$ has a defined parent $z$ in the external tree decomposition, then the parent of $x$ in the merged tree decomposition is defined to be the node in the internal tree decomposition $T_z$ that introduces the border hyperedge of the torso $T/\set z$ which corresponds to node $y$. Otherwise, if $y$ is the root of  the external tree decomposition, then $x$ is the root of the merged tree decomposition. Here is a picture of the merged tree decomposition:
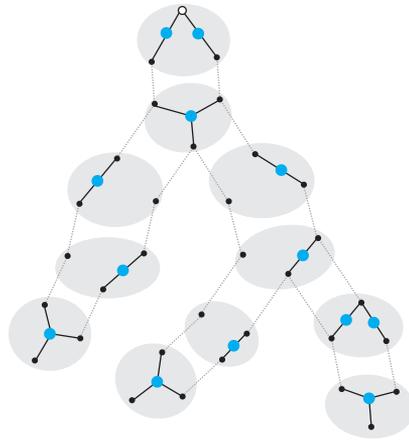
\mypic{114}

We leave it as an exercise for the reader to check that the merged tree decomposition defined this way is indeed a tree decomposition.     Because   bags of the merged tree decomposition are inherited from  bags of the internal tree decompositions, its width is at most $k$, thus proving item~\ref{merging:width}. 
    
    We now prove item~\ref{merging:definable} about definability.  Let $G$ be the underlying hypergraph of the external tree decomposition. An inspection of the definition of the merged tree decomposition shows that its introduction ordering can be defined in  $G$ using \mso formulas that refer to the following relations:  (a) the introduction ordering of the external tree decomposition; and (b) the following  ternary relations on vertices of $G$ that uniformly describes all of the introduction orderings for the  internal tree decompositions:
    \begin{eqnarray*}
        I(u,v,w)&\eqdef&
        \myoverbrace{\text{$u$ is introduced in a node $x$}}{in the external tree decomposition} \text{such that} \myoverbrace{\text{$v$ is introduced before  $w$}}{in the internal tree decomposition $T_x$}.
        \end{eqnarray*}
    The relations from (a) are definable  by the assumption of~\ref{merging:definable}. Therefore to prove definability of the  merged tree decomposition, it remains to show definability of the  ternary relation $I$ from item (b). This is done in the remainder of this proof. 
    
    Let   $\varphi$ be the \mso formula with set parameters, which defines  the internal tree decompositions. Suppose that $\varphi$  has set parameters $Y_1,\ldots,Y_\ell$.  For   every node $x$ of the external tree decomposition, there is a choice of set parameters
    \begin{align*}
    Y_{1,x},\ldots,Y_{\ell,x}  \subseteq \text{vertices and hyperedges in the torso $T/\set x$},
    \end{align*}
    such that the introduction ordering in the internal tree decomposition $T_x$ is obtained by  fixing this choice of set parameters and  evaluating the formula $\varphi$ in the torso $T/\set x$. The main step in the proof will be showing how each family  
    $\set{Y_{i,x} }_x$
    can be represented in the hypergraph $G$ using a constant size formula of \mso with set parameters. 
    
    An issue  is that the set parameters might use border hyperedges, which are not present in the hypergraph $G$.  To solve this issue, we represent border hyperedges using vertices in the following way. We say that a vertex $v$ of $G$ \emph{represents}  a border hyperedge $e$ in a torso $T/\set x$ if:   $v$ is introduced by the external tree decomposition in a node $y$ such that $y$ is the  child node of $x$ that corresponds to the border hyperedge $e$ in the torso $T/\set x$. Every vertex represents at most one border hyperedge, and every border hyperedge is represented by some vertex.
    
    
    Let $i \in \set{1,\ldots,\ell}$.
    Using the above representation, we will show in the following claim  that the set parameters used to define the internal tree decompositions can be represented in a uniform way in the hypergraph $G$. 
    \begin{claim}
        Let $i \in \set{1,\ldots,\ell}$. Each of the following  relations is definable by an \mso formula with set parameters, which depends only on $k$:
            \begin{eqnarray*}
                A_i\myoverbrace{(u,y)}{
                $u$ is a vertex and \\
                \scriptsize $y$ is a vertex or\\
                \scriptsize hyperedge of $G$
            } 
            &\quad \eqdef \quad &
             \myoverbrace{\exists x }{node of the \\
             \scriptsize external tree \\
             \scriptsize decomposition}
             \quad
             \text{$u$ is introduced in $x$ and $y \in Y_{i,x}$} \\
             B_i\myunderbrace{(u,v)}{
                $u,v$ are  \\ \scriptsize vertices \\
                \scriptsize of $G$
            } 
            &\quad \eqdef \quad & 
             \myunderbrace{\exists x }{node of the \\
             \scriptsize external tree \\
             \scriptsize decomposition}
             \quad 
             \txt{$u$ is introduced in $x$ and the border \\ hyperedge represented by $y$ is in $Y_{i,x}$}
            \end{eqnarray*}
    \end{claim}
    Before proving the claim, we show how it implies definability of $I$, and therefore finishes the proof of the Merging Lemma. 
    In order to check   $I(u,v,w)$,  we  run the formula $\varphi$ in the torso corresponding to the node that introduces vertex $u$, with calls to the $i$-th set parameter replaced by calls to its representation in terms of $A_i$ and $B_i$. Therefore, definability of  $A_i$ and $B_i$ implies definability of $I$. The remainder of the proof of the Merging Lemma is devoted to proving the claim.
\begin{proof}
    We begin with $B_i$. The key observation is that, when restricted to border hyperedges, the sets in the family $\set{Y_{i,x}}_x$ are disjoint, because  each border hyperedge belongs to exactly one torso $T/\set x$. For this reason,  $B_i$ can  be viewed as a subset of the  border hyperedges, and is thus definable. More formally,  $B_i$ can be defined in \mso using a single extra set parameter, because  $B_i(u,v)$ holds if and only if 
    \begin{align*} 
    \myunderbrace{\txt{\small $u$ is introduced in the node whose torso \\ \small  contains  the border hyperedge  represented by $v$ }}{definable in \mso using the  external tree decomposition} 
    \ \text{and} \ 
    v \in \myunderbrace{\set{v : B_i(u,v) \text{ for some $u$}}}{a single extra  set parameter}.
    \end{align*}

    We now turn to  $A_i$.  Here the difficulty is that a vertex might  belong  to several  sets from the family  $\set{Y_{i,x}}_x$. This will happen for  vertices from the adhesions in the external tree decomposition.  To solve this difficulty, we will use a colouring of the adhesions that uses at most $k+1$ colours. 

    Partition the  binary relation $A_i$  into two parts as follows:
    \begin{align*}
        \myunderbrace{C_i}{
            pairs $(u,v) \in A_i$ such that\\
            \scriptsize $v$ is in the adhesion of \\
            \scriptsize the node represented by $u$ 
        }
    \qquad \qquad \cup \qquad \qquad
    \myunderbrace{D_i}{
        pairs $(u,y) \in A_i$ such that\\
        \scriptsize $y$ is a vertex or hyperedge \\
        \scriptsize that is  introduced in \\
        \scriptsize the node represented by $u$ 
    }. 
    \end{align*}
    The part $D_i$ does not raise any problems, since every vertex or hyperedge is introduced in exactly one node of the external tree decomposition, and therefore we can use the same argument as for $B_i$ to prove that $D_i$ can be defined in \mso using a single extra set parameter. 
    
    We 
    are left with  $C_i$.     
     Choose a local colouring for the merged  tree decomposition, which maps vertices of $G$ to colours $\set{0,\ldots,k}$ so that no colour is used twice in a  bag of the merged tree decomposition. We will show:
    \begin{description}
        \item[(*)] for every $u$, all vertices that $v$ that satisfy $C_i(u,v)$ have different colours.
    \end{description}
    To see why (*) is true, suppose that in the external tree decomposition, vertex   $v$ is in the adhesion of the node $x$ that introduces $u$. It follows that  $v$ is a port  of the torso $T/\set{x}$, and therefore it must be  in the root bag of the internal tree decomposition $T_{x}$. This bag is also a bag of  the merged tree decomposition, and hence all vertices in this bag must have different colours under the local colouring, thus proving (*).  

    We now use (*) to finish the proof of the claim.
    For a colour $c \in \set{0,\ldots,k}$ and a vertex $u$, define $f_c(u)$ to be the vertex of colour $c$ in the  adhesion of the node that introduces $u$. There is at most one such vertex by (*), and hence $f_c$ can be viewed as a partial function.  The relation $C_i$ can be defined in \mso using $k+1$ extra set parameters, because $C_i(u,v)$ holds if and only 
    \begin{align*}
     \bigvee_{c \in \set{0,\ldots,k}} \qquad 
     \myunderbrace{v = f_c(u)}{can be defined in \mso \\
    \scriptsize  using the external\\
    \scriptsize tree decomposition and\\
    \scriptsize its local colouring} \ \text{ and }\   u \in \myunderbrace{\set{u : C_i(u,f_c(u)) }}{an extra set parameter\\ 
    \scriptsize for each colour $c$}.
    \end{align*}
\end{proof}
\end{proof}

\exercisehead

\mikexercise{ A cut hyperedge in a hypergraph is a hyperedge $e$ such that for some two vertices, every path connecting them must pass through $e$.  Let $L,K$ be sets of hypergraphs, such that for every $G \in K$, if all cut hyperedges are removed from $G$, then the resulting hypergraph is in $L$. Show that if $L$ has definable tree decompositions, then the same is true for $K$. }{}


\subsection{Bounded pathwidth}
\label{sec:bounded-pathwidth}
\newcommand{\pathmonoid}{\mathbb M}
\newcommand{\polgen}[1]{\langle #1 \rangle}
We now proceed to the second step in the proof of Theorem~\ref{thm:definable-tree-decompositions}. In this step,  we prove that bounded pathwidth  implies  definable tree decompositions of bounded width. In other words, we prove a weaker version of Theorem~\ref{thm:definable-tree-decompositions}, where the assumption is strengthened from bounded treewidth to bounded pathwidth (recall that pathwidth is the minimal width of a path decomposition, i.e.~a tree decomposition where all nodes are on a single root-to-leaf path).

\begin{theorem}\label{thm:special-case-of-path-width}
    If a set of hypergraphs has bounded pathwidth, then it has definable tree decompositions of bounded width. 
\end{theorem}

Note the  asymmetry in the above theorem: the assumption uses path decompositions but the conclusion uses tree decompositions. To see the reason for this asymmetry, recall the independent sets from  Example~\ref{example:isolated-vertices-introduction-ordering}. Independent sets have pathwidth zero. Path decompositions of width zero  cannot be defined in independent sets using a constant size formula of \mso, in contrast to tree decompositions of width zero. 

The rest of Section~\ref{sec:bounded-pathwidth} is devoted to proving Theorem~\ref{thm:special-case-of-path-width}. 
To find definable tree decompositions, we will view path decompositions as semigroup, and  use the Factorisation Forest Theorem. 

\paragraph*{Path decompositions as a semigroup.} Fix $k \in \set{1,2,\ldots}$  for the rest of this section. As  was the case for tree decompositions, we assume that every path decomposition comes together with a local colouring that uses colours $\set{0,\ldots,k}$.  Here is a picture of a path decomposition together with a local colouring:
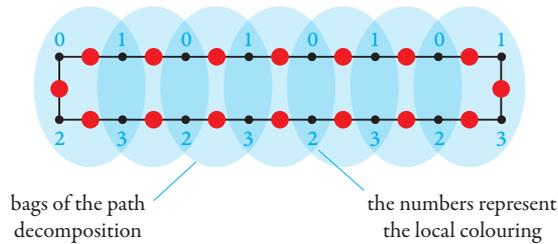
\mypic{104}
 For path decompositions defined this way, we define a semigroup product, as explained in the following picture:
\mypic{105}
It is easy to see that the operation described above  is associative, and hence  path decompositions of width at most $k$ are a semigroup.  We call this the ``semigroup of path decompositions''.

\paragraph*{The reachability   homomorphism.}  On the semigroup of path decompositions, we define a semigroup homomorphism, which stores a finite amount of information about paths.
Define an \emph{inner path} in a hypergraph to be a path of the form 
\begin{align*}
    v_0 \stackrel{e_1}\to 
    \myunderbrace{v_1 \stackrel{e_2}\to \cdots \stackrel{e_{n-1}}\to v_{n-1}}{these vertices are not ports} \stackrel{e_n}\to v_n.
\end{align*}
In other words, an inner path is a path that avoids ports, with the possible exception of the source and target of the path. Define the \emph{reachability homomorphism} to be the function which maps a  path decomposition  of width at most $k$  to the  answers to the following questions, for all $i,j \in \set{0,\ldots,k}$ and $\sigma,\tau \in \set{\text{first, last}}$:
\begin{itemize}
    \item does the $\sigma$ bag contain a vertex with local colour $i$?
    \item is there an inner path from    a  vertex with local colour $i$ in the $\sigma$ bag  to a vertex with local colour $j$ in the  $\tau$ bag?
    \item is there are vertex with local colour $i$ that is both in the first and last bag?
\end{itemize}
The  reachability homomorphism is easily seen to be   compositional in the semigroup sense, and therefore it can be viewed as a  semigroup homomorphism  from the  semigroup of path decompositions into a finite semigroup (which consists of all possible sets of answers to the questions described above).
 
 \paragraph*{Using the Factorisation Forest Theorem.} The semigroup of path decompositions is finitely generated, namely by
 \begin{align*}
 \Delta = \text{path decompositions with at most two nodes.}
 \end{align*}
 We  view the reachability homomorphism as a homomorphism
 \begin{align*}
 h : \Delta^+ \to S.
 \end{align*}
 For $\ell \in \set{0,1,\ldots}$, consider the path  decompositions (corresponding to words in) in $\Delta^+$ that have a Simon tree -- as defined  in the Factorisation Forest Theorem --  of height at most $\ell$, with respect to the homomorphism $h$. Let $L_\ell$ be the underlying hypergraphs of these path decompositions. 
  By the Factorisation Forest Theorem, there is some $\ell$ such that all hypergraphs of pathwidth at most $k$ belong to $L_\ell$. Therefore, to prove that all hypergraphs of pathwidth at most $k$ have  definable tree decompositions of bounded width,  it remains to show the following lemma. 
 \begin{lemma}\label{lem:graphs-simon-induction}
     For every $\ell \in \set{0,1,\ldots}$,  $L_\ell$ has definable tree decompositions of bounded width.
 \end{lemma}
\begin{proof}
    In the proof, we  work  with hypergraphs  where all non-ports can be connected by inner paths, as described in the following definition.  
    For a vertex or  hyperedge $x$ in a hypergraph $G$, define its  \emph{inner component} to be hypergraph obtained from $G$ by restricting to the ports, plus vertices and hyperedges that can be reached from $x$ via an inner path. Here is a picture:
    \mypic{115}
    Every hypergraph is equal to the fusion of its  inner components, where fusion is the operation on hypergraphs of same arity  that was described in Figure~\ref{fig:treewidth-terms}.
We will prove that for every  $\ell \in \set{0,1,\ldots}$,  the language
\begin{align*}
K_\ell = \set{G : \text{$G$ is an inner component of some hypergraph in $L_\ell$}}
\end{align*}
has definable tree decompositions of bounded width. This will imply the lemma, as explained in the following claim.

\begin{claim}\label{claim:reduce-to-inner-connected}
    Let $L$ be a set of hypergraphs. If 
    \begin{align*}
        K = \set{G : \text{$G$ is an inner component of some hypergraph in $L$}}
        \end{align*}
     has definable tree decompositions of bounded width, then so does $L$.
\end{claim}
\begin{proof}
    We use the Merging Lemma to get  a definable tree decomposition for every hypergraph $G \in L$. 
      Define an external tree decomposition for $G$ as follows. For every inner component we have a node, and these nodes are connected by a common root, whose bag is the ports. This external tree decomposition is clearly definable in \mso, even without set parameters.  The torso of the root node has constant size, while the torsos of the  remaining nodes have definable tree decompositions of bounded width thanks to the assumption on $L$.    Therefore, we can apply the Merging Lemma to get a definable tree decomposition for~$G$. \end{proof}

It remains to prove that, for every $\ell \in \set{0,1,\ldots}$, the set   $K_\ell$ has definable tree decompositions of bounded width. 
The proof is by induction on $\ell$. 


In the induction base of $\ell=0$, the hypergraphs  from $K_0$ are inner components of the finitely many generators from $\Delta.$ Therefore, we can use trivial definable tree decompositions where all vertices are in the same bag.

We are left with  the induction step.  Consider a hypergraph $G \in K_{\ell+1}$, which by definition is an inner component of some  hypergraph $H \in L_{\ell+1}$. Again by definition, $H$ has  a path decomposition  which can be factorised in the semigroup of path decompositions as
$P_1 \cdots P_n$,
so that:
\begin{description}
    \item[(i)]   $P_1,\ldots,P_n$ have   their underlying hypergraphs in $L_\ell$;
    \item[(ii)] either $n=2$, or  all $P_1,\ldots,P_n$ have the same value under the  reachability homomorphism, which is furthermore idempotent.
\end{description}
To define a   tree decomposition for $G$, we  use the Merging Lemma. 
 The internal decompositions will be definable thanks to  (i), and the external decomposition will be  definable thanks to  (ii).

The external decomposition $T$, which is a path decomposition\footnote{We are able to produce a path decompositions, and not just a tree decomposition, thanks to the  assumption that $G$ is an inner component. The place in the proof where tree decompositions are needed is the application of Claim\ref{claim:reduce-to-inner-connected} which splits the hypergraph into inner components.},  is defined as follows. The bags are those  $x \in \set{1,\ldots,n}$ such that the path decomposition $P_x$ contains at least one vertex of $G$, ordered in the natural way. The bag of node $x \in \set{1,\ldots,n}$  consists of those vertices of $G$ that appear in $P_x$. The following claim takes care of the internal tree decompositions.

\begin{claim} \label{claim:path-torsos} For all nodes of $T$, the corresponding torsos come from a set of hypergraphs with definable tree decompositions of bounded width. 
\end{claim}
\begin{proof}
    Let $x$ be a node of $T$. The torso $T/\set x$ is obtained as follows: take the underlying hypergraph  $H_x$  of the path decomposition $P_x$, restrict it to vertices that appear in $G$, and then add a border hyperedge. Since every inner path of $H_x$ is also an inner path of $H$, it follows that every inner component of $H_x$ is contained in some inner component of $H$. Therefore, if we define $G_x$ to be  the restriction of $H_x$ to the vertices and hyperedges of the inner component of $G$, then $G_x$ is   a union of inner components of $H_x$.   
    By assumption, every inner component of $H_x$ belongs to $K_{\ell}$, and therefore has a definable tree decomposition thanks to the induction assumption. Thus $G_x$ has a definable tree decomposition, as a fusion of hypergraphs with definable tree decompositions. Finally, adding a single border hyperedge of constant arity does not affect definable tree decompositions, since we can modify the definable tree decomposition by adding to all bags the vertices that are incident with the added border hyperedge.
\end{proof}

Thanks to the above claim and  the Merging Lemma, in order to finish the proof of the lemma, it remains to show that the introduction ordering in the  external tree decomposition $T$ is definable. 

If $n=2$ then there is not much to do: the external path decomposition has two nodes, and therefore  it is definable by an  \mso  formula with  two set parameters that say which vertices are in which bags. 
We are left with the case where $n > 2$.
The key to defining the external tree decomposition  is given in the following claim.   In the claim, the \emph{profile} of a path in the hypergraph $G$ is defined to be  the nodes of the external tree decomposition that introduce at least one vertex used by the path. 

\begin{claim}\label{claim:index-paths}
    Let $v$ and $w$ be non-port vertices of $G$ that are introduced, respectively,  in nodes $x \le y$ of the external tree decomposition $T$. If $\Pi$ is the set of inner paths in $G$ with source $v$ and target $w$,  then:
    \begin{enumerate}
        \item every inner path  in $\Pi$ has  profile that contains   $\set{x,x+1,\ldots,y-1,y}$;
        \item some inner path in $\Pi$ has profile  contained in  $\set{x-1,x,\ldots,y,y+1}$.
    \end{enumerate}
\end{claim}
\begin{proof}
    The proof  crucially depends on the assumptions that: (a)  all path decompositions  $P_1,\ldots,P_n$ have the same idempotent image under the  reachability homomorphism; and (b) the hypergraph $G$ is an inner component.
    \begin{enumerate}
        \item To prove item (1), we will show:
        \begin{description}
            \item[(*)] every non-port vertex $u$ appears in at most two bags of $T$.
        \end{description}
        Suppose that we have proved (*). Because bags containing $u$ must be consecutive by definition of path decompositions, it follows that  consecutive vertices in an inner path must be introduced in consecutive nodes of $T$. This implies that the profile of an inner path is an interval contained in $\set{1,\ldots,n}$. Such an interval must contain all numbers between $x$ and $y$, if the inner path has source introduced in $x$ and target introduced in $y$, thus proving (1).
        
        It remains to  prove (*). Toward a contradiction,  suppose that $u$ appears in at least three bags of the external tree decomposition $T$, which means that it appears in at least three nodes of the path decompositions $P_1,\ldots,P_n$. If we take  $x \in \set{1,\ldots,n}$ to be the second node where $u$ appears, then $u$  appears in the path decompositions both before and after $P_x$.  From the way that path decompositions are composed in the semigroup of path decompositions, it follows that $u$ must appear  in  both the first and the last bags of the path decomposition $P_x$. Therefore,  the reachability homomorphism gives a positive answer to the question ``is there are vertex with local colour $c$ that is both in the first and last bag?'', where $c$ is the colour of $u$ under the local colouring corresponding to the path decomposition $P$.  Since all of the path decompositions $P_1,\ldots,P_n$ have the same abstraction under the reachability homomorphisms, a positive answer to this question is given for all of these path decompositions, and therefore $u$ appears in all of them. This means that  $u$ is a port, contradicting the assumption of (*).
         \item   Consider the following partition of the path decomposition $P_1 \cdots P_n$: 
         \begin{align*}
         \myunderbrace{P_1 \cdots P_{x-1}}{left part}
         \myunderbrace{P_x \cdots P_{y}}{middle part}
         \myunderbrace{P_{x+1} \cdots P_{n}}{right part}
         \end{align*}
         By the assumption, all of the path decompositions $P_1,\ldots,P_n$ have the same image under the reachability homomorphism, and this is also the same as the image of the left, right and middle parts described above. 
         Consider now an inner path in $G$ that connects vertices $v$ and $w$, which must exist because  $G$ is a single inner component. 
           By  the previous item, this path can be split into segments of  three kinds: (a) paths in the left part that begin and end in its last bag; (b) paths in middle part; (c) paths in the right part that begin and end its first bag. 
         Each of these segments can be modified, without changing its source or target, so that its profile is contained in  $\set{x-1,x\ldots,y,y+1}$. For segments of the kind (b), there is nothing to do. For segments of kind (a), we use  the fact  that $P_{x-1}$ has the same image under the reachability homomorphisms  as the left part. A  similar argument holds for segments of kind (c).
    \end{enumerate}
\end{proof}

Using the above claim, we  define  in \mso the introduction ordering  of the external tree decomposition, thus completing the proof of the lemma.
    For a  vertex of $G$ that is introduced in node $x \in \set{1,\ldots,n}$ of the external tree decomposition,  define its \emph{colour} to be $x$ modulo 5. The colours of vertices can be represented using 5 set parameters. We begin with the following observation, which implies that the relation ``introduced in the same node of the external tree decomposition'' is definable by a constant \mso formula with set parameters:
    \begin{description}
        \item[(0)] two vertices $v$ and $w$ are introduced in the same node of the external tree decomposition if and only if they have the same colour, and they can be connected by an inner path that does not use all colours.
    \end{description}
    The right-to-left implication in (0) follows from the first item in Claim~\ref{claim:index-paths}, since vertices introduced in the same node can be connected by a path that uses at most three colours. The left-to-right implication of  (0) follows from the second item.

    A similar argument proves the following characterisation of the successor relation on nodes in the external tree decomposition:
    \begin{description}
        \item[(1)] the node introducing vertex $w$ is the successor of the node introducing $v$  if and only if the colour of $w$ is one plus the colour of $v$ (modulo 5), and there is a path connecting them that does not use all colours.
    \end{description}
    Again, the left-to-right implication uses the first item of Claim~\ref{claim:index-paths}, because vertices introduced in nodes $x$ and $x+1$ can be connected using a path with at most 4 colours. For the right-to-left implication, we use the second item of Claim~\ref{claim:index-paths},  and  the following observation, which explains the need for counting modulo 5:
    \begin{align*}
    \myunderbrace{y \equiv x +1}{modulo 5} \qquad \text{implies} \qquad y=x+1 \quad \text{or} \quad 
    \myunderbrace{y \le x-4 \quad \text{or} \quad y \ge x+6}{any connecting path must use all 5 colours}.
    \end{align*}
    The introduction ordering of the external tree decomposition is the transitive closure of the union of the two relations defined in (0) and (1), and therefore it is definable. 
\end{proof}

\exercisehead
\mikexercise{We say that a hypergraph is a tree if removing any hyperedge increases the number of connected components. Show that trees have unbounded pathwidth.}{}{}

\mikexercise{Give an algorithm which inputs a tree hypergraph, and computes its pathwidth.}{}


\subsection{Unbounded pathwidth}
\label{sec:general-case-treewidth}
In this section, we complete the proof of  Theorem~\ref{thm:definable-tree-decompositions}, using the    Merging Lemma  and the case of bounded pathwidth shown above. The idea is to show that for every tree decomposition of width at most $k$, its nodes can be partitioned into factors in a way that is depicted in the following picture:
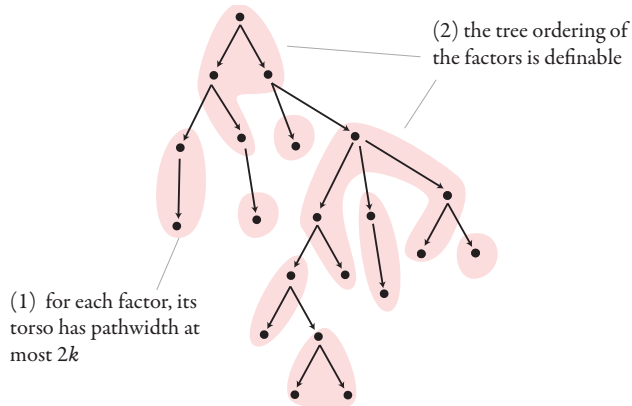
\mypic{101}
Given such a factorisation, we will apply  the Merging Lemma, with the external tree decomposition having the factors as nodes, and with the  internal tree decompositions being definable thanks to  the results on bounded pathwidth.

To define the tree ordering on the factors from item (2) in the picture, we will use   paths with bounded overlap that connect the factors. Paths with bounded overlap, and their application to \mso definability, are explained in the following lemma.

\begin{lemma}[Bounded Overlap Lemma]\label{lem:guidance-system} 
    Let $\Pp$ be a family of paths in a hypergraph $G$ such that every vertex is used by at most $\ell \in \set{1,2,\ldots}$ paths. Then the binary relation
    \begin{align*}
    \set{(s,t) : \text{$s,t$ are vertices such that some path in $\Pp$ has source $s$ and target $t$}}
    \end{align*}
    is definable by  an \mso formula with set parameters that depends only on $\ell$ and the maximal arity of hyperedges, and not on the hypergraph $G$ or the family of paths $\Pp$.
\end{lemma}
\begin{proof} Choose a family of colourings 
    \begin{align*}
    \myunderbrace{\text{vertices used by $P$} \quad
     \to   \quad \set{1,\ldots,\ell}}{one colouring for each path $P \in \Pp$},
    \end{align*}
     so that   if a vertex appears in two different paths from $\Pp$, then it  has different colours under the corresponding colourings. Such a family of colourings can easily be obtained using a greedy algorithm, thanks to the assumption that every vertex is used by at most $\ell$ paths.
      Define  $\to$ to be the  binary relation on   pairs of the form (vertex of $G$, number in $\set{1,\ldots,\ell$}), such that  
      \begin{align}\label{eq:one-step-paths}
      (v,i) \to (w,j)
      \end{align}
      holds if there is some path $P \in \Pp$ where $v$ and $w$ are consecutive vertices which have have colours $i$ and $j$  under the colouring corresponding to $P$. A vertex pair $(s,t)$ belongs to the  relation in the statement of the lemma if and only if there exist  colours $i,j \in \set{1,\ldots,\ell}$ such that:   $(s,i)$  can reach  $(t,j)$ using finitely many steps of $\to$,   $(s,i)$ has no incoming edges with respect to $\to$, and $(t,j)$ has no outgoing edges with respect to $\to$.  Therefore, to prove the  lemma, it will be enough to show that both $\to$ and its transitive closure can be defined in \mso, as described below.

      For every fixed choice of  colours $i$ and $j$, we can view~\eqref{eq:one-step-paths} as a binary relation on vertices. This relation can be described in \mso using set parameters that range over sets of hyperedges, as described below:
      \begin{align*}
      \bigvee_{n,m}\quad  \exists  e \ e \in \myunderbrace{E_{i,j,n,m}}
    {set of hyperedges $e$ such that \\
    \scriptsize for some  path $P \in \Pp$,\\ 
    \scriptsize 
    $P$ uses a step of the form \\
    \scriptsize $e[n] \stackrel e \to e[m]$\\
    \scriptsize  and  the colouring of $P$ satisfies
    \\
    \scriptsize $e[n]\mapsto i$ and $e[m] \mapsto j$. }\ 
    \land\  e[n]=v \ \land\  e[m]=w .
      \end{align*}
      In the above, $n$ and $m$ range over positions in incidence lists, and therefore these numbers are bounded by the maximal arity of hyperedges. 
      Consider now the transitive reflexive closure 
      \begin{align}\label{eq:many-step-paths}
        (v,i) \to^* (w,j)
        \end{align}
       of the relation $\to$. To define this transitive closure, we use the usual \mso formalisation of least fix-points. For every choice of colours $i$ and $j$, the relation~\eqref{eq:many-step-paths} can be defined using \mso in terms of the relations from~\eqref{eq:one-step-paths}, as follows:
      \begin{align*}
        \forall \myoverbrace{V_1,\ldots,V_\ell}{sets of \\
        \scriptsize vertices}\ \myoverbrace{ \big(v \in V_i  \land \bigwedge_{n,m} \forall x\  (x \in V_n  \land (x,n) \to (y,m)) \Rightarrow y \in V_m \big)  \Rightarrow v \in V_j}
        {if the set $\set{(x,i) : x \in V_i}$ contains $(v,i)$ and is closed under $\to$, then it contains $(w,j)$ }
    \end{align*}
      It follows that the transitive closure in~\eqref{eq:many-step-paths} can be defined using a constant formula of \mso with set parameters, for every choice of $i,j \in \set{1,\ldots,\ell}$. 
\end{proof}

\paragraph*{The Two Path Lemma.} The next step in the proof is the  Two Path Lemma, which will provide the paths with small overlap that are needed to connect factors in a tree decomposition. The Two Path Lemma says that for every tree decomposition of width at most $k$, and every choice of source and target vertices from the same  bag in the tree decomposition, one can find a factor of the tree decomposition, such that the corresponding  torso has bounded pathwidth and  contains two roughly disjoint paths connecting  the source and target.  

The Two Path Lemma  uses a mild assumption on tree decompositions, defined as follows. 
We say that a tree decomposition  is \emph{sane} if for every subtree, the  corresponding torso is inner connected, which means that it consists of a single inner component.  If a hypergraph is inner connected, then it has a sane tree decomposition of optimal width. Indeed, if we take any tree decomposition, and we find a subtree whose  torso is not inner connected, then we can distribute  vertices introduced by the subtree into separate subtrees, one for each inner  component. 

\begin{lemma}[Two Path Lemma]\label{lem:biconnected}
    Let $T$ be a sane tree decomposition of width $k$, with two distinguished vertices (call them source and target) that belong to the bag of some  node  $x$.   Then there exists a factor $X$ with root $x$ such that:
    \begin{enumerate}
        \item the torso  $T/X$ has pathwidth at most $2k$;
        \item the source and target can be connected by  two inner paths in the torso  $T/X$, such that every border hyperedge of $T/X$ is used by at most one of the paths.
    \end{enumerate}
\end{lemma}
\begin{proof}
    To create the two paths from the conclusion of the lemma, we  use  (a hypergraph version of) the Menger's Theorem about cuts and disjoint paths, as described below. For a hypergraph with distinguished  source and target vertices, define a   \emph{separating hyperedge} to be any hyperedge that must appear on every path  from the source to the target. One could also talk about separating vertices, but we only need separating hyperedges here. Here is a picture:
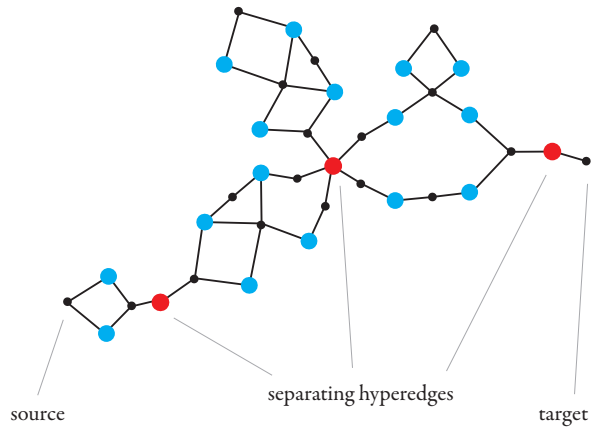
\mypic{95}
Menger's Theorem\footnote{
    The hypergraph version of Menger's Theorem that we use here can be easily inferred from the classical statement of undirected graphs, which can be found here:
    \incite[Theorem 3.3.3.]{diestel}
    For a self-contained proof, see the exercises. 
} says that if the source and target in a hypergraph can be connected by some path, then they can be connected by two paths such that every non-separating hyperedge is used by at most one of the two paths.

    To prove the lemma, we  iterate the following claim starting with $X = \set x$. In the claim, use an operation ``forget all ports'' which inputs a hypergraph, and outputs the same hypergraph, except that no vertices are ports any more. If we forget all ports in a torso, then the accompanying path decompositions do not need to have the adhesion of the torso in the first bag.
    \begin{claim}
        Let $X$ be a factor of the tree decomposition $T$ such that:
        \begin{description}
            \item[(*)] the root of $X$ is $x$ and  there is a path decomposition of 
            \begin{center}
                forget all ports in $T/X$
            \end{center}
            which has width at most $k$ and satisfies the following property:
            \begin{description}
                \item[(**)] the source and target vertices from the assumption lemma are, respectively, in the first and last bags, and every separating hyperedge is covered by some node  such that:
                \mypic{125}
            \end{description}

        \end{description}
        Then either $T/X$ has no border hyperedges  which separate the source from the target, or otherwise one can add a new node to $X$ so that it still  satisfies (*).        
    \end{claim}
    \begin{proof}
        Suppose that there is a border hyperedge  $e$ in the torso $T/X$ which separates the source from the target. This border hyperedge corresponds to some node $y$ of the tree decomposition that is in   the border of $X$. We add $y$ to the factor $X$, while preserving the invariant (*).  Take the path decomposition from the assumption that $X$ satisfies (*).  By  (**), there is a  node $z$ in this path decomposition  that covers the separating hyperedge $e$, and such that every vertex  of the torso $T/X$ appears either only to the left, or only to the right of $z$. Replace node $z$ by the torso  $T/\set y$, with the separating hyperedges put into separate bags so that (**) is still satisfied. 
    \end{proof}

    Start with $\set x$ and keep iterating the  above claim, until a factor $X$ is reached such that $T/X$ has no border hyperedge that separates the source from the target.  Take the path decomposition from the claim and, if necessary, add the ports of the torso to a prefix the path decomposition so that they are also present in the first bag (as required for path decompositions). This modification can make the width of the path decomposition go up from $k$ to $2k$. 
    
    Because the tree decomposition is sane, the torso  $T/X$ contains at least one path that connects the source with the target. Therefore, thanks to Menger's Theorem there are two paths from the source to the target  in $T/X$ which are disjoint on border hyperedges. 
    \end{proof}

\paragraph*{Partitioning a tree decomposition into factors of bounded pathwidth.} To finish the proof of Theorem~\ref{thm:definable-tree-decompositions}, we will iterate the Two Path Lemma to show that every sane  tree decomposition can be partitioned into  factors, so that each factor has a torso with bounded pathwidth, and the partition into factors  can be defined by a constant \mso formula with set parameters. The picture is the same as in the beginning of the section, but we repeat it for the reader's convenience:
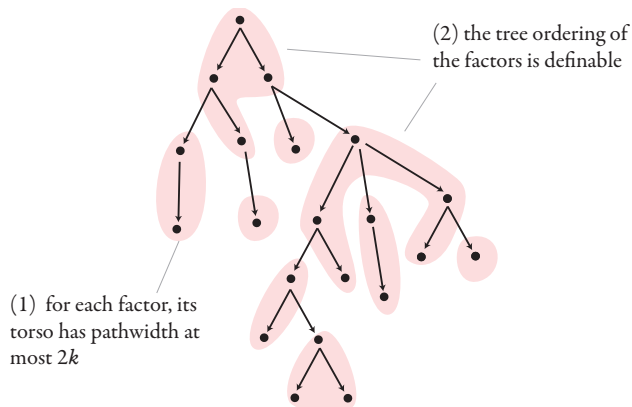
\mypic{101}
\begin{lemma}\label{lem:family-of-connected-factors}
    Let $T$ be a sane tree decomposition of width $k$. There is a family of factors $\Xx$ which partitions the nodes of $T$ such that:
    \begin{enumerate}
        \item For every factor $X \in \Xx$, the  torso $T/X$ has pathwidth at most $2k$.
        \item There is  an \mso formula with set parameters, which only depends on $k$, and which defines the following binary relation on vertices:
        \begin{align*}
        \text{factor from $\Xx$ that introduces $v$} 
        \quad \myunderbrace{\le}{tree ordering on $\Xx$ inherited from $T$}\quad
        \text{factor from $\Xx$ that introduces $w$}.
        \end{align*}
    \end{enumerate}
\end{lemma}

Before proving the lemma, we use it to complete the proof of Theorem~\ref{thm:definable-tree-decompositions}.  By Claim~\ref{claim:reduce-to-inner-connected}, it is enough to show that for every $k \in \set{0,1,\ldots}$, the inner connected hypergraphs of treewidth at most $k$ have 
definable tree decompositions of bounded width. If a hypergraph is inner connected, then it has  a sane tree decomposition of optimal width. Apply Lemma~\ref{lem:family-of-connected-factors} to this tree decomposition, yielding a family of factors $\Xx$. Define a new tree decomposition, where the nodes are the factors from $\Xx$, with the tree ordering inherited from the original tree decomposition. By item (2) of  the lemma, the new tree decomposition is  definable, and by item (1) of the lemma, each node  has a torso of pathwidth at most $2k$.  Therefore, we can apply the Merging Lemma to get a definable tree decomposition of bounded width, with the internal tree decompositions coming from Theorem~\ref{thm:special-case-of-path-width} about bounded pathwidth. 

It
remains to prove the lemma. 

\begin{proof}
    We iterate the  Two Path Lemma, to get the family of factors that satisfy item (1), together with a family of paths with bounded overlap that will be used to ensure item (2).  One  step of the iteration is described in the following claim. 
\begin{claim}
    Suppose that $\Xx$ is a family of disjoint factors in $T$  which satisfies:
    \begin{description}
        \item[(*)] the union $\bigcup \Xx$ is a prefix  of the tree decomposition, i.e.~it is closed under ancestors,  and:
        \begin{enumerate}
            \item for every factor $X \in \Xx$, the  torso $T/X$ has pathwidth at most $2k$;
            \item there is a   family of paths $\Pp$ in the torso $T/(\bigcup \Xx)$ such that:
            \begin{enumerate}
                \item \label{fact-family:stars} for every factor $X \in \Xx$  there is some vertex $u$ introduced in $X$  such that for every vertex $v$ in the root bag of $X$ there is a path from $u$ to $v$ in $\Pp$;
                \item \label{fact-family:disjoint}  every vertex  of $T/(\bigcup \Xx)$ is used by at most $2k^3+k$ paths from $\Pp$;
                \item \label{fact-family:border-disjoint} every border hyperedge of $T/(\bigcup \Xx)$ is used by at most $2k^3$ paths from~$\Pp$.
            \end{enumerate}
        \end{enumerate}
    \end{description}
    Then either $\bigcup \Xx$ is  all nodes of the tree decomposition, or otherwise one can add a new factor to $
    \Xx$ so that the resulting family  still satisfies (*).
\end{claim}
\begin{proof}
    Suppose that the prefix $\bigcup \Xx$ is not all nodes in the tree decomposition. Choose some minimal node $x$ outside this prefix; this node corresponds to a border hyperedge $e$ in the  torso $T/(\bigcup \Xx)$. By definition of torsos, the vertices incident to this border hyperedge are the adhesion of $x$. The adhesion has size at most $k$, since bags have size at most $k+1$ and the adhesion of a non-root node is a proper subset of its bag (otherwise a node would have the same bag as its parent). For two  vertices $v,w$  from the adhesion of $x$, we say that a path $P \in \Pp$ has \emph{profile} $(v,w)$ if it contains an infix of the form
    \begin{align*}
    v \stackrel e \to w.
    \end{align*}
    There  are less than $k^2$ choices for $v,w$. In general a path might have several profiles if it uses the hyperedge $e$ several times. However, by eliminating loops, we can assume without affecting the invariant (*)  that every path from $\Pp$ uses each hyperedge at most once, and therefore every path has at most one profile.
    Choose  a pair of vertices $(s,t)$ in the adhesion of $y$ so that the number of paths with this profile is maximal; let $\ell$ be the number of such paths.  Apply the Two Path Lemma with $s$ being the source and $t$ being the target, yielding  a factor $X$ in $T$ with root node $x$. To prove the claim, we will show that the invariant (*) is still satisfied after adding $X$ to the family $\Xx$. 
    
    Clearly item (1) in the invariant is satisfied, because $T/X$ has pathwidth at most $2k$ thanks to  the Two Path Lemma. It remains to find a family of paths, call it $\Rr$, which  witnesses item (2) of the invariant. This family is constructed as follows:
    \begin{enumerate}[label=(\alph*)]
        \item \label{group:star} Choose a vertex $u$ that is introduced in $x$. For every vertex in the  bag of $x$, choose  an inner path in the torso $T/X$ which goes from that vertex to $u$, and add this path to $\Rr$. Such a path must exist because the  tree decomposition $T$  is sane. The paths added in this step  ensure that item (i) of the invariant will hold. In the remaining steps, we need to take care of the  border hyperedge $e$ which is used by the paths from $\Pp$. This hyperedge will not be allowed in $\Rr$ since it is removed from the torso once we add the factor $X$ to $\Xx$.
        \item \label{group:vw} Consider a path in $\Pp$ which either does not use
         the border  hyperedge $e$, or uses it with a profile  $(v,w)$ that is different from $(s,t)$.  If $e$ is used by this path, then replace the corresponding segment
        \begin{align*}
        v \stackrel e \to w
        \end{align*}
        with some inner path in $T/X$ that goes from $v$ to $w$ (which exists because the tree decomposition is sane). Then, add the resulting path to $\Rr$.
        \item \label{group:v0w0} Consider a path in $\Pp$ which  uses the hyperedge $e$ with  profile  $(s,t)$. We do a similar procedure as in the previous item, except that we use two inner paths instead of one. 
        Let $P_1$ and $P_2$ be the two inner paths in the torso $T/X$ from the conclusion of the Two Path Lemma. Partition the $\ell$ paths in $\Pp$ which have profile $(s,t)$ into two groups, of sizes $\lceil \ell/2 \rceil$ and $\lfloor \ell/2 \rfloor$. For every path in the first group, replace the segment
        \begin{align*}
        s \stackrel e \to t
        \end{align*}
        with the inner path $P_1$, and for every path in the second group replace this segment with the inner path $P_2$. Add the resulting paths to $\Rr$. By construction, if we look at the paths added in this step, then every border hyperedge of $T/X$  is visited by paths from at most one of the two groups, and therefore by  at most $\lceil \ell /2 \rceil$ paths.
    \end{enumerate}
    We now argue that the family $\Xx \cup \set X$ satisfies item (2) of the invariant, with respect to the family $\Rr$ defined above. 
    \begin{enumerate}[(i)]
        \item This part of the  invariant is satisfied  thanks to the paths  that are described  in item~\ref{group:star} of the definition of~$\Rr$.    
        \item This part of the invariant says that every vertex  from the torso of  $\bigcup (\Xx \cup \set X)$ is used by at most $2k^3+k$ paths from $\Rr$. Consider first a vertex $v$ that appears already in $\bigcup \Xx$. Because, in the definition of $\Rr$, we used inner paths from $T/X$ to replace traversals of the border hyperedge $e$, it follows that  
        \begin{align*}
            \text{number of paths in $\Rr$ that uses $v$} = 
        \myunderbrace{\text{number of paths in $\Pp$ that uses $v$}}{at most $2k^3 + k$ by  assumption}.
        \end{align*}
        Consider now a vertex that is introduced in $X$. This vertex is used by at most 
        \begin{align*}
           \myunderbrace{2k^3}{paths from $\Pp$ \\
           \scriptsize that used $e$} \quad + \quad 
           \myunderbrace{k}{paths added  \\ \scriptsize in item \ref{group:star}}
           \end{align*}
           paths and therefore the invariant is also satisfied. 
   
        \item This part of the invariant says that every border hyperedge  in the torso of $\bigcup (\Xx \cup \set X)$ is used by at most $2k^3$ paths from $\Rr$. If this hyperedge is not a border hyperedge of $T/X$, then we use the same argument as for vertices. Otherwise, the number of paths in $\Rr$ that use this border hyperedge is at most 
       \begin{align*}
           \myunderbrace{\text{number of paths  in $\Pp$ that do not have profile $(s,t)$}}{paths added  \\ \scriptsize in item \ref{group:vw}} \ \ + \ \ 
           \myunderbrace{\lceil \ell/2 \rceil}{paths added  \\ \scriptsize in item \ref{group:v0w0}}\ \ + \ \ 
           \myunderbrace{k.}{paths added  \\ \scriptsize in item \ref{group:star}}
           \end{align*}
           The above sum is at most $2k^3$, which is argued by considering two cases. If $\ell < 2k$, then the summands for items~\ref{group:vw} and~\ref{group:v0w0} are at most $k^2 (2k-1)$, and adding $k$ will not exceed the threshold $2k^3$. Otherwise, if $\ell \ge 2k$, then 
       \begin{align*}
           \lceil \ell /2 \rceil + k \le \ell
       \end{align*}
       and therefore the sum of all items at most the number of paths in $\Pp$ which uses the border hyperedge $e$.
    \end{enumerate}
\end{proof}

Iterate the above claim, starting with the empty family, until reaching a family of factors $\Xx$ which satisfies the invariant (*), and  which partitions all nodes of the tree decomposition. We will show that $\Xx$ satisfies the conclusion of the lemma. Item (1) of the lemma, about bounded pathwidth of the torsos,  follows immediately from the corresponding item in the  claim. We are left with proving  item (2) from the lemma,  about definability of  the binary relation 
\begin{align*}
    \text{factor from $\Xx$ that introduces $v$} 
    \quad \myunderbrace{\le}{tree ordering on $\Xx$ inherited from $T$}\quad
    \text{factor from $\Xx$ that introduces $w$}.
    \end{align*}
Here we will use the family of paths  $\Pp$ from item (2) in the invariant (*).  By item~\ref{fact-family:stars} of the invariant, for every factor $X \in \Xx$  we can choose a vertex $v_X$ that is introduced in $X$, such that using paths from $\Pp$ one can connect $v_X$ with every vertex in the root bag of $X$.
By the bounded overlap condition for vertices from  from item~\ref{fact-family:disjoint} of the invariant, we can use the Bounded Overlap Lemma to show that the binary  relation
\begin{align}
\label{eq:adhesion-definable}
\set{(v_X,w) : \text{$X \in \Xx$ and $w$ is in the adhesion of the  root node of $X$}}
\end{align}
can be defined by an \mso formula with set parameters that depends only on $k$. (When applying the Bounded Overlap Lemma, we observe that in a hypergraph with treewidth at most $k$, all hyperedges have arity at most $k+1$.)
Let $X$ be a factor in $\Xx$, with root node $x$. By definition of sane tree decompositions, the torso of the subtree of $x$ in $T$ is inner connected, which means that every two vertices from this torso that are not ports (i.e.~they are not in the  adhesion of $x$) can be connected by an inner path. Since $v_X$ is not in the adhesion of $x$, it follows that a vertex $w$ is introduced in the subtree of  $x$ if and only if there is a path from $w$ to $v_X$ which does not use vertices from the adhesion of $x$. Since this can be formalised in \mso using the relation~\eqref{eq:adhesion-definable}, we can define in \mso the following binary relation:
\begin{align*}
\set{(v_X,w) : \text{$X \in  \Xx$ and $w$ is introduced in root node of $X$ or its descendants}}.
\end{align*}
Using the above relation, we can easily define in \mso the binary relation from item (2) in the statement of the lemma, thus finishing the proof. 
\end{proof}

\exercisehead

\mikexercise{
\label{exercise:mengers-theorem}
Prove Menger's Theorem, in the hypergraph version that is used in the proof of the Two Path Lemma.}{We use  Menger's Theorem about separators and disjoint paths in undirected graphs. Menger's Theorem says that if, in an undirected graph with designated source and target vertices, one cannot find a set of $k$ separating vertices, then there are $k$ paths from the source to target which are disjoint (apart from the source and target). The lemma follows immediately by applying Menger's Theorem, in the case of $k=1$, to the Gaifman graph of a hypergraph, which is an undirected graph that is defined as in the following example:
\mypic{106}
}

\subsection{Application to recognisability}
\label{sec:application-to-recognisability}
A corollary of  Theorem~\ref{thm:definable-tree-decompositions}  is that the converse of Courcelle's Theorem holds for hypergraphs of bounded treewidth, as stated below.

\begin{corollary}\label{cor:courcelle-reverse}
    Let   $L \subseteq \hmonad \Sigma$  be a set of hypergraphs of  bounded treewidth. Then $L$ is recognisable if and only if it is definable in counting \mso. 
\end{corollary}

\begin{proof}   The right-to-left implication, even without the assumption on  bounded treewidth, is Courcelle's Theorem. 
    
    The left-to-right implication is proved using the definable tree decompositions from Theorem~\ref{thm:definable-tree-decompositions}.  The  formula defining the language $L$ guesses the  tree decomposition,  by existentially  quantifying over the set parameters needed to define its introduction order adjacent bag relation, and then does a bottom-up pass through the decomposition to compute the value of the hypergraph with respect to the recognising homomorphism. Suppose that $L \subseteq \hmonad \Sigma$ has bounded treewidth and is recognised by a homomorphism
    \begin{align*}
    h : \hmonad \Sigma \to A
    \end{align*}
    into a hypergraph algebra that is finite on every arity. We need to show that $L$ is definable in counting \mso. 
    By Theorem~\ref{thm:definable-tree-decompositions} there is some $k \in \set{0,1,\ldots}$ and  an \mso formula $\varphi$ 
    with set parameters, such that every hypergraph in $L$ has a tree decomposition  $T$ of width at most $k$ that can  be defined by $\varphi$. 

    The formula of counting \mso that defines the language $L$ works as follows.
    Let $G \in L$ and let $T$ be a tree tree decomposition which can be defined by some choice of set parameters in the formula $\varphi$. For a node $x$ of this tree decomposition, define  $G_x$ to be the torso of the subtree of $x$, and define the \emph{type} of $x$ to  be  the image of $G_x$ under the recognising homomorphism $h$. 
    The formula defining $L$  uses existential set quantification to guess the following sets:
    \begin{enumerate}[(a)]
        \item set parameters for the formula $\varphi$ which define the introduction order of $T$;
        \item set parameters for the formula $\varphi$ which define the bag relation of $T$;
        \item $k+1$ sets which represent a local colouring;
        \item sets which represent the types of the nodes in the  tree decomposition.
    \end{enumerate}
    The types from item (d) are  represented by storing the  type  of node  $x$ in every vertex $v$ of the underlying hypergraph that is introduced in node $x$. 
    Note that the arity of a type is at most  $k+1$, since each torso has at most $k+1$ ports, and therefore there are finitely many possibilities for the types, which ensures  that they can be  represented a bounded number of  sets (number of elements in the algebra that have arity at most $k+1$).  Next, the formula checks that the sets guessed in (a, b, c) indeed describe a tree decomposition of width at most $k$ together with a local colouring; this can be easily done by expressing the appropriate definitions in \mso. The rest of this proof is devoted to explaining how a formula of counting \mso can check that the types  from item (d) were guessed correctly. Once we know this, membership in the language boils down to checking that the type of the  root node in the tree decomposition  belongs to the accepting set for the homomorphism.

    To check that the  types  from item (d)  are guessed correctly, we examine  every node $x$ of the tree decomposition, and check if its guessed type is consistent with the guessed types of its children.  This is done as follows. 
    For a node $x$ in the tree decomposition, decompose the torso  $G_x$ into the following parts:
    \begin{itemize}
        \item For a hyperedge $e$ that is introduced in  $x$, define $E_e$ to be the hypergraph obtained from $G_x$ by restricting it so that (a)  the  vertices  are the bag of $x$; (b) the only hyperedge is $e$; (c) all vertices are ports.
        \item  For a child $y$ node of $x$, define $H_y$ to be the hypergraph obtained from $G_x$ by  restricting it so that: (a) the vertices are the bag of $x$ plus vertices introduced in $y$ and its  descendants of $y$; (b) the hyperedges are those that are introduced in $y$ and its descendants; (c) the ports are the bag of $x$.
    \end{itemize}
    Recall the fusion, forget and rearrangement operations from Figure~\ref{fig:treewidth-terms}.
    Apart from the ports, the parts $E_e$ and $H_y$ defined above  are disjoint, and together they represent all vertices and hyperedges in the torso $G_x$. 
    Therefore,  $G_x$ can be recovered from these hypergraphs by fusing  these  and then forgetting the vertices from the bag of $x$ which are not in the adhesion of the subtree of $x$, as expressed in the following equality:
    \begin{align*}
        G_x  \qquad = \qquad 
        \myoverbrace{f \big(
        \myunderbrace{\sum_{y} H_y}{fusion ranging\\
        \scriptsize over children  of $x$} 
        \quad 
        \myunderbrace{+}{fusion}
        \quad 
        \myunderbrace{\sum_{e} E_e}{fusion ranging\\
        \scriptsize over hyperedges \\
        \scriptsize introduced in $x$}\big)}{the operation $f$ restricts the ports\\
        \scriptsize to the adhesion of the subtree of $x$
        }.
        \end{align*}
    Although fusion is in principle a binary operation, it is associative and commutative when the number of ports is fixed, and therefore unordered sums in the above expression are meaningful. Both fusion and $f$ are special cases of term operations in hypergraph algebra, which means that they commute with the homomorphism $h$.

    The  hypergraph $H_y$ is obtained from $G_y$ by forgetting some ports (namely the ports that are not in the bag of $x$), and then adding some new vertices as ports (namely the vertices that are in the bag of $x$ but not in the bag of $y$). This corresponds to a term operation, call it $f_y$, which  transforms $H_y$ into $G_y$ is a term operation. Summing up, we have the following equation for every node $x$ in the tree decomposition: 
        \begin{align*}
        G_x = f(\sum_y \underbrace{f_y(G_y)}_{H_y} + \sum_e E_e).
        \end{align*}
        The above equation is in the free algebra $\hmonad \Sigma$.
        Since $h$ is a homomorphism, and homomorphisms commute with term operations such as fusion and $f_y$, we have the following equation in the algebra $A$:
        \begin{align}\label{eq:tree-decomposition-recursive-equality}
        \myunderbrace{h(G_x)}{type of $x$}
         = f(\sum_y f_y(\myunderbrace{h(G_y)}{type of $y$}) + \sum_e h(H_e)),
        \end{align}
        where both fusion and the term operations $f$ and $f_y$ are now interpreted in the algebra $A$.  If the guessed types satisfy the equality in~\eqref{eq:tree-decomposition-recursive-equality} for every node $x$, then they are equal to the actual types.  
        Therefore, it remains to show that a formula of counting \mso can check if the guessed types satisfy the equality~\eqref{eq:tree-decomposition-recursive-equality} for every node $x$.

        Define $A^x$ and $A_x$ to be the elements of the algebra $A$ whose arity is equal to sizes of, respectively, the adhesion and the bag of $x$. The term operations in~\eqref{eq:tree-decomposition-recursive-equality} have the following types:
        \begin{align*}
        \myunderbrace{A_x \to A^x}{$f$}
        \qquad
        \myunderbrace{A_x \times A_x \to A_x}{$+$}
        \qquad 
        \myunderbrace{A^y \to A_x}{$f_y$}.
        \end{align*}
        Since fusion $+$ is  associative and commutative, it follows that the fusion from~\eqref{eq:tree-decomposition-recursive-equality}  corresponds to multiplication in a  finite commutative semigroup, and such multiplication can be computed in counting \mso. The result follows, since the term operation $f_y$ can be determined in \mso based on the node $y$ (given by a vertex introduced in it), and the same is true for $h(H_e)$. 
\end{proof}

\exercisehead

\mikexercise{
    We say that a hypergraph algebra $A$ is \emph{aperiodic} if for every $n \in \set{0,1,\ldots}$ the semigroup
    \begin{align*}
        (\text{elements of $A$ with arity $n$}, \text{fusion})
        \end{align*}
    is aperiodic. Show that if $L$ is a set of hypergraphs of bounded treewidth, then $L$ is definable in \mso without counting if and only if it is recognised by a finite aperiodic hypergraph algebra.
}{}

\mikexercise{Suppose that $L$ has bounded treedepth, as described in Exercise~\ref{ex:tree-depth}. Show that if $L$ is recognised by an aperiodic hypergraph algebra, then it is definable  in first-order logic.
}{}

 
 \backmatter
 \addtocontents{toc} {\vspace{\baselineskip}}
 
 \renewcommand{\refname}{Bibliography}
\printbibliography
 \cleardoublepage





 \printindex{authors}{Author index}
 \printindex{subject}{Subject index}
 \end{document}